\documentclass[aps, prx, reprint, amsmath,amssymb,showpacs,floatfix,longbibliography, twocolumn, superscriptaddress,footinbib]{revtex4-2}
\PassOptionsToPackage{dvipsnames,table,xcdraw}{xcolor}
\usepackage{cmap}
\usepackage{simplewick}
\usepackage{algorithmic}
\usepackage{multirow} 
\usepackage{mathrsfs}
\usepackage{oplotsymbl}

\usepackage{graphicx}
\graphicspath{ {./figures/} }
\usepackage{caption}
\usepackage{dcolumn}
\usepackage{bm}
\usepackage{amssymb}
\usepackage{tabularx}
\usepackage{mathtools,halloweenmath}
\usepackage{physics}
\usepackage{graphicx}
\usepackage{amsmath}
\usepackage[version=4]{mhchem}
\usepackage{siunitx}
\usepackage{tabularx}

\usepackage{braket}
\usepackage{wick}
\usepackage{comment}
\usepackage{hyperref} 
\hypersetup{
    colorlinks,
    linkcolor={red!50!black},
    citecolor={blue!50!black},
    urlcolor={blue!80!black}
}
\usepackage{multirow}
\usepackage{lipsum}
\usepackage{array}
\usepackage{microtype}
\usepackage{appendix}
\usepackage{amsmath}
\usepackage{comment}
\usepackage{amsthm}
\usepackage{tikz}
\usepackage{tikz-cd}

\usepackage{enumitem}
\usepackage{eczoo}
\usepackage{footmisc}
\usetikzlibrary{quantikz}

\usetikzlibrary{external}

\usepackage{amscd} 
\usepackage{empheq}
\usepackage{amsthm}
\usepackage{bbm} 
\usepackage{bold-extra}
\usepackage{booktabs} 
\usepackage{circledsteps} 
\usepackage{scalerel,stackengine,relsize}
\usepackage[justification=Justified,singlelinecheck=false]{caption}

\newcommand{\Z}{{\mathbb{Z}}}
\newcommand{\R}{{\mathbb{R}}}
\newcommand{\C}{{\mathbb{C}}}
\newcommand{\T}{{\mathbb{T}}}
\newcommand{\N}{\mathbb{N}}

\stackMath
\providecommand\reallywidehat[1]{%
  \savestack{\tmpbox}{\stretchto{\scaleto{%
    \scalerel*[\widthof{\ensuremath{#1}}]{\kern-.6pt\bigwedge\kern-.6pt}%
      {\rule[-\textheight/2]{1ex}{\textheight}}}{\textheight}}{0.5ex}}%
  \stackon[1pt]{#1}{\tmpbox}%
}

\providecommand{\U}{\mathrm{U}}

\providecommand{\bn}{\noindent\begin{nummer}\rmfamily}
\providecommand{\en}{\end{nummer}}

\usepackage[color=yellow]{todonotes} 

\usepackage[ruled,lined]{algorithm2e}

\usepackage{mathtools}

\global\long\def\zl{\overline{Z}}%
\global\long\def\xl{\overline{X}}%
\definecolor{greenvva}{RGB}{0,128,0}%
\newcommand{\cv}{\textsc{cv}}
\newcommand{\dia}[1]{{\boldsymbol{#1}}}
\newcommand{\rowlab}[1]{\makebox[7.2em][l]{\small #1}}
\newcommand{\rowpairbrace}[1]{\smash{\lower.5\normalbaselineskip\hbox{%
 $\left.\vphantom{\begin{array}{c}1\\1\end{array}}\right\}\ \text{#1}$}}}
\newcommand{\rowquadbrace}[1]{\smash{\lower1.5\normalbaselineskip\hbox{%
 $\left.\vphantom{\begin{array}{c}1\\1\\1\\1\end{array}}\right\}\
 \text{#1}$}}}

\newtheorem{theorem}{Theorem}
\newtheorem{lemma}[theorem]{Lemma}
\newtheorem{proposition}[theorem]{Proposition}

\newtheorem{corollary}[theorem]{Corollary}

\newtheorem{definition}{Definition}

\theoremstyle{definition}\theoremstyle{plain}

\newtheorem{remark}{Remark}

\DeclareUnicodeCharacter{202F}{\,}

\makeatletter
\@ifpackagelater{array}{2024/06/01}{\let\switch@array\relax}{}
\makeatother

\begin{document}

\title{Stabilizer codes over general phase spaces}
\author{Sayan Chakraborty}
\affiliation{Institute for Advancing Intelligence, TCG CREST, Sector V, Salt Lake, Kolkata 700091, India}
\author{Victor V. Albert}
\affiliation{Joint Center for Quantum Information and Computer Science,
NIST/University of Maryland, College Park, Maryland 20742, USA}

\begin{abstract}
We develop a theory of stabilizer codes whose stabilizer groups consist of commuting qudit Paulis, oscillator displacements, and/or planar-rotor displacements.
We consider codes whose stabilizer groups form generalized lattices in quantum phase space, a condition that guarantees a finite logical dimension. 
We construct oscillator-rotor, rotor-qudit, and oscillator-rotor-qubit codes that cannot be decomposed into separate subsystems by generalized Clifford transformations.
We also introduce an oscillator-qubit code arising from a partially lifted Golay code, where Construction A is applied to half of the Golay code coordinates.
We
derive analytical forms of logical operators from the symplectic dual lattice,
determine Clifford-Gaussian gates from lattice symmetries, 
show that the code dimension equals the covolume of the stabilizer lattice, and 
organize the syndrome
subspaces into a vector bundle.
Our technique builds on earlier non-commutative geometric results by Rieffel, but it can often be interpreted simply as 
concatenating a given stabilizer code with a Gottesman-Kitaev-Preskill (GKP) code, applying conventional lattice-theoretic results, and unconcatenating. 
We formulate distances, optimal decoding, and quantum weight enumerators, recovering the Pauli, qudit, and GKP versions together with known MacWilliams identities.
We construct finite-energy codewords whose encoding is an
isometry up to an error exponentially small in the inverse damping
parameter.
\end{abstract}

\maketitle

\onecolumngrid
\raggedbottom

\begingroup
\makeatletter
\setcounter{tocdepth}{2}
\setlength{\parindent}{0pt}
\newlength{\glcatocnumwd}
\settowidth{\glcatocnumwd}{\textbf{VIII.}}
\newlength{\glcatocindent}
\setlength{\glcatocindent}{\dimexpr\glcatocnumwd+0.6em\relax}
\global\let\glcatocsubs\@empty
\newif\ifglcatocmain
\glcatocmaintrue
\newif\ifglcatocsection
\glcatocsectionfalse
\newif\ifglcatocexamples
\glcatocexamplesfalse
\newif\ifglcatocchildren
\glcatocchildrenfalse
\def\glcatocflush{%
  \ifx\glcatocsubs\@empty
  \else
    \nobreak\noindent\hspace*{\glcatocindent}%
    \begin{minipage}{\dimexpr\linewidth-\glcatocindent\relax}%
      \small\let\numberline\@gobble\glcatocsubs
    \end{minipage}\par\vspace{2.5ex}%
    \global\let\glcatocsubs\@empty
  \fi}
\def\glcatocsectionentry\numberline#1#2\@nil#3#4{%
  \glcatocflush
  \glcatocsectionfalse
  \glcatocexamplesfalse
  \ifglcatocmain
    \if\relax\detokenize{#1}\relax
    \else
      \glcatocsectiontrue
      \def\glcatoccurrenttitle{#2}%
      \def\glcatocexamplestitle{Examples}%
      \ifx\glcatoccurrenttitle\glcatocexamplestitle
        \glcatocexamplestrue
      \fi
      \noindent\makebox[\glcatocnumwd][r]{\textbf{#1.}}\hspace{0.6em}%
      \textbf{\hyperlink{#4}{#2}}\dotfill\textbf{\hyperlink{#4}{#3}}%
      \par\nobreak\vspace{1.1ex}%
    \fi
  \fi}
\def\glcatocsubsectionentry#1#2#3{%
  \ifglcatocmain\ifglcatocsection
    \ifglcatocexamples
      \glcatocchildrenfalse
      \ifx\glcatocsubs\@empty\else
        \g@addto@macro\glcatocsubs{\par\vspace{0.65ex}}%
      \fi
      \g@addto@macro\glcatocsubs{\noindent\small\upshape}%
    \else
      \ifx\glcatocsubs\@empty\else
        \g@addto@macro\glcatocsubs{\,\,$\bullet$\,\,}%
      \fi
    \fi
    \g@addto@macro\glcatocsubs{\hyperlink{#3}{#1}}%
  \fi\fi}
\def\glcatocsubsubsectionentry#1#2#3{%
  \ifglcatocmain\ifglcatocexamples
    \ifglcatocchildren
      \g@addto@macro\glcatocsubs{\,\,$\bullet$\,\,}%
    \else
      \glcatocchildrentrue
      \g@addto@macro\glcatocsubs{%
        \par\vspace{0.2ex}\noindent\footnotesize\itshape
        \hangindent=1.8em\hangafter=1\hspace*{1.8em}}%
    \fi
    \g@addto@macro\glcatocsubs{\hyperlink{#3}{#1}}%
  \fi\fi}
\def\contentsline#1#2#3#4{%
  \csname glcatocentry@#1\endcsname{#2}{#3}{#4}}
\def\glcatocentry@section#1#2#3{%
  \glcatocsectionentry#1\@nil{#2}{#3}}
\let\glcatocentry@subsection\glcatocsubsectionentry
\let\glcatocentry@subsubsection\glcatocsubsubsectionentry
\let\glcatocentry@paragraph\@gobblethree
\let\glcatocentry@subparagraph\@gobblethree
\let\glcatocentry@title\@gobblethree
\let\glcatocentry@abstract\@gobblethree
\let\tocdepth@munge\relax
\let\tocdepth@restore\relax
\def\appendix{\glcatocflush\glcatocmainfalse}
\noindent\rule{\linewidth}{0.8pt}\par\vspace{0.4ex}
{\large\textbf{Contents}}\par\vspace{1.0ex}
\@starttoc{toc}%
\glcatocflush
\noindent\rule{\linewidth}{0.8pt}\par
\makeatother
\endgroup
\newpage

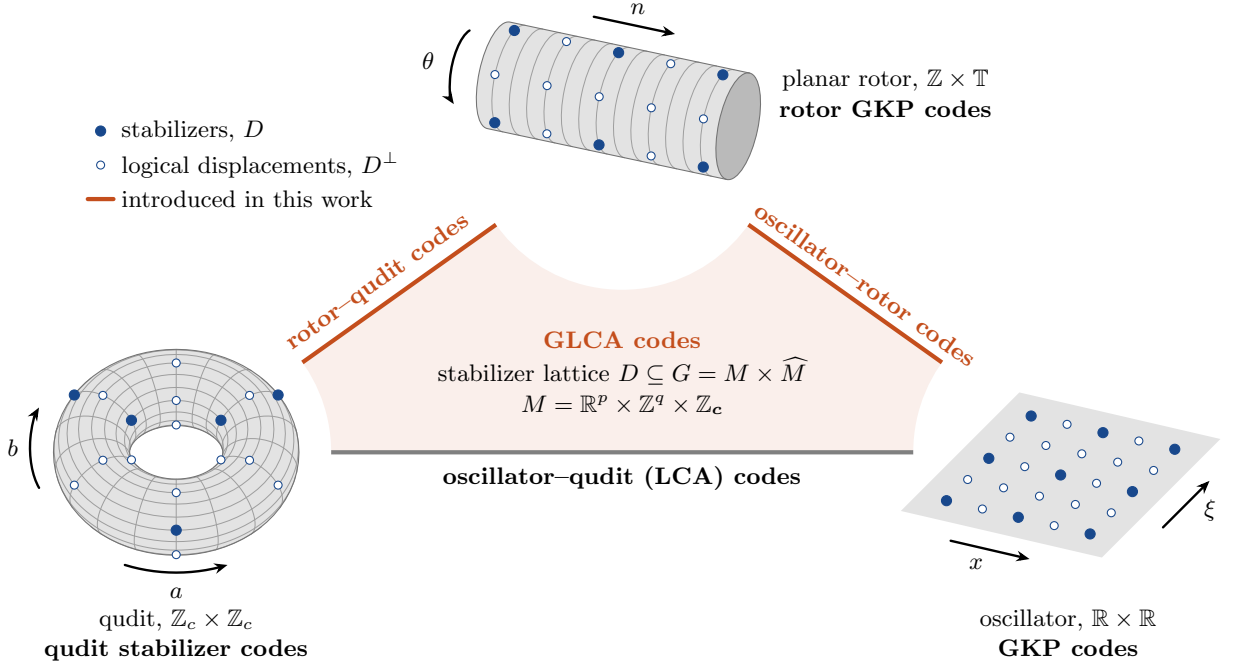
\begin{figure}[t]
\centering
\usetikzlibrary{arrows.meta,calc}
\definecolor{hyStab}{RGB}{24,72,140}   
\definecolor{hyNew}{RGB}{196,78,28}    
\begin{tikzpicture}[
  line cap=round, line join=round,
  surf/.style={fill=black!11},
  rim/.style={draw=black!55, line width=0.4pt},
  mesh/.style={draw=black!38, line width=0.3pt},
  stab/.style={circle, fill=hyStab, draw=hyStab, line width=0.4pt,
    inner sep=0pt, minimum size=3.8pt},
  logi/.style={circle, fill=white, draw=hyStab, line width=0.4pt,
    inner sep=0pt, minimum size=3pt},
  ax/.style={-{Stealth[length=5pt,width=4pt]}, line width=0.9pt},
  axlab/.style={font=\small, inner sep=1.5pt},
  known/.style={draw=black!50, line width=1.6pt},
  new/.style={draw=hyNew, line width=1.6pt},
  lab/.style={font=\small, align=center, inner sep=2pt},
  declare function={
    torX(\u,\v) = (\torR + \torr*cos(\v))*cos(\u);
    torY(\u,\v) = (\torR + \torr*cos(\v))*sin(\u)*sin(\torE)
                  + \torr*sin(\v)*cos(\torE);
    torF(\u,\v) = sin(\v)*sin(\torE) - cos(\v)*sin(\u)*cos(\torE);
    torS(\u)    = atan2(sin(\u)*cos(\torE), sin(\torE));
  },
]
\def\torR{1}\def\torr{0.45}\def\torE{50}

\coordinate (qud) at (-5.9,0);
\coordinate (rot) at (0,4.2);
\coordinate (osc) at (5.9,0);

\fill[hyNew!9] (qud) -- (rot) -- (osc) -- cycle;
\draw[known] (qud) -- (osc);
\draw[new] (qud) -- (rot) -- (osc);

\foreach \tn in {qud,rot,osc}{\fill[white] (\tn) circle[radius=2.05];}

\begin{scope}[shift={(qud)}, scale=1.12]
  \fill[surf, even odd rule]
    plot[domain=0:360, samples=73, variable=\tu]
      ({torX(\tu,torS(\tu))},{torY(\tu,torS(\tu))}) -- cycle
    plot[domain=0:360, samples=73, variable=\tu]
      ({torX(\tu,torS(\tu)+180)},{torY(\tu,torS(\tu)+180)}) -- cycle;
  \foreach \ti in {0,...,17}{
    \pgfmathsetmacro\tu{20*\ti+10}
    \pgfmathsetmacro\tsil{torS(\tu)}
    \draw[mesh] plot[domain=\tsil:\tsil+180, samples=25, variable=\tv]
      ({torX(\tu,\tv)},{torY(\tu,\tv)});
  }
  \foreach \ti in {0,...,17}{
    \pgfmathsetmacro\tv{20*\ti+10}
    \pgfmathsetmacro\tp{cos(\tv)*cos(\torE)}
    \pgfmathsetmacro\tp{ifthenelse(abs(\tp)<0.001,0.001,\tp)}
    \pgfmathsetmacro\tas{asin(max(-1,min(1,sin(\tv)*sin(\torE)/\tp)))}
    \pgfmathsetmacro\tua{ifthenelse(\tp>0,180-\tas,\tas)}
    \pgfmathsetmacro\tub{ifthenelse(\tp>0,360+\tas,180-\tas)}
    \pgfmathtruncatemacro\tok{ifthenelse(\tub-\tua>1,1,0)}
    \ifnum\tok=1
      \draw[mesh] plot[domain=\tua:\tub, samples=61, variable=\tu]
        ({torX(\tu,\tv)},{torY(\tu,\tv)});
    \fi
  }
  \draw[rim]
    plot[domain=0:360, samples=73, variable=\tu]
      ({torX(\tu,torS(\tu))},{torY(\tu,torS(\tu))}) -- cycle
    plot[domain=0:360, samples=73, variable=\tu]
      ({torX(\tu,torS(\tu)+180)},{torY(\tu,torS(\tu)+180)}) -- cycle;
  \foreach \ta in {0,...,17}{
    \foreach \tb in {0,...,17}{
      \pgfmathsetmacro\tu{20*\ta+10}
      \pgfmathsetmacro\tv{20*\tb+10}
      \pgfmathtruncatemacro\tok{ifthenelse(torF(\tu,\tv)>0.1,1,0)}
      \pgfmathtruncatemacro\tlog{ifthenelse(mod(\ta,3)==1 && mod(\tb,3)==1,1,0)}
      \pgfmathtruncatemacro\tstb{ifthenelse(mod(\ta,6)==1 && mod(\tb,6)==1,1,0)}
      \ifnum\tok=1
        \ifnum\tstb=1
          \node[stab] at ({torX(\tu,\tv)},{torY(\tu,\tv)}) {};
        \else\ifnum\tlog=1
          \node[logi] at ({torX(\tu,\tv)},{torY(\tu,\tv)}) {};
        \fi\fi
      \fi
    }
  }
  \draw[ax] plot[domain=248:292, samples=15, variable=\tu]
    ({1.62*cos(\tu)},{1.62*sin(\tu)*sin(\torE)-0.18});
  \node[axlab] at (0,-1.66) {$a$};
  \draw[ax] (-1.62,-0.42) arc[start angle=230, end angle=130,
    x radius=0.3, y radius=0.62];
  \node[axlab] at (-1.93,0.05) {$b$};
\end{scope}

\begin{scope}[shift={($(rot)+(-0.05,0.45)$)}, rotate=-12]
  \def\cylR{0.72}  
  \def\cylE{0.26}  
  \def\cylS{0.235} 
  \def\cylL{1.62}  
  \fill[surf] (-\cylL,\cylR) -- (\cylL,\cylR)
    arc[start angle=90, end angle=-90, x radius=\cylE, y radius=\cylR]
    -- (-\cylL,-\cylR)
    arc[start angle=270, end angle=90, x radius=\cylE, y radius=\cylR]
    -- cycle;
  \foreach \ti in {-6,...,6}{
    \draw[mesh] (\ti*\cylS,\cylR)
      arc[start angle=90, end angle=270, x radius=\cylE, y radius=\cylR];
  }
  \fill[black!27] (\cylL,0) ellipse[x radius=\cylE, y radius=\cylR];
  \draw[rim] (\cylL,0) ellipse[x radius=\cylE, y radius=\cylR];
  \draw[rim] (\cylL,\cylR) -- (-\cylL,\cylR)
    arc[start angle=90, end angle=270, x radius=\cylE, y radius=\cylR]
    -- (\cylL,-\cylR);
  \foreach \ti in {-6,-3,0,3,6}{
    \foreach \tb in {-1,0,1}{
      \pgfmathtruncatemacro\tstb{ifthenelse(mod(\ti,6)==0 && \tb!=0,1,0)}
      \coordinate (pt) at ({\ti*\cylS-\cylE*cos(60*\tb)},{\cylR*sin(60*\tb)});
      \ifnum\tstb=1 \node[stab] at (pt) {};
      \else \node[logi] at (pt) {}; \fi
    }
  }
  \draw[ax] (-0.55,1.02) -- (0.55,1.02);
  \node[axlab] at (0,1.25) {$n$};
  \draw[ax] (-\cylL-0.5,0.5)
    arc[start angle=110, end angle=250, x radius=0.2, y radius=0.53];
  \node[axlab] at (-\cylL-0.95,0) {$\theta$};
\end{scope}

\begin{scope}[shift={($(osc)+(-0.1,-0.3)$)}, x={(0.95cm,-0.22cm)}, y={(0.56cm,0.56cm)}]
  \fill[surf] (-1.4,-1.4) -- (1.4,-1.4) -- (1.4,1.4) -- (-1.4,1.4) -- cycle;
  \foreach \ta in {-2,...,2}{
    \foreach \tb in {-2,...,2}{
      \pgfmathtruncatemacro\tstb{ifthenelse(mod(\ta,2)==0 && mod(\tb,2)==0,1,0)}
      \ifnum\tstb=1 \node[stab] at (0.5*\ta,0.5*\tb) {};
      \else \node[logi] at (0.5*\ta,0.5*\tb) {}; \fi
    }
  }
  \draw[ax] (-0.5,-1.75) -- (0.6,-1.75);
  \node[axlab] at (0.05,-2.08) {$x$};
  \draw[ax] (1.75,-0.55) -- (1.75,0.55);
  \node[axlab] at (2.08,0) {$\xi$};
\end{scope}

\node[lab, anchor=north] at ($(qud)+(0,-2.0)$)
  {qudit, $\Z_c\times\Z_c$\\[1pt]\textbf{qudit stabilizer codes}};
\node[lab, anchor=north] at ($(osc)+(0,-2.0)$)
  {oscillator, $\R\times\R$\\[1pt]\textbf{GKP codes}};
\node[lab, anchor=west] at ($(rot)+(2.0,0.55)$)
  {planar rotor, $\Z\times\T$\\[1pt]\textbf{rotor GKP codes}};

\path (qud) -- (rot) node[lab, pos=0.48, sloped, above=2pt]
  {\textbf{\color{hyNew}rotor--qudit codes}};
\path (rot) -- (osc) node[lab, pos=0.52, sloped, above=2pt]
  {\textbf{\color{hyNew}oscillator--rotor codes}};
\path (qud) -- (osc) node[lab, midway, below=3pt]
  {\textbf{oscillator--qudit (LCA) codes}};

\node[lab] at (0,1.05)
  {\textbf{\color{hyNew}GLCA codes}\\[2pt]
   stabilizer lattice $D\subseteq G=M\times\widehat M$\\[1pt]
   $M=\R^{p}\times\Z^{q}\times\Z_{\dia c}$};

\node[stab] (lg1) at (-6.9,4.25) {};
\node[lab, anchor=west] at ($(lg1)+(0.2,0)$) {stabilizers, $D$};
\node[logi] (lg2) at (-6.9,3.8) {};
\node[lab, anchor=west] at ($(lg2)+(0.2,0)$) {logical displacements, $D^\perp$};
\draw[new] (-7.07,3.35) -- (-6.73,3.35);
\node[lab, anchor=west] at (-6.7,3.35) {introduced in this work};
\end{tikzpicture}
\caption{Phase spaces of the three system types and the stabilizer codes
that they support.
Each phase space pairs a shift label with a phase label: $(a,b)$ for a
qudit of dimension $c$, $(n,\theta)$ for a planar rotor, and $(x,\xi)$ for
an oscillator~\cite{albert2017general}.
Mesh lines mark the discrete labels.
Filled dots form a stabilizer lattice $D$, and open dots complete its
symplectic dual $D^\perp$, which labels the logical displacements.
Each lattice shown encodes a qubit.
Vertices of the triangle represent codes on one type of system
(Sec.~\ref{subsec:examples_single}).
Edges represent codes whose stabilizer lattice couples two types
(Sec.~\ref{subsec:examples_two}), and the interior represents codes on
all three, such as the codes of Secs.~\ref{ex:torsion}
and~\ref{ex:hybrid-intrinsic}.
A stabilizer of a hybrid code can displace several phase spaces at once.
GLCA codes cover the entire triangle, with any numbers of oscillators,
rotors, and qudits.}
\label{fig:hybrid}
\end{figure}

\section{Introduction \& summary}\label{Sec:Intro}

Stabilizer codes store quantum information in the joint $+1$ eigenspace
of a commuting group of displacement operators, a.k.a.\ generalized Pauli
operators.
For qubits and qudits, the stabilizer formalism reduces to symplectic
linear algebra over finite
rings~\cite{gottesman1997stabilizer,calderbank1998quantum,10.1007/3-540-49208-9_27,nadella2012stabilizer}.
For oscillators, the Gottesman--Kitaev--Preskill (GKP) codes use
stabilizer groups labeled by lattices of position and momentum
displacements~\cite{10.1103/physreva.64.012310,Conrad2022gottesmankitaev,10.1103/prxquantum.3.010335}.
Planar rotors admit analogous codes, which already appear in the original
GKP construction~\cite{10.1103/physreva.64.012310} and have since been
developed further~\cite{rotorclifford,vuillotHomologicalQuantumRotor2024}.
Hybrid processors combine several of these system types, for example
oscillators with qubits~\cite{liu2024hybrid} or qudits~\cite{yr64-ypxt}
and qubits with rotors~\cite{luo2026hybridqubitrotor}.

All of these systems share one structure.
The configuration space of each system is a locally compact abelian (LCA)
group $M$, and its phase space is
$G=M\times\widehat M$~\cite{albert2017general,bermejo2016normalizer}.
Here $\widehat M$ is the Pontryagin dual of $M$, i.e., the group of its
continuous characters.
For example, $M=\mathbb Z_2^n$ describes $n$ qubits, $M=\mathbb R^p$
describes $p$ oscillators, and $M=\mathbb Z^q$ describes $q$ planar rotors
in their angular-momentum basis.
Displacement operators $\mathbf D(x)$, labeled by $x\in G$, commute up to
phases given by a skew bicharacter $\rho$ on $G$.
A stabilizer group consists of commuting displacements, so its labels form
an isotropic subgroup $D\subseteq G$.
We focus on the case in which $D$ is a lattice, i.e., a discrete subgroup
with compact quotient $G/D$.
Every such code has a finite-dimensional logical space, whose dimension we
compute below.
Discrete stabilizer groups without compact quotient can have
infinite-dimensional logical spaces, and we do not consider them here.

Code theory for these systems has largely developed one system type at a
time.
Qubit and qudit codes build on classical coding theory through the
Calderbank--Shor--Steane (CSS)
construction~\cite{PhysRevA.54.1098,PhysRevLett.77.793} and its symplectic
generalizations~\cite{css,calderbank1998quantum,gottesman1997stabilizer}.
GKP codes have been analyzed through their lattices and
symmetries~\cite{Conrad2022gottesmankitaev,10.48550/arxiv.2407.03270,fiber-bundle-ft}
and, more recently, through time-frequency
analysis~\cite{luef2026timefrequency} and complex abelian
varieties~\cite{mayrand2026complex}.
Rotor codes have been formulated
homologically~\cite{vuillotHomologicalQuantumRotor2024} and equipped with
Clifford operations~\cite{rotorclifford}.
Quantum weight enumerators likewise come in separate
qubit~\cite{shor1997macwilliams}, qudit~\cite{rains1998weight}, and
continuous-variable~\cite{burchards2025macwilliams} versions.

In previous work~\cite{yr64-ypxt}, we introduced LCA codes on the
oscillator--qudit configuration group $M=\mathbb R^p\times\mathbb Z_{\dia c}$,
where $\mathbb Z_{\dia c}=\mathbb Z_{c_1}\times\cdots\times\mathbb Z_{c_k}$.
Their stabilizer lattices are built from Rieffel's embedding
maps~\cite{rieffel,li2004strong}.
Each such lattice projects isomorphically onto a full-rank lattice in the
oscillator phase space $\mathbb R^{2p}$.
LCA codes therefore require at least one oscillator, and their stabilizer
lattices are torsion-free.
In this setting, the stabilizer and logical algebras are a commutative torus
and a dual noncommutative torus, and Rieffel's Morita equivalence relates
the two~\cite{yr64-ypxt}.
The oscillator requirement excludes many familiar codes.
Qubit and qudit stabilizer codes lie outside the LCA class, as do
planar-rotor GKP codes and all codes whose stabilizer lattices have torsion.

In this work, we develop the theory of stabilizer codes for all LCA
configuration groups of the form
$M=\mathbb R^p\times\mathbb Z^q\times\mathbb Z_{\dia c}$.
Here $p,q\geq0$, and $\mathbb Z_{\dia c}$ is the product of finite cyclic
groups defined above.
Every finite abelian group is isomorphic to such a product.
Oscillators, rotors, and qudits may therefore appear in any numbers,
including zero.
The stabilizer group is labeled by an arbitrary isotropic lattice
$D\subseteq G$, possibly with torsion.
We call the resulting codes generalized LCA (GLCA) codes.
They include qubit and qudit stabilizer codes, GKP codes, the LCA codes of
Ref.~\cite{yr64-ypxt}, and planar-rotor GKP codes.
They also include hybrid codes on all three system types and two new
families, namely oscillator--rotor and rotor--qudit codes.
Figure~\ref{fig:hybrid} organizes these families by the types of systems
that their stabilizer lattices couple.
Our main tool is Rieffel's module~\cite{rieffel}, a Morita-equivalence
bimodule between the twisted group algebras of $D$ and of its symplectic
dual $D^\perp$.
Rieffel's construction applies to every lattice in $G$, so it treats all
GLCA codes uniformly.
Table~\ref{tab:qec_ncg_dictionary} lists the resulting dictionary between
quantum error correction and noncommutative geometry.
It extends the oscillator--qudit dictionary of Ref.~\cite{yr64-ypxt} to all
GLCA codes.

\begin{table}[t]
\caption{Dictionary between stabilizer codes with a finite-dimensional
logical space and Rieffel's noncommutative-geometry framework.
The finite group
$\mathbb Z_{\dia c}=\mathbb Z_{c_1}\times\cdots\times\mathbb Z_{c_k}$
is labeled by positive integers $\dia c=(c_1,\ldots,c_k)$, and $\mu$ is
the self-Plancherel Haar measure on $G$.
When $D$ and $D^\perp$ are torsion-free, as for GKP and LCA codes, the
stabilizer and logical algebras are a commutative torus and a dual
noncommutative torus.
Torsion in either lattice adds a finite factor.
For Pauli codes ($p=q=0$), $\widehat D$ is finite and
$\mathscr E_D$ is a trivial bundle.
Rieffel's module $\Omega_D\cong\Gamma(\mathscr E_D)$ consists of the
continuous sections of the syndrome bundle.}
\label{tab:qec_ncg_dictionary}
\centering
\renewcommand{\arraystretch}{1.15}
\begin{tabular}{cc}
\toprule
Quantum error correction & Noncommutative geometry \\
\midrule
configuration space & $M=\mathbb R^p\times\mathbb Z^q\times\mathbb Z_{\dia c}$ \\
phase space & self-dual LCA group $G=M\times\widehat M$ \\
displacement operators & projective representation of $G$ with cocycle $\beta$ \\
commutation phase & skew bicharacter $\rho$ \\
\midrule
stabilizer group & isotropic lattice $D\subseteq D^\perp$ \\
consistent stabilizer phases & coboundary $c$ that trivializes $\beta|_{D\times D}$ \\
syndrome space & Pontryagin dual $\widehat D\cong G/D^\perp$ \\
stabilizer algebra & commutative algebra $\mathrm C^*(D)\cong C(\widehat D)$ \\
syndrome & character $\chi\in\widehat D$ \\
\midrule
logical Pauli group & symplectic dual lattice $D^\perp$ \\
logical Paulis modulo stabilizers and phases & finite group $D^\perp/D$ of order $K^2$ with bicharacter $\rho_D$ \\
logical algebra on the codespace & $\operatorname{End}(E_{D,1})\cong M_K(\mathbb C)$ \\
logical algebra over all syndromes & twisted group algebra $\mathrm C^*(D^\perp,\overline\beta)$ \\
stabilizer-to-logical mapping & Morita equivalence via Rieffel's module $\Omega_D$ \\
Clifford gates from lattice symmetries & symplectic automorphism group $\operatorname{Aut}_{\mathrm{sp}}(D)$ of $D$ \\
\midrule
physical wavefunctions & square-integrable sections $L^2(\widehat D,\mathscr E_D)$ \\
codespace & fiber $E_{D,1}$ at the trivial character \\
error/syndrome space & fiber $E_{D,\chi}$ of the syndrome bundle $\mathscr E_D\to\widehat D$ \\
syndrome decomposition & Bloch--Floquet direct integral over $\widehat D$ \\
stabilizer-group average & Rieffel's $C(\widehat D)$-valued inner product at $\chi=1$ \\
code dimension $K$ & trace pairing $\tau_D([\Omega_D])$ $=$ covolume $|G/D|_\mu$ \\
\bottomrule
\end{tabular}
\end{table}

\subsection{Summary of results}
\label{subsec:summary_of_results}
We summarize our results in the order in which they appear.

\paragraph{Examples.}
Section~\ref{sec:examples} collects our examples, old and new, and
describes each with the same three tools: displacement operators, a
stabilizer generator matrix $T$, and a logical matrix.
The columns of $T$ form a basis of all lifted labels of the stabilizers,
including the periods of rotor angles and qudit labels.
With the lifted symplectic form $J$, the integer matrix
$\varTheta=T^{\intercal}JT$ records the commutation phases of these
columns.
The code dimension is then $K=|\operatorname{Pf}\varTheta|$ for every such
system and every stabilizer lattice, with or without torsion
(Lemma~\ref{lem:full_lift_pfaffian}).
The columns of the logical matrix $T\varTheta^{-1}$ generate the logical
displacements.
Single-factor examples include qubit and qudit stabilizer codes such as
the toric code, GKP codes, and rotor GKP codes.
They also include an intrinsically non-CSS two-rotor code, which no
generalized Clifford operator maps to a CSS code, and homological rotor
codes with finite homology.
Two-factor examples include the oscillator--qudit LCA codes of
Ref.~\cite{yr64-ypxt}, a new hybrid Golay code, and two new families,
which couple a rotor to an oscillator and to a qudit, respectively.
We build the hybrid Golay code from partially lifted binary and
quaternary Golay codes, which we also introduce.
The resulting code couples twelve oscillators intrinsically to twelve qubits.
Its stabilizer lattice has torsion, so the code lies outside the LCA
construction.
A minimal oscillator--rotor code encodes a ququart in a lattice that no
symplectic automorphism of $G$ splits into separate oscillator and rotor
lattices.
We prove that every one-oscillator--one-rotor lattice code of square-free
logical dimension does split, so dimension four is the first possible
instance of intrinsic oscillator--rotor coupling.
In particular, every ideal lattice stabilizer state of this system is
symplectically equivalent to a product state.
In contrast, a new rotor--qudit code couples a rotor intrinsically to a
qudit of dimension $d$ already at $K=d$.
A related rotor--qudit stabilizer state is not mapped to any product
state by any generalized Clifford operator.
Examples with several factor types include mixed-qudit codes with unequal
or composite local dimensions, which we relate to concatenated GKP codes.
A torsion-extended code places a stabilizer lattice with torsion on an
oscillator, a rotor, and a qubit.
This lattice couples all three systems in its defining coordinates, but a
generalized Clifford operator maps the code to a product of a GKP code and
two stabilizer states.
In contrast, a torsion-free CSS lattice on the same three systems encodes a
ququart that no generalized Clifford operator decomposes along any
bipartition (Proposition~\ref{prop:hybrid-intrinsic-nondecoupling}).

\paragraph{Syndrome bundle and code dimension.}
We first identify the syndromes of a GLCA code with the characters $\chi$
of $D$.
These characters form the compact group $\widehat D\cong G/D^\perp$
(Sec.~\ref{sec:LSC_ncg}).
Rieffel's module is isomorphic to the module of continuous sections of a
Hermitian vector bundle over $\widehat D$, which we call the syndrome
bundle.
Its fiber over $\chi$ consists of the generalized states with syndrome
$\chi$ (Proposition~\ref{prop:rieffel_fiber_covariance}).
A Bloch--Floquet transform decomposes $L^2(M)$ into a direct integral of
these fibers (Theorem~\ref{thm:rieffel_bundle_decomposition}).
The codespace is the fiber at the trivial character
(Lemma~\ref{lem:trivial_fiber_codespace}).
Its elements may be non-normalizable tempered distributions, as are ideal
GKP codewords.
Displacements map fibers unitarily onto one another, so every fiber has
the dimension $K$ of the codespace (Corollary~\ref{cor:vector_bundle}).
We then express this dimension as a trace pairing and as a covolume
(Theorem~\ref{thm:code_dimension}),
\[
 K=\tau_D([\Omega_D])=|G/D|_\mu,
\]
where $\mu$ is the self-Plancherel Haar measure on $G$.
For example, $n$ qubits with $l$ independent stabilizer generators give
covolume $2^{n-l}$.
For the toric code on a closed surface of genus
$g$~\cite{kitaev2003fault,dennis2002topological}, the covolume reproduces
the homological count $K=2^{2g}$
(Sec.~\ref{ex:qudit}).
The same theorem shows that $D^\perp/D$ has order $K^2$.
Encoding each rotor and each qudit into an oscillator turns every GLCA
code into a concatenated GKP code with the same symplectic Gram matrix
$\varTheta$ (Remark~\ref{rem:full_lift_gkp_concatenation}).
In particular, $K=|\operatorname{Pf}\varTheta|$ is the usual GKP dimension
formula for this code.
This correspondence identifies ideal codespaces, stabilizer groups, and
logical displacement groups, but not error sets or Clifford groups.

\paragraph{Logical operators and Clifford gates.}
Displacements labeled by $D^\perp$ preserve every syndrome fiber and
therefore act on the codespace
(Sec.~\ref{sec:logical_displacement_operators}).
We show that the logical phase space $\mathsf L_D=D^\perp/D$ carries a
nondegenerate alternating bicharacter
(Theorem~\ref{thm:logical_displacement_algebra}).
The corresponding logical displacements form an orthogonal basis of the
full logical operator algebra $M_K(\mathbb C)$.
Next, every symplectic automorphism of $G$ that preserves $D$ induces a
symplectic automorphism of $\mathsf L_D$
(Proposition~\ref{prop:lattice_symmetry_logical_phase_space}).
Each such automorphism of $G$ is implemented by a generalized Clifford,
or metaplectic, operator (Proposition~\ref{prop:symplectic_implementer}).
After a displacement correction, this operator preserves the codespace
and acts as a logical Clifford gate
(Theorem~\ref{thm:corrected_clifford_implementer}).
This extends the GKP construction of Burchards, Flammia, and
Conrad~\cite{fiber-bundle-ft} to all GLCA codes.


\paragraph{Distance and decoding.}
Given a displacement cost $w$ on $G$, we define the code distance as the
least cost of a nontrivial logical displacement
(Sec.~\ref{sec:distance_error_correction}).
A Borel set of displacement errors is correctable by a syndrome decoder
exactly when no difference of two of its errors is a nontrivial logical
displacement (Theorem~\ref{thm:distance_correctable_decoder}).
For subadditive Borel costs, minimum-cost decoders correct every error
whose cost is below half the distance
(Theorem~\ref{thm:minimum_cost_half_distance}).
For a given error distribution, the optimal decoders are exactly the
maximum-likelihood decoders of error classes modulo $D$
(Theorem~\ref{thm:maximum_likelihood_decoder}).
Euclidean, symplectic Hamming, and Lee costs recover the GKP, Pauli, and
qudit Lee~\cite{guo_lee_codes} distances, respectively.

\paragraph{Weight enumerators.}
We define enumerators for pairs of operators as functions on the phase
space $G$ (Sec.~\ref{sec:weight_enumerators}).
We show that the two enumerators of each pair are symplectic Fourier
transforms of each other.
Pushing them forward along a displacement cost gives weight measures,
which organize error detection by cost.
This single duality recovers the classical MacWilliams
identity~\cite{macwilliams1977theory} as well as the quantum identities of
Shor and Laflamme~\cite{shor1997macwilliams}, Rains~\cite{rains1998weight},
and Burchards~\cite{burchards2025macwilliams}.
At the level of lattices, Poisson summation on $G$ relates the enumerators
of $D$ and $D^\perp$.
For GKP codes, this relation becomes a theta-series duality that recovers
the distance construction of Conrad et
al.~\cite{Conrad2022gottesmankitaev}.

\paragraph{Finite-energy codewords.}
When $M$ has a continuous or an infinite discrete factor, ideal codewords
can lie outside $L^2(M)$.
We regularize them by damping oscillator and rotor energies with Gaussian
operators (Sec.~\ref{sec:finite_energy}).
A single scalar, determined by $K$, the torsion of $D$, and a partition
function, renormalizes the entire regularized code basis
(Theorem~\ref{thm:fe_asymptotic_isometry}).
The resulting encoding is an isometry up to an error that is
exponentially small in the inverse damping parameter.
For oscillators alone, this recovers the asymptotic isometry of Mayrand
and Royer~\cite{mayrand2026complex}.
For planar rotors, the damping coincides, up to conventions, with the
regularization of Xu, Wang, and Albert~\cite{rotorclifford}.

\paragraph{Outline.}
Section~\ref{sec:examples} presents our examples with minimal notation.
Section~\ref{sec:general_structure} develops the general structure of GLCA
codes.
It collects our notation (Sec.~\ref{sec:notation}), reviews Pauli
stabilizer codes in the language of projective modules
(Sec.~\ref{sec:PSC_basic}), and defines GLCA codes (Sec.~\ref{sec:GSC}).
It then constructs the syndrome bundle (Sec.~\ref{sec:LSC_ncg}), computes
the code dimension (Sec.~\ref{subsec:code_dimension_covolume}), and derives
logical displacement operators (Sec.~\ref{sec:logical_displacement_operators})
and Clifford gates from lattice symmetries (Sec.~\ref{sec:clifford_gates}).
Sections~\ref{sec:distance_error_correction}--\ref{sec:finite_energy} treat
distance and decoding, weight enumerators, and finite-energy codewords,
returning to the examples along the way.
The appendices contain proofs and technical details, the graph-lattice
construction of LCA codes (Appendix~\ref{app:graph_lattices}), and
background on
$\mathrm C^*$-algebras, projective modules, and Morita equivalence.

\section{Examples}
\label{sec:examples}

This section collects stabilizer codes with a finite-dimensional logical
space, both known and new, and treats all of them with one small set of
tools.
Section~\ref{subsec:examples_notation} introduces these tools:
displacement operators, the stabilizer generator matrix, and the logical
matrix.
We then sort the examples by the types of systems on which they act.
Sections~\ref{subsec:examples_single}, \ref{subsec:examples_two},
and~\ref{subsec:examples_several} treat codes on one, two, and several
types of systems, respectively.
The single-type codes live on qudits, oscillators, or rotors and include
an intrinsically non-CSS rotor code.
The two-type codes include a hybrid Golay code and two new families.
The several-type codes include an intrinsically coupled
oscillator--rotor--qubit code.
The examples rely mainly on the notation of
Sec.~\ref{subsec:examples_notation}.
Section~\ref{sec:general_structure} and the appendices prove the general
statements quoted here.
Later sections return to the examples to illustrate Clifford gates
(Sec.~\ref{subsec:clifford_examples}), distances
(Sec.~\ref{sec:distance_error_correction}), weight enumerators
(Sec.~\ref{sec:weight_enumerators}), and finite-energy codewords
(Sec.~\ref{subsec:fe_examples}).

\subsection{Displacement operators and stabilizer lattices}
\label{subsec:examples_notation}
\label{subsec:phase-space-displacement-conventions}

\textit{Systems and displacements.}
We consider $p$ oscillators, $q$ planar rotors, and $k$ qudits of
dimensions $c_1,\ldots,c_k$.
We write $\dia c$ both for the vector $(c_1,\ldots,c_k)$ and for the
diagonal matrix $\operatorname{diag}(c_1,\ldots,c_k)$.
The configuration space of the combined system is
\begin{equation}
 M=\R^p\times\Z^q\times\Z_{\dia c},
 \qquad
 \Z_{\dia c}:=\Z_{c_1}\times\cdots\times\Z_{c_k}.
 \label{eq:examples-configuration-space}
\end{equation}
Thus an oscillator has a real position, a rotor has an integer angular
momentum, and a qudit has a label in $\Z_c$.
States of the physical Hilbert space $L^2(M)$ are expanded in the
generalized position kets listed below.
Each type of system carries shift operators $\mathbf T$ and phase
operators $\mathbf M$, also listed below.
\begin{center}
\small
\renewcommand{\arraystretch}{1.3}
\begin{tabular}{@{}llllll@{}}
\toprule
System & Basis kets & Label & Shift & Phase & Periods\\
\midrule
oscillator & $\ket{x'}_{\mathrm{osc}}$, $x'\in\R$ & $(x,\xi)\in\R^2$
 & $\mathbf T_x\ket{x'}=\ket{x'-x}$
 & $\mathbf M_\xi\ket{x'}=e^{2\pi i\xi x'}\ket{x'}$ & none\\
rotor & $\ket{n'}_{\mathrm{rot}}$, $n'\in\Z$ & $(n,\theta)\in\Z\times\R$
 & $\mathbf T_n\ket{n'}=\ket{n'-n}$
 & $\mathbf M_\theta\ket{n'}=e^{2\pi i\theta n'}\ket{n'}$
 & $\theta\sim\theta+1$\\
qudit & $\ket{j}_{\mathrm{fin}}$, $j\in\Z_c$ & $(a,b)\in\Z^2$
 & $\mathbf T_a\ket{j}=\ket{j-a}$
 & $\mathbf M_b\ket{j}=e^{2\pi ibj/c}\ket{j}$
 & $a\sim a+c$, $b\sim b+c$\\
\bottomrule
\end{tabular}
\end{center}
The displacement operator with label $(m,s)$ is
$\mathbf D(m,s):=\mathbf M_s\mathbf T_m$, where $m$ collects all shifts
and $s$ collects all phases.
On a qudit, $X:=\mathbf T_{-1}$ and $Z:=\mathbf M_1$ are the usual
generalized Pauli operators, with $X\ket j=\ket{j+1}$ and
$Z\ket j=e^{2\pi ij/c}\ket j$.
Thus $\mathbf D(a,b)=Z^bX^{-a}$, and for qubits $X$ and $Z$ are the Pauli
operators.
On a rotor, $\mathbf M_\theta=e^{2\pi i\theta\widehat{\mathbf L}}$, where
$\widehat{\mathbf L}\ket{n'}_{\mathrm{rot}}=n'\ket{n'}_{\mathrm{rot}}$
defines the angular-momentum operator.
On an oscillator, $\mathbf M_\xi=e^{2\pi i\xi\hat x}$ with the position
operator $\hat x$.
The rotor angle is defined modulo one, since angular momenta are integers.
Similarly, qudit labels are defined modulo $c$.

\textit{Lifted labels and commutation.}
We collect the labels of all systems into one lifted label vector
\begin{equation}
 v=(x,\xi\,|\,n,\theta\,|\,a,b)
 \in\widetilde G:=\R^{2p}\times\bigl(\Z^q\times\R^q\bigr)\times\Z^{2k},
 \label{eq:examples-lifted-label}
\end{equation}
where $x,\xi\in\R^p$, $n\in\Z^q$, $\theta\in\R^q$, and $a,b\in\Z^k$.
The entries are ordered by system type, and within each type the shifts
come before the phases.
We write $\pi(v)$ for the point of the phase space $G:=M\times\widehat M$
that $v$ labels, and $\mathbf D(v):=\mathbf D(\pi(v))$.
The map $\pi$ reduces rotor angles modulo one and the labels of the $j$th
qudit modulo $c_j$.
Its kernel is the lattice of \emph{period vectors}, generated by the unit
rotor-angle vectors $e_{\theta_i}$ and by the vectors $c_je_{a_j}$ and
$c_je_{b_j}$.
Displacements commute up to a phase,
\begin{equation}
 \mathbf D(v)\mathbf D(w)
 =e^{2\pi i\,v^{\intercal}Jw}\,\mathbf D(w)\mathbf D(v),
 \label{eq:examples-commutation}
\end{equation}
where $J$ is the block-diagonal matrix
\begin{equation}
 J=\begin{pmatrix}J_{\mathrm{osc}}&&\\&J_{\mathrm{rot}}&\\&&J_{\mathrm{fin}}
 \end{pmatrix},
 \qquad
 J_{\mathrm{osc}}=\begin{pmatrix}0&I_p\\-I_p&0\end{pmatrix},
 \quad
 J_{\mathrm{rot}}=\begin{pmatrix}0&I_q\\-I_q&0\end{pmatrix},
 \quad
 J_{\mathrm{fin}}=\begin{pmatrix}0&\dia c^{-1}\\-\dia c^{-1}&0\end{pmatrix}.
 \label{eq:symplectic-form}
\end{equation}
Its qudit block has entries $\pm1/c_j$, so $J$ is not an integer matrix.
Nevertheless, period vectors pair integrally with every lifted vector, so
the phase in Eq.~\eqref{eq:examples-commutation} does not depend on the
choice of lifts.

\textit{Stabilizer lattices.}
A stabilizer group consists of commuting displacement operators, possibly
multiplied by phases.
Its labels form a subgroup $D\subseteq G$.
We require $D$ to be a lattice, i.e., a discrete subgroup whose quotient
$G/D$ is compact, and call it the \emph{stabilizer lattice}.
Commutation means that $D$ is \emph{isotropic}: $v^{\intercal}Jw\in\Z$ for
all lifts $v$ and $w$ of elements of $D$.
The displacements $\mathbf D(w)$ with $w\in D$ then commute, but they need
not multiply as a group.
Lemma~\ref{lem:phase_triv} provides stabilizer phases that turn them into
a group.
The explicit lattices below need no phases, so their bare displacements
already form the stabilizer group.
The codespace $\mathcal C$ consists of the generalized states fixed by all
stabilizers.
Its dimension $K$ is finite and equals the covolume of $D$
(Theorem~\ref{thm:code_dimension}).

\textit{Stabilizer generator matrix.}
We encode $D$ in one square matrix.
The \emph{full lift} of $D$ is the set $\Lambda:=\pi^{-1}(D)$ of all lifted
labels of all elements of $D$.
It contains every period vector, and it is a full-rank lattice in
$\R^{2(p+q+k)}$ (Lemma~\ref{lem:full_lift_pfaffian}).
A \emph{stabilizer generator matrix} is a square matrix $T$ whose columns
form a basis of $\Lambda$.
We group its rows by system type,
\begin{equation}
 T=\begin{pmatrix}T_x\\ T_\xi\\ T_n\\ T_\theta\\ T_a\\ T_b\end{pmatrix}
 \begin{array}{@{}l@{}}
 \left.\begin{array}{@{}l@{}}\rowlab{positions}\\ \rowlab{momenta}
 \end{array}\right\}\ \text{$p$ oscillators}\\
 \left.\begin{array}{@{}l@{}}\rowlab{angular momenta}\\ \rowlab{angles}
 \end{array}\right\}\ \text{$q$ rotors}\\
 \left.\begin{array}{@{}l@{}}\rowlab{shifts}\\ \rowlab{phases}
 \end{array}\right\}\ \text{$k$ qudits,}
 \end{array}
 \label{eq:examples-generator-matrix}
\end{equation}
where $T_x$ collects the $p$ rows of oscillator positions, and similarly
for the other blocks.
Each column of $T$ labels a stabilizer, and a column that is a period
vector labels the identity operator.
Periods enter only because rotor angles and qudit labels are periodic.
For oscillators alone, $T$ is simply a basis of $D$.
To find $T$ in practice, we lift a generating set of $D$, adjoin the
period vectors, and extract a basis of the lattice that they generate.
If $D$ has \emph{torsion}, i.e., nonzero elements of finite order, then
the columns of $T$ cannot contain a basis of the period vectors.

\textit{Gram matrix and code dimension.}
The \emph{symplectic Gram matrix}
\begin{equation}
 \varTheta:=T^{\intercal}JT
 \label{eq:examples-gram-matrix}
\end{equation}
collects the commutation phases of the columns $t_i$ of $T$, in the sense
that
$\mathbf D(t_i)\mathbf D(t_j)=e^{2\pi i\varTheta_{ij}}
\mathbf D(t_j)\mathbf D(t_i)$.
Thus $D$ is isotropic exactly when $\varTheta$ is an integer matrix.
The code dimension is then a Pfaffian,
\begin{equation}
 K=|\operatorname{Pf}\varTheta|=\frac{|\det T|}{c_1\cdots c_k}.
 \label{eq:examples-dimension}
\end{equation}
Lemma~\ref{lem:full_lift_pfaffian} derives
Eq.~\eqref{eq:examples-dimension} from Theorem~\ref{thm:code_dimension}.
We keep the absolute value because the sign of
$\operatorname{Pf}\varTheta$ depends on the order of the rows and columns
of $T$.
For qudits alone, Eq.~\eqref{eq:examples-dimension} reduces to the
familiar count $K=c_1\cdots c_k/|D|$.
For oscillators alone, it is the usual GKP formula $K=|\det T|$.
With rotors or qudits, the period columns cannot be dropped.
For instance, the Pfaffian of the Gram matrix of the two stabilizer
generators of the LCA code in Sec.~\ref{ex:lca} can differ from its code dimension.
In general, it is the GKP formula for the lifted lattice $\Lambda$ with
each qudit coordinate rescaled by $c_j^{-1/2}$.
Section~\ref{ex:mixed-qudit} and
Remark~\ref{rem:full_lift_gkp_concatenation} make this concatenation
picture precise.

\textit{Logical matrix.}
Logical displacements commute with every stabilizer.
Their labels form the symplectic dual $D^\perp$, which contains $D$.
The finite group $D^\perp/D$ labels the logical displacements modulo
stabilizers and phases (Theorem~\ref{thm:logical_displacement_algebra}).
The \emph{logical matrix}
\begin{equation}
 T_{\mathrm{dual}}:=T\varTheta^{-1}
 \label{eq:examples-logical-matrix}
\end{equation}
has columns that generate all lifted labels of $D^\perp$
(Lemma~\ref{lem:full_lift_pfaffian}).
Every logical label is therefore $T_{\mathrm{dual}}\ell$ for an integer
vector $\ell$.
Two integer vectors give the same logical displacement, up to stabilizers
and phases, exactly when their difference lies in
$\varTheta\Z^{2(p+q+k)}$.
Hence
\begin{equation}
 D^\perp/D\cong\Z^{2(p+q+k)}/\varTheta\,\Z^{2(p+q+k)},
 \qquad
 |D^\perp/D|=\det\varTheta=K^2.
 \label{eq:examples-logical-group}
\end{equation}
The Gram matrix of $T_{\mathrm{dual}}$ is $-\varTheta^{-1}$, so the
logical displacements obey
\begin{equation}
 \mathbf D(T_{\mathrm{dual}}\ell)\mathbf D(T_{\mathrm{dual}}\ell')
 =e^{-2\pi i\,\ell^{\intercal}\varTheta^{-1}\ell'}\,
 \mathbf D(T_{\mathrm{dual}}\ell')\mathbf D(T_{\mathrm{dual}}\ell).
 \label{eq:examples-logical-commutation}
\end{equation}
Some columns of $T_{\mathrm{dual}}$ can be stabilizers, since $D^\perp$
contains $D$.
The alternating Smith normal form of $\varTheta$ further identifies the
logical system
as a collection of qudits whose dimensions multiply to $K$, as explained
after Lemma~\ref{lem:full_lift_pfaffian}.

\textit{Ideal codewords.}
Most of our examples use CSS lattices $D=A\times B$, with a shift part
$A\subseteq M$ and a phase part $B\subseteq\widehat M$.
The two parts commute when every phase in $B$ acts trivially on every
configuration in $A$.
CSS lattices need no stabilizer phases (Sec.~\ref{sec:GSC}).
Let $B^\perp\subseteq M$ be the set of configurations on which every phase
in $B$ acts trivially.
A codeword is fixed by every phase in $B$, so it is supported on
$B^\perp$.
It is also fixed by every shift in $A$, so it is constant on each coset
$r+A$ in $B^\perp$.
The ideal codewords are therefore spanned by the combs
\begin{equation}
 \sum_{a\in A}\ket{r+a},\qquad r\in B^\perp,
 \label{eq:examples-css-codewords}
\end{equation}
one for each coset of $A$ in $B^\perp$.
Distinct cosets give combs with disjoint supports, so the combs are
linearly independent.
Hence $K=|B^\perp/A|$, the number of such cosets.
When $A$ is infinite, these combs are non-normalizable distributions,
like ideal GKP states.
For CSS lattices, the Bloch--Floquet decomposition of
Theorem~\ref{thm:rieffel_bundle_decomposition} reduces to the Zak
transform along $A$~\cite{zak1967finite}.
Its basis states span every syndrome space and are the group coset states
of Refs.~\cite{albert2020robust,culf2022group}, i.e., combs over cosets of
$A$ with phases given by characters of $A$.
Section~\ref{sec:finite_energy} constructs normalizable finite-energy
versions.

\subsection{Stabilizer codes over a single factor}
\label{subsec:examples_single}

We first consider codes on a single type of system.
Most of these codes are known, and we use them to illustrate the
notation of Sec.~\ref{subsec:examples_notation}.

\subsubsection{Qudit stabilizer codes}
\label{ex:qudit}
Qubit and qudit stabilizer codes are the GLCA codes with $p=q=0$.
For clarity, we briefly indicate what the lattice $D$ is in this context.
If the stabilizer group is generated by $l$ independent Pauli operators
$g_1,g_2,\ldots,g_l$ in the $n$-qubit Pauli group $G_n$, then the subgroup
generated by the binary symplectic representations of the operators
$g_1,g_2,\ldots,g_l$ inside $\Z_2^n\times\Z_2^n$ is precisely the lattice
$D$ in this case.
More generally, for $n$ qudits of local dimension $d$, qudit stabilizer
codes are obtained by replacing the binary phase space with
$\Z_d^n\times\Z_d^n$ and taking $D$ to be the subgroup generated by the
symplectic representations of the generalized Pauli
stabilizers~\cite{10.1007/3-540-49208-9_27}.
Every subgroup of this finite phase space is a lattice, and
Eq.~\eqref{eq:examples-dimension} gives $K=d^n/|D|$ when it is isotropic.
For example, $n$ qubits with $l$ independent generators give $|D|=2^l$
and $K=2^{n-l}$.
Unequal local dimensions are treated in Sec.~\ref{ex:mixed-qudit},
where integer lifts give the stabilizer lattice, codespace dimension, and
logical displacement operators.

\textit{A two-qudit repetition code.}
As a small illustration, take two qudits of dimension $d$ and the
stabilizer $Z\otimes Z^{-1}$.
Its label is $w=(0,0\,|\,1,-1)$ in the row order $(a_1,a_2\,|\,b_1,b_2)$
of Sec.~\ref{subsec:examples_notation}.
The full lift is generated by $w$ and the periods $d\,e_{a_1}$,
$d\,e_{a_2}$, $d\,e_{b_1}$, and $d\,e_{b_2}$, where
$e_{a_1},\ldots,e_{b_2}$ are the standard unit vectors.
The last period equals $d\,e_{b_1}-d\,w$, so a stabilizer generator matrix
is
\begin{equation}
 T=\Bigl(\underbrace{\,w\,}_{Z\otimes Z^{-1}}\;\;
   \underbrace{\,d\,e_{a_1}\;\;d\,e_{a_2}\;\;d\,e_{b_1}\,}_{\text{periods}}
 \Bigr)
 =\begin{pmatrix}0&d&0&0\\0&0&d&0\\1&0&0&d\\-1&0&0&0\end{pmatrix}
 \;\begin{array}{@{}l@{\;}l@{}}
 a_1&\rowquadbrace{two qudits.}\\
 a_2&\\
 b_1&\\
 b_2&
 \end{array}
 \label{eq:examples-qudit-T}
\end{equation}
Its Gram matrix and Pfaffian are
\begin{equation}
 \varTheta=\begin{pmatrix}0&-1&1&0\\1&0&0&d\\-1&0&0&0\\0&-d&0&0
 \end{pmatrix},
 \qquad
 \operatorname{Pf}\varTheta=-d,
 \label{eq:examples-qudit-Theta}
\end{equation}
so $K=d$.
Directly, the $+1$ eigenspace of $Z\otimes Z^{-1}$ is spanned by the
states $\ket j\ket j$, so the codewords are $\ket{j_L}=\ket{j}_{\mathrm{fin}}\otimes\ket{j}_{\mathrm{fin}}$
with $j\in\Z_d$.
The logical matrix has the columns $(0,d\,|\,0,0)$, $(0,0\,|\,1,0)$,
$(0,0\,|\,0,1)$, and $(-1,-1\,|\,0,0)$.
The first is a period, and the others label $Z\otimes I$, $I\otimes Z$,
and $X\otimes X$.
Thus $\xl=X\otimes X$ and $\zl=Z\otimes I$ act as
$\xl\ket{j_L}=\ket{(j+1)_L}$ and $\zl\ket{j_L}=e^{2\pi ij/d}\ket{j_L}$.
For $d=2$, this is the two-qubit repetition code with stabilizer $ZZ$.

\textit{The toric code.}
The toric code is the basic example of a topological CSS stabilizer code
\cite{kitaev2003fault,dennis2002topological}.  Start with a finite cellulation
$\mathcal L$ of a closed, connected, orientable surface $\Sigma_g$ of genus
$g$.  Let $V$, $E$, and $F$ be the vertex, edge, and face sets of
$\mathcal L$. Its one-skeleton is the primal graph with vertex set $V$
and edge set $E$; the dual graph exchanges vertices and faces.  Place one
qubit on every edge, so the physical Hilbert space contains $n=|E|$ qubits.
The cellular chain complex over $\mathbb Z_2$ is
\begin{equation}
 C_2(\mathcal L;\mathbb Z_2)
 \xrightarrow{\;\partial_2\;}
 C_1(\mathcal L;\mathbb Z_2)
 \xrightarrow{\;\partial_1\;}
 C_0(\mathcal L;\mathbb Z_2),
 \qquad \partial_1\partial_2=0.
 \label{eq:toric_cellular_chain_complex}
\end{equation}
After identifying the cellular bases with $F$, $E$, and $V$, put
$\delta_0:=\partial_1^{\intercal}$.  For a vertex $v$ and a face $f$, the
usual star and plaquette stabilizers are
\begin{equation}
 \begin{aligned}
 \mathbf{A}_v&:=\prod_{e\in E}\mathbf X_e^{(\partial_1)_{ve}}
     =\mathbf X^{\delta_0 v},\\
 \mathbf{B}_f&:=\prod_{e\in E}\mathbf Z_e^{(\partial_2)_{ef}}
     =\mathbf Z^{\partial_2 f}.
 \end{aligned}
 \label{eq:toric_star_plaquette_stabilizers}
\end{equation}
The exponents are cellular incidence numbers modulo $2$, so repeated
incidences are counted with their multiplicities. On the usual square
lattice on the torus these are the familiar four-body star
and plaquette checks.  The same formulas on $\Sigma_g$ define the closed
surface-code generalization.

We now translate this graph picture into the lattice language of the present
paper.  Set
\begin{equation}
 M:=C_1(\mathcal L;\mathbb Z_2)\cong\mathbb Z_2^{E},
 \qquad
 \widehat M\cong C_1(\mathcal L;\mathbb Z_2),
 \label{eq:toric_configuration_group}
\end{equation}
using the standard pairing $\langle m,s\rangle=(-1)^{m\cdot s}$.  The binary
symplectic representatives of the star and plaquette stabilizers generate the
subgroup
\begin{equation}
 D_{\mathrm{TC}}
 :=\operatorname{im}(\delta_0)\times\operatorname{im}(\partial_2)
 \subseteq M\times\widehat M.
 \label{eq:toric_stabilizer_lattice}
\end{equation}
For $a\in C_0(\mathcal L;\mathbb Z_2)$
and $b\in C_2(\mathcal L;\mathbb Z_2)$,
\[
 (\delta_0 a)\cdot(\partial_2 b)
 =a^{\intercal}\partial_1\partial_2 b=0.
\]
Thus $\operatorname{im}(\partial_2)\subseteq
\operatorname{im}(\delta_0)^\perp$, so $D_{\mathrm{TC}}$ is an isotropic
split CSS lattice. As a CSS lattice, it needs no stabilizer phases.

We now start from the covolume formula and show that the usual homological
count follows from it.  Write
\begin{equation}
 r:=\dim_{\mathbb Z_2}D_{\mathrm{TC}}
 =\operatorname{rank}(\partial_1)+\operatorname{rank}(\partial_2),
 \label{eq:toric_stabilizer_rank}
\end{equation}
where $\operatorname{rank}(\delta_0)=\operatorname{rank}(\partial_1)$.
Since $|D_{\mathrm{TC}}|=2^r$, the qudit count $K=c_1\cdots c_k/|D|$
below Eq.~\eqref{eq:examples-dimension} gives
\begin{align}
 \dim\mathcal C
 =\frac{2^{|E|}}{|D_{\mathrm{TC}}|}
 &=2^{|E|-r}\notag\\
 &=2^{|E|-\operatorname{rank}(\partial_1)
          -\operatorname{rank}(\partial_2)}.
 \label{eq:toric_covolume_homology}
\end{align}
By Theorem~\ref{thm:code_dimension}, this number is also the covolume of
$D_{\mathrm{TC}}$.  Thus, if $\textbf{k}:=\log_2\dim\mathcal C$ is the number of
encoded qubits, the covolume computation first gives
\begin{equation}
 \textbf{k}=|E|-\operatorname{rank}(\partial_1)
       -\operatorname{rank}(\partial_2).
 \label{eq:toric_k_from_covolume}
\end{equation}
The homological formula now emerges directly from rank--nullity:
\begin{align}
 \textbf{k}
 &=\dim\ker(\partial_1)-\dim\operatorname{im}(\partial_2)\notag\\
 &=\dim\bigl(\ker(\partial_1)/\operatorname{im}(\partial_2)\bigr)
 =\dim H_1(\Sigma_g;\mathbb Z_2).
 \label{eq:toric_logical_homology}
\end{align}
For a connected closed cellulation,
$\operatorname{rank}(\partial_1)=|V|-1$ and
$\operatorname{rank}(\partial_2)=|F|-1$.  Hence the same covolume formula,
together with $|V|-|E|+|F|=2-2g$, gives
\begin{equation}
 \textbf{k}=|E|-|V|-|F|+2=2g,
 \qquad \dim\mathcal C=2^{\textbf{k}}=2^{2g}.
 \label{eq:toric_codespace_dimension}
\end{equation}
In other words, the standard homological count is not an independent input:
it is precisely the binary logarithm of the lattice covolume.  In particular,
the ordinary toric code on
$\Sigma_1=\mathbb T^2$ has a four-dimensional codespace and encodes two
logical qubits.
By Eq.~\eqref{eq:examples-css-codewords}, the ideal codewords are the
coset states $\sum_{a\in\operatorname{im}(\delta_0)}\ket{\gamma+a}$, where
$\gamma$ runs over representatives of
$\ker(\partial_2^{\intercal})/\operatorname{im}(\delta_0)$.

\textit{Further properties.}
Section~\ref{sec:PSC_basic} describes Pauli stabilizer codes through
projective modules.
As for the toric code, Eq.~\eqref{eq:examples-css-codewords} gives the
coset states of qudit CSS codes~\cite{albertfaist2026modularquditcss}.
The Pauli distance of qubit codes and the Lee distance of qudit
codes~\cite{guo_lee_codes} appear in Eqs.~\eqref{eq:pauli_hamming_distance}
and~\eqref{eq:qudit_lee_distance}, respectively.
Section~\ref{subsec:we_finite} recovers their weight enumerators.

\subsubsection{GKP codes}
\label{ex:gkp}
The usual GKP
codes~\cite{10.1103/physreva.64.012310,Conrad2022gottesmankitaev,10.1103/prxquantum.3.010335}
use $p$ oscillators, so that $M=\R^p$ and there are no period vectors.
A stabilizer generator matrix is then any basis $T$ of $D$.
The lattice is isotropic exactly when
$\varTheta=T^{\intercal}J_{\mathrm{osc}}T$ is an integer matrix.
Equation~\eqref{eq:examples-dimension} gives
$K=|\det T|=|\operatorname{Pf}\varTheta|$.
The columns of $T_{\mathrm{dual}}=T\varTheta^{-1}$ generate $D^\perp$,
which supplies the logical displacement operators.

A useful family of GKP codes is obtained from
$\varTheta=\lambda J_{\mathrm{osc}}$ with $\lambda\in\mathbb N$.
If $T_0\in\mathrm{Sp}(2p,\R)$, then $T=\sqrt\lambda\,T_0$ satisfies
$T^{\intercal}J_{\mathrm{osc}}T=\varTheta$ and defines a scaled GKP lattice
with logical dimension $K=|\det T|=\lambda^p$.
For a general $T_0$, the stabilizers may need the phases of
Lemma~\ref{lem:phase_triv}.
For one mode, the square-lattice GKP code~\cite{10.1103/physreva.64.012310}
is specified by
\begin{equation}
 \varTheta=\begin{pmatrix}0&2\\-2&0\end{pmatrix},
 \qquad T=\sqrt 2\,I_2,
\end{equation}
whereas the same symplectic Gram matrix with
\begin{equation}
 T=\frac{1}{3^{1/4}}\begin{pmatrix}2&1\\0&\sqrt 3\end{pmatrix}
\end{equation}
gives the hexagonal GKP lattice. These two embeddings illustrate that the
same logical commutation data can be realized by different lattice
geometries.

\textit{The rectangular GKP qubit.}
For $\alpha>0$, choose
\begin{equation}
 A=2\alpha\mathbb Z,
 \qquad
 B=\alpha^{-1}\mathbb Z.
 \label{eq:css_qubit_gkp_lattices}
\end{equation}
The split stabilizer lattice $D=A\times B$ has the stabilizer generator
matrix
\begin{equation}
 T=\begin{pmatrix}2\alpha&0\\0&\alpha^{-1}\end{pmatrix}
 \;\begin{array}{@{}l@{\;}l@{}}
 x&\rowpairbrace{oscillator}\\
 \xi\vphantom{\alpha^{-1}}&
 \end{array}\qquad\quad
 T^{\intercal}J_{\mathrm{osc}}T=\begin{pmatrix}0&2\\-2&0\end{pmatrix}.
 \label{eq:examples-rectangular-gkp-T}
\end{equation}
Thus its covolume, and hence its logical dimension, is $K=2$.
For $\alpha=1/\sqrt2$, this is the square embedding $T=\sqrt2\,I_2$ above.
Other positive values give rectangular GKP lattices related to it by
squeezing.
Here $B^\perp=\alpha\Z$, so Eq.~\eqref{eq:examples-css-codewords} gives the
two ideal codewords
\begin{align}
 \ket{0_{\mathrm{GKP}}}
 &=\sum_{m\in\mathbb Z}\ket{2\alpha m},
 \label{eq:css_qubit_gkp_zero}\\
 \ket{1_{\mathrm{GKP}}}
 &=\sum_{m\in\mathbb Z}\ket{\alpha+2\alpha m}.
 \label{eq:css_qubit_gkp_one}
\end{align}
These combs form the ideal GKP logical
basis~\cite{10.1103/physreva.64.012310}.
The logical matrix
$T_{\mathrm{dual}}=\bigl(\begin{smallmatrix}0&-\alpha\\
1/(2\alpha)&0\end{smallmatrix}\bigr)$ gives the logical operators
$\xl=\mathbf D(-\alpha,0)=\mathbf T_{-\alpha}$ and
$\zl=\mathbf D(0,\tfrac1{2\alpha})=\mathbf M_{1/(2\alpha)}$.
They act as $\xl\ket{j_{\mathrm{GKP}}}=\ket{(j+1)_{\mathrm{GKP}}}$ and
$\zl\ket{j_{\mathrm{GKP}}}=(-1)^j\ket{j_{\mathrm{GKP}}}$, with $j$ read
modulo two.

\textit{Clifford gates.}
Lattice symmetries give logical Clifford
gates~\cite[Secs.~4.2 and~6]{fiber-bundle-ft}.
An integer matrix $W$ with $W^{\intercal}\varTheta W=\varTheta$ yields the
symplectic matrix $S_W:=TWT^{-1}$, which maps $D$ onto itself.
A Gaussian unitary implementing $S_W$, corrected by a displacement when
necessary, preserves the codespace and acts as a logical Clifford gate
(Theorem~\ref{thm:corrected_clifford_implementer}).
On logical labels $T_{\mathrm{dual}}\ell$, this gate acts as
$[\ell]\mapsto[W^{-\intercal}\ell]$, as detailed in
Sec.~\ref{subsec:clifford_examples}.

\textit{Further properties.}
Multimode GKP codes on two-dimensional arrays of modes can realize the
anyons of $\U(1)_{2n}\times\U(1)_{-2m}$ Chern--Simons
theory~\cite{magdalena2024topological}.
These Chern--Simons GKP codes store $4nm$ states on a torus.
Section~\ref{sec:distance_error_correction} recovers the Euclidean GKP
distance and closest-vector decoding,
Eqs.~\eqref{eq:gkp_euclidean_distance}
and~\eqref{eq:gkp_decoder_closest_vector}.
Section~\ref{sec:weight_enumerators} relates the enumerators of $D$ and
$D^\perp$ through theta series (Sec.~\ref{subsec:we_lattice}).
The Gaussian regularization of Sec.~\ref{sec:finite_energy} applies to the
codewords~\eqref{eq:css_qubit_gkp_zero}--\eqref{eq:css_qubit_gkp_one}.
For GKP codes, Theorem~\ref{thm:fe_asymptotic_isometry} recovers the
asymptotic isometry of Mayrand and Royer~\cite{mayrand2026complex}.

\subsubsection{Rotor GKP codes}
\label{ex:rotor-gkp}
The oscillator GKP codes of Sec.~\ref{ex:gkp} have a planar-rotor
counterpart, already described in the original GKP construction
\cite[Sec.~IV, Eqs.~(19)--(20)]{10.1103/physreva.64.012310}.
In this framework, the relevant phase space is
$\Z^q\times\T^q$, replacing $\R^p\times\widehat{\R^p}$.
We focus on CSS lattices and their codewords and end with an
intrinsically non-CSS two-rotor code.

\textit{Rotor GKP codes as CSS lattices.}
For $q$ planar rotors, take
\begin{equation}
 M=\Z^q,
 \qquad
 \widehat M=\T^q,
 \qquad
 G=\Z^q\times\T^q,
 \label{eq:rotor-phase-space}
\end{equation}
with $\T=\R/\Z$ and
$\mathcal H=\ell^2(\Z^q)\cong L^2(\T^q)$.
The displacements $\mathbf T_n$ and
$\mathbf M_\theta=e^{2\pi i\theta\cdot\widehat{\mathbf L}}$ act on the
angular-momentum kets as in the rotor row of the table in
Sec.~\ref{subsec:examples_notation}.
A split rotor GKP lattice has the form
\begin{equation}
 D=A\times B,
 \qquad A\subseteq\Z^q,
 \qquad B\subseteq A^\perp\subseteq\T^q,
 \label{eq:rotor-css-general}
\end{equation}
where $A$ and $B$ are subgroups, $A$ has finite index in $\Z^q$,
$B$ is finite, and $A^\perp\subseteq\T^q$ is the set of phases that act
trivially on $A$.
Thus $D$ is a lattice, and as a CSS lattice it needs no stabilizer phases.
Let $Q:=B^\perp/A$, so that $K=|Q|$ by Sec.~\ref{subsec:examples_notation}.
A displacement $\mathbf D(n,\theta)$ commutes with every shift in $A$
exactly when $\theta\in A^\perp$, and with every phase in $B$ exactly when
$n\in B^\perp$.
Hence $D^\perp=B^\perp\times A^\perp$.
Pontryagin duality identifies $A^\perp/B$ with the dual group $\widehat Q$.
In summary,
\begin{equation}
 \begin{aligned}
 Q&=B^\perp/A,\qquad K=|Q|,\\
 D^\perp&=B^\perp\times A^\perp,\qquad
 D^\perp/D\cong Q\times\widehat Q.
 \end{aligned}
 \label{eq:rotor-css-logical-group}
\end{equation}
In particular, representatives $r_0,\ldots,r_{K-1}\in B^\perp$ of the $K$
cosets give the ideal logical basis
\begin{equation}
 \ket{j_L}=\sum_{a\in A}\ket{r_j+a}_{\mathrm{rot}},
 \qquad j=0,\ldots,K-1.
 \label{eq:rotor-css-general-codewords}
\end{equation}
These are the combs of Eq.~\eqref{eq:examples-css-codewords}.
When $A$ is infinite, they are not square-summable states in
$\ell^2(\Z^q)$.

\textit{From oscillator GKP to rotor GKP.}
The connection with the oscillator GKP codes of Sec.~\ref{ex:gkp} can be
made explicit for any
split oscillator lattice
$D_{\mathrm{osc}}=A_{\mathrm{osc}}\times B_{\mathrm{osc}}$, where
$A_{\mathrm{osc}}\subseteq\R^p$ and
$B_{\mathrm{osc}}\subseteq\widehat{\R^p}$ are full lattices with
$B_{\mathrm{osc}}\subseteq A_{\mathrm{osc}}^\perp$.
Choose a finite-index sublattice $B_0\subseteq B_{\mathrm{osc}}$ and set
\begin{equation}
 \begin{aligned}
 M_{\mathrm{rot}}&:=B_0^\perp\subseteq\R^p,
 \qquad M_{\mathrm{rot}}\cong\Z^p,\\
 \widehat{M_{\mathrm{rot}}}&\cong\widehat{\R^p}/B_0\cong\T^p.
 \end{aligned}
 \label{eq:oscillator-to-rotor-configuration}
\end{equation}
Here $B_0^\perp$ is taken in the oscillator configuration group, so
the construction produces $p$ rotors ($q=p$).  The dual identification
uses the pairing $\langle x,[s]\rangle=\langle x,s\rangle$, which is
well defined because $x\in B_0^\perp$.
The modulation subgroup descends to the finite group
$B_{\mathrm{rot}}:=B_{\mathrm{osc}}/B_0$, and
\begin{equation}
 \begin{aligned}
 A_{\mathrm{osc}}&\subseteq B_{\mathrm{osc}}^\perp
                    \subseteq M_{\mathrm{rot}},\\
 D_{\mathrm{rot}}&:=A_{\mathrm{osc}}\times B_{\mathrm{rot}}
       \subseteq M_{\mathrm{rot}}\times\widehat{M_{\mathrm{rot}}}.
 \end{aligned}
 \label{eq:oscillator-to-rotor-lattice}
\end{equation}
Since $A_{\mathrm{osc}}$ has finite index in $M_{\mathrm{rot}}$ and
$B_{\mathrm{rot}}$ is finite, this is a rotor CSS lattice.
Taking the annihilator of
$B_{\mathrm{rot}}$ inside $M_{\mathrm{rot}}$ gives
\begin{equation}
 B_{\mathrm{rot}}^\perp=B_{\mathrm{osc}}^\perp,
 \qquad
 \frac{B_{\mathrm{rot}}^\perp}{A_{\mathrm{osc}}}
 =\frac{B_{\mathrm{osc}}^\perp}{A_{\mathrm{osc}}}.
 \label{eq:oscillator-to-rotor-logical-quotient}
\end{equation}
Thus discretizing the configuration coordinate and compactifying its
Pontryagin-dual coordinate preserves $Q$ and the logical dimension $K$.
This gives a lattice formulation of the connection; the original GKP
paper obtains the rotor code as a large-dimensional qudit limit.

\textit{The single-rotor CSS family and its codewords.}
Generalize the rectangular oscillator GKP qubit of
Eq.~\eqref{eq:css_qubit_gkp_lattices} to any positive integer
logical dimension $K$ by taking
\begin{equation}
 A_{\mathrm{osc}}=K\alpha\Z,
 \qquad B_{\mathrm{osc}}=\alpha^{-1}\Z,
 \qquad \alpha>0.
 \label{eq:rotor-parent-gkp-lattice}
\end{equation}
Its logical quotient is $\alpha\Z/(K\alpha\Z)\cong\Z_K$.
For an integer $N\geq1$, choose $B_0=N\alpha^{-1}\Z$.
Then $M_{\mathrm{rot}}=(\alpha/N)\Z$, and the identification
$(\alpha/N)n\leftrightarrow n$ puts the resulting rotor lattice in the
standard coordinates
\begin{equation}
 A=KN\Z,
 \qquad B=\tfrac1N\Z/\Z,
 \qquad D=A\times B\subseteq\Z\times\T.
 \label{eq:rotor-css-lattice}
\end{equation}
Its stabilizer generator matrix is
\begin{equation}
 T=\begin{pmatrix}KN&0\\0&\tfrac1N\end{pmatrix}
 \;\begin{array}{@{}l@{\;}l@{}}
 n&\rowpairbrace{rotor}\\
 \theta\vphantom{\tfrac1N}&
 \end{array}\qquad\quad
 \varTheta=\begin{pmatrix}0&K\\-K&0\end{pmatrix},
 \label{eq:rotor-css-full-lift}
\end{equation}
because the angle period $(0,1)=N\,(0,\tfrac1N)$ is already an integer
combination of the columns.
The symplectic rescaling $\operatorname{diag}(N/\alpha,\alpha/N)$ maps the
parent matrix $\operatorname{diag}(K\alpha,\alpha^{-1})$ to this $T$, and
the generator of $B_0$ becomes the angle period.
The stabilizer generators are
\begin{equation}
 \mathbf D(KN,0)=\mathbf T_{KN},
 \qquad
 \mathbf D(0,\tfrac1N)=\mathbf M_{1/N}
 =e^{2\pi i\widehat{\mathbf L}/N}.
 \label{eq:rotor-css-generators}
\end{equation}
The $+1$ modulation constraint restricts angular momentum to $N\Z$,
and the translation makes the amplitudes periodic with period $KN$.
Consequently,
\begin{equation}
 B^\perp=N\Z,
 \qquad Q=N\Z/(KN\Z)\cong\Z_K.
 \label{eq:rotor-css-dimension}
\end{equation}
This agrees with $K=|\operatorname{Pf}\varTheta|$ in
Eq.~\eqref{eq:rotor-css-full-lift}.
The $K$ logical sectors have representatives $r_j=jN$, giving
\begin{equation}
 \ket{j_L}=\sum_{k\in\Z}\ket{N(j+Kk)}_{\mathrm{rot}},
 \quad j=0,\ldots,K-1.
 \label{eq:rotor-css-codewords}
\end{equation}
For $K=2$, these are the rotor counterparts of the two oscillator combs
in Eqs.~\eqref{eq:css_qubit_gkp_zero}--\eqref{eq:css_qubit_gkp_one}.
These $K$ combs have disjoint supports and exhaust the codespace by
Eq.~\eqref{eq:examples-css-codewords}.  The symplectic dual and
logical displacement group are
\begin{equation}
 \begin{aligned}
 D^\perp&=N\Z\times\bigl(\tfrac1{KN}\Z/\Z\bigr),\\
 D^\perp/D&\cong\Z_K\times\Z_K.
 \end{aligned}
 \label{eq:rotor-css-dual}
\end{equation}
Logical generators may therefore be chosen as
$\mathbf D(N,0)=\mathbf T_N$ and
$\mathbf D(0,1/(KN))=\mathbf M_{1/(KN)}$, with action
\begin{equation}
 \begin{aligned}
 \mathbf T_N\ket{j_L}&=\ket{(j-1)_L},\\
 \mathbf M_{1/(KN)}\ket{j_L}&=e^{2\pi ij/K}\ket{j_L},
 \end{aligned}
 \label{eq:rotor-css-logical-action}
\end{equation}
where $j$ is read modulo $K$.
The logical matrix
$T_{\mathrm{dual}}=\bigl(\begin{smallmatrix}0&-N\\1/(KN)&0
\end{smallmatrix}\bigr)$ contains these generators, up to inversion, as
its columns $(0,\tfrac1{KN})$ and $(-N,0)$.
Setting $m=KN$ and $n=K$ in the notation of
\cite[Sec.~IV, Eqs.~(19)--(20)]{10.1103/physreva.64.012310}
gives their stabilizers as $\mathbf T_{-KN}$ and $\mathbf M_{-1/N}$
in our convention.  Their logical $\overline Z$ and $\overline X$
are respectively $\mathbf T_{-N}$ and $\mathbf M_{-1/(KN)}$.
These are the inverses of the generators above and generate the same
stabilizer and logical displacement groups.
Section~\ref{subsec:fe_examples} regularizes these combs into
finite-energy rotor states, Eq.~\eqref{eq:fe_rotor_finite_energy_codeword}.
For $N=1$, the code stores a qudit of dimension $K$ in the angular
momentum modulo $K$.
Reference~\cite{luo2026hybridqubitrotor} uses finite-energy versions of
this code, in its angular-momentum and angle bases, as registers for rotor
quantum Fourier transforms.

\textit{An intrinsically non-CSS rotor code.}
We end with a two-rotor code that no generalized Clifford operator maps
to a CSS code.
Its stabilizer generators are
\begin{equation}
 \mathbf S_1=\mathbf T_2\otimes\mathbf M_{1/2},
 \qquad
 \mathbf S_2=I\otimes\mathbf T_2.
 \label{eq:rotor-noncss-stabilizers}
\end{equation}
They commute because the parity $\mathbf M_{1/2}=(-1)^{\widehat{\mathbf L}}$
of the second rotor is invariant under shifts by two.
Their labels generate a lattice $D\cong\Z^2$ without torsion.
Adjoining the two angle periods gives
\begin{equation}
 T=\begin{pmatrix}2&0&0&0\\0&2&0&0\\0&0&1&0\\\tfrac12&0&0&1\end{pmatrix}
 \;\begin{array}{@{}l@{\;}l@{}}
 n_1&\rowquadbrace{two rotors}\\
 n_2&\\
 \theta_1&\\
 \theta_2\vphantom{\tfrac12}&
 \end{array}\qquad\quad
 \varTheta=\begin{pmatrix}0&-1&2&0\\1&0&0&2\\-2&0&0&0\\0&-2&0&0
 \end{pmatrix}.
 \label{eq:rotor-noncss-T}
\end{equation}
Thus $K=|\operatorname{Pf}\varTheta|=4$.
The Smith normal form $\operatorname{diag}(1,1,4,4)$ of $\varTheta$ shows
that the code stores one ququart, with $D^\perp/D\cong\Z_4^2$.
The ideal codewords are
\begin{equation}
 \ket{r,s}_L=\sum_{a,b\in\Z}(-1)^{as}\,
 \ket{r+2a}_{\mathrm{rot}}\otimes\ket{s+2b}_{\mathrm{rot}},
 \qquad r,s\in\{0,1\}.
 \label{eq:rotor-noncss-codewords}
\end{equation}
The signs $(-1)^{as}$ absorb the phase $(-1)^s$ that $\mathbf S_1$
produces on the second rotor.
The logical displacements $\mathbf T_1\otimes I$ and
$\mathbf M_{1/4}\otimes\mathbf T_1$ have order four modulo stabilizers
and commute up to the phase $i$.

To show that no generalized Clifford operator makes the code CSS, note
that these operators act on labels by symplectic automorphisms of $G$
(Proposition~\ref{prop:symplectic_implementer}).
Each such automorphism preserves the identity component
$\{0\}\times\T^2$ of $G$, which consists of the pure phases.
It therefore acts on angular momenta by some $P\in\mathrm{GL}(2,\Z)$, in
agreement with the rotor Clifford group of Ref.~\cite{rotorclifford}.
The image of $D$ thus has the angular-momentum projection
$P(2\Z^2)=2\Z^2$.
If this image were a CSS lattice $A\times B$, it would be torsion free
like $D$, so its finite phase part $B$ would be trivial.
It would then equal $2\Z^2\times\{0\}$, whose logical group is $\Z_2^4$
by Eq.~\eqref{eq:rotor-css-logical-group}.
This contradicts $D^\perp/D\cong\Z_4^2$, which symplectic automorphisms
preserve.
More generally, the stabilizers $\mathbf T_p\otimes\mathbf M_{1/p}$ and
$I\otimes\mathbf T_p$ with an integer $p\geq2$ give a torsion-free code
with $D^\perp/D\cong\Z_{p^2}^2$.
The same argument applies, since the CSS candidate $p\Z^2\times\{0\}$ has
logical group $\Z_p^4$.

The code also couples the two rotors intrinsically: no symplectic
automorphism of $G$ maps $D$ onto a direct sum of two single-rotor
lattices (Sec.~\ref{subsec:examples_coupling}).
Indeed, a single-rotor lattice is generated by its finite phase part and
one label $(a,\varphi)$ with $a\neq0$.
The symplectic shear $(n,\theta)\mapsto(n,\theta-\varphi n/a)$ fixes the
phase part and maps $(a,\varphi)$ to $(a,0)$.
Every symplectically decomposable two-rotor lattice is therefore
symplectically equivalent to a CSS lattice, unlike $D$.

This obstruction is specific to rotors.
The CSS lattice $(\Z\times4\Z)\times\{0\}$, a trivial rotor next to a
rotor ququart, also has $K=4$ and no torsion.
Its Gram matrix and $\varTheta$ have the same alternating Smith normal
form (see after Lemma~\ref{lem:full_lift_pfaffian}), so the two full lifts
have bases with equal Gram matrices.
For oscillators, bases $T$ and $T'$ with equal Gram matrices are related
by the symplectic matrix $T'T^{-1}$.
Hence the GKP codes obtained from the two lattices by
Remark~\ref{rem:full_lift_gkp_concatenation} are symplectically
equivalent.
Rotor Clifford operations cannot relate the two rotor codes, because they
never add angles to angular momenta.

\subsubsection{Homological rotor codes}
\label{ex:homological-rotor}
Homological rotor codes fit into our framework after one reduction step,
provided that their homology is finite.
The CSS rotor codes of Vuillot, Ciani, and Terhal
\cite{vuillotHomologicalQuantumRotor2024} use integer matrices
$H_X\in\Z^{r_X\times q}$ and $H_Z\in\Z^{r_Z\times q}$ satisfying
$H_XH_Z^{\intercal}=0$.  Their stabilizer subgroup is
$A_C\times B_C\subseteq\Z^q\times\T^q$, where
\begin{equation}
 A_C=H_X^{\intercal}\Z^{r_X},
 \qquad B_C=H_Z^{\intercal}\T^{r_Z}.
 \label{eq:homological-rotor-stabilizers}
\end{equation}
The authors exclude choosing nontrivial finite subgroups as the modulation
stabilizer group, corresponding to modular angular-momentum measurements
\cite[p.~6, remarks after Eq.~(24)]{vuillotHomologicalQuantumRotor2024}.
Here $B_C$ is connected.
If $\dim B_C>0$, then $A_C\times B_C$ is not discrete and hence not a
lattice.
Imposing its modulation constraints restricts the physical
space to $\ell^2(M_C)$, where $M_C:=B_C^\perp=\ker H_Z$.
The remaining stabilizers and logical labels are
\begin{equation}
 \begin{aligned}
 D_{\mathrm{red}}&=A_C\times\{0\}
       \subseteq M_C\times\widehat{M_C},\\
 Q&=M_C/A_C=\ker H_Z/\operatorname{im}H_X^{\intercal},
 \end{aligned}
 \label{eq:homological-rotor-reduction}
\end{equation}
with $Q$ the first homology of the integer complex with boundary maps
$\partial_2=H_X^{\intercal}$ and $\partial_1=H_Z$.
Since $M_C$ is free abelian, $D_{\mathrm{red}}$ is a lattice
exactly when $Q$ is finite.
Then $K=|Q|$, and the codewords are the
combs of Eq.~\eqref{eq:rotor-css-general-codewords} on $M_C$
with $A=A_C$ and representatives in $M_C/A_C$.
A nonzero free part of $Q$ describes logical rotors and prevents
co-compactness.  The torsion encoding logical qudits therefore lies in
$Q$, whereas the lattice $D_{\mathrm{red}}$ in
Eq.~\eqref{eq:homological-rotor-reduction} is torsion free.

\textit{A two-rotor instance.}
The two-rotor code behind the $0$--$\pi$
qubit~\cite[Sec.~5.3]{vuillotHomologicalQuantumRotor2024} illustrates
the finite case.
It has $H_X=(2\;\;{-2})$ and $H_Z=(1\;\;1)$, so the modulation
constraint imposes $n_1+n_2=0$.
Thus $M_C=\{(n,-n):n\in\Z\}\cong\Z$, and the remaining stabilizer
$\mathbf T_2\otimes\mathbf T_{-2}$ shifts along $M_C$ by two units.
Hence $A_C=\langle(2,-2)\rangle$ and $Q=M_C/A_C\cong\Z_2$.
In the coordinate $n$ of $M_C$, the reduced lattice $2\Z\times\{0\}$ has
the stabilizer generator matrix
\begin{equation}
 T=\begin{pmatrix}2&0\\0&1\end{pmatrix}
 \;\begin{array}{@{}l@{\;}l@{}}
 n&\rowpairbrace{rotor}\\
 \theta&
 \end{array}\qquad\quad
 \varTheta=\begin{pmatrix}0&2\\-2&0\end{pmatrix},
 \label{eq:homological-rotor-zero-pi}
\end{equation}
so $K=2$ and the code encodes a qubit.
Its second column is the angle period, which the reduced lattice needs
explicitly.
The ideal codewords are
$\sum_{m\in\Z}\ket{2m}_{\mathrm{rot}}\otimes\ket{-2m}_{\mathrm{rot}}$ and
$\sum_{m\in\Z}\ket{2m+1}_{\mathrm{rot}}\otimes\ket{-2m-1}_{\mathrm{rot}}$.
The logical operators $\xl=\mathbf T_{-1}\otimes\mathbf T_1$ and
$\zl=(-1)^{\widehat{\mathbf L}_1}$ exchange and distinguish these combs.
The single-rotor code of the real projective
plane~\cite{vuillotHomologicalQuantumRotor2024}, with $H_X=2$ and
$H_Z=0$, needs no reduction.
It is the lattice $D=2\Z\times\{0\}$ of Sec.~\ref{ex:rotor-gkp} with
$K=2$ and $N=1$.

\textit{Modular constraints.}
Our framework also allows the finite modulation subgroups excluded in
Ref.~\cite{vuillotHomologicalQuantumRotor2024}.
For instance, the stabilizer $(-1)^{\widehat{\mathbf L}_1}$ imposes the
parity constraint that the first angular momentum is even.
Let $B_{\mathrm{mod}}\subseteq A_C^\perp$ be a finite group of such phases,
and let $B_{\mathrm{mod}}|_{M_C}$ denote their restrictions to $M_C$.
Then $D_{\mathrm{red}}=A_C\times B_{\mathrm{mod}}|_{M_C}$ is a split lattice
of the form~\eqref{eq:rotor-css-general} on $M_C$ whenever $Q$ is finite.
Equation~\eqref{eq:rotor-css-logical-group} then gives
$K=|(\ker H_Z\cap B_{\mathrm{mod}}^\perp)/\operatorname{im}H_X^{\intercal}|$.
Bosonic tiger codes~\cite{xu2024lettingtigercagebosonic} are built from
the same pairs $(H_X,H_Z)$, with $H_Z$ constraining occupation numbers
instead of angular momenta.
The parity-type number constraints added in
Ref.~\cite{poirson2026algebraic}, such as that of the four-legged cat
code~\cite{leghtas2013hardware,mirrahimi2014dynamically}, are the bosonic
analogue of $B_{\mathrm{mod}}$.

\textit{Oscillator realizations.}
Reference~\cite{magdalena2024topological} realizes homological rotor
codes on oscillators by condensing bosons in an $\R$ gauge theory.
Each oscillator carries one GKP stabilizer, which encodes a rotor as in
Remark~\ref{rem:full_lift_gkp_concatenation}.

\subsection{Hybrid stabilizer codes over two factor types}
\label{subsec:examples_two}

We now couple two types of systems through one stabilizer lattice.
The first example is known, and the other three are new.

\subsubsection{LCA codes}
\label{ex:lca}
The LCA codes of Ref.~\cite{yr64-ypxt} combine $p$ oscillators with
qudits, so that $M=\R^p\times\Z_{\dia c}$.
Their stabilizer lattices are graph lattices: each is the graph of a
homomorphism from a full-rank lattice in $\R^{2p}$ to the finite phase
space of the qudits (Appendix~\ref{app:graph_lattices}).
For these lattices $D\cong\Z^{2p}$, so their syndrome spaces are tori
$\T^{2p}$ (Sec.~\ref{sec:LSC_ncg}).
When the finite factor is absent, LCA codes reduce to GKP
codes~\cite{10.1103/physreva.64.012310}.

The simplest genuinely hybrid example has
$M=\R\times\Z_c$~\cite[Appendix~F(2)]{yr64-ypxt}.
For a much broader range of examples, see~\cite{yr64-ypxt}.
Fix integers $c>0$, $d$, and $\theta$ with $\gcd(c,d)=1$ and
$c\theta+d\neq0$.
Here $\theta$ and $d$ are integer parameters, not a rotor angle or a qudit
dimension.
The code has the two stabilizer generators
\begin{equation}
 \mathbf S_1=\mathbf D(1,0\,|\,{-d},0)=\mathbf T_1\otimes X^{d},
 \qquad
 \mathbf S_2=\mathbf D(0,\theta+\tfrac dc\,|\,0,1)
 =e^{2\pi i(\theta+d/c)\hat x}\otimes Z,
 \label{eq:examples-lca-stabilizers}
\end{equation}
which act as $(\mathbf S_1\psi)(x,j)=\psi(x+1,j-d)$ and
$(\mathbf S_2\psi)(x,j)=e^{2\pi i[(\theta+d/c)x+j/c]}\psi(x,j)$.
Their labels are the columns of the graph-lattice embedding of
Ref.~\cite{yr64-ypxt}, with oscillator part
$T_{\cv}=\operatorname{diag}(1,\theta+d/c)$ and qudit part
$T_{\mathrm{fin}}=\operatorname{diag}(-d,1)$.
Any other oscillator part with the same Gram matrix
$T_{\cv}^{\intercal}J_{\mathrm{osc}}T_{\cv}$ also works, possibly with
the stabilizer phases of Lemma~\ref{lem:phase_triv}.
It gives the same Gram matrix and code dimension below.
Adjoining the two qudit periods $c\,e_a$ and $c\,e_b$ gives the square
stabilizer generator matrix
\begin{equation}
 T=\begin{pmatrix}T_{\cv}&0\\T_{\mathrm{fin}}&c\,I_2\end{pmatrix}
 =\begin{pmatrix}1&0&0&0\\0&\theta+\tfrac dc&0&0\\
 -d&0&c&0\\0&1&0&c\end{pmatrix}
 \;\begin{array}{@{}l@{\;}l@{}}
 x&\rowpairbrace{oscillator}\\
 \xi\vphantom{\tfrac dc}&\\
 a&\rowpairbrace{qudit.}\\
 b&
 \end{array}
 \label{eq:examples-lca-T}
\end{equation}
Its Gram matrix is
\begin{equation}
 \varTheta=\begin{pmatrix}0&\theta&0&-d\\-\theta&0&-1&0\\
 0&1&0&c\\d&0&-c&0\end{pmatrix},
 \qquad
 \operatorname{Pf}\varTheta=c\theta+d,
 \label{eq:examples-lca-Theta}
\end{equation}
so a direct computation gives $K=|c\theta+d|$.
Its upper-left block is the Gram matrix of the graph lattice in
Ref.~\cite{yr64-ypxt}.
The Pfaffian $\theta$ of this block generally differs from $K$, so the
qudit period columns cannot be dropped.
The lattice is CSS, with shift part generated by $(1,-d)\in\R\times\Z_c$
and phase part generated by $(\theta+d/c,1)$.
Its ideal codewords are
\begin{equation}
 \ket{\ell_L}=\sum_{s\in\ell+K\Z}
 \ket{s/(c\theta+d)}_{\mathrm{osc}}\otimes\ket{-s}_{\mathrm{fin}},
 \qquad \ell=0,\ldots,K-1,
 \label{eq:examples-lca-codewords}
\end{equation}
with the qudit label read modulo $c$.
The phase $\mathbf S_2$ is trivial on their supports, $\mathbf S_1$
reindexes the sum by $s\mapsto s-(c\theta+d)$, and the $K$ supports are
disjoint.
The logical matrix yields the logical operators
$\xl=\mathbf D(-\tfrac1{c\theta+d},0\,|\,1,0)$ and
$\zl=\mathbf D(0,1\,|\,0,0)=e^{2\pi i\hat x}$.
The first shifts $s$ by one and hence maps $\ket{\ell_L}$ to
$\ket{(\ell+1)_L}$.
The second multiplies $\ket{\ell_L}$ by $e^{2\pi i\ell/(c\theta+d)}$.
As in the purely continuous case, the symplectic dual $D^\perp$
determines the logical displacement operators.
Thus the graph-lattice construction extends ordinary GKP lattices to codes
with coupled oscillator and finite-dimensional degrees of freedom.

The lattice also has hybrid Clifford symmetries.
For example, let $d\neq0$, and let $W\in\mathrm{SL}(2,\Z)$ have a
lower-left entry divisible by $d$.
Then $W$ yields a symplectic map that acts separately on the oscillator
and on the qudit and preserves $D$.
Its implementer is a product of an oscillator Gaussian unitary and a
qudit Clifford operator, followed by a displacement correction
(Sec.~\ref{subsec:clifford_examples}).
Appendix~\ref{app:graph_lattices} reviews the graph-lattice construction
of general LCA codes, and Theorem~\ref{thm:main_dual} there describes
their symplectic duals~\cite{yr64-ypxt}.

\subsubsection{Torsion and intrinsic coupling}
\label{subsec:examples_coupling}
The next three examples are new.
The first couples twelve oscillators to twelve qubits, and the other two
couple a rotor to another system.
Their stabilizer lattices can have torsion, i.e., nonzero displacements of
finite order.
The qubit stabilizer $Z^{\otimes12}$ of Sec.~\ref{ex:hybrid-golay} and
the rotor phase $e^{2\pi i\widehat{\mathbf L}/d}$ of
Sec.~\ref{ex:rotor-qudit} are such displacements.
In the oscillator--rotor family below, torsion appears when its integer
parameter $b$ is at least two.
The LCA construction of Ref.~\cite{yr64-ypxt} excludes lattices with
torsion, since its lattices are torsion-free and project isomorphically
onto oscillator phase space.
Lattices on systems with rotors lie outside that construction for a
second reason.
With $q\geq1$ rotors, the free part of every lattice in $G$ has rank
$2p+q$, which exceeds the rank $2p$ of a lattice in oscillator phase space.
To our knowledge, no earlier work constructs oscillator--rotor or
rotor--qudit stabilizer codes.
The normalizer formalism of Ref.~\cite{bermejo2016normalizer} admits
rotor--qudit stabilizer groups, but for infinite groups it treats
stabilizer states rather than codes.

Hybrid codes raise a structural question: is the coupling intrinsic, or
can a change of frame remove it?
Generalized Clifford operators implement the symplectic automorphisms of
$G$ (Proposition~\ref{prop:symplectic_implementer}).
These are the continuous automorphisms of $G$ that preserve the
commutation phases of Eq.~\eqref{eq:examples-commutation}.
They form the group $\operatorname{Sp}(G,\rho)$ of
Eq.~\eqref{eq:ambient_symplectic_automorphism}, where $\rho$ denotes the
commutation phase.
Two lattices are \emph{symplectically equivalent} if some
$S\in\operatorname{Sp}(G,\rho)$ maps one onto the other.
For a phase space $G=G_1\oplus G_2$ of two systems, an isotropic lattice
$D$ is \emph{symplectically decomposable} if some
$S\in\operatorname{Sp}(G,\rho)$ maps $D$ onto a direct sum $D_1\oplus D_2$
with $D_i\subseteq G_i$.
Such a code can be transformed into independent codes on the two systems
by a symplectic change of coordinates.
Otherwise, $D$ is \emph{indecomposable}, and we say that it couples the
two systems \emph{intrinsically}.

For oscillators and qudits, the question is trivial.
Continuous homomorphisms between $\R^{2p}$ and a finite group vanish, so
every $S\in\operatorname{Sp}(G,\rho)$ acts separately on the two factors.
An oscillator--qudit lattice is therefore decomposable only if it is
already a direct sum.
The torsion-free lattice of Sec.~\ref{ex:lca} is not a direct sum
when $c\geq2$.
In a direct sum, its qudit part would be a finite subgroup of a
torsion-free lattice, hence trivial.
The qudit would then be unencoded, which forces $c\mid K$.
However, $K=|c\theta+d|$ is coprime to $c$.
A similar count shows that the hybrid Golay code below, whose lattice has
torsion, is not a direct sum either.
The question becomes nontrivial when symplectic maps mix the factors.
For a rotor and a qudit of dimension $d$, the controlled shift
$\ket{n}_{\mathrm{rot}}\otimes\ket{j}_{\mathrm{fin}}\mapsto
\ket{n}_{\mathrm{rot}}\otimes\ket{j+n}_{\mathrm{fin}}$ implements such a
map.
On labels, it acts as
$(n,\theta\,|\,a,b)\mapsto(n,\theta-b/d\,|\,a+n,b)$
(Remark~\ref{rem:rotor-qudit-symplectic-maps}).
For qubits, Ref.~\cite{luo2026hybridqubitrotor} classifies all symplectic
automorphisms of this phase space.
In particular, angular-momentum parity can control qubit Paulis, while no
qubit-controlled shift of the angular momentum is Clifford.
For an oscillator and a rotor, the shears of
Eq.~\eqref{eq:oscillator-rotor-shear} in
Appendix~\ref{appx:proof_oscillator_rotor_normal_form} play the same role.
The oscillator--rotor and rotor--qudit examples below nevertheless
couple their systems intrinsically.

\subsubsection{Hybrid Golay code}
\label{ex:hybrid-golay}
Our first new two-type example encodes one logical qubit and couples
twelve oscillators intrinsically to twelve qubits.
It combines the binary and quaternary Golay
codes~\cite{golay1949notes,bonnecaze1995quaternary}.
Its stabilizer lattice has torsion, so the code lies outside the LCA
construction of Sec.~\ref{ex:lca}.
We first introduce the underlying classical objects and then build the
quantum code from them.

\textit{Partially lifted Golay codes.}
The extended Golay code $\mathcal G_{24}\subset\Z_2^{24}$ is a
self-dual binary code of length $24$, dimension $12$, and minimum
distance $8$~\cite{golay1949notes,conway2013sphere}.
Its codewords of weight $12$ are called dodecads, and the complement of
a dodecad is again a dodecad.
Fix a dodecad $I$ and its complement $J$, and write each codeword as
$c=(c_I,c_J)$ with $c_I,c_J\in\Z_2^{12}$.
We define the partially lifted binary Golay code
\begin{equation}
 \mathcal P:=\Bigl\{\bigl(\tfrac{m}{\sqrt2},\,u\bigr)\;:\;
 m\in\Z^{12},\ u\in\Z_2^{12},\ (m\bmod2,\,u)\in\mathcal G_{24}\Bigr\}
 \subset\R^{12}\times\Z_2^{12}.
 \label{eq:golay-partial-lift}
\end{equation}
It applies Construction~A~\cite{conway2013sphere} to the coordinates in
$I$ and keeps the coordinates in $J$ binary.
Lifting all $24$ coordinates instead gives Construction~A of
$\mathcal G_{24}$, which after rescaling by $1/\sqrt2$ is the Niemeier
lattice $N(A_1^{24})$ with root system $A_1^{24}$~\cite{conway2013sphere}.
We pair elements of $\R^{12}\times\Z_2^{12}$ through
$\langle(x,u),(y,u')\rangle:=x\cdot y+\tfrac12u\cdot u'$ modulo one.
Because $\mathcal G_{24}$ is self-dual, $\mathcal P$ is self-dual under
this pairing, i.e., it consists exactly of the vectors that pair
integrally with all of $\mathcal P$.
Two further facts about $\mathcal G_{24}$ shape $\mathcal P$.
First, the restrictions $c_I$ form the even-weight code of length $12$,
so $\sum_im_i$ is even on $\mathcal P$.
Second, the only codewords supported in $J$ are $0$ and the dodecad $J$
itself.

The quaternary Golay code $\mathcal Q_{24}\subset\Z_4^{24}$ is the
self-dual Hensel lift of $\mathcal G_{24}$ to $\Z_4$~\cite{bonnecaze1995quaternary}.
Its codewords reduce modulo two to codewords of $\mathcal G_{24}$.
Their Euclidean weights, i.e., the sums of the squares of their entries
taken in $\{0,\pm1,2\}$, are multiples of eight.
Construction~$A_4$ turns $\mathcal Q_{24}$ into the lattice
$A_4(\mathcal Q_{24})\subset\Z^{24}$ of integer vectors that reduce modulo
four to its codewords.
The rescaled lattice $\Lambda_{24}=A_4(\mathcal Q_{24})/2$ is the Leech
lattice~\cite{bonnecaze1995quaternary}.
We define the partially lifted quaternary Golay code
\begin{equation}
 \mathcal P_0:=\Bigl\{\bigl(\tfrac{m}{\sqrt2},\,u\bigr)\;:\;
 k_I\equiv m\ (\mathrm{mod}\ 4)\ \text{and}\ k_J\equiv u\
 (\mathrm{mod}\ 2)\ \text{for some}\ k\in\mathcal Q_{24}\Bigr\}.
 \label{eq:golay-partial-lift-quaternary}
\end{equation}
It applies Construction~$A_4$ to the coordinates in $I$ and reduces
those in $J$ modulo two.
Equivalently, $\mathcal P_0$ and $\mathcal P$ are the images of
$\Lambda_{24}$ and $N(A_1^{24})$ under
$v\mapsto(\sqrt2v_I,\,2v_J\bmod2)$ and
$w\mapsto(w_I,\,\sqrt2w_J\bmod2)$, respectively.
Since $\mathcal Q_{24}$ reduces to $\mathcal G_{24}$, $\mathcal P_0$ is a
sublattice of $\mathcal P$, and its index is two.
Explicitly, fix $j\in J$, and write $e_j$ for the unit vector at $j$
(similarly $e_i$ for $i\in I$).
Then $\mathcal Q_{24}$ contains a codeword $(s,2e_j)$ with
$s\in\{\pm1\}^{12}$, and
\begin{equation}
 \mathcal P_0=\Bigl\{\bigl(\tfrac m{\sqrt2},u\bigr)\in\mathcal P\;:\;
 \textstyle\sum_is_im_i\equiv2u_j\ (\mathrm{mod}\ 4)\Bigr\}.
 \label{eq:golay-twisted-cut}
\end{equation}
Every element of $\mathcal P_0$ satisfies this congruence, because the
codeword $(s,2e_j)$ is orthogonal to all of $\mathcal Q_{24}$ by self-duality.
By the first fact above, both sides of the congruence are even.
Every self-dual $\Z_4$ lift of $\mathcal G_{24}$ whose Euclidean weights
are multiples of eight gives the same $\mathcal P_0$ up to sign flips
$m_i\mapsto-m_i$.

\textit{The code.}
In the notation of Sec.~\ref{subsec:examples_notation}, take
$M=\R^{12}\times\Z_2^{12}$, with the oscillators labeled by $I$ and the
qubits labeled by $J$.
The hybrid Golay code is the CSS lattice $D:=\mathcal P_0\times\mathcal P$,
whose shift part is $A=\mathcal P_0$ and whose phase part is $B=\mathcal P$.
Its stabilizer group is generated by
\begin{equation}
 \mathbf T_{m/\sqrt2}\otimes X^{u}\ \ \text{for }
 \bigl(\tfrac m{\sqrt2},u\bigr)\in\mathcal P_0,
 \qquad
 \mathbf M_{m/\sqrt2}\otimes Z^{u}\ \ \text{for }
 \bigl(\tfrac m{\sqrt2},u\bigr)\in\mathcal P,
 \label{eq:golay-stabilizers}
\end{equation}
where $X^u=\bigotimes_jX^{u_j}$ and similarly for $Z^u$.
Since $X^{-1}=X$ on a qubit, $X^u$ is the qubit shift $\mathbf T_u$ of
Sec.~\ref{subsec:examples_notation}.
A shift with labels $(m/\sqrt2,u)$ and a phase with labels
$(m'/\sqrt2,u')$ commute up to the sign $(-1)^{m\cdot m'+u\cdot u'}$.
This sign is $+1$, because $(m\bmod2,u)$ and $(m'\bmod2,u')$ are
orthogonal codewords of $\mathcal G_{24}$.
Thus $D$ is isotropic, and as a CSS lattice it needs no stabilizer phases.
By the second fact above, the only purely qubit stabilizers are the
identity and $Z^{\otimes12}$, since the shift $X^{\otimes12}$ violates
Eq.~\eqref{eq:golay-twisted-cut}.
Hence $D\cong\Z^{24}\oplus\Z_2$, with torsion generated by $Z^{\otimes12}$.
The self-duality of $\mathcal P$ gives $B^\perp=\mathcal P$, so the count
of Sec.~\ref{subsec:examples_notation} yields
\begin{equation}
 K=|B^\perp/A|=[\mathcal P:\mathcal P_0]=2.
 \label{eq:golay-dimension}
\end{equation}
The full lift agrees: a stabilizer generator matrix has
$|\det T|=2^{13}$, and Eq.~\eqref{eq:examples-dimension} gives
$K=2^{13}/2^{12}=2$.
Using $\mathcal P$ for the shifts as well would instead give a stabilizer
state.

\textit{Ideal logical basis.}
The two cosets of $\mathcal P_0$ in $\mathcal P$ give the ideal codewords
\begin{equation}
 \ket{\ell_L}=\sum_{\substack{(m/\sqrt2,\,u)\in\mathcal P\\
 \sum_is_im_i\equiv2u_j+2\ell\ (\mathrm{mod}\ 4)}}
 \ket{m/\sqrt2}_{\mathrm{osc}}\otimes\ket u_{\mathrm{fin}},
 \qquad \ell\in\{0,1\},
 \label{eq:golay-codewords}
\end{equation}
which are the combs of Eq.~\eqref{eq:examples-css-codewords}.
Flipping all twelve qubits exchanges them, so
$\ket{1_L}=X^{\otimes12}\ket{0_L}$.
Like ideal GKP states, they are not normalizable.

\textit{Logical Pauli operators.}
The logical $X$ can act on the qubits alone or on a single oscillator
alone,
\begin{equation}
 \overline X:=X^{\otimes12}\simeq\mathbf T_{\sqrt2e_i},
 \qquad
 \overline Z:=\mathbf M_{s/(2\sqrt2)}\otimes Z_j.
 \label{eq:golay-logicals}
\end{equation}
Here $\mathbf T_{\sqrt2e_i}\otimes X^{\otimes12}$ is a stabilizer for every
oscillator $i$, so the two forms of $\overline X$ agree on the codespace.
The phase $\overline Z$ multiplies $\ket{\ell_L}$ by
$i^{s\cdot m+2u_j}=(-1)^\ell$.
Both operators commute with all stabilizers and anticommute with each
other.
Every representative of $\overline Z$ acts on all twelve oscillators, and
its phase part acts on an odd number of qubits.

\textit{Symmetry.}
Every automorphism of $\mathcal G_{24}$ that fixes $I$ setwise permutes
the oscillators and the qubits separately and preserves $\mathcal P$.
These permutations form the Mathieu group $M_{12}$~\cite{conway2013sphere}.
Each of them preserves $\mathcal P_0$ when combined with parity flips
$x_i\mapsto-x_i$ of suitable oscillators, unique up to the parity of all
twelve oscillators.
The symmetries of $D$ of this signed-permutation type therefore form a
double cover of $M_{12}$, whose center is generated by that global parity.
On the oscillators, this cover acts as the monomial automorphism group of
an extended ternary Golay code.
As $j$ ranges over $J$, the sign vectors $s$ of the codewords
$(s,2e_j)\in\mathcal Q_{24}$ span this code over $\Z_3$.
Its monomial automorphism group is the nonsplit double cover
$2.M_{12}$~\cite{conway2013sphere}, so no choice of oscillator signs lets
all of $M_{12}$ act by permutations alone.
The symmetries fix both codewords~\eqref{eq:golay-codewords} and hence
act trivially on the logical qubit.

\textit{Intrinsic coupling.}
Generalized Clifford operators act separately on oscillators and qudits
(Sec.~\ref{subsec:examples_coupling}).
Hence the code decomposes only if $D$ is already the direct sum of its
purely oscillator part and its purely qubit part.
The latter is generated by $Z^{\otimes12}$, so a separate qubit factor
would encode eleven logical qubits.
This would force $2^{11}\mid K$, which contradicts $K=2$.
The coupling between the oscillators and the qubits is therefore
intrinsic.

\textit{Concatenation.}
Each oscillator nevertheless carries the local stabilizers
$\mathbf T_{2\sqrt2}$ and $\mathbf M_{\sqrt2}$, with labels
$(2\sqrt2e_i,0)\in\mathcal P_0$ and $(\sqrt2e_i,0)\in\mathcal P$.
They generate the lattice $2\sqrt2\,\Z\times\sqrt2\,\Z$ of the rectangular
GKP ququart in Eq.~\eqref{eq:rotor-parent-gkp-lattice}, with
$\alpha=1/\sqrt2$.
The hybrid Golay code is therefore a concatenated code.
Its inner codes are these ququarts, and its outer code is a CSS code on
twelve ququarts and twelve qubits built from both Golay codes.
The outer $X$-type checks are the codewords of $\mathcal Q_{24}$ with their
$J$-coordinates reduced modulo two.
The outer $Z$-type checks are the codewords of $\mathcal G_{24}$ with their
$I$-coordinates doubled.
Encoding each qubit in the square GKP code of
Remark~\ref{rem:full_lift_gkp_concatenation} gives a $24$-mode GKP code
with $K=2$.
Its phase lattice is $N(A_1^{24})$, and its shift lattice contains
$\sqrt2\,\Lambda_{24}$ with index $2^{11}$.
The logical operators $\mathbf T_{\sqrt2e_i}$ and $\overline Z$ then
have labels $v/\sqrt2$ with $v$ a minimal vector of $\Lambda_{24}$.

\subsubsection{Oscillator--rotor codes}
\label{ex:osc-rotor}
For simplicity, we consider only one oscillator and one rotor.
We show that they can couple intrinsically, but only in logical dimension
four or higher.
Let
\begin{equation}
 M=\R\times\Z,
 \qquad \widehat M=\R\times\T,
 \qquad \mathcal H\cong L^2(\R)\otimes\ell^2(\Z).
 \label{eq:nondecoupling-configuration-space}
\end{equation}
Write $z=(x,n;\xi,\theta)\in G=M\times\widehat M$, with
$x,\xi\in\R$, $n\in\Z$, and $\theta\in\T=\R/\Z$.
These tuples list both positions before both momenta.
In contrast, the rows of $T$ below use the paired order
$(x,\xi\,|\,n,\theta)$ of Sec.~\ref{subsec:examples_notation}.
In this system, the decomposition defined above asks for some
$S\in\operatorname{Sp}(G,\rho)$ with
\begin{equation}
 S(D)=D_{\mathrm{osc}}\oplus D_{\mathrm{rot}},
 \quad D_{\mathrm{osc}}\subset\R^2,
 \quad D_{\mathrm{rot}}\subset\Z\times\T.
 \label{eq:oscillator-rotor-decoupling}
\end{equation}

\textit{When such codes decompose.}
Theorem~\ref{thm:squarefree-oscillator-rotor-decoupling} below shows that
decomposition is the rule for small or square-free dimensions.
A symplectic change of frame brings every isotropic lattice to a normal
form with $K=mN$ (Lemma~\ref{lem:oscillator-rotor-normal-form}).
Here $N$ is the covolume of an oscillator lattice, and $m=r/t$ is the
angular-momentum spacing $r$ of $D$ divided by its torsion order $t$.
If $\gcd(m,N)=1$, then $D$ decomposes
(Lemma~\ref{lem:oscillator-rotor-coprime-splitting}).
Consequently, every lattice with square-free $K$ decomposes, including all
stabilizer states ($K=1$) and all codes with $K=2$ or $K=3$.

\textit{A minimal indecomposable ququart.}
We now attain the lower bound $K=4$ in
Theorem~\ref{thm:squarefree-oscillator-rotor-decoupling}.  In the phase
space~\eqref{eq:nondecoupling-configuration-space}, fix $\alpha>0$ and
$b\in\Z_{\geq1}$, and set
\begin{equation}
 \begin{gathered}
 w_1=(2\alpha,0;0,0),\qquad w_2=(0,0;\alpha^{-1},0),\\
 w_3=(\alpha,2b;0,0),\qquad w_4=(0,0;0,b^{-1}),\\
 D_{\alpha,b}:=\langle w_1,w_2,w_3,w_4\rangle.
 \end{gathered}
 \label{eq:nondecoupling-stabilizer-lattice}
\end{equation}
Here the parameter $b$ is a positive integer, not a qudit phase.
The shift part $\langle w_1,w_3\rangle$ and the phase part
$\langle w_2,w_4\rangle$ of $D_{\alpha,b}$ commute.
Thus $D_{\alpha,b}$ is an isotropic CSS lattice, which needs no
stabilizer phases.
Writing $\mathbf S_i=\mathbf D(w_i)$, the convention
$\mathbf D=\mathbf M\mathbf T$ gives
\begin{subequations}
\label{eq:nondecoupling-stabilizer-actions}
\begin{align}
 (\mathbf S_1\psi)(x,n)&=\psi(x+2\alpha,n),\\
 (\mathbf S_2\psi)(x,n)&=e^{2\pi ix/\alpha}\psi(x,n),\\
 (\mathbf S_3\psi)(x,n)&=\psi(x+\alpha,n+2b),\\
 (\mathbf S_4\psi)(x,n)&=e^{2\pi in/b}\psi(x,n).
\end{align}
\end{subequations}
The third stabilizer simultaneously translates the oscillator and shifts
the rotor angular momentum.
In the paired row order, the stabilizer generator matrix and its Gram
matrix are
\begin{equation}
 T=\begin{pmatrix}2\alpha&0&\alpha&0\\0&\alpha^{-1}&0&0\\
 0&0&2b&0\\0&0&0&b^{-1}\end{pmatrix}
 \;\begin{array}{@{}l@{\;}l@{}}
 x&\rowpairbrace{oscillator}\\
 \xi\vphantom{\alpha^{-1}}&\\
 n&\rowpairbrace{rotor}\\
 \theta\vphantom{b^{-1}}&
 \end{array}\qquad\quad
 \varTheta=\begin{pmatrix}0&2&0&0\\-2&0&-1&0\\0&1&0&2\\0&0&-2&0
 \end{pmatrix}.
 \label{eq:nondecoupling-T}
\end{equation}
The columns of $T$ are $w_1,\ldots,w_4$.
No period columns are needed, because the angle period $(0,0;0,1)=b\,w_4$
is already generated by the columns.
For all $\alpha$ and $b$, $\operatorname{Pf}\varTheta=4$, so
Eq.~\eqref{eq:examples-dimension} gives
\begin{equation}
 \dim\mathcal C_{\alpha,b}=4.
 \label{eq:nondecoupling-code-dimension}
\end{equation}
A general lattice element is
\begin{equation}
 \bigl(\alpha(2m_1+m_3),2bm_3;m_2/\alpha,j/b\bigr),
 \quad m_1,m_2,m_3\in\Z,\quad j\in\Z_b.
 \label{eq:nondecoupling-general-lattice-element}
\end{equation}
Consequently $D_{\alpha,b}\cong\Z^3\oplus\Z_b$.
Its torsion has order $b$, and the angular momenta of its elements form
$2b\Z$.

\textit{Ideal codewords.}
An ideal basis is
\begin{equation}
 \ket{j_L}\propto\sum_{k,\ell\in\Z}
 \ket{\alpha k}_{\mathrm{osc}}
 \otimes\ket{b(2k+j+4\ell)}_{\mathrm{rot}},
 \quad j\in\Z_4.
 \label{eq:nondecoupling-logical-codewords}
\end{equation}
The phase stabilizers act trivially on these supports, and the
translation stabilizers reindex the sums.  The four supports are
disjoint, so the four states are independent and exhaust the dimension
in Eq.~\eqref{eq:nondecoupling-code-dimension}.
They are non-normalizable combs, and the regularization of
Sec.~\ref{sec:finite_energy} gives finite-energy representatives.
The exact statements below concern the ideal code.
For $b\geq2$, the phase stabilizer $\mathbf S_4$ has order $b$ and
generates the torsion of $D_{\alpha,b}$.
Torsion is not needed for the obstruction below.
For $b=1$, the stabilizer $\mathbf S_4$ is the identity and
$D_{\alpha,1}\cong\Z^3$ is torsion free, yet
Proposition~\ref{prop:oscillator-rotor-nondecoupling} still applies.

\textit{Symplectic dual and logical ququart.}
Pairing with the four generators gives
\begin{equation}
 \begin{split}
 D_{\alpha,b}^{\perp}=\bigl\{(x,n;\xi,\theta):\ &x/\alpha\in\Z,
       \ n/b\in\Z,\\
       &2\alpha\xi\in\Z,\quad
       \alpha\xi+2b\theta=0\text{ in }\T\bigr\}.
 \end{split}
 \label{eq:nondecoupling-symplectic-dual}
\end{equation}
Choose the logical displacement labels and operators
\begin{equation}
 \begin{gathered}
 z_X=(0,-b;0,0),\qquad
 z_Z=(0,0;-\tfrac1{2\alpha},\tfrac1{4b}),\\
 \overline X_4:=\mathbf D(z_X),\qquad
 \overline Z_4:=\mathbf D(z_Z).
 \end{gathered}
 \label{eq:nondecoupling-logical-ququart-operators}
\end{equation}
Both labels lie in $D_{\alpha,b}^{\perp}$.
The logical matrix $T_{\mathrm{dual}}=T\varTheta^{-1}$ contains $z_X$ and
$-z_Z$ among its columns.
Their actions are
\begin{equation}
 \overline X_4\ket{j_L}=\ket{(j+1)_L},\qquad
 \overline Z_4\ket{j_L}=i^j\ket{j_L},\qquad j\in\Z_4,
 \label{eq:nondecoupling-logical-ququart-action}
\end{equation}
where the indices are modulo four.  In particular,
\begin{equation}
 \begin{gathered}
 \overline Z_4\overline X_4=i\,\overline X_4\overline Z_4,\\
 \overline X_4^{\,4}=\mathbf S_1\mathbf S_3^{-2},\qquad
 \overline Z_4^{\,4}=\mathbf S_2^{-2}\mathbf S_4.
 \end{gathered}
 \label{eq:nondecoupling-logical-ququart-relations}
\end{equation}
The classes $[z_X]$ and $[z_Z]$ have order four.
Since $\overline Z_4\overline X_4=i\,\overline X_4\overline Z_4$, their
classes generate a subgroup of order sixteen.
Equation~\eqref{eq:examples-logical-group} gives the same order $K^2=16$
for the full quotient, so
\begin{equation}
 \mathsf L_{D_{\alpha,b}}
 =D_{\alpha,b}^{\perp}/D_{\alpha,b}
 \cong\Z_4\oplus\Z_4.
 \label{eq:nondecoupling-logical-group}
\end{equation}
Thus the code encodes a ququart rather than two qubits.

\begin{theorem}
\label{thm:squarefree-oscillator-rotor-decoupling}
Every isotropic lattice for one oscillator and one rotor whose logical
dimension $K$ is square-free is symplectically decomposable.
Equivalently, the logical dimension $K$ of an isotropic lattice in this
system that is not symplectically decomposable is divisible by the square
of a prime.
In particular, $K\geq4$.
\end{theorem}

\begin{proof}
By Lemma~\ref{lem:oscillator-rotor-normal-form}, $K=mN$ with
$m,N\in\Z_{>0}$.  A prime dividing both $m$ and $N$ would have its
square dividing $K$.  Square-freeness therefore implies
$\gcd(m,N)=1$, and
Lemma~\ref{lem:oscillator-rotor-coprime-splitting} applies.
\end{proof}

Appendix~\ref{appx:proof_oscillator_rotor_normal_form} states and proves
Lemma~\ref{lem:oscillator-rotor-normal-form}, and
Appendix~\ref{appx:proof_oscillator_rotor_coprime_splitting} does the same
for Lemma~\ref{lem:oscillator-rotor-coprime-splitting}.
An \emph{ideal lattice stabilizer state} is a nonzero vector, up to
overall scalar, spanning the ideal codespace $\mathcal C$ of an
isotropic lattice $D$ with $K=\dim\mathcal C=1$.  Thus it is the unique
generalized state, up to scalar, fixed by all stabilizers
(Eq.~\eqref{eq:trivial_fiber_codespace}).

\begin{corollary}
\label{cor:oscillator-rotor-stabilizer-states}
Every ideal lattice stabilizer state of one oscillator and one rotor is
symplectically decomposable.  More precisely, a generalized Clifford
operator maps it to a tensor product of an ideal oscillator
stabilizer state and an ideal rotor stabilizer state.  All lattice codes
with $K=2$ or $K=3$ are also symplectically decomposable.
\end{corollary}

\begin{proof}
This follows from Theorem~\ref{thm:squarefree-oscillator-rotor-decoupling},
since $K=1,2,3$ are square-free, together with the generalized Clifford
implementation of Proposition~\ref{prop:symplectic_implementer}.
\end{proof}

The conclusion concerns equivalence under a
global symplectic change of frame.  Finite-energy regularizations are
treated separately in Sec.~\ref{sec:finite_energy}.

\begin{proposition}
\label{prop:oscillator-rotor-nondecoupling}
For every $\alpha>0$ and $b\in\Z_{\geq1}$, the lattice
$D_{\alpha,b}$ of Sec.~\ref{ex:osc-rotor} is not symplectically
decomposable.  Thus four is the smallest logical dimension permitting
intrinsic coupling in this system.
\end{proposition}

The proof in Appendix~\ref{appx:proof_oscillator_rotor_nondecoupling}
compares logical groups.
A decomposed version of $D_{\alpha,b}$ would have the logical group
$(\Z_2)^4$, which has no element of order four, unlike $[z_X]$.
Moreover, no continuous symplectic automorphism preserving $D_{\alpha,b}$
implements a logical ququart Fourier gate
(Sec.~\ref{subsec:clifford_examples}).

\subsubsection{Rotor--qudit codes}
\label{ex:rotor-qudit}
This example couples a planar rotor intrinsically to a qudit of dimension
$d\geq2$, already at logical dimension $d$.
Reference~\cite{luo2026hybridqubitrotor} develops the Clifford theory and
universal control of hybrid qubit--rotor systems.
For one rotor and one qudit,
\begin{equation}
 M=\Z\times\Z_d,\qquad \widehat M=\T\times\Z_d,\qquad
 \mathcal H=\ell^2(\Z)\otimes\C^d .
 \label{eq:rq-config}
\end{equation}
As in Sec.~\ref{subsec:examples_notation}, we label a displacement by
$(n,\theta\,|\,a,b)$.
Its entries are a rotor angular-momentum shift $n\in\Z$, a rotor angle
$\theta\in\T$, and a qudit shift $a$ and phase $b$, read modulo $d$.
The stabilizer group is generated by
\begin{equation}
 \begin{aligned}
 \mathbf S_1&=\mathbf D(d,0\,|\,1,0)=\mathbf T_d\otimes X^{-1},\\
 \mathbf S_2&=\mathbf D(0,\tfrac1d\,|\,0,0)
   =e^{2\pi i\widehat{\mathbf L}/d}\otimes I ,
 \end{aligned}
 \label{eq:rq-stabilizers}
\end{equation}
that is, $(\mathbf S_1\psi)(n,j)=\psi(n+d,j+1)$ and
$(\mathbf S_2\psi)(n,j)=e^{2\pi i n/d}\psi(n,j)$.
The first generator shifts the rotor angular momentum by $d$ together
with the qudit.
The second restricts the angular momentum to multiples of $d$.
Their labels $w_1=(d,0|1,0)$ and $w_2=(0,\tfrac1d|0,0)$ have
$w_1^{\intercal}Jw_2=d\cdot\tfrac1d=1$, so the two generators commute.
The lattice $D=\langle w_1,w_2\rangle$ is CSS, $D=A\times B$ with
$A=\langle(d,1)\rangle\subset M$ and
$B=\langle(\tfrac1d,0)\rangle\subset\widehat M$.
Its bare displacements therefore already form the stabilizer group.
As a group, $D\cong\Z\oplus\Z_d$, with torsion generated by $w_2$.

Adjoining the qudit periods $d\,e_a=(0,0|d,0)$ and $d\,e_b=(0,0|0,d)$
gives the stabilizer generator matrix, since the rotor-angle period
$(0,1|0,0)=d\,w_2$ is already included:
\begin{equation}
 T=\Bigl(\underbrace{\,w_1\,}_{\mathbf S_1}\;\;
         \underbrace{\,w_2\,}_{\mathbf S_2}\;\;
         \underbrace{\,d\,e_a\;\;d\,e_b\,}_{\text{qudit periods}}\Bigr)
 =\begin{pmatrix} d&0&0&0\\ 0&\tfrac1d&0&0\\ 1&0&d&0\\ 0&0&0&d
 \end{pmatrix}
 \;\begin{array}{@{}l@{\;}l@{}}
 n&\rowpairbrace{rotor}\\
 \theta\vphantom{\tfrac1d}&\\
 a&\rowpairbrace{qudit.}\\
 b&
 \end{array}
 \label{eq:rq-T}
\end{equation}
The Gram matrix is integral,
\begin{equation}
 \varTheta=T^{\intercal}JT=
 \begin{pmatrix}0&1&0&1\\-1&0&0&0\\0&0&0&d\\-1&0&-d&0\end{pmatrix},
 \qquad
 \operatorname{Pf}\varTheta=d=\frac{|\det T|}{d},
 \label{eq:rq-Theta}
\end{equation}
so the code encodes a qudit, $K=d$.

The ideal codewords are
\begin{equation}
 \ket{\ell_L}=\sum_{m\in\Z}\ket{dm}_{\mathrm{rot}}
 \otimes\ket{m+\ell}_{\mathrm{fin}},
 \qquad \ell\in\Z_d ,
 \label{eq:rq-codewords}
\end{equation}
with the qudit label read modulo $d$.
The phase $\mathbf S_2$ is trivial on their supports, $\mathbf S_1$
reindexes the sum, and the $d$ supports are disjoint.
The logical matrix $T_{\mathrm{dual}}=T\varTheta^{-1}$ has the columns
$(0,\tfrac1d|0,0)$, $(-d,0|0,0)$, $(0,-\tfrac1{d^2}|0,1)$, and
$(0,0|{-1},0)$.
The first column is the stabilizer $w_2$, and
$D^\perp/D\cong\Z^4/\varTheta\Z^4\cong\Z_d\times\Z_d$.
Choosing
\begin{equation}
 \xl=\mathbf D(0,0\,|\,{-1},0)=I\otimes X,\qquad
 \zl=\mathbf D(0,-\tfrac1{d^2}\,|\,0,1)
    =e^{-2\pi i\widehat{\mathbf L}/d^2}\otimes Z
 \label{eq:rq-logicals}
\end{equation}
gives
\begin{equation}
 \xl\ket{\ell_L}=\ket{(\ell+1)_L},\qquad
 \zl\ket{\ell_L}=e^{2\pi i\ell/d}\ket{\ell_L},\qquad
 \zl\,\xl=e^{2\pi i/d}\,\xl\,\zl ,
 \label{eq:rq-logical-action}
\end{equation}
with $\xl^{\,d}=I$ and $\zl^{\,d}=\mathbf S_2^{-1}$.
The rotor translation $\mathbf T_{-d}$ also acts as $\xl^{-1}$.
A qudit-only displacement $\mathbf D(0,0\,|\,a,b)$ commutes with
$\mathbf S_1$ only if $b=0$, so a logical $Z$ cannot be supported on the
qudit alone.

Writing $m=j+dk$ shows that
$\ket{\ell_L}=\sum_{j\in\Z_d}\ket{j_L}_{\mathrm{rGKP}}
\otimes\ket{j+\ell}_{\mathrm{fin}}$, where
$\ket{j_L}_{\mathrm{rGKP}}=\sum_{k\in\Z}\ket{d(j+dk)}_{\mathrm{rot}}$
are the rotor GKP codewords of Eq.~\eqref{eq:rotor-css-codewords} with
$K=N=d$.
The code is thus a rotor GKP qudit concatenated into the two-qudit code
stabilized by a joint shift.
Indeed, $\mathbf S_1^{\,d}=\mathbf T_{d^2}$ and $\mathbf S_2$ are the inner
rotor GKP stabilizers, and $\mathbf S_1$ is the product of the inner
logical shift $\mathbf T_d$ with a physical qudit shift.
For $d=2$, the code is a logical qubit in a rotor and a qubit, with
$\mathbf S_1=\mathbf T_2\otimes X$ and
$\mathbf S_2=(-1)^{\widehat{\mathbf L}}$.
Its codewords are
\[
 \ket{0_L}=\sum_{k\in\Z}\bigl(\ket{4k}\ket0+\ket{4k+2}\ket1\bigr),\qquad
 \ket{1_L}=\sum_{k\in\Z}\bigl(\ket{4k}\ket1+\ket{4k+2}\ket0\bigr),
\]
and its logical operators are $\xl=I\otimes X$ and
$\zl=(-i)^{\widehat{\mathbf L}}\otimes Z$.
Proposition~\ref{prop:rotor-qudit-nondecoupling} below shows that no
symplectic automorphism splits $D$ into a rotor lattice and a qudit
lattice.

\begin{proposition}
\label{prop:rotor-qudit-nondecoupling}
For every $d\geq2$, the lattice $D=\langle w_1,w_2\rangle$ of
Sec.~\ref{ex:rotor-qudit} is not symplectically decomposable.
\end{proposition}

We prove Proposition~\ref{prop:rotor-qudit-nondecoupling} in
Appendix~\ref{appx:proof_rotor_qudit}, using
Lemmas~\ref{lem:rotor-qudit-block-structure}
and~\ref{lem:rotor-qudit-obstruction}.
The proof uses only the continuity of the automorphisms.
Intuitively, continuous couplings between a rotor and a qudit only see
the angular momentum modulo $d$.
This residue vanishes for the rotor shift by $d$ in $\mathbf S_1$, so no
such coupling removes the accompanying qudit shift.
The unitary $\ket n_{\mathrm{rot}}\ket j_{\mathrm{fin}}\mapsto
\ket{n-dj}_{\mathrm{rot}}\ket j_{\mathrm{fin}}$, with $j\in\{0,\ldots,d-1\}$,
does map the code to a rotor GKP code times a fixed qudit state.
However, it shifts the angular momentum conditioned on the qudit.
No generalized Clifford operator does so
(Lemma~\ref{lem:rotor-qudit-block-structure} and, for qubits,
Ref.~\cite{luo2026hybridqubitrotor}).
A logical Fourier gate on this code likewise requires non-Clifford
operations (Sec.~\ref{subsec:clifford_examples}).

\begin{corollary}
\label{cor:rotor-qudit-state}
For every $d\geq2$, let $z=(0,-\tfrac1{d^2}|0,1)$ be the label of the
logical operator $\zl$ of Sec.~\ref{ex:rotor-qudit}.
The lattice $D'=\langle w_1,z\rangle$, obtained by adding $\zl$ to the
stabilizers, is isotropic with $K=1$.
It is not symplectically decomposable, and no generalized
Clifford operator maps its ideal stabilizer state
\begin{equation}
 \ket{\psi_d}=\sum_{m\in\Z}\ket{dm}_{\mathrm{rot}}\otimes\ket{m}_{\mathrm{fin}}
 \label{eq:rq-state}
\end{equation}
to a product $u_{\mathrm{rot}}\otimes u_{\mathrm{fin}}$.
\end{corollary}

Appendix~\ref{appx:proof_rotor_qudit} proves
Corollary~\ref{cor:rotor-qudit-state} with the same two lemmas.
The state $\ket{\psi_d}$ is a stabilizer state in the sense of
Ref.~\cite{bermejo2016normalizer}.
What is new, to our knowledge, is that no generalized Clifford operator
disentangles it.
Adding $\xl$ instead of $\zl$ gives
$\langle(d,0|0,0),(0,\tfrac1d|0,0),(0,0|1,0)\rangle$, which is already
the product of the rotor state $\sum_k\ket{dk}_{\mathrm{rot}}$ and the
qudit state $\sum_j\ket{j}_{\mathrm{fin}}$.

These results contrast with those of Sec.~\ref{ex:osc-rotor}.
For one oscillator and one rotor, square-free $K$ and every stabilizer
state decompose, so intrinsic coupling needs $K\geq4$.
For one rotor and one qudit, intrinsic coupling already occurs at $K=d$,
e.g., at $K=2$ for a rotor and a qubit.
By Corollary~\ref{cor:rotor-qudit-state}, it even occurs at $K=1$.
The factor $d$ in $\mathbf S_1$ is essential.
The lattice $\langle\mathbf T_1\otimes X^{-1}\rangle$ has the stabilizer
generator matrix with columns $(1,0\,|\,1,0)$, $e_\theta$, $d\,e_a$, and
$d\,e_b$.
Its Gram matrix is again $\varTheta$, so again $K=d$.
Yet the rotor-controlled qudit shift
$(n,\theta\,|\,a,b)\mapsto(n,\theta+b/d\,|\,a-n,b)$ maps it to a rotor
stabilizer state times an unencoded qudit.
Torsion is not needed either, since the torsion-free lattice
$\langle\mathbf T_d\otimes X^{-1}\rangle$, with $K=d^2$, is not
decomposable.
Our proofs need only the coarser block structure of continuous
automorphisms in Lemma~\ref{lem:rotor-qudit-block-structure}, which holds
for every $d$, rather than the full qubit classification of
Ref.~\cite{luo2026hybridqubitrotor}.

\subsection{Hybrid stabilizer codes over several factor types}
\label{subsec:examples_several}

Our last three examples combine several types of systems.
We treat qudits of unequal dimensions as different types of systems,
because their labels have different periods.
Integer lifts handle such mixed qudits uniformly, and concatenation with
GKP codes turns them into oscillator codes.
The last two codes combine an oscillator, a rotor, and a qubit in one
lattice.
The first decomposes after a change of frame, whereas the second couples
all three systems intrinsically.

\subsubsection{Mixed-qudit codes from integer lifts}
\label{ex:mixed-qudit}
This example specializes Sec.~\ref{subsec:examples_notation} to $p=q=0$
and fixes notation that Secs.~\ref{sec:distance_error_correction}
and~\ref{sec:weight_enumerators} use again.
For $k$ qudits with possibly unequal or composite local dimensions
$c_1,\ldots,c_k\geq2$, write
\[
 M=\Z_{\dia c}:=\prod_{i=1}^k\Z_{c_i},\qquad
 \dia c=\operatorname{diag}(c_1,\ldots,c_k),\qquad
 C_{\dia c}:=\prod_{i=1}^k c_i=\dim L^2(M).
\]
Systems of qudits with unequal dimensions, also called mixed-dimensional
or heterogeneous systems, arise in constructions of $k$-uniform
states~\cite{feng2023kuniform} and in software for simulating and
compiling qudit circuits~\cite{mato2023mixed,mato2024mqtqudits}.
Our aim is to describe mixed-qudit stabilizer codes by integer matrices,
so that their dimensions and logical operators can be computed using
lattice methods. Identify $G=M\times\widehat M$ with
$\Z_{\dia c}\times\Z_{\dia c}$. A physical displacement label in $G$
can be represented by an integer vector in $\Z^{2k}$, with the two
coordinates associated with qudit $i$ understood modulo $c_i$.
Let $\pi_{\dia c}:\Z^{2k}\to G$ denote this reduction map.
Its kernel $\Lambda_{\dia c}$ is the period lattice
$\Lambda_{\mathrm{per}}$ of Lemma~\ref{lem:full_lift_pfaffian} for
$p=q=0$.
This lattice of periods and the matrix
that expresses the qudit symplectic pairing in integer coordinates are
\begin{equation}
 \Lambda_{\dia c}=
 \begin{pmatrix}\dia c&0\\0&\dia c\end{pmatrix}\Z^{2k},
 \qquad
 J_{\mathrm{fin}}=\begin{pmatrix}0&\dia c^{-1}\\-\dia c^{-1}&0\end{pmatrix}.
 \label{eq:mixed-qudit-periods-form}
\end{equation}
For a subgroup $D\subseteq G$, the \emph{full lift}
$\Lambda:=\pi_{\dia c}^{-1}(D)$ is a full-rank integer lattice
containing $\Lambda_{\dia c}$, and
$D\cong\Lambda/\Lambda_{\dia c}$.
Choose an integer basis matrix $T$ for $\Lambda$ and form its symplectic
Gram matrix:
\begin{equation}
 \Lambda_{\dia c}\subseteq\Lambda=T\Z^{2k}\subseteq\Z^{2k},
 \qquad D:=\pi_{\dia c}(\Lambda),\qquad
 \varTheta:=T^{\intercal}J_{\mathrm{fin}}T.
 \label{eq:mixed-qudit-lift-gram}
\end{equation}
As in Sec.~\ref{subsec:examples_notation}, $D$ is isotropic exactly when
$\varTheta$ has integer entries.
After the phase correction of Lemma~\ref{lem:phase_triv}, such a subgroup
defines a stabilizer code $\mathcal C$.
Equivalently, to construct a code, choose a full-rank integer lattice
$\Lambda_{\dia c}\subseteq\Lambda\subseteq\Z^{2k}$ whose Gram matrix
is integral, and reduce it modulo the local dimensions. Every mixed-qudit
stabilizer subgroup arises in this way.
Since $|\det T|=[\Z^{2k}:\Lambda]=C_{\dia c}^2/|D|$,
Eq.~\eqref{eq:examples-dimension} gives the codespace dimension
\begin{equation}
 K:=\dim\mathcal C
 =\frac{C_{\dia c}}{|D|}
 =\frac{|\det T|}{C_{\dia c}}
 =|\operatorname{Pf}(\varTheta)|.
 \label{eq:mixed-qudit-pfaffian-dimension}
\end{equation}

Logical displacement operators are obtained from the symplectic dual
$\Lambda^\perp:=\{v\in\R^{2k}:v^{\intercal}J_{\mathrm{fin}}\Lambda\subseteq\Z\}$.
Its basis matrix and Gram matrix are
\begin{equation}
 T_{\mathrm{dual}}=T\varTheta^{-1},\qquad
 \Lambda^\perp=T_{\mathrm{dual}}\Z^{2k}=\pi_{\dia c}^{-1}(D^\perp),\qquad
 T_{\mathrm{dual}}^{\intercal}J_{\mathrm{fin}}T_{\mathrm{dual}}=-\varTheta^{-1}.
 \label{eq:mixed-qudit-dual-basis}
\end{equation}
Since $\Lambda$ contains $\Lambda_{\dia c}$, its symplectic dual is
contained in $\Lambda_{\dia c}^{\perp}=\Z^{2k}$.  In particular,
$T_{\mathrm{dual}}$ is integral, and reducing its columns by
$\pi_{\dia c}$ gives generators of $D^\perp$.
More precisely,
\begin{equation}
 D^\perp/D\cong\Lambda^\perp/\Lambda,
 \qquad |D^\perp/D|=\det\varTheta=K^2.
 \label{eq:mixed-qudit-logical-quotient}
\end{equation}
For mixed dimensions, this finite abelian group is generally not a
vector space over a single finite field.
The passage from $\varTheta$ to the dual Gram matrix $-\varTheta^{-1}$ in
Eq.~\eqref{eq:mixed-qudit-dual-basis} is the finite integer-lift
counterpart of the split-orthogonal duality for GKP--LCA codes in
Appendix~\ref{app:general-lca}.

\textit{A qubit and a ququart.}
For a concrete instance with unequal dimensions, take a qubit and a
ququart, $\dia c=(2,4)$, and the stabilizer $X\otimes X$.
Its label $w=(-1,-1\,|\,0,0)$ generates
$D=\{I,\,X\otimes X,\,I\otimes X^2,\,X\otimes X^3\}$, so $K=8/4=2$.
With the row order $(a_1,a_2\,|\,b_1,b_2)$, a stabilizer generator matrix
and its Gram matrix are
\begin{equation}
 T=\Bigl(\underbrace{\,w\,}_{X\otimes X}\;\;
   \underbrace{\,2e_{a_1}\;\;2e_{b_1}\;\;4e_{b_2}\,}_{\text{periods}}\Bigr)
 =\begin{pmatrix}-1&2&0&0\\-1&0&0&0\\0&0&2&0\\0&0&0&4\end{pmatrix}
 \;\begin{array}{@{}l@{\;}l@{}}
 a_1&\rowquadbrace{qubit and ququart}\\
 a_2&\\
 b_1&\\
 b_2&
 \end{array}\qquad\qquad\quad
 \varTheta=\begin{pmatrix}0&0&-1&-1\\0&0&2&0\\1&-2&0&0\\1&0&0&0
 \end{pmatrix}.
 \label{eq:mixed-qudit-example-T}
\end{equation}
The ququart shift period $4e_{a_2}=-4w-4e_{a_1}$ is already generated,
and $\operatorname{Pf}\varTheta=-2$ confirms $K=2$.
The codewords are the uniform superpositions over the two orbits of
$X\otimes X$,
\begin{equation}
 \ket{0_L}=\sum_{t\in\Z_4}\ket{t\bmod2}\ket{t},
 \qquad
 \ket{1_L}=\sum_{t\in\Z_4}\ket{t\bmod2}\ket{t+1}.
 \label{eq:mixed-qudit-example-codewords}
\end{equation}
The logical matrix has the columns $(0,0\,|\,0,-4)$, $(0,0\,|\,1,-2)$,
$(-1,0\,|\,0,0)$, and $(0,-1\,|\,0,0)$.
The first is a period, and the second labels $Z\otimes Z^2$.
The last two label $X\otimes I$ and $I\otimes X$, which agree up to a
stabilizer.
Thus $\xl=I\otimes X$ exchanges the two codewords, and $\zl=Z\otimes Z^2$
multiplies $\ket{\ell_L}$ by $(-1)^\ell$.

\textit{Qukit-to-qudit stabilizer codes.}
As a single-qudit instance of the construction above,
let $K,r_1,r_2\in\mathbb N$ and $d=Kr_1r_2>1$. Take $M=\Z_d$, so
$G=\Z_d\times\Z_d$, and let $\pi_d:\Z^2\to G$ denote reduction modulo
$d$. The integer lift and stabilizer subgroup are
\begin{equation}
 T=\begin{pmatrix}Kr_1&0\\0&Kr_2\end{pmatrix}
 \;\begin{array}{@{}l@{\;}l@{}}
 a&\rowpairbrace{qudit}\\
 b&
 \end{array}\qquad\quad \Lambda=T\Z^2,\qquad D=\pi_d(\Lambda).
 \label{eq:qukit-qudit-embedding}
\end{equation}
Using the finite symplectic form~\eqref{eq:mixed-qudit-periods-form}
with $\dia c=(d)$ gives
\begin{equation}
 \varTheta=T^{\intercal}J_{\mathrm{fin}}T
 =\begin{pmatrix}0&K\\-K&0\end{pmatrix}.
 \label{eq:qukit-qudit-gram-matrix}
\end{equation}
Thus $D$ is isotropic. Since $D$ is CSS, the bare
displacements $\mathbf D(Kr_1,0)$ and $\mathbf D(0,Kr_2)$ are commuting
stabilizer generators of orders $r_2$ and $r_1$, respectively. Hence
$D\cong\Z_{r_2}\times\Z_{r_1}$, and
Eq.~\eqref{eq:mixed-qudit-pfaffian-dimension} yields
\begin{equation}
 \dim\mathcal C=|\operatorname{Pf}(\varTheta)|=K.
 \label{eq:qukit-qudit-covolume-dimension}
\end{equation}
The code therefore encodes a $K$-dimensional logical system
(a \emph{qukit}) into one physical qudit of dimension $d$.
This code appears in the original GKP
paper~\cite[Sec.~II]{10.1103/physreva.64.012310}, whose stabilizers
$X^{r_1K}$ and $Z^{r_2K}$ generate the same group.
Its codewords are
$\ket{\ell_L}=\sum_{m=0}^{r_2-1}\ket{r_1\ell+Kr_1m}$ with $\ell\in\Z_K$.
The logical matrix
$T_{\mathrm{dual}}=\bigl(\begin{smallmatrix}0&-r_1\\r_2&0
\end{smallmatrix}\bigr)$ gives $\xl=\mathbf D(-r_1,0)=X^{r_1}$ and
$\zl=\mathbf D(0,r_2)=Z^{r_2}$.
They act as $\xl\ket{\ell_L}=\ket{(\ell+1)_L}$ and
$\zl\ket{\ell_L}=e^{2\pi i\ell/K}\ket{\ell_L}$.

\textit{Concatenation with GKP codes.}
The following construction extends to mixed qudit dimensions the
Construction~A description of concatenated square GKP codes
in~\cite[Sec.~2, Eqs.~(43)--(44)]{Conrad2022gottesmankitaev}.
A related rescaled Construction~A lift of qudit CSS codes appears in
the quantum low-density lattice codes of
\cite[Sec.~III.D, Eqs.~(34)--(35)]{hillmann2026quantumldlc}, connecting the integer-lift
description here to GKP constructions with sparse generators.
An integer lift also gives a concatenated GKP code: encode each physical
qudit of the finite code in an ideal single-mode square GKP code of
logical dimension $c_i$, then impose the outer stabilizers indexed by $D$.
For $\Lambda=\pi_{\dia c}^{-1}(D)=T\Z^{2k}$, the resulting oscillator
stabilizer lattice is
\begin{equation}
 R_{\dia c}:=\begin{pmatrix}\dia c^{-1/2}&0\\0&\dia c^{-1/2}\end{pmatrix},
 \qquad D_{\mathrm{GKP}}:=R_{\dia c}\Lambda.
 \label{eq:integer_lift_real_gkp_lattice}
\end{equation}
It contains the inner GKP stabilizer lattice
\begin{equation}
 D_{\mathrm{in}}:=R_{\dia c}\Lambda_{\dia c}
 =\begin{pmatrix}\dia c^{1/2}&0\\0&\dia c^{1/2}\end{pmatrix}\Z^{2k},
 \label{eq:integer_lift_inner_gkp_lattice}
\end{equation}
and $D_{\mathrm{GKP}}/D_{\mathrm{in}}\cong D$ identifies the additional
stabilizers with those of the outer code, with matching stabilizer phases.
Since $R_{\dia c}^{\intercal}J_{\mathrm{osc}}R_{\dia c}=J_{\mathrm{fin}}$,
this lattice has Gram matrix $\varTheta$ and encodes the same logical
dimension $K=|\det T|/C_{\dia c}$.

Remark~\ref{rem:full_lift_gkp_concatenation} extends this construction to
every configuration space of the form~\eqref{eq:examples-configuration-space}.
There, each rotor becomes an oscillator whose position is confined to the
integers by the stabilizer $e^{2\pi i\hat x}$.
The rescaled full lift is then a GKP lattice with the same Gram matrix
$\varTheta$.
Thus Eq.~\eqref{eq:examples-dimension} is the GKP dimension formula
applied to the rescaled full lift.
The correspondence identifies ideal codespaces, stabilizer groups, and
logical displacement groups, but it does not identify error sets or
Clifford groups.
For instance, the rotor--qudit code of Sec.~\ref{ex:rotor-qudit}
becomes a two-mode GKP code with the Gram matrix~\eqref{eq:rq-Theta}.
Integer lifts also yield the Pauli and Lee distances of
Sec.~\ref{sec:distance_error_correction}, in
Eqs.~\eqref{eq:pauli_A2_lift_distance}
and~\eqref{eq:qudit_lee_lift_distance}.
Remark~\ref{rem:we_integer_lifts} uses them for weight enumerators.

\subsubsection{Torsion-extended hybrid code}
\label{ex:torsion}
Stabilizer lattices with both free and torsion components already appear
in Secs.~\ref{ex:rotor-gkp}, \ref{ex:hybrid-golay}, \ref{ex:osc-rotor},
and~\ref{ex:rotor-qudit}.
The following example places such a lattice on an oscillator, a rotor,
and a physical qubit at once.

Fix $K\in\mathbb N$ and let
\begin{equation}
M=\R\times\Z\times\Z_2,
\qquad
\widehat M=\R\times\T\times\Z_2,
\qquad
\mathcal H=L^2(M)
\cong L^2(\R)\otimes\ell^2(\Z)\otimes\C^2.
\label{eq:hybrid-torsion-configuration-space}
\end{equation}
Write a point of $G=M\times\widehat M$ as
$z=(x,n,j;\xi,\theta,b)$, listing all positions before all momenta.
Define
\begin{equation}
\begin{aligned}
w_1&=(1,1,1;0,0,0),
&w_2&=(0,0,0;\tfrac12,0,1),\\
w_3&=(0,0,0;0,\tfrac12,1),
&w_4&=(2K,0,0;0,0,0),
\end{aligned}
\qquad
D_K:=\langle w_1,w_2,w_3,w_4\rangle.
\label{eq:hybrid-torsion-lattice}
\end{equation}
The shift part $\langle w_1,w_4\rangle$ and the phase part
$\langle w_2,w_3\rangle$ of $D_K$ commute.
Thus $D_K$ is an isotropic CSS lattice, and its bare displacements
already form the stabilizer group.
Since $2w_3=0$, while $w_1,w_2,$ and
$w_4$ are independent elements of infinite order,
\begin{equation}
D_K\cong\Z^3\oplus\Z_2.
\label{eq:hybrid-torsion-group-structure}
\end{equation}

To display the hybrid character of the stabilizers, set
$\mathbf S_i:=\mathbf D(w_i)$. With the convention $\mathbf D=\mathbf M\mathbf T$
fixed in Sec.~\ref{subsec:examples_notation}, their action on $\psi\in L^2(M)$ is
\begin{subequations}
\label{eq:hybrid-torsion-stabilizer-actions}
\begin{align}
(\mathbf S_1\psi)(x,n,j)&=\psi(x+1,n+1,j+1),\\
(\mathbf S_2\psi)(x,n,j)&=e^{\pi i x}(-1)^j\psi(x,n,j),\\
(\mathbf S_3\psi)(x,n,j)&=(-1)^{n+j}\psi(x,n,j),\\
(\mathbf S_4\psi)(x,n,j)&=\psi(x+2K,n,j),
\end{align}
\end{subequations}
where the qubit coordinate is understood modulo $2$.

In the paired row order $(x,\xi\,|\,n,\theta\,|\,a,b)$ of
Sec.~\ref{subsec:examples_notation}, where the qubit shift $a$ is the
coordinate $j$ of $z$, a stabilizer generator matrix is
\begin{equation}
 T=\Bigl(\underbrace{\,w_1\;\;w_2\;\;w_3\;\;w_4\,}_{\text{stabilizers}}\;\;
   \underbrace{\,e_\theta\;\;2e_a\,}_{\text{periods}}\Bigr)
 =\begin{pmatrix}1&0&0&2K&0&0\\0&\tfrac12&0&0&0&0\\1&0&0&0&0&0\\
 0&0&\tfrac12&0&1&0\\1&0&0&0&0&2\\0&1&1&0&0&0\end{pmatrix}
 \;\begin{array}{@{}l@{\;}l@{}}
 x&\rowpairbrace{oscillator}\\
 \xi\vphantom{\tfrac12}&\\
 n&\rowpairbrace{rotor}\\
 \theta\vphantom{\tfrac12}&\\
 a&\rowpairbrace{qubit.}\\
 b&
 \end{array}
 \label{eq:hybrid-torsion-T}
\end{equation}
Here $e_\theta$ and $2e_a$ are the rotor-angle and qubit-shift periods.
The qubit-phase period $2e_b=2w_3-e_\theta$ is already generated.
Because of the torsion element $w_3$, the periods $e_\theta$ and $2e_b$
cannot both be columns of $T$.
The Gram matrix
\begin{equation}
 \varTheta=\begin{pmatrix}0&1&1&0&1&0\\-1&0&0&-K&0&-1\\-1&0&0&0&0&-1\\
 0&K&0&0&0&0\\-1&0&0&0&0&0\\0&1&1&0&0&0\end{pmatrix},
 \qquad
 \operatorname{Pf}\varTheta=K,
 \label{eq:hybrid-torsion-Theta}
\end{equation}
gives, through Eq.~\eqref{eq:examples-dimension}, the code dimension
\begin{equation}
 \dim\mathcal C_K=|\operatorname{Pf}\varTheta|=K,
 \label{eq:hybrid-torsion-code-dimension}
\end{equation}
and Eq.~\eqref{eq:examples-logical-group} gives $|D_K^\perp/D_K|=K^2$.

\textit{Ideal logical basis.}
Let $\ket{x,n,j}$ denote the generalized basis associated with $M$.  The
ideal logical codewords are
\begin{equation}
\ket{k_L}\propto
\sum_{\ell,n\in\Z}
\ket{2(k+K\ell)+n}_{\mathrm{osc}}
\otimes\ket n_{\mathrm{rot}}
\otimes\ket{n\bmod2}_{\mathrm{fin}},
\qquad k\in\Z_K.
\label{eq:hybrid-torsion-logical-basis}
\end{equation}
This can be checked directly from
Eq.~\eqref{eq:hybrid-torsion-stabilizer-actions}: every $\mathbf S_i$ fixes
each state in~\eqref{eq:hybrid-torsion-logical-basis}.  Moreover, states with
distinct $k\in\Z_K$ have disjoint comb supports and are therefore linearly
independent.  There are exactly $K$ such states, which agrees with
Eq.~\eqref{eq:hybrid-torsion-code-dimension}.
Consequently, they form a basis of the codespace.
Like ideal GKP states, these combs are not normalizable. For $K=2$,
Eq.~\eqref{eq:hybrid-torsion-logical-basis} gives an explicit logical qubit
with basis
\begin{subequations}
\begin{align}
\ket{0_L}&\propto\sum_{\ell,n\in\Z}
\ket{4\ell+n}_{\mathrm{osc}}\otimes\ket n_{\mathrm{rot}}
\otimes\ket{n\bmod2}_{\mathrm{fin}},\\
\ket{1_L}&\propto\sum_{\ell,n\in\Z}
\ket{4\ell+2+n}_{\mathrm{osc}}\otimes\ket n_{\mathrm{rot}}
\otimes\ket{n\bmod2}_{\mathrm{fin}}.
\end{align}
\label{eq:hybrid-torsion-logical-qubit}
\end{subequations}

\textit{Logical Pauli operators.}
Up to periods, the logical matrix $T_{\mathrm{dual}}=T\varTheta^{-1}$
contains the labels $(2,0,0;0,0,0)$ and
$(0,0,0;\tfrac1{2K},-\tfrac1{2K},0)$ among its columns.
For every $K$, their displacements generate
$D_K^\perp/D_K\cong\Z_K\times\Z_K$.
For $K=2$, choose the physical displacement representatives
\begin{equation}
 \overline X:=\mathbf D(2,0,0;0,0,0),\qquad
 \overline Z:=\mathbf D(0,0,0;\tfrac14,-\tfrac14,0).
 \label{eq:hybrid-torsion-logical-paulis}
\end{equation}
Their actions on wavefunctions are
\[
 (\overline X\psi)(x,n,j)=\psi(x+2,n,j),\qquad
 (\overline Z\psi)(x,n,j)=e^{\pi i(x-n)/2}\psi(x,n,j).
\]
Both commute with all four stabilizers, while
$\overline X\,\overline Z=-\overline Z\,\overline X$,
$\overline X^2=\mathbf S_4$, and
$\overline Z^2=\mathbf S_2\mathbf S_3$.
Thus their restrictions to the codespace are logical Pauli operators:
\begin{equation}
 \overline X\ket{k_L}=\ket{(k+1\bmod2)_L},\qquad
 \overline Z\ket{k_L}=(-1)^k\ket{k_L},\qquad k\in\Z_2.
 \label{eq:hybrid-torsion-logical-pauli-action}
\end{equation}

\textit{Logical Clifford gates.}
For $K=2$, the combination $y=x-n$ supports the rectangular GKP comb
$y=4\ell+2k$.  A scaled Fourier transform in this coordinate and a
quadratic phase implement logical Hadamard and phase gates:
\begin{align}
 (U_H\psi)(x,n,j)
 &=\frac1{\sqrt8}\int_{\R}
 e^{-2\pi i(x-n)(x'-n)/8}\psi(x',n,j)\,dx',\notag\\
 (U_S\psi)(x,n,j)&=e^{i\pi(x-n)^2/8}\psi(x,n,j).
 \label{eq:hybrid-torsion-clifford-gates}
\end{align}
Section~\ref{subsec:clifford_examples} shows that their symplectic
automorphisms preserve $D_2$ and that no displacement correction is
needed.
On the logical basis they act as
\begin{equation}
 U_H\ket{k_L}=\frac{\ket{0_L}+(-1)^k\ket{1_L}}{\sqrt2},\qquad
 U_S\ket{k_L}=i^k\ket{k_L},\qquad k\in\Z_2,
 \label{eq:hybrid-torsion-clifford-action}
\end{equation}
and hence generate the single-qubit Clifford group up to overall phase.

\textit{A change of frame.}
The lattice $D_K$ couples all three systems in the coordinates above.
However, the symplectic automorphism
\begin{equation}
 S(x,n,j;\xi,\theta,b)=\bigl(x-n,\,n,\,j-n;\,\xi,\,\theta+\xi+\tfrac b2,\,b\bigr)
 \label{eq:hybrid-torsion-decoupling}
\end{equation}
maps $D_K$ onto a direct sum of three lattices.
They are the rectangular GKP lattice with $A=2K\Z$ and $B=\tfrac12\Z$, the
lattice of the rotor stabilizer $\mathbf T_1$, and the lattice of the
qubit stabilizer $Z$.
Thus a generalized Clifford operator maps the code onto a GKP code with
$K$ logical states times stabilizer states of the rotor and the qubit.
The example therefore illustrates a hybrid stabilizer lattice with
torsion rather than intrinsic coupling.
Section~\ref{subsec:fe_examples} constructs finite-energy versions of the
codewords~\eqref{eq:hybrid-torsion-logical-qubit}.


\subsubsection{Intrinsically coupled hybrid code}
\label{ex:hybrid-intrinsic}
The lattice $D_K$ of Sec.~\ref{ex:torsion} couples its three systems
only in its defining coordinates.
We now construct a torsion-free CSS lattice on the same systems that no
generalized Clifford operator decomposes.

We keep the configuration space~\eqref{eq:hybrid-torsion-configuration-space}
and the coordinates $z=(x,n,j;\xi,\theta,b)$ of Sec.~\ref{ex:torsion}.
Define
\begin{equation}
w_1=(1,0,0;0,0,0),
\qquad
w_2=(0,0,0;1,\tfrac14,1),
\qquad
w_3=(0,2,1;0,0,0),
\qquad
D_{\mathrm{int}}:=\langle w_1,w_2,w_3\rangle.
\label{eq:hybrid-intrinsic-lattice}
\end{equation}
The corresponding stabilizers $\mathbf S_i:=\mathbf D(w_i)$ are
\begin{equation}
 \mathbf S_1=\mathbf T_1\otimes I\otimes I,
 \qquad
 \mathbf S_2=\mathbf M_1\otimes\mathbf M_{1/4}\otimes Z,
 \qquad
 \mathbf S_3=I\otimes\mathbf T_2\otimes X,
 \label{eq:hybrid-intrinsic-stabilizers}
\end{equation}
where $\mathbf M_{1/4}=i^{\widehat{\mathbf L}}$ acts on the rotor.
The shift part $\langle w_1,w_3\rangle$ and the phase part
$\langle w_2\rangle$ of $D_{\mathrm{int}}$ commute.
Indeed, $\mathbf S_1$ and $\mathbf S_2$ have oscillator pairing $1$,
while $\mathbf S_2$ and $\mathbf S_3$ anticommute on both the rotor and
the qubit.
Thus $D_{\mathrm{int}}$ is an isotropic CSS lattice, and its bare
displacements already form the stabilizer group.
The element $m_1w_1+m_2w_2+m_3w_3$ has oscillator labels $(m_1,m_2)$ and
rotor shift $2m_3$, so it vanishes only when $m_1=m_2=m_3=0$.
Hence
\begin{equation}
D_{\mathrm{int}}\cong\Z^3,
\label{eq:hybrid-intrinsic-group-structure}
\end{equation}
and the lattice has no torsion.

With the convention $\mathbf D=\mathbf M\mathbf T$, the stabilizers act
on $\psi\in L^2(M)$ as
\begin{subequations}
\label{eq:hybrid-intrinsic-stabilizer-actions}
\begin{align}
(\mathbf S_1\psi)(x,n,j)&=\psi(x+1,n,j),\\
(\mathbf S_2\psi)(x,n,j)&=e^{2\pi ix}\,i^n(-1)^j\psi(x,n,j),\\
(\mathbf S_3\psi)(x,n,j)&=\psi(x,n+2,j+1),
\end{align}
\end{subequations}
where the qubit coordinate is understood modulo $2$.

In the paired row order $(x,\xi\,|\,n,\theta\,|\,a,b)$ of
Sec.~\ref{subsec:examples_notation}, where the qubit shift $a$ is the
coordinate $j$ of $z$, a stabilizer generator matrix is
\begin{equation}
 T=\Bigl(\underbrace{\,w_1\;\;w_2\;\;w_3\,}_{\text{stabilizers}}\;\;
   \underbrace{\,e_\theta\;\;2e_a\;\;2e_b\,}_{\text{periods}}\Bigr)
 =\begin{pmatrix}1&0&0&0&0&0\\0&1&0&0&0&0\\0&0&2&0&0&0\\
 0&\tfrac14&0&1&0&0\\0&0&1&0&2&0\\0&1&0&0&0&2\end{pmatrix}
 \;\begin{array}{@{}l@{\;}l@{}}
 x&\rowpairbrace{oscillator}\\
 \xi&\\
 n&\rowpairbrace{rotor}\\
 \theta\vphantom{\tfrac14}&\\
 a&\rowpairbrace{qubit.}\\
 b&
 \end{array}
 \label{eq:hybrid-intrinsic-T}
\end{equation}
Because $D_{\mathrm{int}}$ has no torsion, $T$ can contain all three
periods $e_\theta$, $2e_a$, and $2e_b$ as columns.
The Gram matrix
\begin{equation}
 \varTheta=\begin{pmatrix}0&1&0&0&0&0\\-1&0&-1&0&-1&0\\0&1&0&2&0&1\\
 0&0&-2&0&0&0\\0&1&0&0&0&2\\0&0&-1&0&-2&0\end{pmatrix},
 \qquad
 \operatorname{Pf}\varTheta=4,
 \label{eq:hybrid-intrinsic-Theta}
\end{equation}
gives, through Eq.~\eqref{eq:examples-dimension}, the code dimension
\begin{equation}
 \dim\mathcal C=|\operatorname{Pf}\varTheta|=4.
 \label{eq:hybrid-intrinsic-code-dimension}
\end{equation}
The Smith normal form of $\varTheta$ is $\operatorname{diag}(1,1,1,1,4,4)$,
so Eq.~\eqref{eq:examples-logical-group} gives
$D_{\mathrm{int}}^\perp/D_{\mathrm{int}}\cong\Z_4^2$.

\textit{Ideal logical basis.}
The ideal logical codewords are
\begin{equation}
\ket{k_L}\propto
\sum_{\ell,m\in\Z}
\ket{\ell-\tfrac k4}_{\mathrm{osc}}
\otimes\ket{k+2m}_{\mathrm{rot}}
\otimes\ket{m\bmod2}_{\mathrm{fin}},
\qquad k\in\Z_4.
\label{eq:hybrid-intrinsic-logical-basis}
\end{equation}
Equation~\eqref{eq:hybrid-intrinsic-stabilizer-actions} shows that every
$\mathbf S_i$ fixes each of these states.
For $\mathbf S_2$, the oscillator phase $e^{2\pi ix}=i^{-k}$ cancels the
rotor and qubit phases $i^{k+2m}(-1)^m=i^k$.
States with distinct $k$ have oscillator positions in distinct cosets of
$\Z$, so their supports are disjoint.
These four states therefore form a basis of the codespace, in agreement
with Eq.~\eqref{eq:hybrid-intrinsic-code-dimension}.
Like ideal GKP states, they are not normalizable.
Each codeword is a product of an oscillator comb and a rotor--qubit state,
but both factors depend on $k$, so logical superpositions are entangled
across this cut.
The rotor- and qubit-controlled oscillator displacement $\ket{x,n,j}\mapsto\ket{x+\tfrac{n-2j}4,n,j}$, with
$j\in\{0,1\}$, removes $k$ from the oscillator, but it is not a generalized Clifford operator
(Proposition~\ref{prop:hybrid-intrinsic-nondecoupling}).

\textit{Logical Pauli operators.}
Up to periods and inverses, the logical matrix
$T_{\mathrm{dual}}=T\varTheta^{-1}$ contains the labels
$(\tfrac14,-1,0;0,0,0)$ and $(0,0,0;0,\tfrac14,1)$ among its columns.
We choose the physical displacement representatives
\begin{equation}
 \overline X:=\mathbf D(\tfrac14,-1,0;0,0,0)
 =\mathbf T_{1/4}\otimes\mathbf T_{-1}\otimes I,
 \qquad
 \overline Z:=\mathbf D(0,0,0;0,\tfrac14,1)
 =I\otimes\mathbf M_{1/4}\otimes Z.
 \label{eq:hybrid-intrinsic-logical-paulis}
\end{equation}
Their actions on wavefunctions are
\[
 (\overline X\psi)(x,n,j)=\psi(x+\tfrac14,n-1,j),\qquad
 (\overline Z\psi)(x,n,j)=i^n(-1)^j\psi(x,n,j).
\]
Both commute with all three stabilizers, while
$\overline X\,\overline Z=-i\,\overline Z\,\overline X$,
$\overline X^4=\mathbf S_1\mathbf S_3^{-2}$, and $\overline Z^4=I$.
Their squares are not stabilizers, so both have order four modulo
stabilizers.
They act on the logical basis as
\begin{equation}
 \overline X\ket{k_L}=\ket{(k+1)_L},\qquad
 \overline Z\ket{k_L}=i^k\ket{k_L},\qquad k\in\Z_4,
 \label{eq:hybrid-intrinsic-logical-pauli-action}
\end{equation}
and thus generate the Pauli group of a logical ququart.

\textit{Indecomposability.}
Unlike $D_K$, the lattice $D_{\mathrm{int}}$ couples its three systems
intrinsically.
The oscillator--rotor shear
$(x,n;\xi,\theta)\mapsto(x+\tfrac n4,n;\xi,\theta-\tfrac\xi4)$ of
Eq.~\eqref{eq:oscillator-rotor-shear} maps $D_{\mathrm{int}}$ to the
torsion-free rotor--qubit code $\langle\mathbf T_2\otimes X\rangle$ of
Sec.~\ref{ex:rotor-qudit} with its qubit encoded in the oscillator--qubit
LCA code $\langle\mathbf T_1,\mathbf M_1\otimes Z\rangle$.
Encoding that qubit in the rectangular GKP qubit of Sec.~\ref{ex:gkp}
instead gives the ququart $D_{\alpha,1}$ of Sec.~\ref{ex:osc-rotor}.

\begin{proposition}
\label{prop:hybrid-intrinsic-nondecoupling}
The lattice $D_{\mathrm{int}}$ is not symplectically decomposable for any
bipartition of the oscillator, the rotor, and the qubit.
In particular, no generalized Clifford operator maps the code to a
tensor product of codes on complementary sets of these systems.
\end{proposition}

The proof in Appendix~\ref{appx:proof_hybrid_intrinsic} rests on two
invariants.
First, every generalized Clifford operator maps the stabilizers without
rotor shift to stabilizers without rotor shift.
Among these, $\mathbf S_2$ acts on the qubit, and no change of frame
removes this qubit action.
Because $D_{\mathrm{int}}$ has no torsion, splitting off the oscillator or
the qubit would leave the image of $\mathbf S_2$ without qubit labels.
This rules out both splits.
Second, the rotor shifts $2\Z$ of the stabilizers are invariant.
A separate rotor factor would therefore encode a qubit.
The logical group would then be $\Z_2^4$, which has no element of order
four, unlike $\overline X$.
As for the ququart $D_{\alpha,1}$ of Sec.~\ref{ex:osc-rotor}, this
intrinsic coupling needs neither stabilizer torsion nor a non-CSS lattice.
The Gaussian regularization of Sec.~\ref{sec:finite_energy} applies to
the combs~\eqref{eq:hybrid-intrinsic-logical-basis} as well.

\section{General structure of GLCA codes}
\label{sec:general_structure}

We now develop the general structure of GLCA codes for every configuration
group $M=\mathbb R^p\times\mathbb Z^q\times\mathbb Z_{\dia c}$, where
$p,q\geq0$ and $\mathbb Z_{\dia c}$ is a finite abelian group written as a
product of cyclic groups.
Sections~\ref{sec:notation}--\ref{sec:GSC} fix notation, review the
projective-module picture of Pauli stabilizer codes, and define GLCA codes
for every isotropic lattice $D\subseteq G$, including lattices with
torsion.

Sections~\ref{sec:LSC_ncg}--\ref{sec:clifford_gates} then derive four
general results.
First, the syndrome subspaces form a Hermitian vector bundle over the
syndrome space $\widehat D$, with Rieffel's module as its module of
continuous sections (Sec.~\ref{sec:LSC_ncg}).
Second, the code dimension equals the covolume of $D$
(Theorem~\ref{thm:code_dimension}) and hence the Pfaffian
$|\operatorname{Pf}\varTheta|$ used throughout Sec.~\ref{sec:examples}
(Lemma~\ref{lem:full_lift_pfaffian}).
Third, the symplectic dual lattice $D^\perp$ labels logical displacement
operators that generate the full logical operator algebra
(Theorem~\ref{thm:logical_displacement_algebra}).
Fourth, lattice symmetries, i.e., symplectic automorphisms of $G$ that
preserve $D$, yield logical Clifford gates
(Theorem~\ref{thm:corrected_clifford_implementer}).
The examples of Sec.~\ref{sec:examples} illustrate these results.
Most proofs appear in
Appendices~\ref{appx:proof_rieffel_bundle_decomposition}--\ref{appx:proof_corrected_clifford_implementer}.
Sections~\ref{sec:distance_error_correction}--\ref{sec:finite_energy}
build on this section, treating distance and decoding, weight enumerators,
and finite-energy codewords in turn.

\subsection{Notation}
\label{sec:notation}
Table~\ref{tab:notation_normalization} collects the main notation and
normalization conventions. The indicated sections and equations give the
definitions and constructions.
Table~\ref{tab:qec_ncg_dictionary} matches stabilizer-code notions with
their counterparts in Rieffel's noncommutative-geometry framework.

\begin{table}[htbp]
\caption{Main notation and normalization conventions for an isotropic
lattice $D\subseteq G=M\times\widehat M$.}
\label{tab:notation_normalization}
\centering
\small
\renewcommand{\arraystretch}{1.18}
\begin{tabularx}{\linewidth}{@{}p{0.20\linewidth}Xl@{}}
\toprule
Symbol & Meaning or convention & Reference \\
\midrule
$M$ & Configuration group $\mathbb R^p\times\mathbb Z^q\times\mathbb Z_{\dia c}$;
physical Hilbert space $L^2(M)$. & Sec.~\ref{sec:GSC} \\
$G$ & Phase space $M\times\widehat M$, with commutator bicharacter $\rho$.
& Sec.~\ref{sec:GSC} \\
$D$, $\mathcal S_D$ & Isotropic stabilizer lattice and its operator group
$\mathcal S_D=\{c(w)\mathbf D(w):w\in D\}$. & \eqref{eq:stabgroup} \\
$D^\perp$ & Symplectic dual of $D$ with respect to $\rho$;
its displacements preserve every syndrome fiber. & Sec.~\ref{sec:logical_displacement_operators} \\
$\widehat D$, $\chi$ & Syndrome parameter space and a fiber label.
On $E_{D,\chi}$, the stabilizer indexed by $w$ has eigenvalue
$\overline{\chi(w)}$. & \eqref{eq:rieffel_fiber_generalized_eigenvalue} \\
$\kappa_z$ & Syndrome shift caused by displacement $z$:
$\kappa_z(w)=\overline{\rho(w,z)}$, sending $\chi$ to $\kappa_z\chi$.
& \eqref{eq:rieffel_fiber_covariance} \\
$E_{D,\chi}$ & Syndrome fiber; $\iota_1(E_{D,1})=\mathcal C$
is the ideal codespace. & \eqref{eq:trivial_fiber_codespace} \\
$\mathscr E_D$, $\Omega_D$ & Syndrome bundle and Rieffel's module,
with $\Omega_D\cong\Gamma(\mathscr E_D)$. & \eqref{eq:serre_swan_evaluation_map} \\
$\mathsf L_D$ & Logical phase space $D^\perp/D$, of order $K^2$.
& \eqref{eq:logical_phase_space} \\
$K$ & Code dimension $\dim\mathcal C=\dim E_{D,\chi}$,
independent of $\chi$. & \eqref{eq:rank_covolume_formula} \\
$s_D$ & Covolume $|G/D|_\mu$; for an isotropic lattice, $s_D=K$. & \eqref{eq:covolume_definition} \\
$\sigma_D$ & Order $|\operatorname{Tor}(D)|$ of the torsion subgroup;
$\sigma_D=1$ when $D$ is torsion-free. & \eqref{eq:fe_torsion_order} \\
$\mu$ & Self-Plancherel Haar measure $\mu_M\otimes\mu_{\widehat M}$
on $G$; for finite $M$, it is $|M|^{-1}$ times counting measure.
& Sec.~\ref{subsec:code_dimension_covolume} \\
$d\chi$ & Normalized Haar measure on $\widehat D$;
Haar measure on $D$ is counting measure. & \eqref{eq:canonical_trace_haar} \\
$\tau_D$, $\mathrm{Tr}_D$ & Corresponding normalized traces:
$\tau_D(f)=\int_{\widehat D}f(\chi)\,d\chi$ and
$\mathrm{Tr}_D(a)=a(0)$. Matrix extensions use the unnormalized matrix trace.
& \eqref{eq:canonical_trace_haar} \\
$J$ & Lifted symplectic form
$\operatorname{diag}(J_{\mathrm{osc}},J_{\mathrm{rot}},J_{\mathrm{fin}})$ in the
paired coordinates $(x,\xi\,|\,n,\theta\,|\,a,b)$. & \eqref{eq:symplectic-form} \\
$T$ & Stabilizer generator matrix, whose columns form a basis of the full
lift $\Lambda$ of $D$, period vectors included. & Sec.~\ref{subsec:examples_notation} \\
$\varTheta$ & Symplectic Gram matrix $T^{\intercal}JT$, integral exactly
when $D$ is isotropic. Then $K=|\operatorname{Pf}\varTheta|$. & \eqref{eq:examples-dimension} \\
$T_{\mathrm{dual}}$ & Logical matrix $T\varTheta^{-1}$, whose columns
generate the lifted labels of $D^\perp$. & \eqref{eq:examples-logical-matrix} \\
\bottomrule
\end{tabularx}
\end{table}

\subsection{Pauli stabilizer codes and projective modules}\label{sec:PSC_basic}
Before introducing \textit{generalized LCA codes}, we first review Pauli stabilizer
codes (see Gottesman's thesis~\cite[Sections~3.2 and~3.4]{gottesman1997stabilizer}
and Nielsen and Chuang~\cite[Section~10.5]{nielsenchuang2010})
and their projective-module description, which will be useful in developing
the general framework. Definitions and background concerning $\mathrm C^*$-algebras are
collected in Appendix~\ref{appx:c_star}, while vector bundles, finitely generated
projective modules, and trace pairings are reviewed in
Appendix~\ref{sec:app_prelim}.

In the check-matrix formalism, a binary Pauli stabilizer code is specified by
an injective linear map $T\colon\mathbb Z_2^l\to G$, where
$G=\mathbb Z_2^n\times\mathbb Z_2^n$ carries the symplectic form
\[
 J=\begin{pmatrix}0&I\\I&0\end{pmatrix},
 \qquad \langle x,y\rangle_S=x^{\mathsf T}Jy,
\]
and $D:=T(\mathbb Z_2^l)$ is isotropic:
\[
 \langle x,y\rangle_S=0\quad(x,y\in D),
 \qquad\text{equivalently,}\quad T^{\mathsf T}JT=0.
\]
In this subsection, $T$ and $J$ denote these binary objects rather than the
lifted matrices of Sec.~\ref{subsec:examples_notation}.

On the Hilbert space  $\mathcal H=\ell^2(\mathbb Z_2^n)\cong(\mathbb C^2)^{\otimes n}$,
the space of complex functions on $\mathbb Z_2^n$, let
$m=(m_1,\ldots,m_n)$ and $s=(s_1,\ldots,s_n)$ be binary vectors in
$\mathbb Z_2^n$.
Define the corresponding shift and phase operators on $\mathcal H$ by
\[
 (\mathbf T_m\xi)(t)=\xi(t+m),\qquad
 (\mathbf M_s\xi)(t)=(-1)^{s\cdot t}\xi(t).
\]
Here $t\in\mathbb Z_2^n$, addition is componentwise modulo $2$, and
$s\cdot t=\sum_{i=1}^n s_i t_i$ is evaluated modulo $2$.
These are the usual Pauli operators
$\mathbf X^m=\bigotimes_{i=1}^n\mathbf X_i^{m_i}$ and
$\mathbf Z^s=\bigotimes_{i=1}^n\mathbf Z_i^{s_i}$, respectively.
The displacement operators
\begin{equation}\label{eq:displacement_operator}
 \mathbf D(m,s):=\mathbf M_s\mathbf T_m
\end{equation}
satisfy the Heisenberg commutation relation
\begin{equation}\label{eq:ccr_qubit}
 \mathbf D(x)\mathbf D(y)
 =(-1)^{\langle x,y\rangle_S}\mathbf D(y)\mathbf D(x).
\end{equation}
Thus isotropy makes $\{\mathbf D(w):w\in D\}$ a commuting family.
One can choose phases $c(w)$ such that  $u_w=c(w)\mathbf D(w)$ (cf. Lemma~\ref{lem:phase_triv}) so that
$u\colon D\to\mathscr U(\mathcal H)$ is a unitary representation,
$u_wu_{w'}=u_{w+w'}$. The \textit{codespace} is its common $+1$ eigenspace,
\[
 \mathcal C=\{\xi\in\mathcal H:u_w\xi=\xi\text{ for all }w\in D\}.
\]
Every classical binary linear code $C=\ker H$ arises in this formalism.
Taking a full-row-rank $l\times n$ parity-check matrix $H$ over
$\mathbb Z_2$, viewed as a linear map
$H\colon\mathbb Z_2^n\to\mathbb Z_2^l$, set
\[
 T_{\mathrm{cl}}=H^{\mathsf T},\qquad
 T(u)=(0,T_{\mathrm{cl}}u),\qquad
 D=\{0\}\times\operatorname{Im}(T_{\mathrm{cl}}).
\]
Here $D$ is isotropic and $u_{(0,s)}=\mathbf M_s$ needs no phase adjustment.
In the following, $\operatorname{Im}(T_{\mathrm{cl}})^\perp$ denotes the
orthogonal complement with respect to the ordinary binary dot product on
$\mathbb Z_2^n$.
Since $\mathbf M_s\ket{t}=\ket{t}$ for all
$s\in\operatorname{Im}(T_{\mathrm{cl}})$ exactly when
$t\in\operatorname{Im}(T_{\mathrm{cl}})^\perp=\ker H$, the quantum codespace is
\[
 \mathcal C=\operatorname{span}\{\ket{t}:t\in C\}.
\]

\subsubsection{Syndrome fibers and the finite projective-module picture}
\label{subsec:pauli_fiber_model}

The simultaneous spectral decomposition of the stabilizers $u_w$ gives a finite
version of the geometric picture used in Section~\ref{sec:LSC_ncg}. 
Using the coordinates supplied by the isomorphism
$T\colon\mathbb Z_2^l\to D$, we identify $D$ with its Pontryagin dual group
$\widehat D:=\operatorname{Hom}(D,\mathbb T)$, the group of homomorphisms (characters)
from $D$ to the unit circle $\mathbb T=\{z\in\mathbb C:|z|=1\}$ under
pointwise multiplication, via $t\mapsto\chi_t$, where
\[
 t\cdot s:=T^{-1}(t)\cdot T^{-1}(s),\qquad
 \chi_t(s):=(-1)^{t\cdot s},\quad t,s\in D.
\]
The dot product on the right is the standard binary dot product on
$\mathbb Z_2^l$; thus the pairing on $D$ is nondegenerate and depends on
the chosen coordinates, rather than being the restriction of the ambient
dot product on $G$.
The syndrome projections and fibers are
\begin{equation}
 P_t:=\frac{1}{|D|}\sum_{s\in D}(-1)^{t\cdot s}u_s,
 \qquad
 \mathcal H(t):=\operatorname{Im}(P_t),\quad t\in D.
 \label{eq:pauli_syndrome_projection}
\end{equation}
It is easy to verify that $P_t^\dagger=P_t$,
$P_tP_{t'}=\delta_{t,t'}P_t$, and $\sum_{t\in D}P_t=I$.
Since $u_sP_t=(-1)^{t\cdot s}P_t$, the map
\begin{equation}
 \mathcal U_D\colon\mathcal H
 \longrightarrow\bigoplus_{t\in D}\mathcal H(t),
 \qquad
 (\mathcal U_D\varphi)(t):=P_t\varphi,
 \label{eq:pauli_syndrome_decomposition}
\end{equation}
is unitary, with
$\mathcal C=\mathcal H(0)=\operatorname{Im}(P_0)$ the codespace.

These fibers $\mathcal H(t)$ form a vector bundle
$\mathscr E_D:=\bigsqcup_{t\in D}\mathcal H(t)\to D$
over the finite discrete set $D$, called the \emph{syndrome bundle}.
Since nonidentity Pauli operators have trace zero, every fiber has dimension
\begin{equation}
 \dim\mathcal H(t)=\operatorname{Tr}(P_t)
 =\frac{1}{2^l}\operatorname{Tr}(I)=2^{n-l}.
 \label{eq:pauli_fiber_dimension}
\end{equation}
Therefore $\mathscr E_D$ is a trivial
rank-$2^{n-l}$ bundle. Its sections form a module $\Gamma(\mathscr E_D)=\bigoplus_{t\in D}\mathcal H(t)$ over the commutative $\mathrm C^*$-algebra $C(D)$:
\[
(f\cdot\xi)(t)=f(t)\xi(t), \quad f\in C(D),\ \xi\in\Gamma(\mathscr E_D).
\]
The projection $p=(P_t)_{t\in D}\in M_{2^n}(C(D))$ realizes this as a
finitely generated projective module, with
\[
 \Gamma(\mathscr E_D)\cong p\,C(D)^{2^n}.
\]

The normalized trace $\tau(f)=|D|^{-1}\sum_{t\in D}f(t)$ on $C(D)$
extends by the ordinary, unnormalized matrix trace to
\[
 \tau_*\colon M_{2^n}(C(D))\longrightarrow\mathbb C,
 \qquad
 \tau_*((a_{ij})):=\sum_{i=1}^{2^n}\tau(a_{ii}).
\]
Its value on $p$ depends only on the corresponding module isomorphism
class $[p]$; writing $\tau_*([p]):=\tau_*(p)$, this pairing recovers the code dimension:
\begin{equation}
 \tau_*([p])
 =\frac{1}{|D|}\sum_{t\in D}\operatorname{rank}(P_t)
 =2^{n-l}
 =\dim\mathcal C.
 \label{eq:pauli_trace_pairing_dimension}
\end{equation}
The syndrome bundle thus encodes the codespace in its zero fiber, the other
syndrome spaces in the remaining fibers, and the code dimension in its rank.
The generalized construction of Secs.~\ref{sec:GSC} and~\ref{sec:LSC_ncg}
retains the Pontryagin dual
$\widehat D$ as the syndrome space, but allows $D$ to be a general isotropic
lattice. Unlike the finite case, the resulting bundle need not be trivial,
and its section module need not be free.

\subsection{Generalized LCA codes (GLCA codes)}\label{sec:GSC}
\label{a_sec:glca}

Let $M=\mathbb R^p\times\mathbb Z^q\times\mathbb Z_{\dia c}$, where $p,q\geq0$
and $\mathbb Z_{\dia c}=\mathbb Z_{c_1}\times\cdots\times\mathbb Z_{c_k}$ is a
finite abelian group written as a product of cyclic groups.
This form covers every finite abelian group, up to isomorphism.
The physical Hilbert space is $L^2(M)$,
with respect to a fixed Haar measure $\mu_M$, and the phase space is
$G=M\times\widehat M$. Here
$\widehat M=\operatorname{Hom}(M,\mathbb T)$ is the Pontryagin dual of
$M$, consisting of its continuous characters, and
$\langle m,s\rangle:=s(m)$ denotes the canonical pairing between
$M$ and $\widehat M$, where $m \in M$ and $s \in \widehat M$. Pauli stabilizer codes correspond to
$M=\mathbb Z_2^n$. The group operation combines addition in $M$ with
pointwise multiplication in $\widehat M$; for $x=(m,s)$ and $y=(l,t)$,
we denote it by $x+y:=(m+l,st)$.

The Heisenberg cocycle and its associated skew bicharacter are
\[
 \beta((m,s),(l,t)):=\langle m,t\rangle,
 \qquad \rho(x,y):=\beta(x,y)\overline{\beta(y,x)}.
\]
On $L^2(M)$, the shift and modulation operators
$(\mathbf T_m\xi)(r)=\xi(r+m)$ and
$(\mathbf M_s\xi)(r)=\langle r,s\rangle\xi(r)$ satisfy
\begin{equation}\label{eq:ccr_time_modu}
 \mathbf M_s\mathbf T_m
 =\overline{\langle m,s\rangle}\,\mathbf T_m\mathbf M_s.
\end{equation}
The displacement operators $\mathbf D(m,s):=\mathbf M_s\mathbf T_m$
therefore obey
\begin{align}
 \mathbf D(x)\mathbf D(y)&=\beta(x,y)\mathbf D(x+y),
 \label{eq:ccr_semi_displacement_general}\\
 \mathbf D(x)\mathbf D(y)&=\rho(x,y)\mathbf D(y)\mathbf D(x).
 \label{eq:ccr_displacement_general}
\end{align}

A \emph{lattice} $D\subset G$ is a discrete subgroup with compact quotient
$G/D$. Its symplectic dual
\[
 D^\perp:=\{x\in G:\rho(x,w)=1\text{ for all }w\in D\}
\]
is also a lattice, by the duality identifications
$D^\perp\cong\widehat{G/D}$ and
$D\cong\widehat{G/D^\perp}$~\cite[p.~270]{rieffel}.
The bicharacter $\rho$ identifies $G$ with $\widehat G$, under which this
symplectic dual is the ordinary Pontryagin annihilator of $D$
\cite[pp.~264--265, 270]{rieffel}.  The double-annihilator theorem therefore
gives the symplectic double-dual identity
\begin{equation}\label{eq:symplectic_double_annihilator}
 (D^\perp)^\perp=D,
\end{equation}
since the lattice $D$ is closed in $G$.
The lattice is \emph{isotropic} if $D\subseteq D^\perp$.
Its displacements commute, but their multiplier $\beta|_{D\times D}$
must be removed to obtain a representation.

\begin{lemma}\label{lem:phase_triv}
For every isotropic lattice $D\subset G$, there is a function
$c\colon D\to\mathbb T$, with $c(0)=1$, such that
\begin{equation}\label{eq:phase_coboundary}
 \beta(w_1,w_2)=\frac{c(w_1+w_2)}{c(w_1)c(w_2)},
 \qquad w_1,w_2\in D.
\end{equation}
A multiplier of the form~\eqref{eq:phase_coboundary}, for a phase function
$c\colon D\to\mathbb T$ with $c(0)=1$, is called a \emph{coboundary}.
Consequently,
\begin{equation}\label{eq:tilde_displacement}
 \widetilde{\mathbf D}(w):=c(w)\mathbf D(w)
\end{equation}
defines a unitary representation of $D$ on $L^2(M)$.
\end{lemma}
\begin{proof}
Isotropy makes $\beta|_{D\times D}$ symmetric, hence a coboundary
by~\cite[Lemma 7.2]{kleppner1965multipliers}.
Then Equation~\eqref{eq:ccr_semi_displacement_general} gives
$\widetilde{\mathbf D}(w_1)\widetilde{\mathbf D}(w_2)
=\widetilde{\mathbf D}(w_1+w_2)$; unitarity follows from
$|c(w)|=1$.
\end{proof}

\begin{definition}\label{def:stab_group}
For an isotropic lattice $D$ and a phase function $c$ satisfying
\eqref{eq:phase_coboundary}, the \emph{lattice stabilizer group} is
\begin{equation}\label{eq:stabgroup}
 \mathcal S_D:=\{\widetilde{\mathbf D}(w):w\in D\}.
\end{equation}
\end{definition}
The phase function is unique up to multiplication by a character of $D$.
For a subgroup $A\subseteq M$, its \emph{annihilator} is
\[
 A^\perp:=\{s\in\widehat M:\langle a,s\rangle=1
 \text{ for all }a\in A\}\subseteq\widehat M.
\]
This differs from the symplectic dual $D^\perp\subseteq M\times\widehat M$
defined above. A split lattice $D=A\times B$, with $B\subseteq A^\perp$,
is isotropic: indeed, $\beta|_{D\times D}=1$. One may therefore take
$c=1$, recovering the CSS construction.

Ideal codewords of this generalized construction may be distributions rather than elements of $L^2(M)$.
We therefore use the rigged Hilbert space
$\mathscr S(M)\subseteq L^2(M)\subseteq\mathscr S'(M)$.
The Schwartz--Bruhat space $\mathscr S(M)$ consists of functions smooth
in $x\in\mathbb R^p$ with finite seminorms
\begin{equation}\label{eq:schwartz_seminorms}
 q_R(\varphi):=
 \max_{|\alpha|\leq R}
 \sup_{\substack{x\in\mathbb R^p,\ n\in\mathbb Z^q\\f\in\mathbb Z_{\dia c}}}
 (1+|x|+|n|)^R|\partial_x^\alpha\varphi(x,n,f)|,
 \qquad R\in\mathbb Z_{\geq0},
\end{equation}
which define its Fr\'echet topology; $\alpha$ is a multi-index and
$|x|,|n|$ are Euclidean norms. We take $\mathscr S'(M)$ to be the
continuous anti-dual, with the embedding of $\psi\in L^2(M)$ given by
$\varphi\mapsto\langle\varphi,\psi\rangle_{L^2(M)}
=\int_M\overline{\varphi(r)}\psi(r)\,d\mu_M(r)$.
We write $\langle u,\varphi\rangle:=u(\varphi)$ for anti-dual evaluation,
which is linear in $u\in\mathscr S'(M)$ and conjugate-linear in
$\varphi\in\mathscr S(M)$.

If a bounded operator $\mathbf{A}$ and its adjoint preserve $\mathscr S(M)$
continuously, its distributional extension is
\begin{equation}\label{eq:distributional_extension}
 \langle\widehat{\mathbf{A}}u,\varphi\rangle
 :=\langle u,\mathbf{A}^\dagger\varphi\rangle,
 \qquad u\in\mathscr S'(M),\quad\varphi\in\mathscr S(M).
\end{equation}
This applies to the displacement operators and their phase corrections,
with $\widetilde{\mathbf D}(w)^\dagger=\widetilde{\mathbf D}(-w)$.
The \emph{generalized LCA code (GLCA code)} is
the invariant subspace
\[
 \mathcal C:=\{u\in\mathscr S'(M):
 \widehat{\widetilde{\mathbf D}(w)}u=u\text{ for all }w\in D\}.
\]
We often suppress the extension hat when writing stabilizer conditions as
$g\ket\psi=\ket\psi$.

For $r=(x,n,f)\in M=\mathbb R^p\times\mathbb Z^q\times\mathbb Z_{\dia c}$,
with $x\in\mathbb R^p$, $n\in\mathbb Z^q$, and $f\in\mathbb Z_{\dia c}$, the generalized
position ket $\ket r\in\mathscr S'(M)$ is defined by
\begin{equation}\label{eq:canonical_position_kets}
 \langle\ket r,\varphi\rangle:=\overline{\varphi(r)},
 \qquad \varphi\in\mathscr S(M).
\end{equation}
The pairing denotes anti-dual evaluation. We write
\begin{equation}\label{eq:position_subsystem_notation}
 \ket{x,n,f}
 =\ket x_{\mathrm{osc}}\otimes\ket n_{\mathrm{rot}}
 \otimes\ket f_{\mathrm{fin}},
\end{equation}
where the factors denote oscillator position, rotor angular momentum,
and the finite-system computational basis.
The corresponding bra $\bra r$ is the continuous linear functional on
$\mathscr S(M)$ defined by $\bra r\varphi\rangle:=\varphi(r)$.
The position expansion and resolution of the identity are
\begin{equation}\label{eq:position_resolution}
 \ket\psi=\int_M\psi(r)\ket r\,d\mu_M(r),
 \qquad \mathbf1=\int_M\ket r\!\bra r\,d\mu_M(r),
\end{equation}
where $\mathbf1$ denotes the identity operator on $\mathcal H=L^2(M)$,
and the identities are understood weakly on Schwartz test functions.
For $p>0$ the kets are distributions;
for discrete $M$ with counting measure they form the usual orthonormal basis.

\paragraph*{Coordinates.}
Explicit computations use real lifts of the labels in $G$.
Recall that the finite factor is
$\Z_{\dia c}=\Z_{c_1}\times\cdots\times\Z_{c_k}$, with
$\dia c=\operatorname{diag}(c_1,\ldots,c_k)$.
We identify $\widehat{\Z^q}\cong\T^q$ and
$\widehat{\Z_{\dia c}}\cong\Z_{\dia c}$ as in
Sec.~\ref{subsec:examples_notation}.
A lift of a label is a column vector $v=(x,\xi\,|\,n,\theta\,|\,a,b)$.
Its entries list the oscillator positions and momenta, the rotor angular
momenta and angles, and the qudit shifts and phases.
The entries $n$, $a$, and $b$ are integers, so the lifts form the
covering group
\begin{equation}
 \pi:\widetilde G:=\R^{2p}\times\bigl(\Z^q\times\R^q\bigr)\times\Z^{2k}
 \longrightarrow G.
 \label{eq:full_lift_covering}
\end{equation}
The covering map $\pi$ reduces the rotor angles modulo $\Z^q$ and the
shift and phase of the $j$th qudit modulo $c_j$.
The group $\widetilde G$ lies in the real vector space
\begin{equation}
 H^*:=\R^p\times\R^{*p}\times\R^q\times\R^{*q}\times\R^k\times\R^{*k},
 \label{eq:h}
\end{equation}
where $\R^{*r}$ denotes the real dual vector space of $\R^r$.
The matrix $J$ of Eq.~\eqref{eq:symplectic-form} has size $2(p+q+k)$ and
defines an alternating $2$-form on $H^*$.
Recall that $J=\operatorname{diag}(J_{\cv},J_{\mathrm{fin}})$, where
$J_{\cv}=\operatorname{diag}(J_{\mathrm{osc}},J_{\mathrm{rot}})$ is the
oscillator--rotor block and
$J_{\mathrm{fin}}=\left(\begin{smallmatrix}0&\dia c^{-1}\\-\dia c^{-1}&0\end{smallmatrix}\right)$
is the qudit block.
Let $J'$ be obtained by replacing the negative entries of $J$ by zero, so
that $J=J'-(J')^{\intercal}$.

For $v=(x,\xi\,|\,n,\theta\,|\,a,b)$ and
$w=(x',\xi'\,|\,n',\theta'\,|\,a',b')$ in $\widetilde G$, write
$\pi(v)=(m,s)$ and $\pi(w)=(l,t)$.
Then $m$ is represented by $(x,n,a)$, and the character $t$ is
represented by $(\xi',\theta',b')$.
The canonical character pairing is therefore
\begin{equation}
 \langle m,t\rangle
 =\exp\!\left[2\pi i\left(
 x\cdot\xi'+n\cdot\theta'+a^{\intercal}\dia c^{-1}b'
 \right)\right].
\end{equation}
The expression in parentheses is $v^{\intercal}J'w$, so the definitions
of $\beta$ and $\rho$ at the start of this subsection give
\begin{equation}
 \beta(\pi(v),\pi(w))=e^{2\pi i\,v^{\intercal}J'w},
 \qquad
 \rho(\pi(v),\pi(w))=e^{2\pi i\,v^{\intercal}Jw}.
 \label{eq:lifted_cocycle_bicharacter}
\end{equation}
Changing either lift adds only integers to these bilinear expressions:
rotor angles change by integers and qudit coordinates by multiples of
$c_j$, which cancel the denominators in $\dia c^{-1}$. Hence both phases
are independent of the chosen lifts, although $J'w$ itself need not be.
These are precisely the multiplier and commutation
phase of the displacement operators $\mathbf D(m,s)=\mathbf M_s\mathbf T_m$
in Eqs.~\eqref{eq:ccr_semi_displacement_general}
and~\eqref{eq:ccr_displacement_general}.
Lemma~\ref{lem:full_lift_pfaffian} below uses these lifts to compute code
dimensions.

\subsection{Rieffel's module, syndrome bundle, and code fibers}
\label{sec:LSC_ncg}
\label{a_sec:glca_ncg}

\subsubsection{Rieffel's module and syndrome bundle}

For a lattice $D\subset G=M\times\widehat M$, set
$\mathfrak{A}_D:=\mathrm C^*(D,\beta)$ and $\mathfrak{B}_D:=\mathrm C^*(D^\perp,\overline\beta)$.
These algebras contain the dense Schwartz $*$-subalgebras
$\mathscr S(D,\beta)\subset\mathfrak{A}_D$ and
$\mathscr S(D^\perp,\overline\beta)\subset\mathfrak{B}_D$,
respectively~\cite[Section~2]{rieffel}.

Let $u_w\in\mathfrak{A}_D$ and $u_z\in\mathfrak{B}_D$ denote the
canonical unitaries associated with $w\in D$ and $z\in D^\perp$ in
these twisted group $\mathrm C^*$-algebras, respectively
(see Appendix~\ref{appx:c_star}).
The actions $u_w\cdot\xi=\mathbf D(w)\xi$ and
$\xi\cdot u_z=\mathbf D(z)^\dagger\xi$, for $w\in D$ and $z\in D^\perp$,
make $\mathscr S(M)$ a pre-equivalence bimodule over
$\mathscr S(D,\beta)$ and $\mathscr S(D^\perp,\overline\beta)$
(see Appendix~\ref{appx:morita} for Morita equivalence and inner-product
conventions) with inner products
\begin{align}
 {}_{\mathfrak{A}_D}\!\langle\xi,\eta\rangle(w)
 &=\langle\mathbf D(w)\eta,\xi\rangle_{L^2(M)},
 \label{eq:rieffel_left_inner}\\
 \langle\xi,\eta\rangle_{\mathfrak{B}_D}(z)
 &=\langle\xi,\mathbf D(z)\eta\rangle_{L^2(M)},
 \label{eq:rieffel_right_inner}
\end{align}
using compatible Haar measures~\cite[Theorem~2.15]{rieffel}.
We take the Haar measure on $D$ to be counting measure.
Its completion $\Omega_D$, called \emph{Rieffel's module}, is a strong $\mathfrak{A}_D$--$\mathfrak{B}_D$ Morita equivalence bimodule,
finitely generated and projective on both sides since the algebras are
unital~\cite[Section~3]{rieffel}.

Henceforth assume $D$ is isotropic. By Lemma~\ref{lem:phase_triv}, the
phase-corrected operators $\widetilde{\mathbf D}(w)=c(w)\mathbf D(w)$
form a genuine unitary representation of $D$. Transporting the left
module structure through this phase trivialization gives
\[
 \mathfrak{A}_D\cong \mathrm C^*(D)\cong C(\widehat D).
\]
This identification carries $\mathscr S(D,\beta)$ onto the untwisted
Schwartz algebra $\mathscr S(D)\subset\mathrm C^*(D)$, whose action on
$\mathscr S(M)$ is understood in the phase-corrected sense.
We continue to write $u_w$
for the canonical unitary associated with $w\in D$ after passing to the
untwisted algebra $\mathrm C^*(D)$. Then
\begin{equation}\label{eq:phase_corrected_AD_action}
 u_w\cdot\xi:=\widetilde{\mathbf D}(w)\xi=c(w)\mathbf D(w)\xi,
 \qquad \xi\in\mathscr S(M).
\end{equation}
We also transport the left inner product, retaining its notation; henceforth
\[
 {}_{\mathfrak{A}_D}\langle\xi,\eta\rangle(w)
 =\langle\widetilde{\mathbf D}(w)\eta,\xi\rangle_{L^2(M)}
 =\overline{c(w)}\langle\mathbf D(w)\eta,\xi\rangle_{L^2(M)}.
\]
For $a=\sum_{w\in D}a(w)u_w\in\mathscr S(D)$, its left action is
\[
 a\cdot\xi=\sum_{w\in D}a(w)c(w)\mathbf D(w)\xi.
\]
For $\xi\in\mathscr S(M)$, this series converges in the Schwartz topology.
For such $a$ and $\chi\in\widehat D$, we use the Fourier convention
\begin{equation}\label{eq:fourier_transform_convention}
 \widehat a(\chi):=(\mathcal F_Da)(\chi)
 =\sum_{w\in D}a(w)\overline{\chi(w)},
 \qquad \widehat u_w(\chi)=\overline{\chi(w)},
\end{equation}
This transform extends uniquely to the $\mathrm C^*$-algebra isomorphism
$\mathcal F_D:\mathrm C^*(D)\to C(\widehat D)$. Transporting the left module
structure through this Fourier transform makes Rieffel's module
$\Omega_D$ a finitely generated projective module over $C(\widehat D)$,
with $C(\widehat D)$-valued inner product
\begin{equation}\label{eq:fourier_transformed_rieffel_inner_product}
 {}_{\mathrm R}\langle\xi,\eta\rangle
 :=\mathcal F_D\!\left({}_{\mathfrak{A}_D}\langle\xi,\eta\rangle\right).
\end{equation}

We now describe the vector bundle over $\widehat D$ associated with the
projective module $\Omega_D$, beginning with its fibers.\\
For $\chi\in\widehat D$, let $I_\chi:=\{f\in C(\widehat D):f(\chi)=0\}$.
The fiber is
\begin{equation}\label{eq:rieffel_bundle_fiber}
 E_{D,\chi}:=\Omega_D/N_\chi,
 \qquad
 N_\chi:=\overline{I_\chi\Omega_D}
 =\{\xi\in\Omega_D: {}_{\mathrm R}\langle\xi,\xi\rangle(\chi)=0\},
\end{equation}
with scalar inner product and quotient norm
\begin{align}
 \langle[\xi]_\chi,[\eta]_\chi\rangle_\chi
 &:={}_{\mathrm R}\langle\eta,\xi\rangle(\chi),
 \label{eq:rieffel_fiber_scalar_inner_product}\\
 \|[\xi]_\chi\|
 &=\inf_{\zeta\in N_\chi}\|\xi+\zeta\|_{\Omega_D}
 ={}_{\mathrm R}\langle\xi,\xi\rangle(\chi)^{1/2}.
 \label{eq:rieffel_fiber_norm}
\end{align}
The argument reversal makes the scalar inner product conjugate-linear in
its first argument.
By Serre--Swan, the fibers form a finite-rank Hermitian \emph{syndrome bundle}
\begin{equation}\label{eq:rieffel_syndrome_bundle}
 \mathscr E_D:=\bigsqcup_{\chi\in\widehat D}E_{D,\chi}\longrightarrow\widehat D,
\end{equation}
whose continuous-section module is
\begin{equation}\label{eq:serre_swan_evaluation_map}
 \Omega_D\cong\Gamma(\mathscr E_D),
 \qquad \xi\longmapsto\bigl(\chi\longmapsto[\xi]_\chi\bigr).
\end{equation}

\subsubsection{Bloch--Floquet decomposition}

For $a=\sum_{w\in D}a(w)u_w\in\mathbb C[D]$ and
$f\in C(\widehat D)$, define the traces using normalized Haar measure
on $\widehat D$ by
\begin{equation}\label{eq:canonical_trace_haar}
 \mathrm{Tr}_D(a):=a(0),
 \qquad \tau_D(f):=\int_{\widehat D}f(\chi)\,d\chi.
\end{equation}
$\mathrm{Tr}_D$
extends continuously to $\mathrm C^*(D)$, and the traces satisfy
$\mathrm{Tr}_D=\tau_D\circ\mathcal F_D$. In particular,
\begin{equation}\label{eq:rieffel_trace_L2_norm}
 \|\xi\|_{L^2(M)}^2
 =\mathrm{Tr}_D\!\left({}_{\mathfrak{A}_D}\langle\xi,\xi\rangle\right)
 =\tau_D\!\left({}_{\mathrm R}\langle\xi,\xi\rangle\right),
 \qquad \xi\in\mathscr S(M).
\end{equation}

\begin{theorem}\label{thm:rieffel_bundle_decomposition}
The evaluation map
\begin{equation}\label{eq:rieffel_fiberization_map}
 \mathcal U_D^0:\mathscr S(M)\to\Gamma(\mathscr E_D),
 \qquad (\mathcal U_D^0\xi)(\chi)=[\xi]_\chi,
\end{equation}
extends to a unitary Bloch--Floquet transform
\begin{equation}\label{eq:rieffel_direct_integral}
 \mathcal U_D:L^2(M)\xrightarrow{\;\cong\;}
 \int_{\widehat D}^{\oplus}E_{D,\chi}\,d\chi.
\end{equation}
For $a\in \mathrm C^*(D)$, $w\in D$, and $\xi\in L^2(M)$,
\begin{align}
 (\mathcal U_D(a\cdot\xi))(\chi)
 &=\widehat a(\chi)(\mathcal U_D\xi)(\chi),
 \label{eq:rieffel_fiberwise_action}\\
 (\mathcal U_D\widetilde{\mathbf D}(w)\xi)(\chi)
 &=\overline{\chi(w)}(\mathcal U_D\xi)(\chi)
 \label{eq:rieffel_stabilizer_diagonalization}
\end{align}
almost everywhere. For $\xi\in\Omega_D$, these identities hold everywhere
for the canonical continuous representatives.
\end{theorem}
The proof is given in Appendix~\ref{appx:proof_rieffel_bundle_decomposition}.

\subsubsection{Codespace and syndrome fibers}

The continuous inclusion $\mathscr S(M)\to\Omega_D$
(Appendix~\ref{appx:schwartz_heisenberg_continuity}) induces
\begin{equation}\label{eq:rieffel_fiber_distribution_embedding}
 \iota_\chi:E_{D,\chi}\hookrightarrow\mathscr S'(M),
 \qquad
 \langle\iota_\chi(v),\xi\rangle:=\langle[\xi]_\chi,v\rangle_\chi.
\end{equation}
Continuity follows from
$|\langle\iota_\chi(v),\xi\rangle|\leq\|v\|\,\|\xi\|_{\Omega_D}$,
and injectivity from density of the Schwartz images in each fiber.
Using~\eqref{eq:distributional_extension},
\eqref{eq:rieffel_stabilizer_diagonalization}, and
$\widetilde{\mathbf D}(w)^\dagger=\widetilde{\mathbf D}(-w)$ gives
\begin{equation}\label{eq:rieffel_fiber_generalized_eigenvalue}
 \widehat{\widetilde{\mathbf D}(w)}\,\iota_\chi(v)
 =\overline{\chi(w)}\,\iota_\chi(v).
\end{equation}

\begin{lemma}\label{lem:trivial_fiber_codespace}
The trivial-character fiber is the codespace:
\begin{equation}\label{eq:trivial_fiber_codespace}
 \iota_1(E_{D,1})=\mathcal C
 =\{u\in\mathscr S'(M):\widehat{\widetilde{\mathbf D}(w)}u=u
 \text{ for all }w\in D\}.
\end{equation}
\end{lemma}

The proof is given in Appendix~\ref{appx:proof_trivial_fiber_codespace}.

Define
\begin{equation}\label{eq:kappa_definition}
 \kappa_z(w):=\overline{\rho(w,z)},\qquad z\in G,\quad w\in D.
\end{equation}
Pontryagin duality gives~\cite[p.~264]{rieffel} the following exact sequence of abelian groups:
\begin{equation}\label{eq:kappa_exact_sequence}
 0\longrightarrow D^\perp\longrightarrow G
 \xrightarrow{\;\kappa\;}\widehat D\longrightarrow0,
 \qquad G/D^\perp\cong\widehat D.
\end{equation}

\begin{proposition}\label{prop:rieffel_fiber_covariance}
Displacement by $z\in G$ induces a unitary
\begin{equation}\label{eq:rieffel_fiber_covariance}
 V_{z,\chi}:E_{D,\chi}\xrightarrow{\;\cong\;}E_{D,\kappa_z\chi},
 \qquad [\xi]_\chi\longmapsto[\mathbf D(z)\xi]_{\kappa_z\chi},
\end{equation}
intertwining the distributional embeddings:
\begin{equation}\label{eq:fiber_embedding_covariance}
 \iota_{\kappa_z\chi}V_{z,\chi}
 =\widehat{\mathbf D(z)}\,\iota_\chi.
\end{equation}
Each fiber is the full generalized syndrome eigenspace,
\begin{equation}\label{eq:all_fibers_full_eigenspaces}
 \iota_\chi(E_{D,\chi})
 =\{u\in\mathscr S'(M):\widehat{\widetilde{\mathbf D}(w)}u
 =\overline{\chi(w)}u\text{ for all }w\in D\}.
\end{equation}
\end{proposition}
The proof is given in Appendix~\ref{appx:proof_rieffel_fiber_covariance}.

\begin{corollary}\label{cor:vector_bundle}
The syndrome bundle has constant rank
$K=\dim E_{D,\chi}=\dim\mathcal C$ for every $\chi\in\widehat D$.
\end{corollary}

\subsection{Code dimension and covolume}
\label{subsec:code_dimension_covolume}

Let $\mu$ be the self-Plancherel Haar measure on $G$: this is the product
$\mu_M\otimes\mu_{\widehat M}$, where $\mu_{\widehat M}$ is the dual Haar
measure making the Fourier transform
$L^2(M,\mu_M)\to L^2(\widehat M,\mu_{\widehat M})$
unitary~\cite[p.~270]{rieffel}.

For a lattice $D\subset G$, a measurable subset $B\subseteq G$ is called
a \emph{fundamental domain} for $D$ if the translates
\[
 \{B+w:w\in D\}
\]
form a partition of $G$; equivalently, every $x\in G$ has a unique
expression $x=b+w$ with $b\in B$ and $w\in D$.
The \emph{covolume} of $D$ is
\begin{equation}\label{eq:covolume_definition}
 s_D:=|G/D|_\mu:=\mu(B),
\end{equation}
which is independent of the choice of measurable fundamental domain $B$.

\begin{theorem}\label{thm:code_dimension}
The code dimension is
\begin{equation}\label{eq:rank_covolume_formula}
 K=\dim\mathcal C
 =\tau_D([\Omega_D])=\mathrm{Tr}_D([\Omega_D])
 =|G/D|_\mu=s_D.
\end{equation}
The dual covolumes and logical displacement group satisfy
\begin{align}
 |G/D|_\mu\,|G/D^\perp|_\mu&=1,
 \label{eq:dual_covolume_reciprocity}\\
 [D^\perp:D]&=K^2.
 \label{eq:logical_quotient_order}
\end{align}
\end{theorem}
The proof is given in Appendix~\ref{appx:proof_code_dimension}.
Thus $\widehat D$ parametrizes the syndrome fibers, while $D^\perp/D$
labels the $K^2$ logical displacements.

The covolume becomes a determinant once the lattice is lifted to the
covering group $\widetilde G$ of Eq.~\eqref{eq:full_lift_covering}.
The following lemma justifies the matrix computations of
Sec.~\ref{subsec:examples_notation}.

\begin{lemma}
\label{lem:full_lift_pfaffian}
Let $M=\R^p\times\Z^q\times\Z_{\dia c}$ with $p,q,k\geq0$, and put
$C_{\dia c}:=c_1\cdots c_k$.
Let $\pi\colon\widetilde G\to G$ be the covering map of
Eq.~\eqref{eq:full_lift_covering}, with lifted coordinates
$(x,\xi\,|\,n,\theta\,|\,a,b)$.
Its kernel is the period lattice
\begin{equation}
 \Lambda_{\mathrm{per}}:=\ker\pi
 =\{0\}^{2p+q}\times\Z^q\times\dia c\,\Z^k
 \times\dia c\,\Z^k.
 \label{eq:full_lift_periods}
\end{equation}
For a lattice $D\subset G$, we call $\Lambda:=\pi^{-1}(D)$ its
\emph{full lift}.
By part~(i) below, $\Lambda=T\Z^{2(p+q+k)}$ for an invertible
\emph{stabilizer generator matrix} $T$.
Put $\varTheta:=T^{\intercal}JT$, with $J$ as in
Eq.~\eqref{eq:symplectic-form}, and call
$T_{\mathrm{dual}}:=T\varTheta^{-1}$ the \emph{logical matrix}.
\begin{enumerate}[label=(\roman*)]
 \item $\Lambda$ is a full-rank lattice in $\R^{2(p+q+k)}$ with
 $\Lambda_{\mathrm{per}}\subseteq\Lambda\subseteq\widetilde G$ and
 $\Lambda/\Lambda_{\mathrm{per}}\cong D$.
 Conversely, every full-rank lattice between $\Lambda_{\mathrm{per}}$
 and $\widetilde G$ is the full lift of a lattice in $G$.
 \item $D$ is isotropic if and only if $\varTheta$ is an integer matrix.
 \item The covolume is
 $|G/D|_\mu=|\det T|/C_{\dia c}=|\operatorname{Pf}\varTheta|$.
 In particular, $K=|\det T|/C_{\dia c}=|\operatorname{Pf}\varTheta|$
 when $D$ is isotropic.
 \item The symplectic dual
 $\Lambda^\perp:=\{v\in\R^{2(p+q+k)}:v^{\intercal}J\Lambda\subseteq\Z\}$
 equals $T_{\mathrm{dual}}\Z^{2(p+q+k)}=\pi^{-1}(D^\perp)$.
 If $D$ is isotropic, then $\ell\mapsto\pi(T_{\mathrm{dual}}\ell)$
 induces an isomorphism
 $\Z^{2(p+q+k)}/\varTheta\Z^{2(p+q+k)}\cong D^\perp/D$, and
 $\rho\bigl(\pi(T_{\mathrm{dual}}\ell),\pi(T_{\mathrm{dual}}\ell')\bigr)
 =e^{-2\pi i\,\ell^{\intercal}\varTheta^{-1}\ell'}$.
\end{enumerate}
\end{lemma}
The proof is given in Appendix~\ref{appx:proof_full_lift_pfaffian}.
The sign of $\operatorname{Pf}\varTheta$ depends on the ordering of the
lifted coordinates and of the columns of $T$, so we always quote
$K=|\operatorname{Pf}\varTheta|$.

The lemma applies to every lattice in $G$ for every configuration group
$M=\R^p\times\Z^q\times\Z_{\dia c}$.
Since every finite abelian group is a product of cyclic groups, the lemma
covers every finite abelian factor.
It needs no condition on torsion or on the parity of $q$.
The period columns make this possible.
A $\Z$-basis of $D$ itself exists only when $D$ is torsion free, and it then
has $2p+q$ elements.
The Pfaffian of its Gram matrix gives $K$ in general only for $M=\R^p$.
Similarly, the graph-lattice formula of Appendix~\ref{app:graph_lattices}
needs $q=0$ and a torsion-free lattice.
The method relies on the covering group $\widetilde G$, which exists because
$M$ is built from copies of $\R$, $\Z$, and finite cyclic groups.
It does not apply directly to other locally compact abelian groups, such as
an infinite product of qubits.

Two consequences of Lemma~\ref{lem:full_lift_pfaffian} are convenient in
examples.
First, for $\lambda\in\Lambda$ the label $\pi(\lambda)$ has finite order
exactly when $\lambda\in\operatorname{span}_{\R}\Lambda_{\mathrm{per}}$,
since $\Lambda\cap\operatorname{span}_{\R}\Lambda_{\mathrm{per}}$
contains $\Lambda_{\mathrm{per}}$ with finite index.
Hence
$\operatorname{Tor}(D)\cong
(\Lambda\cap\operatorname{span}_{\R}\Lambda_{\mathrm{per}})/\Lambda_{\mathrm{per}}$.
In particular, $D$ is torsion free exactly when $T$ can be chosen with a
basis of $\Lambda_{\mathrm{per}}$ among its columns.
Second, for isotropic $D$, the alternating Smith normal form of $\varTheta$
provides $U\in\mathrm{GL}(2(p+q+k),\Z)$ with
$U^{\intercal}\varTheta U
=\bigoplus_i\bigl(\begin{smallmatrix}0&h_i\\-h_i&0\end{smallmatrix}\bigr)$
and $h_1\mid h_2\mid\cdots$.
Then $D^\perp/D\cong\bigoplus_i\Z_{h_i}^2$, and distinct pairs commute.
The two generators of the $i$th pair have commutation phase
$e^{2\pi i/h_i}$.
Thus $(\mathsf L_D,\rho_D)$ is isomorphic to the logical phase space of
qudits of dimensions $h_i$, and $K=\prod_ih_i$.

\begin{remark}
\label{rem:full_lift_gkp_concatenation}
Lemma~\ref{lem:full_lift_pfaffian} extends the concatenation picture of
Sec.~\ref{ex:mixed-qudit} to every configuration
group $M=\R^p\times\Z^q\times\Z_{\dia c}$.
We keep each oscillator and replace each rotor and each qudit by one
oscillator.
The $i$th qudit is encoded in the square GKP code with
stabilizers $\mathbf D(\sqrt{c_i},0)$ and $\mathbf D(0,\sqrt{c_i})$, via
$\ket{j}\mapsto\sum_{m\in\Z}\ket{(j+c_im)/\sqrt{c_i}}_{\mathrm{osc}}$.
A rotor is encoded in the oscillator states fixed by the single
displacement $\mathbf D(0,1)=e^{2\pi i\hat x}$, via
$\ket{n}_{\mathrm{rot}}\mapsto\ket{n}_{\mathrm{osc}}$.
This constraint confines position to the integers.
Rotor angular momentum thus becomes oscillator position, and the rotor
angle becomes oscillator momentum modulo one.
The resulting injective encoding
$V\colon\mathscr S'(M)\to\mathscr S'(\R^{p+q+k})$ satisfies
\begin{equation}
 V\,\mathbf D(\pi(v))=\mathbf D(Rv)\,V,
 \qquad
 R:=I_{2p}\oplus I_{2q}\oplus\dia c^{-1/2}\oplus\dia c^{-1/2},
 \label{eq:full_lift_concatenation}
\end{equation}
for every lift $v\in\widetilde G$.
Let $J_{\mathrm{osc}}^{(r)}:=\left(\begin{smallmatrix}0&I_r\\-I_r&0\end{smallmatrix}\right)$,
and let $J_{\mathrm{std}}:=J_{\mathrm{osc}}^{(p)}\oplus J_{\mathrm{osc}}^{(q)}
\oplus J_{\mathrm{osc}}^{(k)}$ be the standard form of $p+q+k$
oscillators.
Then $R^{\intercal}J_{\mathrm{std}}R=J$.
Likewise, $R^{\intercal}J'_{\mathrm{std}}R=J'$, where $J'_{\mathrm{std}}$ is
obtained from $J_{\mathrm{std}}$ by replacing its negative entries by zero.
The inner stabilizers $\mathbf D(Rv)$ with $v\in\Lambda_{\mathrm{per}}$,
together with the encoded stabilizers $c(\pi(v))\mathbf D(Rv)$ with
$v\in\Lambda$, generate the concatenated stabilizer group.
By the second identity and Eq.~\eqref{eq:lifted_cocycle_bicharacter}, the
phases $c(\pi(v))$ trivialize the oscillator cocycle on $R\Lambda$.
The label lattice of this group is the rescaled full lift $R\Lambda$, a GKP
lattice with symplectic Gram matrix $\varTheta$.
By Eq.~\eqref{eq:full_lift_concatenation}, the GKP code of $R\Lambda$
with phases $c(\pi(v))$ contains $V(\mathcal C)$.
Theorem~\ref{thm:code_dimension}, applied to $M$ and to
$\R^{p+q+k}$, gives both spaces the dimension
$|\det RT|=|\det T|/C_{\dia c}=K$.
Hence the two codes coincide, and
$(R\Lambda)^\perp/R\Lambda=R\Lambda^\perp/R\Lambda
\cong\Z^{2(p+q+k)}/\varTheta\Z^{2(p+q+k)}$.
Conversely, an isotropic full-rank lattice in $\R^{2(p+q+k)}$
that contains $R\Lambda_{\mathrm{per}}$ lies in
$(R\Lambda_{\mathrm{per}})^\perp=R\widetilde G$.
By Lemma~\ref{lem:full_lift_pfaffian}, it equals $R\pi^{-1}(D)$ for a
unique isotropic lattice $D\subset G$.
The case $p=q=0$ is the concatenation of Sec.~\ref{ex:mixed-qudit}.
For split rotor lattices, the case $p=k=0$ inverts the oscillator-to-rotor
construction of Sec.~\ref{ex:rotor-gkp} with $B_0=\Z^q$.

The correspondence identifies ideal codespaces, stabilizer groups, and
logical displacement groups.
It does not identify error sets.
Indeed, an oscillator displacement $\mathbf D(v)$ with
$v\in\R^{2(p+q+k)}\setminus R\widetilde G$ fails to commute with the inner
stabilizers.
It also does not identify Clifford groups.
Every $W\in\mathrm{GL}(2(p+q+k),\Z)$ with $W^{\intercal}\varTheta W=\varTheta$
gives a Gaussian symplectic map $RTWT^{-1}R^{-1}$ that preserves
$R\Lambda$.
However, $TWT^{-1}$ descends to an element of
$\operatorname{Aut}_{\mathrm{sp}}(D)$ exactly when it maps
$\Lambda_{\mathrm{per}}$ into itself.
For example, the two-mode GKP code of the ququart lattice $D_{\alpha,b}$
of Sec.~\ref{ex:osc-rotor} has a Gaussian logical Fourier gate.
In the basis $(w_1,\ldots,w_4)$ of its full lift, this gate corresponds to
\[
 W=\begin{pmatrix}0&0&0&1\\2&1&0&2\\2&0&1&-2\\-1&0&0&0\end{pmatrix},
\]
which maps the logical classes $[z_X]\mapsto-[z_Z]$ and
$[z_Z]\mapsto[z_X]$ but does not preserve $\Lambda_{\mathrm{per}}$.
In contrast, no element of $\operatorname{Aut}_{\mathrm{sp}}(D_{\alpha,b})$
implements a logical Fourier gate (Sec.~\ref{subsec:clifford_examples}).
\end{remark}

\subsection{Logical displacement operators}
\label{sec:logical_displacement_operators}

We now describe the logical operators of an arbitrary lattice stabilizer code.
With the stabilizer phase function $c$ fixed as above, the construction
requires neither lattice generators nor coordinates on $G$.  The logical
phase space $D^\perp/D$ and its commutation pairing are intrinsic to the
embedded isotropic lattice $D\subseteq G$ with the ambient pairing $\rho$;
the concrete codespace and its logical operators also use the choice of $c$.

The action of the displacement operators on the syndrome bundle was
described in Proposition~\ref{prop:rieffel_fiber_covariance}.  For
$z\in G$, recall the character
\[
 \kappa_z(w)=\overline{\rho(w,z)},
 \qquad w\in D.
\]
Displacement by $z$ induces a unitary from the fiber over $\chi$ to the fiber
over $\kappa_z\chi$,
\begin{equation}
 V_{z,\chi}:E_{D,\chi}\longrightarrow E_{D,\kappa_z\chi},
 \qquad
 [\xi]_\chi\longmapsto[\mathbf D(z)\xi]_{\kappa_z\chi}.
 \label{eq:logical_fiber_action}
\end{equation}
This is the precise sense in which displacement operators respect the
Bloch--Floquet decomposition.  It is a pointwise statement about the fibers
of the syndrome bundle.  By the exact sequence
\eqref{eq:kappa_exact_sequence}, $\kappa_z=1$ exactly when
$z\in D^\perp$.  Hence every displacement indexed by $D^\perp$ preserves
each individual syndrome fiber and acts there by
\begin{equation}
 V_{z,\chi}:E_{D,\chi}\longrightarrow E_{D,\chi}.
 \label{eq:dual_lattice_fiberwise_action}
\end{equation}
If $z\notin D^\perp$, then $\kappa_z$ is nontrivial and the displacement
moves every syndrome fiber to a different fiber.

We now specialize this fiberwise action to the trivial syndrome.  By
Lemma~\ref{lem:trivial_fiber_codespace}, the logical Hilbert space is
identified with
\[
 \mathcal C\cong E_{D,1}.
\]
For $z\in D^\perp$, Equation~\eqref{eq:logical_fiber_action} gives a unitary
endomorphism of this fiber.  We define the corresponding logical
displacement by
\begin{equation}
 \mathbf L_z:=V_{z,1}\in\mathscr U(E_{D,1}),
 \qquad
 \mathbf L_z[\xi]_1=[\mathbf D(z)\xi]_1.
 \label{eq:logical_displacement_definition}
\end{equation}
Under the distributional embedding $\iota_1$, this is precisely the
restriction of the distributional action of $\mathbf D(z)$ to the
codespace.  This preservation can also be seen directly from the Weyl
relations.  If $u\in\mathcal C$, $w\in D$, and $z\in D^\perp$, then
\begin{align}
 \widehat{\widetilde{\mathbf D}(w)}\,
 \widehat{\mathbf D(z)}u
 &=\rho(w,z)\widehat{\mathbf D(z)}
   \widehat{\widetilde{\mathbf D}(w)}u \notag\\
 &=\widehat{\mathbf D(z)}u,
 \label{eq:logical_displacement_preserves_code}
\end{align}
because $\rho(w,z)=1$.  Thus $\widehat{\mathbf D(z)}u$ again lies in
$\mathcal C$.

Conversely, if $z\notin D^\perp$, there is a $w\in D$ for which
$\rho(w,z)\neq1$.  The same calculation shows that displacement by $z$
changes the stabilizer eigenvalue and therefore moves the state out of the
trivial-syndrome fiber.  Consequently $D^\perp$ is exactly the subgroup of
phase-space displacements that centralizes the stabilizers, while $D$ itself
acts trivially on the code after the phase correction.  The quotient
$D^\perp/D$ is therefore the natural analogue of the logical Pauli group
modulo phases.  The next theorem makes this statement precise and identifies
the full algebra generated by these logical displacements.

\begin{theorem}
\label{thm:logical_displacement_algebra}
Let
\begin{equation}
 \mathsf L_D:=D^\perp/D.
 \label{eq:logical_phase_space}
\end{equation}
Then $\mathsf L_D$ is a finite abelian group of order $K^2$, where
$K=\dim\mathcal C$.  Write
$\operatorname{PU}(E_{D,1}):=\mathscr U(E_{D,1})/\mathbb T$ for the
projective unitary group, obtained by quotienting the unitary group by
scalar phases.  The projective class of $\mathbf L_z$ depends only on
$[z]\in\mathsf L_D$, and the map
\begin{equation}
 \mathsf L_D\longrightarrow\operatorname{PU}(E_{D,1}),
 \qquad
 [z]\longmapsto[\mathbf L_z]
 \label{eq:logical_projective_representation}
\end{equation}
is a faithful group homomorphism.  Equivalently, the operators
$\mathbf L_z$ give a faithful projective unitary representation of
$\mathsf L_D$.

The bicharacter $\rho$ descends to a nondegenerate alternating bicharacter
\begin{equation}
 \rho_D:\mathsf L_D\times\mathsf L_D\longrightarrow\mathbb T,
 \qquad
 \rho_D([z],[z']):=\rho(z,z'),
 \label{eq:logical_symplectic_bicharacter}
\end{equation}
and the logical displacements obey the Weyl relations
\begin{align}
 \mathbf L_z\mathbf L_{z'}
 &=\beta(z,z')\mathbf L_{z+z'},
 \label{eq:logical_projective_product}\\
 \mathbf L_z\mathbf L_{z'}
 &=\rho_D([z],[z'])
   \mathbf L_{z'}\mathbf L_z.
 \label{eq:logical_weyl_relation}
\end{align}
Finally, if $R\subset D^\perp$ is any set of representatives for
$D^\perp/D$, then
\begin{equation}
 \{\mathbf L_z:z\in R\}
 \label{eq:logical_displacement_basis_set}
\end{equation}
is an orthogonal basis of
$\operatorname{End}(E_{D,1})\cong M_K(\mathbb C)$ for the
Hilbert--Schmidt inner product.  In particular, the logical displacement
operators generate the full logical operator algebra.
\end{theorem}

The proof is given in Appendix~\ref{appx:proof_logical_displacement_algebra}.

The theorem yields the usual stabilizer trichotomy for displacement
operators.  Elements of $D$ act trivially on the code up to phase; elements
of $D^\perp\setminus D$ are nontrivial logical displacements; and elements of
$G\setminus D^\perp$ move the trivial-syndrome fiber to a nontrivial syndrome
fiber.  Choosing generators of the finite logical phase space $D^\perp/D$
and evaluating their displacement operators gives generating logical
displacements.  To obtain standard paired logical Pauli generators
$X_i,Z_i$, choose generator pairs adapted to the nondegenerate alternating
pairing $\rho_D$.

\subsection{Clifford gates from lattice symmetries}
\label{sec:clifford_gates}

Section~\ref{sec:logical_displacement_operators} identified the finite
abelian group
$\mathsf L_D=D^\perp/D$, equipped with the nondegenerate alternating
bicharacter $\rho_D$, as the intrinsic logical phase space of the code.
We now ask which symmetries of the physical phase space preserve the lattice
$D$ and therefore induce Clifford transformations of $\mathsf L_D$.
For GKP codes, the relation between lattice symmetries and logical Clifford
operations, including the displacement corrections needed to preserve
stabilizer phases, is described by Burchards, Flammia, and
Conrad~\cite[Sec.~4.2]{fiber-bundle-ft}.  We formulate this construction
for the general LCA phase spaces considered here.

\subsubsection{Symplectic automorphisms and their implementers}
\label{subsec:clifford_general}

Let $\operatorname{Sp}(G,\rho)$ denote the group of continuous group
automorphisms $S:G\to G$ satisfying
\begin{equation}
 \rho(Sx,Sy)=\rho(x,y)
 \qquad(x,y\in G).
 \label{eq:ambient_symplectic_automorphism}
\end{equation}
The symplectic automorphism group of the code lattice is
\begin{equation}
 \operatorname{Aut}_{\mathrm{sp}}(D)
 :=\{S\in\operatorname{Sp}(G,\rho):S(D)=D\}.
 \label{eq:lattice_symplectic_automorphism_group}
\end{equation}
This definition depends only on the embedded lattice $D\subseteq G$ and not
on a choice of lattice generators.

\begin{proposition}
\label{prop:lattice_symmetry_logical_phase_space}
Every $S\in\operatorname{Aut}_{\mathrm{sp}}(D)$ satisfies
$S(D^\perp)=D^\perp$ and hence induces an automorphism
\begin{equation}
 \overline S:\mathsf L_D\longrightarrow\mathsf L_D,
 \qquad
 \overline S[z]=[Sz].
 \label{eq:induced_logical_symplectic_map}
\end{equation}
Moreover, $\overline S$ preserves the logical bicharacter:
\begin{equation}
 \rho_D(\overline S[z],\overline S[z'])
 =\rho_D([z],[z']).
 \label{eq:induced_map_preserves_logical_bicharacter}
\end{equation}
Consequently there is a canonical homomorphism
\begin{equation}
 \Phi_D:\operatorname{Aut}_{\mathrm{sp}}(D)
 \longrightarrow\operatorname{Sp}(\mathsf L_D,\rho_D),
 \qquad
 S\longmapsto\overline S.
 \label{eq:lattice_to_logical_symplectic_homomorphism}
\end{equation}
\end{proposition}

\begin{proof}
Let $x\in D^\perp$ and $w\in D$.  Since $S(D)=D$, we may write
$w=Sv$ for some $v\in D$.  Symplecticity gives
\[
 \rho(Sx,w)=\rho(Sx,Sv)=\rho(x,v)=1.
\]
Thus $Sx\in D^\perp$, and applying the same argument to $S^{-1}$ gives
$S(D^\perp)=D^\perp$.  Therefore
Equation~\eqref{eq:induced_logical_symplectic_map} is well defined.  Finally,
\[
 \rho_D([Sz],[Sz'])
 =\rho(Sz,Sz')
 =\rho(z,z')
 =\rho_D([z],[z']),
\]
which proves the remaining assertions.
\end{proof}

The Segal--Shale--Weil construction provides unitary implementers for the
symplectic automorphisms considered here
\cite[Chap.~I, \S\S10--12]{weil1964unitaires}.

\begin{proposition}
\label{prop:symplectic_implementer}
Every $S\in\operatorname{Sp}(G,\rho)$ admits a continuous phase function
$a_S:G\to\mathbb T$ and a unitary $\mathcal U_S$ on $L^2(M)$ such that
both $\mathcal U_S$ and $\mathcal U_S^\dagger$ preserve $\mathscr S(M)$
continuously and
\begin{equation}
 \mathcal U_S\mathbf D(x)\mathcal U_S^\dagger
 =a_S(x)\mathbf D(Sx),
 \qquad x\in G.
 \label{eq:clifford_implementer}
\end{equation}
For a fixed phase function $a_S$, the implementer $\mathcal U_S$ is unique
up to an overall scalar phase.
\end{proposition}

\begin{proof}
Symplecticity makes
$\delta_S(x,y):=\beta(Sx,Sy)/\beta(x,y)$ a continuous symmetric cocycle.
Thus $G\times\mathbb T$, with multiplication
$(x,t)(y,u)=(x+y,tu\delta_S(x,y))$, is an abelian LCA group.
By the Pontryagin character-extension theorem, the identity character
on $\{0\}\times\mathbb T$ extends to a continuous character
$h(x,t)=b(x)t$ of this group. Its multiplicativity gives
$b(x)b(y)=\delta_S(x,y)b(x+y)$, so $a_S=b^{-1}$ satisfies
\[
 a_S(x+y)=a_S(x)a_S(y)\delta_S(x,y).
\]
Consequently $(x,t)\mapsto(Sx,a_S(x)t)$ is an automorphism of the
Heisenberg group with multiplier $\beta$.
The Segal--Shale--Weil construction supplies its unitary implementer
and continuous Schwartz preservation
\cite[Chap.~I, \S\S10--12]{weil1964unitaires}.
Uniqueness up to scalar phase follows from irreducibility of the
displacement representation.
\end{proof}

We call such an implementer a generalized Clifford, or metaplectic,
operator associated with $S$.  The proposition asserts existence for each
$S$; it does not specify a compatible choice of implementers as $S$ varies.

There is a small phase subtlety.  Even when $S(D)=D$, the implementer
$\mathcal U_S$ need not preserve the trivial-syndrome fiber by itself,
because the stabilizers were defined using the phase correction $c$.  The
following theorem shows that a displacement always removes this mismatch.

\begin{theorem}
\label{thm:corrected_clifford_implementer}
Let $S\in\operatorname{Aut}_{\mathrm{sp}}(D)$ and choose a unitary
implementer $\mathcal U_S$ as in
Proposition~\ref{prop:symplectic_implementer}.
There exists $z_S\in G$ such that
\begin{equation}
 \mathcal U_{S,D}:=\mathbf D(z_S)\mathcal U_S
 \label{eq:corrected_clifford_operator}
\end{equation}
preserves the codespace $\mathcal C$.  Its restriction to $\mathcal C$ is a
logical Clifford gate and satisfies
\begin{equation}
 \mathcal U_{S,D}\mathbf L_x\mathcal U_{S,D}^\dagger
 =\lambda_S(x)\mathbf L_{Sx},
 \qquad x\in D^\perp,
 \label{eq:logical_clifford_conjugation}
\end{equation}
for phases $\lambda_S(x)\in\mathbb T$.  Thus its action on logical
displacement labels is precisely
\begin{equation}
 [x]\longmapsto[Sx]=\overline S[x]
 \qquad\text{in }D^\perp/D.
 \label{eq:logical_clifford_label_action}
\end{equation}
The choice of $z_S$ is unique only modulo $D^\perp$; changing it by an
element of $D^\perp$ changes $\mathcal U_{S,D}$ by a logical displacement, up to
an overall phase.
\end{theorem}

The proof is given in Appendix~\ref{appx:proof_corrected_clifford_implementer}.


\subsubsection{Examples}
\label{subsec:clifford_examples}

We now determine lattice symmetries and their logical actions for several
codes of Sec.~\ref{sec:examples}.

\paragraph{GKP and LCA codes.}
The GKP and LCA codes of Secs.~\ref{ex:gkp} and~\ref{ex:lca} give
concrete instances of the lattice symmetries of
Sec.~\ref{subsec:clifford_general}; see also the
noncommutative-torus
description in~\cite[Appendix~G]{yr64-ypxt} (see also \cite{chakluef2019metaplectic}, \cite{chaktracing2023tracing}, \cite{2505.19869}).
For a purely continuous GKP code, define
\begin{equation}
 \operatorname{Aut}(\varTheta)
 :=\{W\in\mathrm{GL}(2p,\mathbb Z):
       W^{\intercal}\varTheta W=\varTheta\}.
 \label{eq:gkp_gram_automorphisms}
\end{equation}
Since $T^{\intercal}J_{\mathrm{osc}}T=\varTheta$, each such $W$ gives
$S_W:=TWT^{-1}\in\mathrm{Sp}(2p,\mathbb R)$ with $S_W(D)=D$.
Thus $W\mapsto S_W$ identifies $\operatorname{Aut}(\varTheta)$ with
the subgroup $\operatorname{Aut}_{\mathrm{sp}}(D)\subseteq
\operatorname{Sp}(G,\rho)$ defined in
Eq.~\eqref{eq:lattice_symplectic_automorphism_group}.
Its metaplectic implementer, followed when necessary by the displacement
correction of Theorem~\ref{thm:corrected_clifford_implementer}, therefore
acts as a logical Clifford gate. For the GKP Clifford action and its
geometric fault-tolerance interpretation, see~\cite[Secs.~4.2 and~6]{fiber-bundle-ft}.
To describe this action explicitly, distinguish a physical displacement
vector $v\in D^\perp$ from its integer coordinates and its logical class.
As in Sec.~\ref{ex:gkp}, we generate the symplectic dual with the
logical matrix,
\[
 T_{\mathrm{dual}}=T\varTheta^{-1},\qquad D^\perp=T_{\mathrm{dual}}(\mathbb Z^{2p}).
\]
Thus every $v\in D^\perp$ has a unique expression $v=T_{\mathrm{dual}}\ell$, with
$\ell\in\mathbb Z^{2p}$. Here $\ell$ is the coordinate vector of $v$
in the dual-lattice basis given by the columns of $T_{\mathrm{dual}}$; equivalently,
$\ell=T^{\intercal}J_{\mathrm{osc}}v$.
Two such displacements have the same logical action up to phase exactly
when their vectors differ by a stabilizer vector $Tm$, with
$m\in\mathbb Z^{2p}$. Since $T=T_{\mathrm{dual}}\varTheta$, in these coordinates
$v\mapsto v+Tm$ becomes $\ell\mapsto\ell+\varTheta m$. Consequently,
\[
 D^\perp/D\cong\mathbb Z^{2p}/\varTheta\mathbb Z^{2p},
 \qquad [v]\longleftrightarrow[\ell].
\]
The class $[\ell]$ therefore labels a logical displacement operator
modulo phase.
Finally, $W^{\intercal}\varTheta W=\varTheta$ implies
$W\varTheta^{-1}=\varTheta^{-1}W^{-\intercal}$, so
$S_W T_{\mathrm{dual}}=T_{\mathrm{dual}}W^{-\intercal}$ and hence $S_Wv=T_{\mathrm{dual}}W^{-\intercal}\ell$.
The induced action is
\begin{equation}
 [\ell]\longmapsto[W^{-\intercal}\ell].
 \label{eq:gkp_clifford_dual_labels}
\end{equation}

For the single-mode oscillator--qudit code of Sec.~\ref{ex:lca}
(see also~\cite[Appendix G(2)]{yr64-ypxt}), we use the integers $c$, $d$,
and $\theta$ of that code and assume $d\ne0$.
A sufficient condition for such a symmetry to descend to the hybrid phase
space is then particularly simple. Choose
$W=\left(\begin{smallmatrix}w_{11}&w_{12}\\w_{21}&w_{22}\end{smallmatrix}\right)
\in\mathrm{SL}(2,\mathbb Z)$ with $d\mid w_{21}$, and set
\begin{equation}
 \begin{aligned}
 W_c&:=T_{\mathrm{fin}}WT_{\mathrm{fin}}^{-1}
      =\begin{pmatrix}w_{11}&-dw_{12}\\-w_{21}/d&w_{22}\end{pmatrix},\\
 S_W&:=\operatorname{diag}
       \bigl(T_{\cv}WT_{\cv}^{-1},W_c\bigr).
 \end{aligned}
 \label{eq:hybrid_graph_clifford_lift}
\end{equation}
Here $T_{\cv}$ and $T_{\mathrm{fin}}=\operatorname{diag}(-d,1)$ are the oscillator
and qudit blocks of the graph-lattice embedding of Sec.~\ref{ex:lca}.
The $W_c$ block of $S_W$ acts modulo $c$ on the physical qudit
phase space. The condition $d\mid w_{21}$ makes $W_c$ integral, and
$\det W_c=1$, so its reduction modulo $c$ is a symplectic automorphism.
Moreover, $S_W$ preserves the symplectic pairing and satisfies
$S_W\binom{T_{\cv}}{T_{\mathrm{fin}}}=\binom{T_{\cv}}{T_{\mathrm{fin}}}W$,
so it preserves $D$. Its implementer can be chosen as a product of an
oscillator metaplectic operator and a qudit Clifford operator, followed
by the displacement correction in
Theorem~\ref{thm:corrected_clifford_implementer}. Its logical action is
$[v]\mapsto[S_Wv]$ on $D^\perp/D$, whose order is $K^2$ with
$K=|c\theta+d|$; the modulus $c$ above describes the physical qudit.
This gives a concrete family of hybrid Clifford symmetries, including
when $\theta=0$. The explicit construction and circuit synthesis of
gates for general GKP--LCA families are beyond the scope of
this paper and will be presented in future work.

\paragraph{Oscillator--rotor ququart.}
The logical ququart Fourier gate cannot be implemented by a continuous
symplectic automorphism preserving the lattice $D_{\alpha,b}$ of
Sec.~\ref{ex:osc-rotor}.  The rotor component
$n$ changes only by a sign under any continuous automorphism of $G$.
On $D_{\alpha,b}^{\perp}$, the residue $n\bmod2b$ is well defined
modulo $D_{\alpha,b}$ and distinguishes $[z_X]$, with residue $b$,
from $[z_Z]$, with residue zero.  A Fourier gate would exchange these
logical classes up to inversion, which is impossible.
In contrast, the two-mode GKP code of the full lift of $D_{\alpha,b}$
has a Gaussian logical Fourier gate
(Remark~\ref{rem:full_lift_gkp_concatenation}).

\paragraph{Rotor--qudit code.}
No generalized Clifford operator implements a logical Fourier gate on the
code of Sec.~\ref{ex:rotor-qudit}.
For $d=2$, this gate is the logical Hadamard gate.
The argument has two steps.

First, let $\mathcal U$ be a generalized Clifford operator for some
$S\in\operatorname{Sp}(G,\rho)$, and suppose that $\mathcal U$ preserves
the codespace $\mathcal C$ of an isotropic lattice $D$.
Conjugation by $\mathcal U$ maps each phase-corrected stabilizer
$\widetilde{\mathbf D}(w)$ to a multiple of $\mathbf D(Sw)$.
This multiple acts on $\mathcal C$ as the identity, because $\mathcal U$
maps $\mathcal C$ onto itself.
A displacement acts on $\mathcal C$ as a scalar only if its label lies in
$D$, by Sec.~\ref{sec:logical_displacement_operators} and
Theorem~\ref{thm:logical_displacement_algebra}.
Hence $S(D)\subseteq D$.
The same argument for $\mathcal U^\dagger$ gives
$S\in\operatorname{Aut}_{\mathrm{sp}}(D)$.

Second, recall the logical operators
$\xl=\mathbf D(0,0\,|\,{-1},0)$ and
$\zl=\mathbf D(0,-\tfrac1{d^2}\,|\,0,1)$ of
Sec.~\ref{ex:rotor-qudit}.
Every element of $G$ with nonzero angular momentum has infinite order.
The representatives with zero angular momentum of the class
$k[\zl]+l[\xl]\in D^\perp/D$ are
$(0,\tfrac md-\tfrac k{d^2}\,|\,{-l},k)$ with $m\in\Z$.
Such an element has order dividing $d$ exactly when $d\mid k$.
Thus the classes that contain an element of order dividing $d$ form the
subgroup generated by $[\xl]$.
Every $S\in\operatorname{Aut}_{\mathrm{sp}}(D)$ preserves the orders of
elements of $G$ and maps classes onto classes.
The induced map $\overline S$ of
Eq.~\eqref{eq:induced_logical_symplectic_map} therefore preserves this
subgroup.
A logical Fourier gate would instead map $[\xl]$ to $\pm[\zl]$, which lies
outside it.
A logical Fourier gate therefore requires non-Clifford operations.
For a rotor and a qubit, a rotor cosine potential and one fixed
qubit--rotor conditional phase extend local Clifford operations to universal
control~\cite{luo2026hybridqubitrotor}.

\paragraph{Torsion-extended hybrid code.}
The logical Hadamard and phase gates $U_H$ and $U_S$ of
Sec.~\ref{ex:torsion}, given in
Eq.~\eqref{eq:hybrid-torsion-clifford-gates}, have the following
symplectic automorphisms.
In the paired coordinates $(x,\xi\,|\,n,\theta\,|\,a,b)$ of
Sec.~\ref{subsec:examples_notation}, where $(a,b)$ are the qubit shift and
phase,
\begin{align}
 S_H(x,\xi\,|\,n,\theta\,|\,a,b)
 &=\bigl(n-8\xi,\tfrac{x-n}{8}\,\big|\,n,
          \theta+\xi-\tfrac{x-n}{8}\,\big|\,a,b\bigr),\notag\\
 S_S(x,\xi\,|\,n,\theta\,|\,a,b)
 &=\bigl(x,\xi-\tfrac{x-n}{8}\,\big|\,n,
          \theta+\tfrac{x-n}{8}\,\big|\,a,b\bigr).
 \label{eq:hybrid-torsion-clifford-automorphisms}
\end{align}
Both preserve $D_2$, and the displayed implementers preserve its
stabilizer phases without an additional displacement correction.

\section{Code distance and error correction}
\label{sec:distance_error_correction}

Section~\ref{sec:logical_displacement_operators} identifies
$\mathsf L_D=D^\perp/D$ as the logical phase space and
$D^\perp\setminus D$ as the set of nontrivial logical displacement labels.
We now equip the physical phase space with a displacement cost and use
this structure to formulate distance and decoding.  Throughout, $D\subseteq
G=M\times\widehat M$ is the isotropic lattice defining the code, and
syndrome extraction is ideal. Write $K=\dim\mathcal C$ as before.
With $\kappa_z(d)=\overline{\rho(d,z)}$ as in~\eqref{eq:kappa_definition},
displacement by $z$ induces the fiber map
$V_{z,1}:E_{D,1}\to E_{D,\kappa_z}$ of~\eqref{eq:logical_fiber_action}.
The phase-corrected stabilizers act on the target fiber with eigenvalue
character $\overline{\kappa_z}$.

\subsection{Correctable displacement sets and distance}

For an error set $\mathcal E\subseteq G$, put
$\mathcal E-\mathcal E=\{e_2-e_1:e_1,e_2\in\mathcal E\}$.
The relative operator $\mathbf D(e_1)^\dagger\mathbf D(e_2)$ is a scalar
phase times $\mathbf D(e_2-e_1)$.  The stabilizer trichotomy established in
Section~\ref{sec:logical_displacement_operators} therefore gives the exact
correctability condition
\begin{equation}
 (\mathcal E-\mathcal E)\cap(D^\perp\setminus D)=\varnothing.
 \label{eq:distance_relative_error_criterion}
\end{equation}
For normalizable code subspaces this is the displacement form of the
Knill--Laflamme criterion: each relative error either changes the syndrome
or acts as a scalar on the code.  For the ideal distributional codes of
Section~\ref{sec:LSC_ncg}, we use
\eqref{eq:distance_relative_error_criterion} as the algebraic condition
for exact correction.  Theorem~\ref{thm:distance_correctable_decoder}
below identifies it with the existence of a syndrome decoder for Borel
error sets.

\begin{definition}
A function $w:G\to[0,\infty]$ is a displacement cost if
$w(0)=0$ and $w(-z)=w(z)$.  Its associated code distance is
\begin{equation}
 d_w(D):=\inf_{z\in D^\perp\setminus D}w(z),
 \qquad \inf\varnothing:=\infty.
 \label{eq:displacement_cost_distance}
\end{equation}
\end{definition}
Thus the lattice determines which displacements are logical, while the
cost specifies which errors are small.  Define the quotient cost on
stabilizer-equivalence classes by
\begin{equation}
 \overline w(z+D):=\inf_{d\in D}w(z+d).
 \label{eq:distance_quotient_cost}
\end{equation}
Its restriction to $\mathsf L_D$ gives
$d_w(D)=\inf_{\ell\in\mathsf L_D\setminus\{0\}}\overline w(\ell)$.
Indeed, choosing a representative $z_\ell\in D^\perp$ for each
$\ell\in\mathsf L_D=D^\perp/D$ gives
$D^\perp\setminus D=\bigsqcup_{\ell\ne0}(z_\ell+D)$, hence
$\inf_{z\in D^\perp\setminus D}w(z)
=\inf_{\ell\ne0}\inf_{d\in D}w(z_\ell+d)$.
Neither subadditivity nor attainment of the infimum is required for these
definitions.

\begin{proposition}
\label{prop:distance_correctable_ball}
Suppose $w$ is subadditive, so that $w(x+y)\leq w(x)+w(y)$.
For $r>0$, let $B_w(r):=\{z\in G:w(z)<r\}$ be the strict cost sublevel set.  If $2r\leq d_w(D)$,
then $B_w(r)$ satisfies the correctability criterion~\eqref{eq:distance_relative_error_criterion}.
\end{proposition}
\begin{proof}
For $e_1,e_2\in B_w(r)$, symmetry and subadditivity imply
\[
 w(e_2-e_1)\leq w(e_2)+w(e_1)<2r\leq d_w(D).
\]
Consequently $e_2-e_1$ cannot lie in $D^\perp\setminus D$.
\end{proof}

\subsection{Syndrome decoders and exact correction}

Using $G/D^\perp\cong\widehat D$ from~\eqref{eq:kappa_exact_sequence},
write the quotient map in~\eqref{eq:quotient_exact_sequence} as
\begin{equation}
 \varphi:G/D\longrightarrow\widehat D,
 \qquad \varphi(z+D)=\kappa_z,
 \qquad \pi_D:G\longrightarrow G/D,\quad \pi_D(z)=z+D.
 \label{eq:decoding_quotient_map}
\end{equation}
The kernel of $\varphi$ is $\mathsf L_D$, and every fiber
$\varphi^{-1}(\chi)$ contains $K^2$ compatible error classes, where
$K=\dim\mathcal C$.  Decoding selects one such class from the measured
syndrome.  A continuous selection is not required; a Borel selection
suffices for the probabilistic statements below.

Lemma~\ref{lem:decoding_borel_section} in
Appendix~\ref{appx:decoding_borel_section} establishes the existence of
the Borel selections needed to construct syndrome decoders.

\begin{definition}
A decoder is a Borel section
\begin{equation}
 \operatorname{Dec}:\widehat D\longrightarrow G/D,
 \qquad \varphi\circ\operatorname{Dec}=\operatorname{id}_{\widehat D}.
 \label{eq:syndrome_decoder_definition}
\end{equation}
\end{definition}
Lemma~\ref{lem:decoding_borel_section} gives a reference decoder $\sigma_0$.
Every decoder is uniquely of the form
\begin{equation}
 \operatorname{Dec}_{\ell}(\chi)=\sigma_0(\chi)+\ell(\chi),
 \qquad \ell:\widehat D\longrightarrow\mathsf L_D\ \text{Borel}.
 \label{eq:decoder_normal_form}
\end{equation}
Indeed, the difference of two sections lies in $\ker\varphi=\mathsf L_D$,
and adding any such Borel function to $\sigma_0$ gives another section.

If the true error is $z$ and $\operatorname{Dec}(\kappa_z)=u+D$, applying
$\mathbf D(-u)$ leaves, up to phase, the residual displacement
$\mathbf D(z-u)$.  The section property ensures $z-u\in D^\perp$.
By Theorem~\ref{thm:logical_displacement_algebra}, its action on $E_{D,1}$
is scalar exactly when $z-u\in D$.  Hence
\begin{equation}
 \text{correction succeeds on }z
 \quad\Longleftrightarrow\quad
 \operatorname{Dec}(\kappa_z)=\pi_D(z).
 \label{eq:decoder_success_criterion}
\end{equation}
This condition is independent of the representative $u$; on failure the
residual class $[z-u]\in\mathsf L_D$ is nonzero.  A Borel representative
$u$ of each chosen class can be selected using
Lemma~\ref{lem:decoding_borel_section} for $\pi_D$.

\begin{theorem}
\label{thm:distance_correctable_decoder}
For a Borel set $\mathcal E\subseteq G$, the following are equivalent:
\begin{enumerate}[label=(\roman*)]
 \item There exists a Borel decoder that succeeds on every error in $\mathcal E$.
 \item For $e_1,e_2\in\mathcal E$, equality $\kappa_{e_1}=\kappa_{e_2}$
 implies $e_1+D=e_2+D$.
 \item $(\mathcal E-\mathcal E)\cap(D^\perp\setminus D)=\varnothing$.
\end{enumerate}
\end{theorem}
\begin{proof}
A successful decoder must assign the correct class to each syndrome, so
(i) implies (ii) by~\eqref{eq:decoder_success_criterion}.  Conditions
(ii) and (iii) are equivalent because $\ker\kappa=D^\perp$.

Assume (ii) and put $\Sigma_{\mathcal E}:=\kappa(\mathcal E)$.
The restriction $\kappa|_{\mathcal E}$ has countable fibers, each contained
in a coset of $D^\perp$.  Lemma~\ref{lem:decoding_borel_section} shows that
$\Sigma_{\mathcal E}$ is Borel and provides a Borel map
$r:\Sigma_{\mathcal E}\to\mathcal E$ with $\kappa_{r(\chi)}=\chi$.
Condition (ii) ensures that $\pi_D(r(\chi))$ is the correct class for every
error in $\mathcal E$ with syndrome $\chi$.  Define
\begin{equation}
 \operatorname{Dec}(\chi)=
 \begin{cases}
  \pi_D(r(\chi)),&\chi\in\Sigma_{\mathcal E},\\
  \sigma_0(\chi),&\chi\notin\Sigma_{\mathcal E}.
 \end{cases}
 \label{eq:decoder_correctable_set_extension}
\end{equation}
This is a Borel section of $\varphi$ and succeeds on all of $\mathcal E$.
\end{proof}

\subsection{Minimum-cost decoding}

Assume henceforth that the displacement cost $w$ is Borel.  The quotient
cost~\eqref{eq:distance_quotient_cost} is then Borel: if
$r_D:G/D\to G$ is a Borel section of $\pi_D$, then
\[
 \overline w(a)=\inf_{d\in D}w(r_D(a)+d)
\]
is a countable infimum of Borel functions.

\begin{definition}
A minimum-cost decoder for $w$ is a decoder satisfying
\begin{equation}
 \operatorname{Dec}_w(\chi)\in
 \operatorname*{arg\,min}_{a\in\varphi^{-1}(\chi)}\overline w(a).
 \label{eq:minimum_cost_decoder}
\end{equation}
\end{definition}
The fiber $\varphi^{-1}(\chi)=\sigma_0(\chi)+\mathsf L_D$ is finite, so
this minimum is attained even if the infimum defining $\overline w$ is
not.  Fixing an ordering of $\mathsf L_D$ and choosing the first minimizer
of the finitely many Borel functions
$\overline w(\sigma_0(\chi)+\ell)$ yields a Borel minimum-cost decoder.

\begin{remark}
If $f:[0,\infty]\to[0,\infty]$ is continuous and strictly increasing
with $f(0)=0$, then $\overline{f\circ w}=f\circ\overline w$.
Consequently $w$ and $f\circ w$ have the same minimum-cost decoders,
although their numerical distances and subadditivity properties may
differ.  In particular, a norm and its square have the same minimizers;
the norm satisfies the subadditivity hypothesis of
Theorem~\ref{thm:minimum_cost_half_distance}, whereas its square generally does not.
\end{remark}

\begin{theorem}
\label{thm:minimum_cost_half_distance}
If $w$ is Borel, symmetric, and subadditive, every minimum-cost decoder
succeeds on each error $z$ satisfying
\begin{equation}
 w(z)<\frac{d_w(D)}{2}.
 \label{eq:minimum_cost_half_distance}
\end{equation}
\end{theorem}
\begin{proof}
Suppose the decoder chooses $u+D\neq z+D$ over $\kappa_z$.  Then
$u-z\in D^\perp\setminus D$, while minimality gives
$\overline w(u+D)\leq\overline w(z+D)\leq w(z)<d_w(D)/2$.
For every $\varepsilon>0$ there is a $d\in D$ with
$w(u+d)<\overline w(u+D)+\varepsilon$.  Since
$u+d-z\in D^\perp\setminus D$, we obtain
\[
 d_w(D)\leq w(u+d-z)
 \leq w(u+d)+w(z)<2w(z)+\varepsilon.
\]
If $d_w(D)<\infty$, choose $\varepsilon$ small enough to contradict
$2w(z)<d_w(D)$.  If $d_w(D)=\infty$, the right-hand side is finite,
which is again a contradiction.  Thus the chosen class is $z+D$.
\end{proof}

For Borel subadditive costs, Proposition~\ref{prop:distance_correctable_ball}
and Theorem~\ref{thm:distance_correctable_decoder} guarantee some decoder
on each $B_w(r)$ with $2r\leq d_w(D)$. The preceding theorem shows that
every minimum-cost decoder succeeds throughout that set, including all
choices made when costs tie.

\subsection{Maximum-likelihood decoding}

Let $\mathbb P$ be a Borel probability measure on $G$ describing random
displacement noise.  Define
\begin{align}
 \operatorname{Succ}(\operatorname{Dec})
 &:=\{z\in G:\operatorname{Dec}(\kappa_z)=\pi_D(z)\},
 \label{eq:decoder_success_set}\\
 P_{\mathrm{succ}}(\operatorname{Dec};\mathbb P)
 &:=\mathbb P(\operatorname{Succ}(\operatorname{Dec})),\qquad
 P_{\mathrm{opt}}(D;\mathbb P)
 :=\sup_{\operatorname{Dec}}P_{\mathrm{succ}}(\operatorname{Dec};\mathbb P).
 \label{eq:decoder_success_probability}
\end{align}
A decoder $\operatorname{Dec}$ is called optimal for $\mathbb P$ if
$P_{\mathrm{succ}}(\operatorname{Dec};\mathbb P)=P_{\mathrm{opt}}(D;\mathbb P)$.

The laws of the error class and syndrome are, respectively,
$\overline{\mathbb P}:=(\pi_D)_*\mathbb P$ and
$\nu:=\varphi_*\overline{\mathbb P}=\kappa_*\mathbb P$.
Disintegration on the compact metrizable spaces $G/D$ and $\widehat D$
gives conditional probability measures $\mathbb P_\chi$, supported on
$\varphi^{-1}(\chi)$ for $\nu$-almost every $\chi$, such that
\begin{equation}
 \overline{\mathbb P}(A)
 =\int_{\widehat D}\mathbb P_\chi
       (A\cap\varphi^{-1}(\chi))\,d\nu(\chi)
 \label{eq:decoder_noise_disintegration}
\end{equation}
for Borel $A\subseteq G/D$.  Here $\mathbb P_\chi$ is the conditional
distribution of the error class $z+D$ given the syndrome $\kappa_z=\chi$;
thus $\mathbb P_\chi(\{a\})$ is the posterior probability of the class
$a\in\varphi^{-1}(\chi)$, defined for $\nu$-almost every $\chi$.
In the parametrization
$\sigma_0(\chi)+\mathsf L_D$, the finitely many conditional masses
$\mathbb P_\chi(\{\sigma_0(\chi)+\ell\})$ may be chosen Borel in $\chi$.

The success set is
$\pi_D^{-1}(\operatorname{im}\operatorname{Dec})$.
A Borel section is injective and has Borel image, and its image meets each
fiber of $\varphi$ in exactly its chosen point.  Therefore
\begin{equation}
 P_{\mathrm{succ}}(\operatorname{Dec};\mathbb P)
 =\overline{\mathbb P}(\operatorname{im}\operatorname{Dec})
 =\int_{\widehat D}\mathbb P_\chi(\{\operatorname{Dec}(\chi)\})\,d\nu(\chi).
 \label{eq:decoder_conditional_success}
\end{equation}

\begin{theorem}
\label{thm:maximum_likelihood_decoder}
The optimum $P_{\mathrm{opt}}(D;\mathbb P)$ is attained.  A decoder $\operatorname{Dec}$ is
optimal if and only if, for $\nu$-almost every syndrome $\chi$,
\begin{equation}
 \operatorname{Dec}(\chi)\in
 \operatorname*{arg\,max}_{a\in\varphi^{-1}(\chi)}\mathbb P_\chi(\{a\}).
 \label{eq:maximum_likelihood_decoder}
\end{equation}
In particular,
\begin{equation}
 P_{\mathrm{opt}}(D;\mathbb P)
 =\int_{\widehat D}\max_{a\in\varphi^{-1}(\chi)}
       \mathbb P_\chi(\{a\})\,d\nu(\chi).
 \label{eq:optimal_decoding_probability}
\end{equation}
\end{theorem}
\begin{proof}
Order $\mathsf L_D$ and choose the first maximizer of its finitely many
Borel conditional masses.  This gives a Borel decoder; on a Borel
$\nu$-null exceptional set, use $\sigma_0$.  Equation~\eqref{eq:decoder_conditional_success}
shows that it attains the displayed integral.  Conversely, failure to
maximize on a set of positive $\nu$-measure strictly decreases the integral.
\end{proof}

This rule maximizes the posterior probability of the error class given
the syndrome; it is conventionally called degenerate maximum-likelihood
decoding. Minimum-cost decoding uses the chosen displacement geometry,
whereas this rule uses the conditional mass of each entire
stabilizer-equivalence class. For continuous noise, these are conditional
probabilities, even though individual classes may have zero unconditional
probability.

\subsection{GKP, Pauli, and Lee-distance specializations}

In the formulas below, a minimum over an empty set of nontrivial logical
labels is understood as $\infty$, consistently with
\eqref{eq:displacement_cost_distance}. This includes the case $K=1$.

\subsubsection{Euclidean GKP distance and closest-vector decoding}

In the purely continuous case of Sec.~\ref{ex:gkp}, take
$M=\R^p$, $G=\R^{2p}$, and $D=T\Z^{2p}$, with $T$ invertible and
$J=J_{\mathrm{osc}}$.  The Euclidean cost gives
\begin{equation}
 d_{\mathrm{GKP}}(D)
 =\min_{z\in D^\perp\setminus D}\|z\|_2,
 \label{eq:gkp_euclidean_distance}
\end{equation}
with value $\infty$ when there is no nontrivial logical class.  This is the
usual GKP displacement distance in our phase-space normalization
\cite{Conrad2022gottesmankitaev}.

To express a decoder in the column-basis convention of
Sec.~\ref{subsec:examples_notation}, represent $\chi\in\widehat D$
by $s_\chi\in\R^{2p}/\Z^{2p}$ through
$\chi(Tn)=e^{2\pi i n^{\intercal}s_\chi}$.  Then
\begin{equation}
 s_{\kappa_z}=-T^{\intercal}Jz\pmod{\Z^{2p}}.
 \label{eq:gkp_decoder_syndrome_coordinates}
\end{equation}
The measured phases of the phase-corrected stabilizers have the opposite sign,
$T^{\intercal}Jz$, because their character is $\overline{\kappa_z}$.
Choose the Borel representative $\widetilde s_\chi\in[0,1)^{2p}$ and set
\begin{equation}
 \eta(\chi):=-(T^{\intercal}J)^{-1}\widetilde s_\chi,
 \qquad \kappa_{\eta(\chi)}=\chi.
 \label{eq:gkp_decoder_reference_error}
\end{equation}
The compatible error classes are
$\eta(\chi)+\lambda+D$, indexed by $\lambda+D\in\mathsf L_D$.
For centered isotropic Gaussian displacement noise with density
$p_\sigma(z)\propto e^{-\|z\|_2^2/(2\sigma^2)}$, their conditional
likelihoods are proportional to
\begin{equation}
 \mathcal L_\sigma(\lambda+D\mid\chi)
 :=\sum_{d\in D}p_\sigma(\eta(\chi)+\lambda+d).
 \label{eq:gkp_decoder_theta_likelihood}
\end{equation}
The maximum-likelihood decoder maximizes this translated-lattice
Gaussian sum over $\mathsf L_D$, as in the theta-series formulation
of~\cite{Conrad2022gottesmankitaev}.

The minimum-cost decoder instead minimizes
\begin{equation}
 \min_{\lambda+D\in\mathsf L_D}\ \min_{d\in D}
       \|\eta(\chi)+\lambda+d\|_2
 =\min_{y\in D^\perp}\|\eta(\chi)+y\|_2.
 \label{eq:gkp_decoder_closest_vector}
\end{equation}
Thus a closest vector $y$ in $D^\perp$ to $-\eta(\chi)$ yields the
chosen class $\eta(\chi)+y+D$.  For each fixed syndrome $\chi$, the smallest
norm terms dominate~\eqref{eq:gkp_decoder_theta_likelihood} as $\sigma\to0$;
if the minimizing logical class is unique, there exists $\sigma_0(\chi)>0$
such that the maximum-likelihood choice agrees with this rule for
$0<\sigma<\sigma_0(\chi)$.  At finite noise,
the full sums, including multiplicities, determine the likelihood.

\subsubsection{Binary Pauli distance and the $A_2$ integer lift}

For $M=\Z_2^n$ and the finite stabilizer lattice
$D\subseteq G=\Z_2^n\times\Z_2^n$ of Section~\ref{sec:PSC_basic}, define
the symplectic Hamming weight of $z=(a,b)$ by
\begin{equation}
 \operatorname{sw}(a,b):=\#\{i:(a_i,b_i)\neq(0,0)\}.
 \label{eq:pauli_distance_weight}
\end{equation}
Taking $w=\operatorname{sw}$ recovers the ordinary stabilizer distance
\begin{equation}
 d_{\mathrm{Pauli}}(D)
 =\min_{z\in D^\perp\setminus D}\operatorname{sw}(z).
 \label{eq:pauli_hamming_distance}
\end{equation}
Since this weight is subadditive, Theorem~\ref{thm:minimum_cost_half_distance}
gives correction of all Pauli errors of weight at most
$\lfloor(d_{\mathrm{Pauli}}(D)-1)/2\rfloor$ when the distance is finite.

Use the integer-lift notation of
Sec.~\ref{ex:mixed-qudit}, with $k=n$ and
$\dia c=2\dia 1_n$: $\Lambda=\pi_{\dia c}^{-1}(D)$ and
$\Lambda^\perp=\pi_{\dia c}^{-1}(D^\perp)$, where the lifted form is
$J_{\mathrm{fin}}=\frac12\left(\begin{smallmatrix}0&\dia 1_n\\-\dia 1_n&0\end{smallmatrix}\right)$.
For $v=(a,b)\in\Z^{2n}$, define
\begin{equation}
 Q_{A_2}(a,b):=\sum_{i=1}^n(a_i^2+b_i^2-a_i b_i).
 \label{eq:pauli_A2_quadratic_cost}
\end{equation}
Each nonzero binary pair has a lift among $(1,0),(0,1),(1,1)$ of
quadratic cost one; every nonzero integral pair has cost at least one.
Coordinatewise minimization therefore gives
\begin{equation}
 \min_{v\in\pi_{\dia c}^{-1}(z)}Q_{A_2}(v)=\operatorname{sw}(z),
 \qquad
 d_{\mathrm{Pauli}}(D)
 =\min_{v\in\Lambda^\perp\setminus\Lambda}Q_{A_2}(v).
 \label{eq:pauli_A2_lift_distance}
\end{equation}
\begin{remark}
\label{rem:pauli_eisenstein_interpretation}
Let $\omega=e^{2\pi i/3}$.  Under $(a,b)\mapsto a+b\omega\in\Z[\omega]$,
the identity $|a+b\omega|^2=a^2+b^2-ab$ identifies each one-qubit
block with the Eisenstein, or hexagonal, lattice.
Since $t^2+t+1$ is irreducible over $\mathbb F_2$, the prime $2$ is
inert in $\Z[\omega]$, and
$\Z[\omega]/2\Z[\omega]\cong\mathbb F_4$.
Its three nonzero classes have unit-norm representatives
$1,\omega,1+\omega$, explaining geometrically why $X$, $Y$, and $Z$
all have unit Pauli cost.
This realizes the description of binary stabilizer codes by additive
codes over $GF(4)$~\cite{calderbank1998quantum} in the integer-lift geometry,
in analogy with Construction~A over the Eisenstein
integers~\cite{conway2013sphere}.
These additive codes need not be $GF(4)$-linear.
\end{remark}

The quadratic cost reproduces the integer Pauli distance
exactly.  Its square root is a norm, and hence subadditive, with the same
minimizers but distance $\sqrt{d_{\mathrm{Pauli}}(D)}$ on the lift;
the usual Pauli correction radius uses $\operatorname{sw}$ on the finite
phase space.

\subsubsection{Qudit Lee distance and the $\ell^1$ integer lift}

For $M=\Z_{\dia c}$ as in
Sec.~\ref{ex:mixed-qudit}, define the Lee magnitude
of $a_i\in\Z_{c_i}$ by
\begin{equation}
 |a_i|_{\mathrm L,c_i}:=\min_{m\in\Z}|\widetilde a_i+m c_i|,
 \label{eq:qudit_lee_magnitude}
\end{equation}
where $\widetilde a_i$ is any integer representative.  For
$z=(a,b)\in\Z_{\dia c}\times\Z_{\dia c}$, put
\begin{equation}
 \operatorname{wt}_{\mathrm{Lee}}(a,b)
 :=\sum_{i=1}^k\bigl(|a_i|_{\mathrm L,c_i}+|b_i|_{\mathrm L,c_i}\bigr),
 \qquad
 d_{\mathrm{Lee}}(D)
 :=\min_{z\in D^\perp\setminus D}\operatorname{wt}_{\mathrm{Lee}}(z).
 \label{eq:qudit_lee_distance}
\end{equation}
This symmetric, subadditive cost measures cyclic shifts in each coordinate.
For equal dimensions $c_i=d$ it is the usual qudit Lee cost.  Retaining
$\Lambda=\pi_{\dia c}^{-1}(D)$ and its dual with respect to $J_{\mathrm{fin}}$,
coordinatewise minimization gives
\begin{equation}
 \operatorname{wt}_{\mathrm{Lee}}(z)
 =\min_{v\in\pi_{\dia c}^{-1}(z)}\|v\|_1,
 \qquad
 d_{\mathrm{Lee}}(D)
 =\min_{v\in\Lambda^\perp\setminus\Lambda}\|v\|_1.
 \label{eq:qudit_lee_lift_distance}
\end{equation}
Thus the same integer-lift construction accommodates unequal local
dimensions, with the geometry supplied by the coordinatewise Lee cost.
\section{Weight enumerators}
\label{sec:weight_enumerators}

The distance of Section~\ref{sec:distance_error_correction} records the
least cost of a nontrivial logical displacement.  Weight enumerators
retain the distribution of costs and their multiplicities.  We formulate
them on the same phase space $G=M\times\widehat M$, using the
displacements $\mathbf D(z)$ and commutator $\rho$ of
Section~\ref{sec:GSC}.  This places the quantum enumerators of
Shor and Laflamme~\cite{shor1997macwilliams}, their nonbinary extension
by Rains~\cite{rains1998weight}, and the continuous-variable distributions
of Burchards~\cite{burchards2025macwilliams} in a common Fourier framework.
For finite stabilizers the enumerators count labels in $D$ and $D^\perp$;
for ideal GLCA codes the corresponding objects are lattice counting
measures.  The latter definition does not require a projection onto a
distributional codespace.

\subsection{Phase-space operator enumerators}
\label{subsec:we_phase_space}

Throughout this section we use the self-Plancherel Haar measure
$\mu$ defined in Subsection~\ref{subsec:code_dimension_covolume}.
In particular, on a finite phase space with $N=\dim \ell^2(M)=|M|$,
this is $N^{-1}$ times counting measure.  Define the symplectic Fourier
transform~\cite[p.~270]{rieffel} by
\begin{equation}
 (\mathcal F_\rho f)(z)
 :=\int_G f(u)\rho(u,z)\,d\mu(u).
 \label{eq:we_symplectic_fourier}
\end{equation}
With this measure it is unitary on $L^2(G)$ and
$\mathcal F_\rho^2=I$.  These are the usual Plancherel and inversion
theorems under the identification of $G$ with its dual induced by
$\rho$; the alternating bicharacter makes the square the identity.

For trace-class operators $O_1,O_2$ on $\mathcal H=L^2(M)$, set
\begin{align}
 \Xi_O(z)&:=\operatorname{Tr}(O\mathbf D(z)),
 \label{eq:we_characteristic}\\
 a_{12}(z)&:=\Xi_{O_1}(z)\overline{\Xi_{O_2}(z)},
 &
 b_{12}(z)&:=\operatorname{Tr}\bigl(
   \mathbf D(z)O_1\mathbf D(z)^\dagger O_2^\dagger\bigr).
 \label{eq:we_full_enumerators}
\end{align}
The characteristic function has the same displacement convention as
Section~\ref{subsec:fe_characteristic} below.  Both enumerators are unchanged
by replacing $\mathbf D(z)$ with any phase multiple of itself.

The standard Weyl--Plancherel identity, stated in
\cite[Lemma~2.1]{spriet2026kirkwood}, reads in our displacement convention
\begin{equation}
 \int_G\Xi_{O_1}(u)\overline{\Xi_{O_2}(u)}\,
       d\mu(u)
 =\operatorname{Tr}(O_1O_2^\dagger).
 \label{eq:we_weyl_plancherel}
\end{equation}
Its application to the conjugated operator in
\eqref{eq:we_full_enumerators} gives the phase-space form of quantum
MacWilliams duality:
\begin{equation}
 b_{12}=\mathcal F_\rho a_{12},\qquad
 a_{12}=\mathcal F_\rho b_{12}.
 \label{eq:we_operator_duality}
\end{equation}
This follows from the relation
$\Xi_{\mathbf D(z)O_1\mathbf D(z)^\dagger}(u)
=\rho(u,z)\Xi_{O_1}(u)$.  The trace-class hypotheses ensure that both
functions are continuous and integrable: Weyl--Plancherel and
Cauchy--Schwarz control $a_{12}$, while nuclear decompositions give
$\|b_{12}\|_1\leq\|O_1\|_1\|O_2\|_1$.  One can also check that
\begin{equation}
 \int_G a_{12}\,d\mu
 =\operatorname{Tr}(O_1O_2^\dagger),\qquad
 \int_G b_{12}\,d\mu
 =\operatorname{Tr}(O_1)\overline{\operatorname{Tr}(O_2)}.
 \label{eq:we_operator_masses}
\end{equation}

\subsection{Weight measures and error detection}
\label{subsec:we_weight_measures}

Let $w:G\to[0,\infty)$ be a finite-valued Borel displacement cost in
the sense of Section~\ref{sec:distance_error_correction}: $w(z)$ is the
weight assigned to the displacement label $z$, and finite-valued means
that $w(z)<\infty$ for every $z$, not that $w$ has finite range.
The functions $a_{12}$ and $b_{12}$ retain information about each
individual displacement.  To form weight enumerators, we collect this
information according to its cost.  Define
finite complex Borel measures on $[0,\infty)$ by pushforward:
\begin{equation}
 \mathcal A^w_{12}:=w_*(a_{12}\mu),\qquad
 \mathcal B^w_{12}:=w_*(b_{12}\mu).
 \label{eq:we_weight_measures}
\end{equation}
As before, the subscript $12$ records the pair of operators $(O_1,O_2)$,
and $a_{12}\mu$ denotes the measure with density $a_{12}$
relative to the self-Plancherel Haar measure $\mu$.  By definition, the pushforward $w_*$ assigns
to a Borel set $E\subseteq[0,\infty)$ the mass of all labels whose
cost lies in $E$:
\[
 \mathcal A^w_{12}(E)
 =\int_{w^{-1}(E)}a_{12}(z)\,d\mu(z),\qquad
 \mathcal B^w_{12}(E)
 =\int_{w^{-1}(E)}b_{12}(z)\,d\mu(z).
\]

For a bounded Borel function $h$ on $[0,\infty)$, the bracket notation
means integration against the weight measure:
\[
 \mathcal A^w_{12}[h]
 :=\int_{[0,\infty)}h(t)\,d\mathcal A^w_{12}(t)
 =\int_G h(w(z))a_{12}(z)\,d\mu(z),
\]
and similarly for $\mathcal B^w_{12}$.  These measures need not be
probability measures, or even positive for general $O_1,O_2$.
Costs taking the value
$\infty$ can instead be treated by pushforward to $[0,\infty]$.

For $h\circ w\in\mathscr S(G)$, the Schwartz space on $G$, the
evenness of $w$ and Fourier duality give
\[
 \mathcal B^w_{12}[h]
 =\int_G a_{12}(z)\bigl(\mathcal F_\rho(h\circ w)\bigr)(z)
       \,d\mu(z).
\]
Here Fourier duality refers to the identity
$b_{12}=\mathcal F_\rho a_{12}$ in~\eqref{eq:we_operator_duality}.
Thus a sufficient condition for a closed identity in the weight variable is
\begin{equation}
 \mathcal F_\rho(h\circ w)=\widetilde h\circ w
 \quad\Longrightarrow\quad
 \mathcal B^w_{12}[h]=\mathcal A^w_{12}[\widetilde h],
 \label{eq:we_weight_mask_duality}
\end{equation}
where $\widetilde h$ is bounded and Borel.

The motivation is the same as for classical weight enumerators.
For a binary code $Q\subseteq\{0,1\}^n$, let
$\nu_Q=\sum_{q\in Q}\delta_q$ be its counting measure and let
$\operatorname{wt}_{\mathrm H}(q)$ count the nonzero coordinates of $q$.
Then
\[
 (\operatorname{wt}_{\mathrm H})_*\nu_Q
 =\sum_{j=0}^n A_j(Q)\delta_j,\qquad
 A_j(Q)=\#\{q\in Q:\operatorname{wt}_{\mathrm H}(q)=j\},
\]
where $\delta_j$ is unit point mass at $j$.  Its masses are exactly
the coefficients of the usual weight-enumerator polynomial
\[
 W_Q(X,Y)=\sum_{j=0}^n A_j(Q)X^{n-j}Y^j.
\]
Equation~\eqref{eq:we_weight_measures} applies this same construction
to the operator densities $a_{12}$ and $b_{12}$.  For qubit codes and
symplectic Hamming cost, its point masses recover the quantum weight
enumerators of Shor and Laflamme~\cite{shor1997macwilliams}, with the
normalization specified in~\eqref{eq:we_finite_coefficients} in
Subsection~\ref{subsec:we_finite}.
The measure formulation also accommodates continuously varying costs,
for which there need not be a finite list of weight coefficients or a
polynomial.  In both settings it lets us organize error detection by
cost thresholds, as follows.

Suppose now that $P$ is the orthogonal projection onto a
$K$-dimensional subspace of $\mathcal H$, and write $a_P,b_P$ for
$O_1=O_2=P$.  The standard positivity and error-detection relation for
quantum enumerators~\cite{shor1997macwilliams,rains1998weight}, in its
pointwise form, is
\begin{align}
 a_P(z)&=|\operatorname{Tr}(P\mathbf D(z))|^2,
 & b_P(z)&=\|P\mathbf D(z)P\|_{\mathrm{HS}}^2,
 \label{eq:we_projection_functions}\\
 Kb_P(z)-a_P(z)
 &=K\left\|P\mathbf D(z)P
       -\frac{\operatorname{Tr}(P\mathbf D(z))}{K}P
     \right\|_{\mathrm{HS}}^2\geq0.
 \label{eq:we_detection_defect}
\end{align}
The right side measures the traceless part of the compressed error.
Its vanishing is exactly the scalar-compression condition for detection
used in Section~\ref{sec:distance_error_correction}. In particular, $\mathcal A^w_P$, $\mathcal B^w_P$, and
\begin{equation}
 \mathcal E^w_P:=K\mathcal B^w_P-\mathcal A^w_P
 \label{eq:we_defect_measure}
\end{equation}
are positive measures, with total masses $K$, $K^2$, and $K^3-K$,
respectively.  Their underlying functions satisfy $a_P(0)=K^2$ and
$b_P(0)=K$.

For a continuous cost, the detection distance of this normalizable
code is therefore
\begin{equation}
 d_w(P):=\inf\{w(z):P\mathbf D(z)P\notin\C P\}
       =\inf\operatorname{supp}\mathcal E^w_P,
 \label{eq:we_projection_distance}
\end{equation}
where the infimum of the empty set is $\infty$. For $K=1$ every compression is scalar and this
detection distance is $\infty$.  The relative-error criterion of
Section~\ref{sec:distance_error_correction} expresses correctability
through detection of all relative displacements; its half-distance
conclusion applies when $w$ is also subadditive.

\subsection{Finite-system specializations}
\label{subsec:we_finite}

\subsubsection{Complete and Hamming enumerators}

Use $M=\Z_{\dia c}$, $N=C_{\dia c}=\prod_i c_i$, and $z=(a,b)$ as in
Sec.~\ref{ex:mixed-qudit}, with local labels
$z_i=(a_i,b_i)\in\Z_{c_i}^2$.  Extend the symplectic Hamming weight
of~\eqref{eq:pauli_distance_weight} to these labels by counting the
sites with $z_i\neq0$.  For a rank-$K$ projection define the normalized
Shor--Laflamme coefficients
\begin{equation}
 A_j:=\frac1{K^2}\sum_{\operatorname{sw}(z)=j}a_P(z),\qquad
 B_j:=\frac1K\sum_{\operatorname{sw}(z)=j}b_P(z).
 \label{eq:we_finite_coefficients}
\end{equation}
The coefficients $A_j$ and $B_j$ equal the corresponding point masses
of~\eqref{eq:we_weight_measures} multiplied by $C_{\dia c}/K^2$ and
$C_{\dia c}/K$, respectively.  In this normalization the familiar relations are
\begin{equation}
 A_0=B_0=1,\qquad 0\leq A_j\leq B_j,\qquad
 \sum_jA_j=\frac{C_{\dia c}}K,\qquad
 \sum_jB_j=C_{\dia c}K.
 \label{eq:we_finite_constraints}
\end{equation}
For $K\geq2$, the first $j$ with $A_j<B_j$ is the Hamming detection
distance~\cite{shor1997macwilliams,rains1998weight}.

Introduce a variable $X_{i,v}$ for each site $i$ and each
$v\in\Z_{c_i}^2$, and set
\begin{equation}
 A_P^{\mathrm{comp}}(X):=\frac1{K^2}\sum_{z\in G}a_P(z)
           \prod_{i=1}^kX_{i,z_i},\qquad
 B_P^{\mathrm{comp}}(X):=\frac1K\sum_{z\in G}b_P(z)
           \prod_{i=1}^kX_{i,z_i}.
 \label{eq:we_complete_enumerators}
\end{equation}
Each monomial $\prod_iX_{i,z_i}$ records the local displacement label
at every site, so these complete enumerators retain more information
than total Hamming weight.  Let
$\rho_i((a,b),(a',b'))=e^{2\pi i(ab'-ba')/c_i}$ be the symplectic
character on site $i$.  Their MacWilliams identity is
\begin{equation}
 \widehat X_{i,u}:=\sum_{v\in\Z_{c_i}^2}\rho_i(u,v)X_{i,v},
 \qquad
 B_P^{\mathrm{comp}}(X)
 =\frac{K}{C_{\dia c}}A_P^{\mathrm{comp}}(\widehat X).
 \label{eq:we_complete_macwilliams}
\end{equation}
Here $u,v\in\Z_{c_i}^2$ are local displacement labels, and
$A_P^{\mathrm{comp}}(\widehat X)$ means replacing every variable
$X_{i,u}$ by the displayed linear combination $\widehat X_{i,u}$.
The substitution acts independently at each site because
$\rho(u,z)=\prod_i\rho_i(u_i,z_i)$, so the Fourier sum
in~\eqref{eq:we_operator_duality} factors over the sites.
The prefactor $K/C_{\dia c}$ combines the Fourier normalization
$1/C_{\dia c}$ with the normalizations $1/K^2$ and $1/K$ of the two
enumerators. 

If all $c_i=d$, the substitution $X_{i,0}=X$ and $X_{i,v}=Y$ for
$v\neq0$ gives homogeneous polynomials
$A_P(X,Y)=\sum_jA_jX^{k-j}Y^j$ and $B_P(X,Y)$.
The known nonbinary quantum MacWilliams identity~\cite{rains1998weight}
takes the form
\begin{equation}
 B_P(X,Y)=\frac{K}{d^k}
       A_P\bigl(X+(d^2-1)Y,X-Y\bigr).
 \label{eq:we_hamming_macwilliams}
\end{equation}
For $n=k$ qubits this is precisely the Shor--Laflamme transform
$B_P(X,Y)=K2^{-n}A_P(X+3Y,X-Y)$
\cite{shor1997macwilliams}.

\subsubsection{Lee enumerators}
\label{subsubsec:we_lee_enumerators}
For arbitrary local dimensions, the Lee cost
of~\eqref{eq:qudit_lee_distance} defines
\begin{equation}
 A_P^{\mathrm{Lee}}(t):=\frac1{K^2}\sum_{z\in G}
                 a_P(z)t^{\operatorname{wt}_{\mathrm{Lee}}(z)},\qquad
 B_P^{\mathrm{Lee}}(t):=\frac1K\sum_{z\in G}
                 b_P(z)t^{\operatorname{wt}_{\mathrm{Lee}}(z)}.
 \label{eq:we_lee_enumerators}
\end{equation}
These are obtained from the complete enumerators by substituting
$X_{i,(a,b)}=t^{|a|_{\mathrm L,c_i}+|b|_{\mathrm L,c_i}}$.
The coefficients of $B_P^{\mathrm{Lee}}-A_P^{\mathrm{Lee}}$ are
nonnegative, and the smallest exponent with a nonzero coefficient is
the Lee detection distance (taken as $\infty$ if the difference
vanishes).  Moreover, the complete enumerator $A_P^{\mathrm{comp}}$
always determines $B_P^{\mathrm{Lee}}$ through the exact substitution
\begin{equation}
 B_P^{\mathrm{Lee}}(t)=\frac K{C_{\dia c}}
 A_P^{\mathrm{comp}}\left(
  \left\{\sum_{v=(a,b)\in\Z_{c_i}^2}\rho_i(u,v)
        t^{|a|_{\mathrm L,c_i}+|b|_{\mathrm L,c_i}}\right\}_{i,u}
                      \right).
 \label{eq:we_lee_complete_transform}
\end{equation}

However, Fourier duality does not in general give a closed transform
using only Lee-weight sums: it requires the Fourier transform of each
shell indicator to be constant on every Lee-weight
shell~\cite{gluesing2015fourier}.  For example, in $G=\Z_5^2$ the
Lee-weight-one shell is $S_1=\{(\pm1,0),(0,\pm1)\}$, and
\begin{equation}
 (\mathcal F_\rho\mathbf1_{S_1})(a,b)
 =\frac25\left(\cos\frac{2\pi a}5+\cos\frac{2\pi b}5\right).
 \label{eq:we_lee_counterexample}
\end{equation}
At $(2,0)$ and $(1,1)$, both of Lee weight two, its values are
$(3-\sqrt5)/10$ and $(\sqrt5-1)/5$, respectively, violating this
condition.  Thus no transform from Lee-weight sums alone works for
all phase-space functions or all operator pairs; the latter can realize
point masses in $a_{12}$ by taking $O_1=O_2$ to be a displacement.
This obstruction does not rule out identities for particular code
families.  The complete enumerator supplies sufficient information
for the dual substitution above; a suitable refinement of the Lee-weight
data may also suffice.
Even for qubits, Lee and Hamming enumerators differ: a $Y$ label has
Lee weight two and symplectic Hamming weight one.

\subsection{Euclidean enumerators}
\label{subsec:we_euclidean}

For $M=\R^p$ with $p\geq1$, use the canonical coordinates
$z=(m,s)$ of Sec.~\ref{subsec:examples_notation}, so that
$\rho(u,z)=e^{2\pi i u^{\intercal}J_{\mathrm{osc}}z}$ and
$d\mu=dz$.  For the radius $r=\|z\|_2$, write
\begin{equation}
 A_{12}^{\mathrm{rad}}(r):=\int_{\|z\|_2=r}a_{12}(z)\,d\sigma_r(z),
 \qquad
 B_{12}^{\mathrm{rad}}(r):=\int_{\|z\|_2=r}b_{12}(z)\,d\sigma_r(z),
 \label{eq:we_radial_densities}
\end{equation}
where $d\sigma_r$ is Euclidean surface area, rather than probability
measure on the sphere.  Thus
$\mathcal A^{\|\cdot\|_2}_{12}(dr)=A_{12}^{\mathrm{rad}}(r)\,dr$,
and similarly for $\mathcal B$.  Radial functions are preserved by
$\mathcal F_\rho$, since $J_{\mathrm{osc}}$ is orthogonal.  Burchards'
continuous-variable quantum MacWilliams identity
\cite[Theorem~4]{burchards2025macwilliams}, expressed in these
coordinates, is
\begin{equation}
 B_{12}^{\mathrm{rad}}(r)
 =2\pi r^p\int_0^\infty
       \frac{J_{p-1}(2\pi r t)}{t^{p-1}}
           A_{12}^{\mathrm{rad}}(t)\,dt,\qquad r>0,
 \label{eq:we_radial_macwilliams}
\end{equation}
where $J_{p-1}$ is the Bessel function of the first kind.
To derive this from~\eqref{eq:we_operator_duality}, substitute
$b_{12}(z)=\int_{\R^{2p}}a_{12}(v)\rho(v,z)\,dv$ into the definition
of $B_{12}^{\mathrm{rad}}(r)$ and interchange the integrals.
This is justified by $a_{12}\in L^1$ and the finite surface area of the
sphere.
Since $J_{\mathrm{osc}}$ is orthogonal, the standard spherical integral is
\[
 \int_{\|z\|_2=r}\rho(v,z)\,d\sigma_r(z)
 =2\pi r^p\frac{J_{p-1}(2\pi r\|v\|_2)}{\|v\|_2^{p-1}},
\]
with the value at $v=0$ understood by continuity.  This kernel depends
only on $t=\|v\|_2$.  Polar integration therefore replaces
$a_{12}(v)\,dv$ by $A_{12}^{\mathrm{rad}}(t)\,dt$, giving
\eqref{eq:we_radial_macwilliams}.

\subsection{Intrinsic lattice enumerators and theta series}
\label{subsec:we_lattice}

Return to an isotropic lattice $D\subseteq G$.  The preceding operator
definitions apply to finite-rank projections in $L^2(M)$, whereas the
ideal code fiber need not be a subspace of $L^2(M)$.  Its intrinsic
enumerators are instead defined directly from $D$ and $D^\perp$.
For any lattice $L\subseteq G$ put
\begin{equation}
 \delta_L:=\sum_{z\in L}\delta_z,\qquad
 \mathcal W_L[f]:=\sum_{z\in L}f(z).
 \label{eq:we_lattice_distribution}
\end{equation}
These define a counting measure and its action on
Schwartz test functions.  If $w$ has finite sublevel sets on
both $D$ and $D^\perp$, the intrinsic weight measures
\begin{equation}
 \mathcal A_D^w:=w_*\delta_D,\qquad
 \mathcal B_D^w:=w_*\delta_{D^\perp}
 \label{eq:we_intrinsic_measures}
\end{equation}
are locally finite.  Isotropy makes their difference positive:
\begin{equation}
 \mathcal B_D^w-\mathcal A_D^w
 =\sum_{z\in D^\perp\setminus D}\delta_{w(z)},\qquad
 d_w(D)=\inf\operatorname{supp}
             (\mathcal B_D^w-\mathcal A_D^w).
 \label{eq:we_intrinsic_distance}
\end{equation}
The last equality is the enumerator interpretation of
\eqref{eq:displacement_cost_distance}; the stabilizer labels have
been removed by the subtraction.

In the finite case these are exactly the usual counting enumerators
of a stabilizer and its normalizer.  Indeed, for the code projection
associated with the phase-corrected stabilizers, trace orthogonality
gives the standard support formulas
\cite{shor1997macwilliams,rains1998weight}
\begin{equation}
 a_P=K^2\mathbf1_D,\qquad b_P=K\mathbf1_{D^\perp}.
 \label{eq:we_stabilizer_supports}
\end{equation}
Thus $A_j$ and $B_j$ count labels of weight $j$ in $D$ and
$D^\perp$, respectively.  For any displacement cost $w$,
$\mathcal A_P^w=(K^2/N)\mathcal A_D^w$ and
$\mathcal B_P^w=(K/N)\mathcal B_D^w$.  These factors distinguish
the operator measures from the intrinsic counts and give
\[
 \mathcal E_P^w=\frac{K^2}{N}
     \bigl(\mathcal B_D^w-\mathcal A_D^w\bigr),\qquad
 d_w(P)=d_w(D).
\]
For example, the classical binary linear code $C=\ker H$ of
Section~\ref{sec:PSC_basic} has $D=\{0\}\times C^\perp$ and
$D^\perp=C\times\widehat M$.  Its enumerators are
$A_P(X,Y)=W_{C^\perp}(X,Y)$ and $B_P(X,Y)=W_C(X+Y,2Y)$.
The qubit case of~\eqref{eq:we_hamming_macwilliams} therefore reduces to
the classical binary MacWilliams identity~\cite{macwilliams1977theory}.

Recall $s_D=|G/D|_\mu$ from~\eqref{eq:covolume_definition}; by
Theorem~\ref{thm:code_dimension}, $s_D=K$ for a stabilizer code.
The classical Poisson summation theorem for LCA groups
\cite[Theorem~5.7(iii) and Remark~5.10 in the
\href{https://arxiv.org/abs/1608.04566v2}{arXiv:1608.04566v2} version]{jakobsen2017feichtinger}
takes the symplectic form
\begin{equation}
 \mathcal W_D[f]
 =s_D^{-1}\mathcal W_{D^\perp}[\mathcal F_\rho f],\qquad
 \mathcal F_\rho\delta_D=s_D^{-1}\delta_{D^\perp}.
 \label{eq:we_poisson}
\end{equation}
The distributional equality uses Fourier duality on test functions;
the reflection in its transpose disappears because $\delta_D$ is even.
Together with~\eqref{eq:dual_covolume_reciprocity}, it yields
\begin{equation}
 \mathcal B_D^w[h]
 =s_D\sum_{z\in D}\bigl(\mathcal F_\rho(h\circ w)\bigr)(z)
 \label{eq:we_lattice_mask_duality}
\end{equation}
whenever $h\circ w\in\mathscr S(G)$.  Under the factorization condition
in~\eqref{eq:we_weight_mask_duality}, this gives a closed identity
involving only the weight measures:
\[
 \mathcal B_D^w[h]=s_D\mathcal A_D^w[\widetilde h].
\]

For a full Euclidean lattice $L\subset\R^{2p}$ define
\begin{equation}
 \vartheta_L(t):=\sum_{z\in L}e^{-\pi t\|z\|_2^2},
 \qquad t>0.
 \label{eq:we_lattice_theta}
\end{equation}
The standard theta transformation, obtained by applying Poisson
summation to the Gaussian $e^{-\pi t\|z\|_2^2}$, whose symplectic
Fourier transform is $t^{-p}e^{-\pi\|u\|_2^2/t}$, is
\begin{equation}
 \vartheta_{D^\perp}(t)=s_Dt^{-p}\vartheta_D(1/t).
 \label{eq:we_theta_duality}
\end{equation}
The series $\vartheta_{D^\perp}-\vartheta_D$ counts precisely the
nontrivial logical vectors.  When $D^\perp\setminus D\neq\varnothing$,
its least squared-length shell is $d_{\mathrm{GKP}}(D)^2$
from~\eqref{eq:gkp_euclidean_distance}; if $D=D^\perp$, the series
vanishes and $d_{\mathrm{GKP}}(D)=\infty$.
This recovers the theta-series distance construction of Conrad et
al.~\cite[Section~5.3, Eqs.~(83)--(88)]{Conrad2022gottesmankitaev},
with $\vartheta_L(t)=\Theta_L(it)$ in their notation.

\begin{remark}
\label{rem:we_integer_lifts}
Let $P$ be the finite stabilizer code projection for $D$, and use 
$\Lambda=\pi_{\dia c}^{-1}(D)$ from
Sec.~\ref{ex:mixed-qudit}.  Duality for $J_{\mathrm{fin}}$ gives
$\Lambda^\perp=\pi_{\dia c}^{-1}(D^\perp)$
by~\eqref{eq:mixed-qudit-dual-basis}.
By~\eqref{eq:we_stabilizer_supports}, the normalized enumerators
$A_P,B_P$ and their complete versions count labels in $D,D^\perp$,
respectively.  Fix $0<q<1$.
For $n=k$ qubits and the $A_2$ cost~\eqref{eq:pauli_A2_quadratic_cost}, put
\[
 \psi_r(q):=\sum_{(a,b)\in r+2\Z^2}q^{a^2+b^2-ab},
 \qquad r\in\Z_2^2.
\]
The symmetries $(a,b)\mapsto(b,a)$ and $(a,b)\mapsto(a,a-b)$ give
\[
 \begin{aligned}
 \psi_{(0,0)}(q)&=:\psi_0(q)=1+O(q^4),\\
 \psi_{(1,0)}(q)=\psi_{(0,1)}(q)=\psi_{(1,1)}(q)
 &=:\psi_1(q)=2q+O(q^3).
 \end{aligned}
\]
Consequently, decomposition into residue classes yields
\begin{equation}
 \begin{aligned}
 \sum_{v\in\Lambda}q^{Q_{A_2}(v)}
 &=A_P\bigl(\psi_0(q),\psi_1(q)\bigr),\\
 \sum_{v\in\Lambda^\perp}q^{Q_{A_2}(v)}
 &=B_P\bigl(\psi_0(q),\psi_1(q)\bigr).
 \end{aligned}
 \label{eq:we_A2_lift_substitution}
\end{equation}
Conrad et al.'s Construction~A substitution uses binary Hamming weight
on $2n$ coordinates, assigning weight two to $Y=(1,1)$
\cite[Section~5.3, Eqs.~(90)--(91)]{Conrad2022gottesmankitaev}.
Here each of $X,Z,Y$ has minimum $A_2$ cost one, so the least logical
exponent is the Pauli distance for arbitrary qubit stabilizer codes
by~\eqref{eq:pauli_A2_lift_distance}; the corresponding recovery of the
distance from the weight distribution
in~\cite[Corollary~3]{Conrad2022gottesmankitaev} is stated for CSS codes.
For $d=d_{\mathrm{Pauli}}(D)<\infty$, put
$N_d:=B_d-A_d=\#\{z\in D^\perp\setminus D:\operatorname{sw}(z)=d\}$.
Then~\eqref{eq:we_A2_lift_substitution} gives
\[
 \sum_{v\in\Lambda^\perp\setminus\Lambda}q^{Q_{A_2}(v)}
 =(B_P-A_P)(\psi_0(q),\psi_1(q))
 =2^dN_d\,q^d+O(q^{d+1})\qquad(q\downarrow0).
\]
Each nonzero local residue class in $\Z_2^2$ has two minimum-cost lifts,
giving the factor $2^d$.

Returning to $k$ qudits of local dimensions $c_i$, the Lee residue series are
\[
 f_{c,r}(q):=\sum_{m\in\Z}q^{|r+cm|}
 =\begin{cases}
 (1+q^c)/(1-q^c),&r=0,\\
 (q^r+q^{c-r})/(1-q^c),&1\leq r<c.
 \end{cases}
\]
Set $X_{i,(a,b)}(q):=f_{c_i,a}(q)f_{c_i,b}(q)$.
The complete enumerators of~\eqref{eq:we_complete_enumerators} give
\begin{equation}
 \begin{aligned}
 \sum_{v\in\Lambda}q^{\|v\|_1}
 &=A_P^{\mathrm{comp}}\bigl(X(q)\bigr),\\
 \sum_{v\in\Lambda^\perp}q^{\|v\|_1}
 &=B_P^{\mathrm{comp}}\bigl(X(q)\bigr).
 \end{aligned}
 \label{eq:we_Lee_lift_substitution}
\end{equation}
Subtracting,
\[
 \sum_{v\in\Lambda^\perp\setminus\Lambda}q^{\|v\|_1}
 =\bigl(B_P^{\mathrm{comp}}-A_P^{\mathrm{comp}}\bigr)\bigl(X(q)\bigr).
\]
Its least exponent is $d_{\mathrm{Lee}}(D)$ when finite,
by~\eqref{eq:qudit_lee_lift_distance}, and the corresponding coefficient
counts all minimum-cost integer lifts.
\end{remark}
\section{Finite-energy codewords}
\label{sec:finite_energy}

When $M$ has a continuous or an infinite discrete factor, ideal codewords
in $\mathcal C$ may lie in $\mathscr S'(M)$ outside $L^2(M)$, as do the
combs of Eq.~\eqref{eq:examples-css-codewords}. This section constructs finite-energy
representatives and an asymptotically orthonormal encoding; an exact
isometry is obtained by a final finite-dimensional Gram correction.
We damp the oscillator and rotor energies with
a one-parameter family $\mathbf E_\Delta$ of positive smoothing operators,
apply it to the ideal codewords, and show that a \emph{single} scalar
$\nu_\Delta$, computed from the covolume and torsion data of $D$ and the partition
function of $\mathbf E_\Delta^2$, renormalizes the whole regularized code
basis at once, with an error that is exponentially small in $\Delta^{-1}$,
for an ideal encoding normalized by the factorization
$\operatorname{Enc}\operatorname{Enc}^\dagger=\Pi_D$ of the stabilizer form
in Proposition~\ref{prop:fe_encoding}.
Throughout, $D\subseteq G$ is an isotropic lattice, $K=\dim\mathcal C$ as in
Theorem~\ref{thm:code_dimension}, and $\mathsf L_D=D^\perp/D$ has order $K^2$.
For the explicit formulas we use Lebesgue measure on $\mathbb R^p$ and
counting measure on $\mathbb Z^q$ and $\mathbb Z_{\dia c}$, with their Pontryagin-dual
measures on $\widehat M$. Their product is the canonical self-Plancherel
measure on $G$ used in Section~\ref{subsec:code_dimension_covolume}. In particular, the finite
coordinate kets are orthonormal and
$\operatorname{Tr}\mathbf 1_{\ell^2(\mathbb Z_{\dia c})}=|\mathbb Z_{\dia c}|=C_{\dia c}$.

\subsection{The Gaussian regularizer}
\label{subsec:fe_regularizer}

Write $M=\mathbb R^p\times\mathbb Z^q\times\mathbb Z_{\dia c}$ as in
Section~\ref{sec:GSC}, so that
\[
 L^2(M)=L^2(\mathbb R^p)\otimes\ell^2(\mathbb Z^q)\otimes\ell^2(\mathbb Z_{\dia c}).
\]
Let $h_k$, $k\in\mathbb Z_{\geq0}^p$, be the tensor Hermite basis of
$L^2(\mathbb R^p)$, with vacuum $h_0(x)=\pi^{-p/4}e^{-|x|^2/2}$, and set
\begin{equation}
 \mathbf N:=\tfrac12\bigl(-\nabla_x^2+|x|^2-p\bigr),
 \qquad
 \mathbf Q_j\delta_n:=n_j\delta_n,
 \label{eq:fe_energy_operators}
\end{equation}
where $\nabla_x^2=\sum_{j=1}^p\partial_{x_j}^2$ is the Laplacian on
$\mathbb R^p$, so that $\mathbf N h_k=|k|h_k$ with $|k|=k_1+\cdots+k_p$.
Here $\mathbf N$ is the total oscillator number operator of the $p$ modes.
The operators $\mathbf Q_j=\widehat{\mathbf L}_j$ are the rotor angular
momenta introduced in Sec.~\ref{subsec:examples_notation}, so
$|\mathbf Q|^2=\sum_j\mathbf Q_j^2$ is the total angular-momentum energy of the
$q$ rotors.  The vectors
\begin{equation}
 e_{k,n,w}:=h_k\otimes\delta_n\otimes\delta_w,
 \qquad
 (k,n,w)\in\mathbb Z_{\geq0}^p\times\mathbb Z^q\times\mathbb Z_{\dia c},
 \label{eq:fe_basis}
\end{equation}
form an orthonormal basis of $L^2(M)$.

\begin{definition}
\label{def:fe_regularizer}
For $\Delta>0$ set
\begin{equation}
 \mathbf E_\Delta
 :=e^{-\Delta\mathbf N}\otimes e^{-\pi\Delta|\mathbf Q|^2/2}\otimes
   \mathbf 1_{\ell^2(\mathbb Z_{\dia c})},
 \qquad
 \mathbf E_\Delta e_{k,n,w}=a_\Delta(k,n)\,e_{k,n,w},
 \label{eq:fe_regularizer}
\end{equation}
with $a_\Delta(k,n)=e^{-\Delta|k|-\pi\Delta|n|^2/2}$.
\end{definition}

The generators of the damping are quadratic in the physical variables:
$\mathbf N$ in the
oscillator quadratures and $|\mathbf Q|^2$ in the rotor angular momenta.  The
resulting effective ranges of $(x,\xi)$ and of $n$ are both of order
$\Delta^{-1/2}$, which is what makes the Gaussian estimates of
Section~\ref{subsec:fe_characteristic} uniform across the two infinite
factors.  The finite factor $\mathbb Z_{\dia c}$ needs no damping.

By Proposition~\ref{prop:fe_smoothing} in
Appendix~\ref{appx:proof_fe_smoothing}, $\mathbf E_\Delta$ extends to
$\mathscr S'(M)$ by transposition, and every $u\in\mathscr S'(M)$ satisfies
\begin{equation}
 \mathbf E_\Delta u\in\mathscr S(M)\subseteq L^2(M),
 \qquad
 \mathbf E_\Delta u\longrightarrow u
 \ \text{ in }\mathscr S'(M)
 \quad(\Delta\downarrow0).
 \label{eq:fe_smoothing_convergence}
\end{equation}
In particular every $\mathbf E_\Delta u$, once normalized, has finite
expectation of $\mathbf N+|\mathbf Q|^2$ and of all of its powers: the
coefficients decay rapidly, while the eigenvalues grow polynomially.  This is
the precise sense in which the regularized codewords produced below have
finite energy.

\subsection{The normalized Weyl characteristic function}
\label{subsec:fe_characteristic}

Set
\begin{equation}
 \mathbf A_\Delta:=\mathbf E_\Delta^2
 =e^{-2\Delta\mathbf N}\otimes e^{-\pi\Delta|\mathbf Q|^2}\otimes
   \mathbf 1_{\ell^2(\mathbb Z_{\dia c})},
 \qquad
 \mathcal Z_\Delta:=\operatorname{Tr}\mathbf A_\Delta .
 \label{eq:fe_partition_operator}
\end{equation}
The diagonal series factors, so
\begin{equation}
 \mathcal Z_\Delta=C_{\dia c}\,(1-e^{-2\Delta})^{-p}\,\vartheta_q(\Delta),
 \qquad
 \vartheta_q(\Delta):=\sum_{n\in\mathbb Z^q}e^{-\pi\Delta|n|^2}.
 \label{eq:fe_partition_function}
\end{equation}
Both factors are finite for every $\Delta>0$; hence $\mathbf A_\Delta$ is
trace class and $\mathbf E_\Delta$ is Hilbert--Schmidt.  Poisson summation
\cite[Eqs.~1.8.14 and 1.8.16]{DLMF} in the Fourier convention
$\widehat f(\xi)=\int f(x)e^{-2\pi ix\cdot\xi}\,dx$ gives
\begin{equation}
 \vartheta_q(\Delta)
 =\Delta^{-q/2}\sum_{\ell\in\mathbb Z^q}e^{-\pi|\ell|^2/\Delta}
 =\Delta^{-q/2}\bigl(1+O(e^{-\pi/\Delta})\bigr),
 \qquad
 \mathcal Z_\Delta\sim C_{\dia c}\,(2\Delta)^{-p}\Delta^{-q/2}.
 \label{eq:fe_partition_asymptotics}
\end{equation}
Empty products and sums over trivial factors have value $1$.

\begin{definition}
For $z\in G$ put
\begin{equation}
 \Xi_\Delta(z)
 :=\frac{\operatorname{Tr}\bigl(\mathbf A_\Delta\mathbf D(z)\bigr)}
        {\mathcal Z_\Delta}.
 \label{eq:fe_characteristic_function}
\end{equation}
\end{definition}

In the notation of~\eqref{eq:we_characteristic} in the weight-enumerator
Section~\ref{sec:weight_enumerators}, this is
$\Xi_\Delta=\mathcal Z_\Delta^{-1}\Xi_{\mathbf A_\Delta}
=\Xi_{\mathbf A_\Delta/\mathcal Z_\Delta}$, the characteristic function
of the normalized density operator $\mathbf A_\Delta/\mathcal Z_\Delta$.

Since $\mathcal Z_\Delta^{-1}\mathbf A_\Delta$ is a density operator and
$\mathbf D(z)$ is unitary, $\Xi_\Delta(0)=1$ and $|\Xi_\Delta(z)|\leq1$.
Because $\mathbf A_\Delta$ is a tensor product and
$\mathbf D(z)=\mathbf M_s\mathbf T_m$ factors along
$M=\mathbb R^p\times\mathbb Z^q\times\mathbb Z_{\dia c}$, writing
$z=(x,n,w;\xi,\theta,\chi)$ gives
\begin{equation}
 \Xi_\Delta(z)
 =\Xi^{\mathbb R}_\Delta(x,\xi)\,
  \Xi^{\mathbb Z}_\Delta(n,\theta)\,
  \Xi^{\mathrm{fin}}(w,\chi).
 \label{eq:fe_characteristic_factorization}
\end{equation}

The vanishing of the rotor and finite factors
(Lemma~\ref{lem:fe_three_factors} in
Appendix~\ref{appx:proof_fe_three_factors}) implies that
$\Xi_\Delta(z)=0$ outside the closed subgroup
\begin{equation}
 H:=\{z=(x,n,w;\xi,\theta,\chi)\in G: n=0,\ w=0,\ \chi=1\}
 \cong\mathbb R^{2p}\times\mathbb T^q.
 \label{eq:fe_contributing_subgroup}
\end{equation}
Thus only points of $H$ can contribute nonzero terms to sums of
$\Xi_\Delta$ over a lattice.  Define an extended radius
\begin{equation}
 \operatorname{rad}(z)^2
 :=|x|^2+|\xi|^2+\operatorname{dist}(\theta,0)^2 \quad(z\in H),
 \qquad
 \operatorname{rad}(z):=+\infty \quad(z\notin H),
 \label{eq:fe_radius}
\end{equation}
with $e^{-\infty}=0$.  Lemma~\ref{lem:fe_three_factors} then supplies
constants $C,c>0$, independent of $z$ and of $0<\Delta\leq1$, with
\begin{equation}
 |\Xi_\Delta(z)|\leq C\,e^{-c\operatorname{rad}(z)^2/\Delta}.
 \label{eq:fe_pointwise_decay}
\end{equation}
This is pointwise decay away from the identity; $\Xi_\Delta$ is \emph{not} an
approximate identity for Haar measure, and pointwise decay alone does not
control an infinite lattice sum.  The next lemma supplies the missing
summable majorant.

\begin{lemma}
\label{lem:fe_lattice_sums}
Let $\Gamma\subseteq G$ be a closed discrete subgroup.  Then there is $a>0$
with
\begin{equation}
 \sum_{\gamma\in\Gamma\setminus\{0\}}|\Xi_\Delta(\gamma)|
 =O(e^{-a/\Delta}),
 \label{eq:fe_lattice_sum}
\end{equation}
and for every fixed $u\in G\setminus\Gamma$ there is $a_u>0$ with
\begin{equation}
 \sum_{\gamma\in\Gamma}|\Xi_\Delta(u+\gamma)|=O(e^{-a_u/\Delta}).
 \label{eq:fe_coset_sum}
\end{equation}
Both series converge absolutely for every $\Delta>0$.
\end{lemma}

The proof is given in Appendix~\ref{appx:proof_fe_lattice_sums}.

\subsection{The stabilizer form and the normalized ideal encoding}
\label{subsec:fe_encoding}

We recall the following well-known fact about closed discrete subgroups
of elementary LCA groups.
\begin{lemma}
\label{lem:fe_torsion}
A closed discrete subgroup $\Gamma\subseteq G$ is finitely generated; in
particular $D\cong\mathbb Z^r\oplus\operatorname{Tor}(D)$ with
$\operatorname{Tor}(D)$ finite.  We write
\begin{equation}
 \sigma_D:=|\operatorname{Tor}(D)|.
 \label{eq:fe_torsion_order}
\end{equation}
\end{lemma}

\begin{lemma}
\label{lem:fe_matrix_coefficients}
For every integer $L\geq0$ there are $R_L\in\mathbb Z_{\geq0}$ and $C_L>0$
such that, for all $\varphi,\psi\in\mathscr S(M)$ and all
$z=(x,n,w;\xi,\theta,\chi)\in G$,
\begin{equation}
 |\langle\varphi,\mathbf D(z)\psi\rangle|
 \leq C_L\,q_{R_L}(\varphi)q_{R_L}(\psi)\,
   (1+|x|+|\xi|+|n|)^{-L},
 \label{eq:fe_matrix_coefficient_decay}
\end{equation}
uniformly in $(w,\theta,\chi)$.  Consequently, for every closed discrete
$\Gamma\subseteq G$ there are $R,C$ with
\begin{equation}
 \sum_{\gamma\in\Gamma}|\langle\varphi,\mathbf D(\gamma)\psi\rangle|
 \leq C\,q_R(\varphi)q_R(\psi).
 \label{eq:fe_matrix_coefficient_summability}
\end{equation}
\end{lemma}

The proof is given in Appendix~\ref{appx:proof_fe_matrix_coefficients}.

\begin{definition}
\label{def:fe_averaging_form}
Let $\Pi_D:\mathscr S(M)\to\mathscr S'(M)$ be defined by
\begin{equation}
 \langle\Pi_D\psi,\varphi\rangle
 :=\frac{1}{\sigma_D}
   \sum_{s\in D}\langle\varphi,\widetilde{\mathbf D}(s)\psi\rangle .
 \label{eq:fe_averaging_form}
\end{equation}
\end{definition}

This is Rieffel's $C(\widehat D)$-valued inner product
from~\eqref{eq:fourier_transformed_rieffel_inner_product}, evaluated at
the trivial character $1\in\widehat D$ and divided by $\sigma_D$:
\[
 \langle\Pi_D\psi,\varphi\rangle
 =\frac{1}{\sigma_D}\,{}_{\mathrm R}\langle\psi,\varphi\rangle(1).
\]
The order of the arguments reflects the left-module inner-product
convention; the equality follows by reindexing $s\mapsto-s$ and using
$\widetilde{\mathbf D}(-s)=\widetilde{\mathbf D}(s)^\dagger$.

The series converges absolutely by
Lemma~\ref{lem:fe_matrix_coefficients}.  Absolute convergence also follows
from Rieffel's inner product taking values in $\mathscr S(D)\subset\ell^1(D)$.
It defines a jointly continuous
sesquilinear form, and reindexing it gives
$\widetilde{\mathbf D}(s)\Pi_D=\Pi_D$ for $s\in D$, so
$\operatorname{ran}\Pi_D\subseteq\mathcal C$; and since $\rho(s,z)=1$ for
$s\in D$, $z\in D^\perp$,
\begin{equation}
 \mathbf D(z)\,\Pi_D\,\mathbf D(z)^\dagger=\Pi_D
 \qquad(z\in D^\perp).
 \label{eq:fe_form_dual_invariance}
\end{equation}
For infinite $D$ the form $\Pi_D$ is not a bounded operator on $L^2(M)$ and
its square is not defined; we never use either.  The factor $\sigma_D^{-1}$
fixes a normalization, and only that choice: writing
$D\cong\mathbb Z^r\oplus\operatorname{Tor}(D)$, the sum
\eqref{eq:fe_averaging_form} is a bare Fourier comb in the $r$ free
directions and the \emph{averaged} projection
$\sigma_D^{-1}\sum_{s\in\operatorname{Tor}(D)}\widetilde{\mathbf D}(s)$ onto
the trivial torsion character in the remaining ones.  For finite $D$ it is the
usual orthogonal projection onto $\mathcal C$.

\begin{proposition}
\label{prop:fe_encoding}
Let $U:\mathbb C^K\to E_{D,1}$ be unitary.  Then
\[
 \operatorname{Enc}:=\sigma_D^{-1/2}\iota_1 U:
 \mathbb C^{K}\longrightarrow
 \mathcal C\subseteq\mathscr S'(M),
\]
is injective and satisfies
\begin{equation}
 \operatorname{Enc}\operatorname{Enc}^\dagger=\Pi_D
 \qquad\text{as maps }\mathscr S(M)\to\mathscr S'(M),
 \label{eq:fe_encoding_normalization}
\end{equation}
where, for $\Psi_j:=\operatorname{Enc}(e_j)$ and an orthonormal logical basis
$(e_j)$, $(\operatorname{Enc}^\dagger\varphi)_j=\overline{\langle\Psi_j,\varphi\rangle}$.
Equivalently, for all $\varphi,\psi\in\mathscr S(M)$,
\begin{equation}
 \sum_{j=1}^{K}\langle\Psi_j,\varphi\rangle
 \overline{\langle\Psi_j,\psi\rangle}
 =\frac{1}{\sigma_D}\sum_{s\in D}
   \langle\varphi,\widetilde{\mathbf D}(s)\psi\rangle.
 \label{eq:fe_encoding_factorization}
\end{equation}
\end{proposition}

\begin{proof}
The map $\iota_1$ is injective and has range $\mathcal C$ by
Lemma~\ref{lem:trivial_fiber_codespace}.  Since $(Ue_j)$ is an orthonormal
basis of $E_{D,1}$, the embedding~\eqref{eq:rieffel_fiber_distribution_embedding}
and the fiber inner product~\eqref{eq:rieffel_fiber_scalar_inner_product} give
\[
 \begin{aligned}
 \sum_{j=1}^K\langle\Psi_j,\varphi\rangle
     \overline{\langle\Psi_j,\psi\rangle}
 &=\sigma_D^{-1}\langle[\varphi]_1,[\psi]_1\rangle_1\\
 &=\sigma_D^{-1}\,{}_{\mathrm R}\langle\psi,\varphi\rangle(1)
 =\langle\Pi_D\psi,\varphi\rangle.
 \end{aligned}
\]
This proves~\eqref{eq:fe_encoding_normalization}.
\end{proof}

We call an ideal encoding \emph{normalized} when it satisfies the
factorization~\eqref{eq:fe_encoding_normalization}. This convention fixes
the scale of the distributional codewords. For codewords outside $L^2(M)$,
it specifies a normalization without assigning a Hilbert-space norm.
Such an encoding exists for every isotropic lattice under consideration
and is fixed up to a unitary change of logical basis.
Section~\ref{subsec:fe_examples} gives explicit combs in this normalization
for the rotor GKP family.

\subsection{Asymptotic isometry}
\label{subsec:fe_isometry}

Throughout this subsection choose a normalized encoding as in
Proposition~\ref{prop:fe_encoding}.
Define the regularized encoding and its Gram matrix,
\begin{equation}
 \operatorname{Enc}_\Delta:=\mathbf E_\Delta\operatorname{Enc}
 :\mathbb C^{K}\longrightarrow\mathscr S(M)\subseteq L^2(M),
 \qquad
 \mathsf G_\Delta:=\operatorname{Enc}_\Delta^\dagger\operatorname{Enc}_\Delta .
 \label{eq:fe_regularized_encoding}
\end{equation}
These are ordinary Hilbert-space objects: the columns of
$\operatorname{Enc}_\Delta$ are the finite-energy codewords
$\Psi_j^\Delta:=\mathbf E_\Delta\Psi_j\in\mathscr S(M)$, and
$\mathsf G_\Delta$ is the positive definite $K\times K$ matrix of their inner
products. The exact stabilizer and decoding results of
Section~\ref{sec:distance_error_correction} concern the ideal code; this
finite-$\Delta$ encoding gives an approximate realization of that code.

\begin{lemma}
\label{lem:fe_trace_expansion}
For $\tilde u\in G$ the expression
$\mathsf T_{\Delta,\tilde u}
 :=\mathbf E_\Delta\mathbf D(\tilde u)^\dagger\Pi_D\mathbf E_\Delta$
extends to an operator of rank at most $K$ on $L^2(M)$, and
\begin{equation}
 \operatorname{Tr}\mathsf T_{\Delta,\tilde u}
 =\frac{\mathcal Z_\Delta}{\sigma_D}
  \sum_{s\in D}\zeta_{\tilde u}(s)\,\Xi_\Delta(s-\tilde u),
 \qquad |\zeta_{\tilde u}(s)|=1,
 \label{eq:fe_trace_expansion}
\end{equation}
the series converging absolutely.  For $\tilde u=0$ one has
$\zeta_0(0)=1$.
\end{lemma}

The proof is given in Appendix~\ref{appx:proof_fe_trace_expansion}.

\begin{theorem}
\label{thm:fe_asymptotic_isometry}
For a normalized encoding as in Proposition~\ref{prop:fe_encoding}, put
\begin{equation}
 \nu_\Delta:=\frac{K\sigma_D}{\mathcal Z_\Delta}
 =\frac{K\sigma_D}{C_{\dia c}}\,
   \frac{(1-e^{-2\Delta})^{p}}{\vartheta_q(\Delta)}
 \ \sim\ \frac{K\sigma_D}{C_{\dia c}}\,(2\Delta)^{p}\Delta^{q/2}.
 \label{eq:fe_normalization_constant}
\end{equation}
Then there are $a,C>0$ and $\Delta_0>0$ such that
\begin{equation}
 \bigl\|\nu_\Delta\,\mathsf G_\Delta-\mathbf 1\bigr\|_{\mathrm{op}}
 \leq Ce^{-a/\Delta}
 \qquad(0<\Delta\leq\Delta_0),
 \label{eq:fe_gram_estimate}
\end{equation}
equivalently, for all $\varphi,\psi\in\mathbb C^{K}$,
\begin{equation}
 \bigl|\nu_\Delta\langle\operatorname{Enc}_\Delta\varphi,
   \operatorname{Enc}_\Delta\psi\rangle-\langle\varphi,\psi\rangle\bigr|
 \leq Ce^{-a/\Delta}\|\varphi\|\,\|\psi\|.
 \label{eq:fe_isometry_estimate}
\end{equation}
\end{theorem}

The proof is given in Appendix~\ref{appx:proof_fe_asymptotic_isometry}.

For $M=\mathbb R^p$, this recovers the asymptotic-isometry result of
~\cite[Theorem~5.2]{mayrand2026complex}: the normalization
reduces to $\nu_\Delta=K(1-e^{-2\Delta})^p$, matching their
$c(\beta)=\det(\Lambda)(1-e^{-2\beta})^n$ with
$\beta=\Delta$, $n=p$, and $\det(\Lambda)=K$.
The present theorem extends this estimate to mixed oscillator, rotor,
and finite factors.

\begin{remark}
\label{rem:fe_interpretation}
The map $\sqrt{\nu_\Delta}\operatorname{Enc}_\Delta$ is an asymptotic
isometry.  In a logical orthonormal basis the finite-energy codewords
$\Psi_j^\Delta$ therefore have the \emph{same} squared norm
$\mathcal Z_\Delta/(K\sigma_D)$ to relative accuracy $O(e^{-a/\Delta})$, and
after rescaling by $\sqrt{\nu_\Delta}$, distinct columns have inner
products $O(e^{-a/\Delta})$. Thus the rescaled Gram matrix agrees with the
identity to every algebraic order in $\Delta$.  For
each fixed $\Delta>0$ an exact isometry is
$\operatorname{Enc}_\Delta\mathsf G_\Delta^{-1/2}$; the content of
Theorem~\ref{thm:fe_asymptotic_isometry} is that the scalar
$\sqrt{\nu_\Delta}$ already reproduces this matrix correction to
exponential relative accuracy.  Replacing $\nu_\Delta$ by its leading
asymptotics $K\sigma_D C_{\dia c}^{-1}(2\Delta)^p\Delta^{q/2}$ still gives
$\lim_{\Delta\downarrow0}
 \nu_\Delta\langle\operatorname{Enc}_\Delta\varphi,
 \operatorname{Enc}_\Delta\psi\rangle=\langle\varphi,\psi\rangle$,
but introduces a relative $O(\Delta)$ error when $p>0$.  The theorem asserts
the existence of some $a>0$ for each fixed code; it does not claim the
optimal exponent, nor a rate uniform over a family of lattices, which would
require tracking the quadratic form in \eqref{eq:fe_oscillator_factor}
against the geometry of the contributing cosets.
\end{remark}

\subsection{Examples}
\label{subsec:fe_examples}

\subsubsection*{Planar-rotor GKP codes}

For the single-rotor GKP lattice
$D=KN\Z\times(\tfrac1N\Z/\Z)$ of
Sec.~\ref{ex:rotor-gkp}, one has $\sigma_D=N$ and
$\mathcal Z_\Delta=\vartheta_1(\Delta)$.  Averaging the finite modulation
subgroup projects onto $N\Z$, so the combs~\eqref{eq:rotor-css-codewords}
satisfy
\begin{equation}
 \Pi_D=\sum_{j=0}^{K-1}\ket{j_L}\bra{j_L},
 \qquad
 \operatorname{Enc}(e_j)=\ket{j_L},
 \label{eq:fe_rotor_encoding}
\end{equation}
verifying~\eqref{eq:fe_encoding_normalization}.  Applying the regularizer
gives the finite-energy states
\begin{equation}
 \Psi_j^\Delta:=\mathbf E_\Delta\ket{j_L}
 =\sum_{k\in\Z}e^{-\frac\pi2\Delta N^2(j+Kk)^2}
   \ket{N(j+Kk)}_{\mathrm{rot}}.
 \label{eq:fe_rotor_finite_energy_codeword}
\end{equation}
Their disjoint supports make them exactly orthogonal, and Poisson
summation gives
\begin{equation}
 \begin{aligned}
 \|\Psi_j^\Delta\|^2
 &=\frac{1}{KN\sqrt\Delta}\sum_{\ell\in\Z}
   e^{-\pi\ell^2/(K^2N^2\Delta)}e^{2\pi i\ell j/K}\\
 &=\frac{1}{KN\sqrt\Delta}
   \bigl(1+O(e^{-\pi/(K^2N^2\Delta)})\bigr).
 \end{aligned}
 \label{eq:fe_rotor_norm}
\end{equation}
Thus Theorem~\ref{thm:fe_asymptotic_isometry} yields
\begin{equation}
 \nu_\Delta=\frac{K\sigma_D}{\mathcal Z_\Delta}
 =\frac{KN}{\vartheta_1(\Delta)}\sim KN\sqrt\Delta,
 \qquad
 \nu_\Delta\langle\Psi_i^\Delta,\Psi_j^\Delta\rangle
 =\delta_{ij}+O(e^{-\pi/(K^2N^2\Delta)}).
 \label{eq:fe_rotor_normalization}
\end{equation}
The Gaussian angular-momentum envelope is the planar-rotor
regularization studied by Xu, Wang, and Albert~\cite{rotorclifford},
up to the damping-parameter convention and choice of logical basis.
Albert, Covey, and Preskill~\cite{albert2020robust} use analogous
angular-momentum damping for rigid-rotor codes on $\mathrm{SO}(3)$, whose
nonabelian configuration space lies beyond the present LCA setting.

\subsubsection*{The torsion-extended hybrid qubit}

For the $K=2$ code of Sec.~\ref{ex:torsion},
apply $\mathbf E_\Delta$ to the ideal combs in
\eqref{eq:hybrid-torsion-logical-qubit} and normalize.  The resulting
finite-energy wavefunctions are
\begin{equation}
 \psi_{k,\Delta}(x,n,j)\propto
 e^{-\pi\Delta n^2/2}\,\mathbf1_{\{j=n\bmod2\}}
 \sum_{\ell\in\Z}M_\Delta(x,4\ell+2k+n),
 \qquad k=0,1,
 \label{eq:fe_hybrid_qubit_codewords}
\end{equation}
where, for our normalized Hermite functions, Mehler's
formula~\cite[Eq.~18.18.28]{DLMF} gives the kernel of
$e^{-\Delta\mathbf N}$ in the completed-square form
\begin{equation}
 M_\Delta(x,y)
 =\frac{1}{\sqrt{\pi(1-e^{-2\Delta})}}\,
  e^{-\frac12(\tanh\Delta)y^2}\,
  e^{-\frac12(\coth\Delta)(x-y\operatorname{sech}\Delta)^2}.
 \label{eq:fe_mehler_kernel}
\end{equation}
It smooths and envelopes the oscillator comb, while the factor
$e^{-\pi\Delta n^2/2}$ damps the rotor angular momentum and the physical
qubit is unchanged.  By Proposition~\ref{prop:fe_smoothing}, these states
lie in $\mathscr S(M)$ and have finite expectation of
$\mathbf N+\mathbf Q^2$ for every $\Delta>0$.

\section{Conclusion \& Future work}

We formalize stabilizer codes on general phase spaces --- generalized locally compact Abelian (GLCA) codes --- encompassing stabilizer codes on oscillators, rotors, and qudits, as well as their various hybrid combinations. 
We also provide several interesting examples of hybrid codes that are supported, in part, by rotor factors.
Several of these examples, like oscillator-qudit LCA codes, couple their subsystems intrinsically: no generalized Clifford operator transforms them into separate subsystem codes.

LCA codes were shown numerically to provide advantages in their error-correcting capabilities vs GKP codes.
It would be interesting to investigate whether GLCA codes inherit similar advantages and under what conditions, especially given that well-controlled rotor systems are not yet readily available or scalable.

The hybrid Golay code of Sec.~\ref{ex:hybrid-golay} builds one hybrid lattice from the binary and quaternary Golay codes.
Its signed-permutation symmetry group is the double cover \(2.M_{12}\) of the sporadic Mathieu group \(M_{12}\).
It would be interesting to catalog other sporadic or exceptional hybrid lattices over various combinations of factor types.

\begin{acknowledgments}
S.C.\@ acknowledges funding from the Anusandhan National Research Foundation (ANRF), Government of India, under the
Prime Minister Early Career Research Grant (File no. ANRF/ECRG/2024/006557/PMS).
V.V.A.\@ acknowledges NSF
grants OMA-2120757 (QLCI) and CIF-2330909.

Parts of this manuscript were prepared with the assistance of ChatGPT 5.6,
developed by OpenAI, and Claude Code, developed by Anthropic. These agents
were used to verify proofs and edit the writing of some sections, in
accordance with the authors' instructions. All ideas presented in this work
originate with the authors. All content, claims, and
conclusions have been reviewed and verified by the authors to ensure
accuracy and originality. Certain products, commercial and otherwise, are
mentioned in this publication. These mentions are for informational
purposes only, and do not imply recommendation or endorsement by NIST.
\end{acknowledgments}

\vspace{.2cm}

\bibliography{references}

\newpage
\appendix
\onecolumngrid


\newpage
 \section*{Appendices}

\section{Proofs and technical details}

\subsection{Proof of Theorem~\ref{thm:rieffel_bundle_decomposition}}
\label{appx:proof_rieffel_bundle_decomposition}

\begin{proof}
Equations~\eqref{eq:rieffel_fiber_norm},
\eqref{eq:canonical_trace_haar}, and \eqref{eq:rieffel_trace_L2_norm} give
\begin{align}
 \|\mathcal U_D^0\xi\|_{L^2(\widehat D,\mathscr E_D)}^2
 &=\int_{\widehat D}\|[\xi]_\chi\|^2\,d\chi \notag\\
 &=\tau_D\!\left({}_{\mathrm R}\langle\xi,\xi\rangle\right)
 =\|\xi\|_{L^2(M)}^2.
 \label{eq:rieffel_fiberization_isometry}
\end{align}
Thus $\mathcal U_D^0$ is an isometry for the scalar $L^2$-norm.
Since $\mathscr S(M)$ is dense in $L^2(M)$, it extends uniquely to an
isometry
\[
 \mathcal U_D:L^2(M)\longrightarrow
 L^2(\widehat D,\mathscr E_D).
\]
The range of $\mathcal U_D$ is closed.  Indeed, if
$\mathcal U_Dx_n\to y$, then
\[
 \|x_n-x_m\|_{L^2(M)}
 =\|\mathcal U_Dx_n-\mathcal U_Dx_m\|_{L^2(\widehat D,\mathscr E_D)},
\]
so $(x_n)$ is Cauchy.  Completeness gives $x_n\to x\in L^2(M)$, and
continuity gives $y=\mathcal U_Dx$.  It remains to prove that this closed
range is dense.

Set
\[
 \mathcal A:=\mathcal U_D^0(\mathscr S(M)),
 \qquad
 \mathcal B:=\Gamma(\mathscr E_D).
\]
Since $\mathscr S(M)$ is dense in $\Omega_D$ in the module norm, the
Serre--Swan identification~\eqref{eq:serre_swan_evaluation_map} shows that
$\mathcal A$ is uniformly dense in $\mathcal B$.  Uniform convergence implies
$L^2$-convergence because Haar measure on $\widehat D$ is normalized:
\[
 \|b-a\|_{L^2(\widehat D,\mathscr E_D)}^2
 =\int_{\widehat D}\|b(\chi)-a(\chi)\|^2\,d\chi
 \leq\|b-a\|_\infty^2.
\]
Consequently,
\begin{equation}
 \mathcal B\subseteq
 \overline{\mathcal A}^{\,L^2(\widehat D,\mathscr E_D)}.
 \label{eq:uniform_density_implies_L2_density}
\end{equation}
It is well known that $\mathcal B$ is dense in
$L^2(\widehat D,\mathscr E_D)$.  
Equation~\eqref{eq:uniform_density_implies_L2_density} now shows that
$\mathcal A$, and hence the range of $\mathcal U_D$, is dense in $L^2(\widehat D,\mathscr E_D)$.  Since that
range is also closed, it equals the whole target space.  Thus $\mathcal U_D$
is the unitary in \eqref{eq:rieffel_direct_integral}.

It remains to prove the two intertwining formulas.  Let
$f\in C(\widehat D)$ and $\xi\in\Omega_D$.  Since
$f-f(\chi)1\in I_\chi$, the definition
$E_{D,\chi}=\Omega_D/\overline{I_\chi\Omega_D}$ gives
\begin{equation}
 [f\cdot\xi]_\chi-f(\chi)[\xi]_\chi
 =[(f-f(\chi)1)\cdot\xi]_\chi=0.
 \label{eq:evaluation_turns_module_action_scalar}
\end{equation}
Thus evaluation at $\chi$ turns the $C(\widehat D)$-module action into
multiplication by the scalar $f(\chi)$.

Now let $a\in\mathscr S(D)$ and $\xi\in\mathscr S(M)$, so that
$a\cdot\xi\in\mathscr S(M)$. Under the
Fourier identification, the action of $a$ is the action of
$\widehat a=\mathcal F_Da\in C(\widehat D)$.  Therefore
\begin{align}
 \bigl(\mathcal U_D(a\cdot\xi)\bigr)(\chi)
 &=[a\cdot\xi]_\chi \notag\\
 &=[\widehat a\cdot\xi]_\chi \notag\\
 &=\widehat a(\chi)[\xi]_\chi \notag\\
 &=\widehat a(\chi)(\mathcal U_D\xi)(\chi),
 \label{eq:fiberwise_action_calculation}
\end{align}
where the third equality is
\eqref{eq:evaluation_turns_module_action_scalar}.  This proves
\eqref{eq:rieffel_fiberwise_action} on the dense subspace
$\mathscr S(M)$ for every $\chi$, because all terms in
\eqref{eq:fiberwise_action_calculation} are continuous sections.
For a general $\xi\in L^2(M)$, choose
$\xi_n\in\mathscr S(M)$ with $\xi_n\to\xi$ in $L^2(M)$.  The action of $a$
is bounded, multiplication by $\widehat a$ is bounded on the direct integral,
and $\mathcal U_D$ is unitary.  Passing to the limit proves
\eqref{eq:rieffel_fiberwise_action} in the direct-integral sense, hence for
almost every $\chi$, initially for $a\in\mathscr S(D)$.
For arbitrary $a\in\mathrm C^*(D)$, choose $a_n\in\mathscr S(D)$
with $\|a_n-a\|_{\mathrm C^*(D)}\to0$. Their actions converge in
operator norm on $L^2(M)$, and $\widehat a_n\to\widehat a$ uniformly.
Passing to the limit therefore proves the intertwining formula for
every $a\in\mathrm C^*(D)$.

The trace identity also gives $\|\xi\|_2\leq\|\xi\|_{\Omega_D}$
on $\mathscr S(M)$. Approximation in the module norm thus converges
both in $L^2(M)$ and uniformly as sections, identifying $\mathcal U_D\xi$
with the section $\chi\mapsto[\xi]_\chi$ for $\xi\in\Omega_D$.
This identification is injective because Haar measure has full support.
For such $\xi$, both sides of the intertwining formula are continuous
sections, so their almost-everywhere equality holds at every $\chi$.

Finally take $a=u_w$.  By the phase-corrected action
\eqref{eq:phase_corrected_AD_action},
$u_w\cdot\xi=\widetilde{\mathbf D}(w)\xi$, while
\eqref{eq:fourier_transform_convention} gives
$\widehat u_w(\chi)=\overline{\chi(w)}$.  Substitution into
\eqref{eq:rieffel_fiberwise_action} yields
\[
 \bigl(\mathcal U_D\widetilde{\mathbf D}(w)\xi\bigr)(\chi)
 =\overline{\chi(w)}(\mathcal U_D\xi)(\chi),
\]
which is precisely \eqref{eq:rieffel_stabilizer_diagonalization}.
As above, this last identity holds for every $\chi$ when
$\xi\in\Omega_D$, and for almost every $\chi$ when $\xi$ is only in
$L^2(M)$.
\end{proof}

\subsection{Proof of Lemma~\ref{lem:trivial_fiber_codespace}}
\label{appx:proof_trivial_fiber_codespace}

\begin{proof}
Equation~\eqref{eq:rieffel_fiber_generalized_eigenvalue} gives one inclusion.
For the converse, Rieffel's smooth projectivity result
\cite[Propositions~3.7]{rieffel}, with the module actions interchanged,
gives finite dual frames $x_j,y_j\in\mathscr S(M)$ such that
\begin{equation}
 \xi=\sum_{j=1}^N a_j(\xi)\cdot x_j,
 \qquad
 a_j(\xi):={}_{\mathfrak{A}_D}\!\langle\xi,y_j\rangle\in\mathscr S(D).
 \label{eq:rieffel_smooth_dual_frame}
\end{equation}
If $u\in\mathcal C$ and $a=\sum_{w\in D}a(w)u_w\in\mathscr S(D)$,
then conjugate-linearity in the test function gives
$\langle u,a\cdot x\rangle=\sum_{w\in D}\overline{a(w)}
\langle u,\widetilde{\mathbf D}(w)x\rangle$, since the series for
$a\cdot x$ converges in the Schwartz topology.
The invariance of $u$ implies
$\langle u,\widetilde{\mathbf D}(w)x\rangle=\langle u,x\rangle$
for every $w\in D$.  Since $1$ is the trivial character,
$\widehat a(1)=\sum_{w\in D}a(w)$, and hence
$\langle u,a\cdot x\rangle=\overline{\widehat a(1)}\langle u,x\rangle$.
Hence
\begin{align}
 \langle u,\xi\rangle
 &=\sum_{j=1}^N\overline{\widehat{a_j(\xi)}(1)}
 \langle u,x_j\rangle \notag\\
 &=\sum_{j=1}^N
 \langle[\xi]_1,[y_j]_1\rangle\langle u,x_j\rangle.
 \label{eq:codespace_fiber_bound}
\end{align}
This is bounded by $C_u\|[\xi]_1\|$ and depends only on $[\xi]_1$.
Consequently,
\[
 L_u([\xi]_1):=\overline{\langle u,\xi\rangle}
\]
defines a bounded linear functional on $E_{D,1}$.  By the Riesz
representation theorem, there is a $v\in E_{D,1}$ such that
$L_u([\xi]_1)=\langle v,[\xi]_1\rangle$.  Taking complex conjugates gives
$\langle u,\xi\rangle=\langle[\xi]_1,v\rangle
=\langle\iota_1(v),\xi\rangle$, and hence $u=\iota_1(v)$.
\end{proof}

\subsection{Proof of Proposition~\ref{prop:rieffel_fiber_covariance}}
\label{appx:proof_rieffel_fiber_covariance}

\begin{proof}
For $w\in D$, the Weyl relation
\[
 \widetilde{\mathbf D}(w)\mathbf D(z)
 =\rho(w,z)\mathbf D(z)\widetilde{\mathbf D}(w)
\]
shows that conjugation by $\mathbf D(z)$ multiplies the $w$-th Fourier
coefficient of the $C(\widehat D)$-valued inner product by the character
$\kappa_z(w)=\overline{\rho(w,z)}$.  With the Fourier convention in
\eqref{eq:fourier_transform_convention}, this gives
\begin{equation}
 \bigl({}_{\mathrm R}\langle
 \mathbf D(z)\xi,\mathbf D(z)\eta\rangle\bigr)(\kappa_z\chi)
 =\bigl({}_{\mathrm R}\langle\xi,\eta\rangle\bigr)(\chi),
\end{equation}
first for $\xi,\eta\in\mathscr S(M)$ and then, by continuity, for
$\xi,\eta\in\Omega_D$.  In particular,
$\mathbf D(z)N_\chi=N_{\kappa_z\chi}$: one inclusion follows from the
displayed identity, and the reverse inclusion follows by applying the same
argument to $\mathbf D(z)^\dagger$.  Hence the formula in
\eqref{eq:rieffel_fiber_covariance} is independent of the representative
$\xi$.  The identity also gives
\[
 \langle V_{z,\chi}[\xi]_\chi,V_{z,\chi}[\eta]_\chi\rangle_{\kappa_z\chi}
 =\langle[\xi]_\chi,[\eta]_\chi\rangle_\chi.
\]
Thus $V_{z,\chi}$ is an isometry, and the map induced by
$\mathbf D(z)^\dagger$ is its inverse; therefore it is unitary.

We next check compatibility with the distributional embeddings.  Using
\eqref{eq:rieffel_fiber_distribution_embedding}, the preceding covariance
identity, and the definition of the distributional action, one obtains, for
every $v\in E_{D,\chi}$,
\begin{equation}
 \iota_{\kappa_z\chi}\!\left(V_{z,\chi}v\right)
 =\widehat{\mathbf D(z)}\,\iota_\chi(v).
\end{equation}
This is the asserted intertwining in $\mathscr S'(M)$.

Finally, let $\chi\in\widehat D$.  Surjectivity of $\kappa$ provides
$z\in G$ with $\kappa_z=\chi$.  If $u\in\mathcal C$, then for every $w\in D$,
\[
 \widehat{\widetilde{\mathbf D}(w)}\,
 \widehat{\mathbf D(z)}u
 =\rho(w,z)\widehat{\mathbf D(z)}
   \widehat{\widetilde{\mathbf D}(w)}u
 =\overline{\chi(w)}\widehat{\mathbf D(z)}u.
\]
Thus displacement by $z$ sends the codespace into the $\chi$-eigenspace.
Conversely, applying $\widehat{\mathbf D(z)}^{-1}$ to any vector in that
eigenspace gives a vector fixed by every
$\widehat{\widetilde{\mathbf D}(w)}$.  Lemma~\ref{lem:trivial_fiber_codespace}
and \eqref{eq:fiber_embedding_covariance} now identify these two spaces with
$E_{D,1}$ and $E_{D,\chi}$, respectively, proving
\eqref{eq:all_fibers_full_eigenspaces}.
\end{proof}

\subsection{Continuity of the Schwartz inclusion into Rieffel's module}
\label{appx:schwartz_heisenberg_continuity}

We prove here the continuity assertion used in
Section~\ref{a_sec:glca_ncg}.  Recall that
$M=\mathbb R^p\times\mathbb Z^q\times\mathbb Z_{\dia c}$ and that
$D\subset M\times\widehat M$ is an isotropic lattice.  In particular, $D$ is
a finitely generated discrete abelian group.  Throughout this appendix we use
the phase trivialization of Lemma~\ref{lem:phase_triv}: thus
$\mathfrak{A}_D$ is identified with the untwisted algebra $\mathrm C^*(D)$, and its canonical
unitary $u_w$ acts as the genuine representation
$\widetilde{\mathbf D}(w)=c(w)\mathbf D(w)$. Recall the Schwartz algebra
$\mathscr S(D)$ introduced in Section~\ref{a_sec:glca_ncg}.
Both $\mathscr S(M)$ and
$\mathscr S(D)$ are Fr\'echet spaces,
the inclusion $\mathscr S(M)\hookrightarrow L^2(M)$ is continuous, and
\begin{equation}
 \mathscr S(D)\hookrightarrow\ell^1(D)
 \label{eq:schwartz_D_l1_inclusion}
\end{equation}
is continuous.  Rieffel's construction
\cite[Section~2]{rieffel} also gives
\begin{equation}
 (\xi,\eta)\longmapsto
 \bigl(w\longmapsto{}_{\mathfrak{A}_D}\!\langle\xi,\eta\rangle(w)\bigr)
 =\bigl(w\longmapsto
 \langle\widetilde{\mathbf D}(w)\eta,\xi\rangle_{L^2(M)}\bigr)
 \label{eq:schwartz_rieffel_coefficient_map}
\end{equation}
as a well-defined sesquilinear map from
$\mathscr S(M)\times\mathscr S(M)$ to $\mathscr S(D)$.

We first record the continuity hidden in this well-definedness statement.
Fix $\eta\in\mathscr S(M)$ and denote by $T_\eta$ the linear map in the first
variable in \eqref{eq:schwartz_rieffel_coefficient_map}.  Its graph is closed.
Indeed, suppose that $\xi_n\to\xi$ in $\mathscr S(M)$ and
$T_\eta\xi_n\to a$ in $\mathscr S(D)$.  For every $w\in D$, continuity of
$\mathscr S(M)\hookrightarrow L^2(M)$, Cauchy--Schwarz, and unitarity of
$\widetilde{\mathbf D}(w)$ give
\begin{align}
 \left|(T_\eta\xi_n)(w)-(T_\eta\xi)(w)\right|
 &=\left|\left\langle
 \widetilde{\mathbf D}(w)\eta,\xi_n-\xi
 \right\rangle_{L^2(M)}\right| \notag\\
 &\leq\|\widetilde{\mathbf D}(w)\eta\|_2\,\|\xi_n-\xi\|_2 \notag\\
 &=\|\eta\|_2\,\|\xi_n-\xi\|_2\longrightarrow0.
 \label{eq:coefficient_pointwise_L2_limit}
\end{align}
On the other hand, evaluation at $w$ is continuous on $\mathscr S(D)$ by
\eqref{eq:schwartz_D_l1_inclusion}, since
$|b(w)|\leq\|b\|_{\ell^1(D)}$.  Hence
$(T_\eta\xi_n)(w)\to a(w)$, and uniqueness of pointwise limits yields
$a=T_\eta\xi$.  The closed graph theorem
\cite[Theorem~2.15]{RudinFunctionalAnalysis} now shows that $T_\eta$ is
continuous.  The same argument applies in the second variable.  After
composing that variable with complex conjugation, we obtain a separately
continuous bilinear map between Frechet spaces.  The separate-to-joint
continuity theorem \cite[Theorem~2.17]{RudinFunctionalAnalysis}, followed by
\eqref{eq:schwartz_D_l1_inclusion}, shows that the coefficient map into
$\ell^1(D)$ is jointly continuous.

\begin{proposition}
\label{prop:schwartz_heisenberg_continuity}
The identity inclusion $\mathscr S(M)\hookrightarrow\Omega_D$ is continuous.
\end{proposition}

\begin{proof}
For every $b\in\ell^1(D)$, the triangle inequality and unitarity of the
canonical generators give
\begin{equation}
 \left\|\sum_{w\in D}b(w)u_w\right\|_{\mathrm C^*(D)}
 \leq\sum_{w\in D}|b(w)|.
 \label{eq:group_Cstar_l1_bound}
\end{equation}
Indeed, this estimate holds in every unitary representation of $D$, and
therefore in the universal norm.

Now suppose that $\xi_n\to0$ in the Schwartz topology.  Joint continuity of
the coefficient map proved above gives
\begin{equation}
 \bigl(w\longmapsto
 \langle\widetilde{\mathbf D}(w)\xi_n,\xi_n\rangle_{L^2(M)}\bigr)
 \longrightarrow0
 \qquad\text{in }\ell^1(D).
 \label{eq:diagonal_coefficients_l1_zero}
\end{equation}
Using \eqref{eq:group_Cstar_l1_bound} and the definition of the module norm,
we therefore obtain
\begin{align}
 \|\xi_n\|_{\Omega_D}^2
 &=\left\|{}_{\mathfrak{A}_D}\!\langle\xi_n,\xi_n\rangle\right\|_{\mathfrak{A}_D}
 \notag\\
 &\leq\sum_{w\in D}
 \left|\langle\widetilde{\mathbf D}(w)\xi_n,
 \xi_n\rangle_{L^2(M)}\right|
 \longrightarrow0.
 \label{eq:module_norm_sequence_zero}
\end{align}
Thus $\xi_n\to0$ in $\Omega_D$.  Since $\mathscr S(M)$ is Frechet and hence
metrizable, sequential continuity is equivalent to continuity.  This proves
that the linear inclusion $\mathscr S(M)\hookrightarrow\Omega_D$ is
continuous.
\end{proof}

\subsection{Proof of Theorem~\ref{thm:code_dimension}}
\label{appx:proof_code_dimension}
\begin{proof}
A projection $P$ representing $\Omega_D$ has pointwise rank $K$, so
$\tau_D([\Omega_D])=\int_{\widehat D}\operatorname{Tr}P(\chi)\,d\chi=K$.
Fourier transformation identifies the trace pairings, and Rieffel's trace
formula~\cite[Theorem~3.4]{rieffel}, applied with the two
Morita-equivalent algebras interchanged,
gives $\mathrm{Tr}_D([\Omega_D])=|G/D|_\mu$.
The dual-covolume identity~\eqref{eq:dual_covolume_reciprocity}
follows directly from~\cite[p.~277]{rieffel}, with the same
self-Plancherel Haar normalization.

Finally, isotropy gives the exact sequence
\begin{equation}\label{eq:quotient_exact_sequence}
 0\longrightarrow D^\perp/D\longrightarrow G/D
 \longrightarrow G/D^\perp\longrightarrow0,
\end{equation}
and the finite kernel contributes its index to the quotient volume:
\begin{equation}\label{eq:logical_quotient_covolume}
 [D^\perp:D]
 =\frac{|G/D|_\mu}{|G/D^\perp|_\mu}=K^2.
\end{equation}
\end{proof}

\subsection{Proof of Lemma~\ref{lem:full_lift_pfaffian}}
\label{appx:proof_full_lift_pfaffian}

\begin{proof}
Split $\widetilde G=\widetilde G_{\mathrm c}\times\widetilde G_{\mathrm d}$,
where $\widetilde G_{\mathrm c}\cong\R^{2p+q}$ carries the continuous
coordinates $(x,\xi,\theta)$ and $\widetilde G_{\mathrm d}\cong\Z^{q+2k}$
carries the discrete coordinates $(n,a,b)$.
The map $\pi$ is an open continuous homomorphism with discrete kernel, so
$G\cong\widetilde G/\Lambda_{\mathrm{per}}$ as topological groups.
By Eq.~\eqref{eq:lifted_cocycle_bicharacter},
$\rho(\pi(v),\pi(w))=e^{2\pi i\,v^{\intercal}Jw}$ for all
$v,w\in\widetilde G$.

(i) For $0<\epsilon<1/2$, let $U_\epsilon\subset\widetilde G$ be the set
of points with $n=a=b=0$ and $|x|,|\xi|,|\theta|<\epsilon$.
Its image $\pi(U_\epsilon)$ is an open neighborhood of $0$, so
discreteness of $D$ gives $\pi(U_\epsilon)\cap D=\{0\}$ for small
$\epsilon$.
Hence $\Lambda\cap U_\epsilon\subseteq\Lambda_{\mathrm{per}}\cap U_\epsilon=\{0\}$.
Since $U_\epsilon$ is the intersection of $\widetilde G$ with an open subset
of $\R^{2(p+q+k)}$, the group $\Lambda$ is discrete.
Suppose that a nonzero linear form $f$ on $\R^{2(p+q+k)}$ vanishes on
$\Lambda$.
Then $f$ vanishes on $\Lambda_{\mathrm{per}}$, hence on the $\theta$, $a$,
and $b$ directions.
It therefore descends to a continuous homomorphism $\bar f\colon G\to\R$
that vanishes on $D$.
Compactness of $G/D$ makes $\bar f(G)$ a compact subgroup of $\R$,
so $\bar f=0$.
Then $f$ vanishes on $\widetilde G$, which spans $\R^{2(p+q+k)}$, a
contradiction.
Thus $\Lambda$ is a full-rank lattice, and $\pi|_\Lambda$ maps onto $D$
with kernel $\Lambda_{\mathrm{per}}$.
Conversely, let $\Lambda'=T'\Z^{2(p+q+k)}$ be a full-rank lattice with
$\Lambda_{\mathrm{per}}\subseteq\Lambda'\subseteq\widetilde G$.
Then $\pi^{-1}(\pi(\Lambda'))=\Lambda'+\Lambda_{\mathrm{per}}=\Lambda'$.
If $\Lambda'\cap U_\epsilon=\{0\}$, then
$\pi(\Lambda')\cap\pi(U_\epsilon)=\{0\}$ (since
$\Lambda_{\mathrm{per}}\subseteq\Lambda'$), so $\pi(\Lambda')$ is discrete.
It is also cocompact, since $G=\pi(\mathcal B')+\pi(\Lambda')$ for the
relatively compact set
$\mathcal B':=T'[0,1)^{2(p+q+k)}\cap\widetilde G$.

(ii) For $v=Tm$ and $w=Tm'$ with $m,m'\in\Z^{2(p+q+k)}$, we have
$\rho(\pi(v),\pi(w))=e^{2\pi i\,m^{\intercal}\varTheta m'}$.
Since $\pi(\Lambda)=D$, isotropy of $D$ is equivalent to
$m^{\intercal}\varTheta m'\in\Z$ for all $m,m'\in\Z^{2(p+q+k)}$.
Taking $m$ and $m'$ to be standard basis vectors shows that this holds
exactly when $\varTheta$ is an integer matrix.

(iii) The set $\mathcal B_T:=T[0,1)^{2(p+q+k)}\cap\widetilde G$ is a
fundamental domain for $\Lambda$ in $\widetilde G$.
If $\pi(v)-\pi(w)\in D$ for $v,w\in\mathcal B_T$, then $v-w\in\Lambda$, so
$v=w$.
Thus $\pi$ maps $\mathcal B_T$ bijectively onto a fundamental domain for
$D$ in $G$.
The self-Plancherel measure $\mu$ is Lebesgue measure on $\R^{2p}$,
counting measure times probability Haar measure on
$\Z^q\times\T^q$, and mass $C_{\dia c}^{-1}$ on each point of
$\Z_{\dia c}\times\widehat{\Z_{\dia c}}$.
For $\delta\in\widetilde G_{\mathrm d}$, put
$\mathcal B_\delta:=\{u\in\widetilde G_{\mathrm c}:(u,\delta)\in\mathcal B_T\}$.
The restriction of $\pi$ to $\mathcal B_T$ is injective, and it acts on
the continuous coordinates of each sheet
$\widetilde G_{\mathrm c}\times\{\delta\}$ only by reducing $\theta$
modulo $\Z^q$.
Eq.~\eqref{eq:covolume_definition} therefore gives
\begin{equation}
 |G/D|_\mu=\mu(\pi(\mathcal B_T))
 =C_{\dia c}^{-1}\sum_{\delta\in\widetilde G_{\mathrm d}}
 \operatorname{vol}(\mathcal B_\delta).
 \label{eq:full_lift_covolume_sheets}
\end{equation}
Adding $[0,1)$ to every discrete coordinate turns $\mathcal B_T$ into a
fundamental domain for $\Lambda$ in $\R^{2(p+q+k)}$, of Lebesgue measure
$\sum_{\delta}\operatorname{vol}(\mathcal B_\delta)$.
Comparing with $T[0,1)^{2(p+q+k)}$ yields
$|\det T|=C_{\dia c}|G/D|_\mu$.
Moreover, $\operatorname{Pf}\varTheta=\det(T)\operatorname{Pf}J$ and
$|\operatorname{Pf}J|=\det\dia c^{-1}=C_{\dia c}^{-1}$.
For isotropic $D$, Theorem~\ref{thm:code_dimension} identifies
$|G/D|_\mu$ with $K$.

(iv) Pairing $v=(x,\xi\,|\,n,\theta\,|\,a,b)\in\R^{2(p+q+k)}$ with the
periods $e_{\theta_i}$, $c_je_{a_j}$, and $c_je_{b_j}$ gives $n_i$,
$-b_j$, and $a_j$, respectively.
Hence $\Lambda_{\mathrm{per}}^\perp=\widetilde G$, and
$\Lambda\supseteq\Lambda_{\mathrm{per}}$ implies
$\Lambda^\perp\subseteq\widetilde G$.
For $v\in\widetilde G$, the formula for $\rho$ shows that
$v\in\Lambda^\perp$ if and only if $\pi(v)\in D^\perp$.
Solving $T^{\intercal}Jv\in\Z^{2(p+q+k)}$ gives
$\Lambda^\perp=J^{-1}T^{-\intercal}\Z^{2(p+q+k)}
=T\varTheta^{-1}\Z^{2(p+q+k)}$.
If $D$ is isotropic, then $\pi$ maps $\Lambda^\perp$ onto $D^\perp$ with
$\pi^{-1}(D)=\Lambda$, so $D^\perp/D\cong\Lambda^\perp/\Lambda$.
Since $T_{\mathrm{dual}}\ell\in\Lambda$ exactly when
$\ell\in\varTheta\Z^{2(p+q+k)}$, this quotient is
$\Z^{2(p+q+k)}/\varTheta\Z^{2(p+q+k)}$.
Its order $\det\varTheta=(\det T)^2\det J=K^2$ agrees with
Eq.~\eqref{eq:logical_quotient_order}.
The phase formula follows from
$T_{\mathrm{dual}}^{\intercal}JT_{\mathrm{dual}}=-\varTheta^{-1}$.
Finally, the logical matrix of $T_{\mathrm{dual}}$ is
$T_{\mathrm{dual}}(-\varTheta^{-1})^{-1}=-T$.
Hence $(\Lambda^\perp)^\perp=\Lambda$, consistent with
Eq.~\eqref{eq:symplectic_double_annihilator}.
\end{proof}

\subsection{Proof of Theorem~\ref{thm:logical_displacement_algebra}}
\label{appx:proof_logical_displacement_algebra}

\begin{proof}
Fix the stabilizer phase function $c$ used to define the code fiber $E_{D,1}$
and the logical displacements $\mathbf L_z$.
The order formula $|D^\perp/D|=K^2$ is
Equation~\eqref{eq:logical_quotient_order}.  The multiplication law for the
displacement operators gives
\[
 \mathbf L_z\mathbf L_{z'}
 =\beta(z,z')\mathbf L_{z+z'},
\]
and exchanging $z$ and $z'$ gives
Equation~\eqref{eq:logical_weyl_relation}.

If $w\in D$, then the phase-corrected stabilizer
$\widetilde{\mathbf D}(w)=c(w)\mathbf D(w)$ acts as the identity on
$E_{D,1}$.  Therefore
\begin{equation}
 \mathbf L_w=c(w)^{-1}\mathbf 1_{E_{D,1}}.
 \label{eq:lattice_displacement_scalar_on_code}
\end{equation}
For $z\in D^\perp$ and $w\in D$, Equation~\eqref{eq:logical_projective_product}
then implies
\begin{equation}
 \mathbf L_{z+w}
 =\beta(z,w)^{-1}c(w)^{-1}\mathbf L_z.
 \label{eq:logical_representative_phase}
\end{equation}
Thus two representatives of the same coset give logical operators that
differ only by an overall phase, proving that
\eqref{eq:logical_projective_representation} is well defined.

The value of $\rho(z,z')$ is unchanged if either argument is shifted by an
element of $D$, because both $z$ and $z'$ lie in $D^\perp$.  Hence the
descended bicharacter $\rho_D$ in
\eqref{eq:logical_symplectic_bicharacter} is well defined and independent of
$c$.  Its radical is
trivial: if $z\in D^\perp$ satisfies $\rho(z,z')=1$ for every
$z'\in D^\perp$, then
\[
 z\in(D^\perp)^\perp=D,
\]
where the last equality is the symplectic double-annihilator identity
\eqref{eq:symplectic_double_annihilator}.  This proves nondegeneracy.  It also proves faithfulness of the
projective representation.  Indeed, if $\mathbf L_z$ were scalar, then
Equation~\eqref{eq:logical_weyl_relation} would force
$\rho(z,z')=1$ for all $z'\in D^\perp$, and hence $z\in D$.

It remains to identify the algebra generated by these operators.  Let
$z\in D^\perp\setminus D$.  Nondegeneracy provides $y\in D^\perp$ such that
$\rho(y,z)\neq1$.  Conjugating $\mathbf L_z$ by the unitary $\mathbf L_y$
and using Equation~\eqref{eq:logical_weyl_relation} gives
\[
 \operatorname{Tr}(\mathbf L_z)
 =\operatorname{Tr}(\mathbf L_y\mathbf L_z\mathbf L_y^{-1})
 =\rho(y,z)\operatorname{Tr}(\mathbf L_z),
\]
so $\operatorname{Tr}(\mathbf L_z)=0$.  It follows that logical
displacements belonging to distinct cosets are Hilbert--Schmidt orthogonal:
\begin{equation}
 \operatorname{Tr}(\mathbf L_z^\dagger\mathbf L_{z'})=0
 \qquad(z+D\neq z'+D).
 \label{eq:logical_displacement_orthogonality}
\end{equation}
There are $|D^\perp/D|=K^2$ such operators, while
$\dim_{\mathbb C}\operatorname{End}(E_{D,1})=K^2$.  They therefore form an
orthogonal basis of the full logical operator algebra.
\end{proof}

\subsection{Proof of Theorem~\ref{thm:corrected_clifford_implementer}}
\label{appx:proof_corrected_clifford_implementer}

\begin{proof}
Choose $\mathcal U_S$ and $a_S$ using
Proposition~\ref{prop:symplectic_implementer}.
For $w\in D$, define
\begin{equation}
 \eta_S(w)
 :=c(S^{-1}w)a_S(S^{-1}w)c(w)^{-1}.
 \label{eq:clifford_syndrome_character}
\end{equation}
Equations~\eqref{eq:tilde_displacement} and
\eqref{eq:clifford_implementer} give
\begin{equation}
 \mathcal U_S\widetilde{\mathbf D}(S^{-1}w)\mathcal U_S^\dagger
 =\eta_S(w)\widetilde{\mathbf D}(w).
 \label{eq:clifford_stabilizer_phase_mismatch}
\end{equation}
Both sides, without the scalar $\eta_S(w)$, are genuine unitary
representations of $D$.  It follows that $\eta_S\in\widehat D$ is a
character.

If $u\in\mathcal C$, Equation~\eqref{eq:clifford_stabilizer_phase_mismatch}
implies
\[
 \widehat{\widetilde{\mathbf D}(w)}\,\widehat{\mathcal U_S}u
 =\overline{\eta_S(w)}\,\widehat{\mathcal U_S}u.
\]
Therefore $\mathcal U_S$ sends the trivial-syndrome fiber to the fiber with
syndrome $\eta_S$.  The map $\kappa:G\to\widehat D$ in
Equation~\eqref{eq:kappa_exact_sequence} is surjective, so we may choose
$z_S\in G$ such that
\begin{equation}
 \kappa_{z_S}=\eta_S^{-1}.
 \label{eq:clifford_displacement_correction}
\end{equation}
By Proposition~\ref{prop:rieffel_fiber_covariance}, displacement by $z_S$
then sends the $\eta_S$-fiber back to the trivial fiber.  Hence
$\mathcal U_{S,D}=\mathbf D(z_S)\mathcal U_S$ preserves $\mathcal C$.

Now let $x\in D^\perp$.  Proposition~\ref{prop:lattice_symmetry_logical_phase_space}
gives $Sx\in D^\perp$.  Conjugating $\mathbf D(x)$ first by
$\mathcal U_S$ and then by $\mathbf D(z_S)$ gives, up to a scalar phase,
$\mathbf D(Sx)$.  Restriction to the code fiber therefore yields
Equation~\eqref{eq:logical_clifford_conjugation}, and hence the action
\eqref{eq:logical_clifford_label_action}.  Since the logical displacements
form the logical Weyl system, this proves that
$\mathcal U_{S,D}|_{\mathcal C}$ is a logical Clifford gate.

Finally, if $z'_S$ is another choice satisfying
\eqref{eq:clifford_displacement_correction}, then
$\kappa_{z'_S-z_S}=1$.  Thus $z'_S-z_S\in\ker\kappa=D^\perp$, and the two
corrected implementers differ, up to phase, by the corresponding logical
displacement.
\end{proof}

%
%
\subsection{Normal form for one oscillator and one rotor}
\label{appx:proof_oscillator_rotor_normal_form}

This subsection and the next two prove the results on symplectic
decomposability for one oscillator and one rotor used in
Sec.~\ref{ex:osc-rotor}.
As in that section, we write points of the phase space
$G\cong\R^2\oplus(\Z\times\T)$ as $(x,n;\xi,\theta)$, listing both
configuration coordinates before both momentum coordinates.
The commutation phase is $\rho(z,z')=e^{2\pi i\omega(z,z')}$, where
$\omega(z,z')=x\xi'-x'\xi+n\theta'-n'\theta$ in $\T$.
For oscillator coordinates $u=(x,\xi)$, we write
$\sigma(u,u')=x\xi'-x'\xi\in\R$.
Symplectic decomposability refers to
Eq.~\eqref{eq:oscillator-rotor-decoupling}.
Here $K=|G/D|_\mu$ denotes the ideal code dimension of
Theorem~\ref{thm:code_dimension}.

\begin{lemma}
\label{lem:oscillator-rotor-normal-form}
Let $D\subset G$ be an isotropic lattice, and let
$p(x,n;\xi,\theta)=n$.  Define $r,t\geq1$ by
\begin{equation}
 p(D)=r\Z,\qquad t=|\operatorname{Tor}(D)|.
 \label{eq:oscillator-rotor-invariants}
\end{equation}
Then $t\mid r$.  Put $m=r/t$.  After a continuous symplectic
automorphism, $D$ has the form
\begin{equation}
 \begin{split}
 D=\bigl\{(\lambda_x,kr;\lambda_\xi,
       \widetilde\gamma(\lambda)/r+j/t):\ &\lambda\in\Lambda_{\mathrm{osc}},\\[-2pt]
       &k,j\in\Z\bigr\},
 \end{split}
 \label{eq:oscillator-rotor-normal-form}
\end{equation}
where $\Lambda_{\mathrm{osc}}\subset\R^2$ is a rank-two lattice with
$\sigma(\Lambda_{\mathrm{osc}},\Lambda_{\mathrm{osc}})\subset\Z$, and
$\gamma:\Lambda_{\mathrm{osc}}\to\Z_m$ is a homomorphism.
Here $\widetilde\gamma(\lambda)$ is any integer lift of
$\gamma(\lambda)$; changing it does not change the set in
Eq.~\eqref{eq:oscillator-rotor-normal-form}.
For a positively oriented basis $\lambda_1,\lambda_2$ of $\Lambda_{\mathrm{osc}}$,
\begin{equation}
 N:=\sigma(\lambda_1,\lambda_2)=|\R^2/\Lambda_{\mathrm{osc}}|\in\Z_{>0},
 \qquad K=mN.
 \label{eq:oscillator-rotor-dimension-factorization}
\end{equation}
\end{lemma}

Here $\operatorname{Tor}(D)$ denotes the torsion subgroup of $D$, consisting
of its finite-order elements; see Lemma~\ref{lem:fe_torsion}.

\begin{proof}
The map $G/D\to\Z/p(D)$ is surjective.  Compactness therefore gives
$p(D)=r\Z$ for some $r\geq1$.  The intersection of $D$ with the
compact circle subgroup is finite and equals its torsion subgroup, so
\[
 \operatorname{Tor}(D)=\langle q\rangle,
 \qquad q=(0,0;0,1/t).
\]
Choose $v=(x_0,r;\xi_0,\theta_0)\in D$.  Isotropy with $q$ implies
$r/t\in\Z$, proving $t\mid r$.

For $a=(a_x,a_\xi)\in\R^2$ and $c\in\T$, define
\begin{equation}
 \begin{split}
 S_{a,c}(x,n;\xi,\theta)
 :=\bigl(x+na_x,n;\xi+na_\xi,\qquad\\[-2pt]
       \theta-a_x\xi+a_\xi x+cn\bigr).
 \end{split}
 \label{eq:oscillator-rotor-shear}
\end{equation}
Direct substitution shows that $S_{a,c}$ preserves $\omega$; its inverse
is $S_{-a,-c}$.  It also preserves the product Haar measure: for each
fixed $n$, it translates the real coordinates and then the circle
coordinate.  Taking $a=-(x_0,\xi_0)/r$ and choosing $c$ with
$rc=-\theta_0$ sends $v$ to $(0,r;0,0)$ and fixes $q$.
We henceforth work with this transformed lattice.

Set $H=D\cap\{n=0\}$.  Then
$D=H+\Z(0,r;0,0)$.  Isotropy with $(0,r;0,0)$ forces the angle
coordinate of each element of $H$ into $r^{-1}\Z/\Z$.
Let $\Lambda_{\mathrm{osc}}$ be the oscillator projection of $H$.
Its kernel is $\langle q\rangle$.  Since the allowed angle coordinates
form a finite set, discreteness of $D$ implies discreteness of
$\Lambda_{\mathrm{osc}}$.  The oscillator projection of all of $D$ is also
$\Lambda_{\mathrm{osc}}$, so $\R^2/\Lambda_{\mathrm{osc}}$ is a continuous image of the compact
space $G/D$.  Thus $\Lambda_{\mathrm{osc}}$ is a rank-two lattice.

For a lift $(\lambda_x,0;\lambda_\xi,\theta_\lambda)\in H$, choose
a real representative of $\theta_\lambda$ and set
\[
 \gamma(\lambda)=r\theta_\lambda\pmod m.
\]
This is well defined: changing the lift by $q$ changes the integer
$r\theta_\lambda$ by $m$, and changing the circle representative
changes it by $r$.  It defines a homomorphism and gives
Eq.~\eqref{eq:oscillator-rotor-normal-form}.
Pairing two elements of $H$ shows that
$\sigma(\Lambda_{\mathrm{osc}},\Lambda_{\mathrm{osc}})\subset\Z$, so the oriented area $N$ is a
positive integer.

A fundamental domain in the reordered coordinates $(u;n,\theta)$ is
\[
 F_{\Lambda_{\mathrm{osc}}}\times\{0,\ldots,r-1\}\times[0,1/t),
\]
where $F_{\Lambda_{\mathrm{osc}}}$ is a half-open fundamental parallelogram of $\Lambda_{\mathrm{osc}}$.
Indeed, first reduce $u$ using a lift in $H$, then reduce $n$ and
$\theta$ using the two pure rotor generators.  Consequently
$|G/D|_\mu=Nr/t=Nm$.
Equation~\eqref{eq:rank_covolume_formula} proves the dimension formula.
\end{proof}

\subsection{Coprime splitting for one oscillator and one rotor}
\label{appx:proof_oscillator_rotor_coprime_splitting}

\begin{lemma}
\label{lem:oscillator-rotor-coprime-splitting}
In Lemma~\ref{lem:oscillator-rotor-normal-form}, if
$\gcd(N,m)=1$, then $D$ is symplectically decomposable.
\end{lemma}

\begin{proof}
If $m=1$, the residual homomorphism $\gamma$ is zero and the normal
form already splits.  Otherwise, reduce the integral oscillator pairing
modulo $m$ to obtain
\begin{equation}
 \begin{split}
 B_m:\Lambda_{\mathrm{osc}}/m\Lambda_{\mathrm{osc}}&\longrightarrow
       \operatorname{Hom}(\Lambda_{\mathrm{osc}},\Z_m),\\
 B_m(\nu)(\lambda)&=\sigma(\nu,\lambda)\pmod m.
 \end{split}
 \label{eq:oscillator-rotor-reduced-pairing}
\end{equation}
In the oriented basis of the preceding lemma, its matrix is
\[
 \begin{pmatrix}0&-N\\N&0\end{pmatrix}\quad\text{over }\Z_m.
\]
Its determinant $N^2$ is a unit modulo $m$ when $\gcd(N,m)=1$.
Choose $\nu\in\Lambda_{\mathrm{osc}}$ with $B_m(\nu)=\gamma$ and apply
$S_{\nu/r,0}$.  A kernel lift over $\lambda$ acquires angle
\[
 \frac{\widetilde\gamma(\lambda)-\sigma(\nu,\lambda)}{r}
 \in\frac1t\Z\pmod1.
\]
Subtracting a torsion element therefore makes it a pure oscillator
element.  The generator $(0,r;0,0)$ becomes
$(\nu_x,r;\nu_\xi,0)$; subtracting the pure oscillator element
over $\nu$ makes it pure rotor again.  Hence
\begin{equation}
 S_{\nu/r,0}(D)
 =\Lambda_{\mathrm{osc}}\oplus\langle(r,0),(0,1/t)\rangle.
 \label{eq:oscillator-rotor-coprime-product}
\end{equation}
\end{proof}

\subsection{Proof of Proposition~\ref{prop:oscillator-rotor-nondecoupling}}
\label{appx:proof_oscillator_rotor_nondecoupling}

We use the lattice $D_{\alpha,b}$ and the logical label $z_X$ of
Sec.~\ref{ex:osc-rotor}, in the coordinates $(x,n;\xi,\theta)$ of
Appendix~\ref{appx:proof_oscillator_rotor_normal_form}.

\begin{proof}
The identity-component quotient is $G/G^\circ\cong\Z$.  Every
continuous automorphism induces multiplication by $\pm1$ on this
quotient.  It therefore preserves the component spacing $r=2b$.
It also preserves the torsion order $t=b$.

Suppose $S(D_{\alpha,b})=D_{\mathrm{osc}}\oplus D_{\mathrm{rot}}$.
The oscillator factor is torsion-free and lies in $G^\circ$, so
$D_{\mathrm{rot}}$ has component projection $2b\Z$ and torsion order
$b$.  A rotor lattice $E$ with projection $r\Z$ and torsion order
$t$ has generators $(r,\varphi)$ and $(0,1/t)$ for some
$\varphi\in\T$.  Reducing the integer coordinate and then the angle
gives a fundamental domain
$\{0,\ldots,r-1\}\times[0,1/t)$, and hence
\[
 |(\Z\times\T)/E|=r/t.
\]
Thus $K_{\mathrm{rot}}=2$ by Theorem~\ref{thm:code_dimension}.

Symplecticity implies $S(D^\perp)=S(D)^\perp$, by the same argument
as in Proposition~\ref{prop:lattice_symmetry_logical_phase_space}.
Hence $S$ preserves the logical quotient and its order.  By
Eq.~\eqref{eq:logical_quotient_order}, the transformed code still has
dimension four.  The quotient for a direct sum is the product of the
factor quotients, so $K_{\mathrm{osc}}K_{\mathrm{rot}}=4$ and
$K_{\mathrm{osc}}=2$.

Theorem~\ref{thm:logical_displacement_algebra} equips each factor's
logical group of order four with a nondegenerate alternating
bicharacter.  Such a group cannot be cyclic: if $g$ generates
$\mathsf L_E$ for either factor lattice $E$, alternation and bilinearity
give $\rho_E(ag,cg)=\rho_E(g,g)^{ac}=1$ for all $a,c\in\Z$.
Each factor's logical group must therefore be $\Z_2\oplus\Z_2$,
and a decomposed code would have
\[
 \mathsf L_{D_{\mathrm{osc}}\oplus D_{\mathrm{rot}}}
 \cong(\Z_2)^4.
\]
Every element of this group has order at most two, contradicting the
order-four class $[z_X]$.  The lower bound on the dimension follows
from Theorem~\ref{thm:squarefree-oscillator-rotor-decoupling}.
\end{proof}

\subsection{Proofs for the rotor--qudit code}
\label{appx:proof_rotor_qudit}

We now prove Proposition~\ref{prop:rotor-qudit-nondecoupling} and
Corollary~\ref{cor:rotor-qudit-state}.
Write $G=G_R\oplus G_Q$ with $G_R=\Z\times\T$ and $G_Q=\Z_d\times\Z_d$.
We write its elements as $(n,\theta\,|\,a,b)$, as in
Sec.~\ref{ex:rotor-qudit}.
As in Sec.~\ref{subsec:examples_coupling}, an isotropic lattice $D\subset G$ is
\emph{symplectically decomposable} if some $S\in\operatorname{Sp}(G,\rho)$
satisfies $S(D)=D_{\mathrm{rot}}\oplus D_{\mathrm{fin}}$ with
$D_{\mathrm{rot}}\subset G_R$ and $D_{\mathrm{fin}}\subset G_Q$.
The proofs rest on two lemmas.
The first one constrains how a continuous automorphism of $G$ can mix
the two factors.

\begin{lemma}\label{lem:rotor-qudit-block-structure}
Let $S$ be a continuous automorphism of $G$, written in blocks as
\[
 S(g_R\,|\,g_Q)=\bigl(A_{RR}g_R+A_{RQ}g_Q\,\big|\,A_{QR}g_R+A_{QQ}g_Q\bigr),
 \qquad g_R\in G_R,\quad g_Q\in G_Q,
\]
with continuous homomorphisms $A_{XY}\colon G_Y\to G_X$.  Then:
\begin{enumerate}
\item[(a)] $A_{QR}(n,\theta)=n\,v$ with $v:=A_{QR}(1,0)\in G_Q$.
In particular, $A_{QR}$ vanishes on the circle $\{0\}\times\T$, and
$A_{QR}(n,\theta)=0$ whenever $d\mid n$.
\item[(b)] $A_{RQ}(G_Q)\subseteq\{0\}\times(\tfrac1d\Z/\Z)$.
\item[(c)] $A_{RR}(\{0\}\times\T)\subseteq\{0\}\times\T$, and the
angular-momentum component of $S(n,\theta\,|\,a,b)$ is $\pm n$.
\item[(d)] $A_{QQ}$ is an automorphism of $G_Q$.
\end{enumerate}
\end{lemma}

\begin{proof}
(a) The restriction of $A_{QR}$ to the connected group $\{0\}\times\T$ is
a continuous homomorphism into the discrete group $G_Q$, so it is trivial.
Hence $A_{QR}(n,\theta)=A_{QR}(n,0)=nv$, and $nv=0$ when $d\mid n$ because
$G_Q$ has exponent $d$.
(b) Every $g_Q\in G_Q$ satisfies $dg_Q=0$, so $A_{RQ}(g_Q)$ has order
dividing $d$.
A finite-order element $(n,\theta)$ of $G_R$ has $n=0$, and then
$d\theta=0$.
(c) The image of $\{0\}\times\T$ under $A_{RR}$ is connected and contains
$0$, so it lies in the identity component $\{0\}\times\T$ of $G_R$.  Write
$A_{RR}(1,0)=(f,\varphi)$.  By (b) and the first part of (c), the
angular-momentum component of $S(n,\theta\,|\,a,b)$ is $fn$.  Since $S$ is
surjective, $f\Z=\Z$, so $f=\pm1$.
(d) Let $g_Q'\in G_Q$ and choose $(n,\theta\,|\,a,b)$ with
$S(n,\theta\,|\,a,b)=(0,0\,|\,g_Q')$.  By (c), $\pm n=0$, so $n=0$, and
then (a) gives $g_Q'=A_{QQ}(a,b)$.  Thus $A_{QQ}$ is surjective, hence
bijective because $G_Q$ is finite.
\end{proof}

The second lemma turns the torsion of $D$ into an obstruction.

\begin{lemma}\label{lem:rotor-qudit-obstruction}
Let $D\subset G$ be an isotropic lattice with the following two
properties.
(i) Every $g\in D$ with $dg=0$ lies in the circle
$\{0\}\times\T\times\{0\}$.
(ii) $D$ contains an element $g_0=(n_0,\theta_0\,|\,a_0,b_0)$ with
$d\mid n_0$ and $(a_0,b_0)\neq0$.
Then no continuous automorphism $S$ of $G$ satisfies
$S(D)=D_{\mathrm{rot}}\oplus D_{\mathrm{fin}}$ with subgroups
$D_{\mathrm{rot}}\subseteq G_R$ and $D_{\mathrm{fin}}\subseteq G_Q$.  In
particular, $D$ is not symplectically decomposable.
\end{lemma}

\begin{proof}
Suppose $S(D)=D_{\mathrm{rot}}\oplus D_{\mathrm{fin}}$.  Every element of
$\{0\}\oplus D_{\mathrm{fin}}\subseteq S(D)$ is annihilated by $d$.  Since
$S$ is an injective homomorphism, the elements of $S(D)$ annihilated by
$d$ are the images of the elements of $D$ annihilated by $d$, which lie
in the circle by (i).  By Lemma~\ref{lem:rotor-qudit-block-structure}(a),
$S(0,\theta\,|\,0,0)=(A_{RR}(0,\theta)\,|\,0,0)$ has zero qudit component.
Hence $D_{\mathrm{fin}}=0$, and every element of $S(D)$ has zero qudit
component.  But the qudit component of $S(g_0)$ is
$A_{QR}(n_0,\theta_0)+A_{QQ}(a_0,b_0)=A_{QQ}(a_0,b_0)$ by
Lemma~\ref{lem:rotor-qudit-block-structure}(a), and it is nonzero by
Lemma~\ref{lem:rotor-qudit-block-structure}(d).  This is a contradiction.
\end{proof}

\begin{proof}[Proof of Proposition~\ref{prop:rotor-qudit-nondecoupling}]
Let $w_1=(d,0\,|\,1,0)$ and $w_2=(0,\tfrac1d\,|\,0,0)$ be the generators of
$D$ in Sec.~\ref{ex:rotor-qudit}.
Every element of $D$ is $kw_1+lw_2=(kd,\tfrac ld\,|\,k,0)$ with $k\in\Z$
and $l\in\Z_d$.  If $d(kw_1+lw_2)=0$, the angular-momentum component
$kd^2$ vanishes, so $k=0$ and the element lies in the circle.
This is property~(i) of Lemma~\ref{lem:rotor-qudit-obstruction}.
The generator $w_1$ has property~(ii), because $(1,0)\neq0$ in
$\Z_d\times\Z_d$ for $d\geq2$.
Lemma~\ref{lem:rotor-qudit-obstruction} therefore applies.
\end{proof}

\begin{proof}[Proof of Corollary~\ref{cor:rotor-qudit-state}]
Let $D'=\langle w_1,z\rangle$ with $z=(0,-\tfrac1{d^2}\,|\,0,1)$, the label
of the logical operator $\zl$.
Since $z$ is a logical label, $D'$ is isotropic.  The Heisenberg cocycle
$\beta$ is trivial on the generators, since
$\beta(w_1,z)=e^{2\pi i(-d/d^2+1/d)}=1$, so $c=1$.  Every element of $D'$
has the form $kw_1+lz=(kd,-\tfrac{l}{d^2}\,|\,k,l)$, so $z$ has order
$d^2$ and $dz=-w_2$.  A fundamental domain is
$\{0,\ldots,d-1\}\times[0,\tfrac1{d^2})\times\Z_d\times\Z_d$, of measure
$d\cdot\frac1{d^2}\cdot d^2\cdot\frac1d=1$, so $K=1$.
Equivalently, $T'=(w_1,z,de_a,de_b)$ has
$|\det T'|/d=1=|\operatorname{Pf}(T'^{\intercal}JT')|$, where
$de_a=(0,0\,|\,d,0)$ and $de_b=(0,0\,|\,0,d)$ are the qudit periods.
If $d(kw_1+lz)=0$, then $k=0$ and $d\mid l$, so the element lies in
$\langle dz\rangle=\langle w_2\rangle$, inside the circle.  The generator
$w_1$ has property~(ii) of Lemma~\ref{lem:rotor-qudit-obstruction}, so
$D'$ is not symplectically decomposable.

For the state, let
$\operatorname{Stab}(u)=\{g\in G:\mathbf D(g)u\in\C u\}$
for a nonzero generalized state $u$.  If
$\mathbf D(g)\ket{\psi_d}=\lambda\ket{\psi_d}$, commuting $\mathbf D(g)$
past $\mathbf D(w)$ for $w\in D'$ gives $\rho(w,g)=1$.  Hence
$D'\subseteq\operatorname{Stab}(\ket{\psi_d})\subseteq D'^{\perp}=D'$,
where the last equality holds because $[D'^\perp:D']=K^2=1$.
A generalized Clifford operator $\mathcal U$ for
$S\in\operatorname{Sp}(G,\rho)$, possibly followed by a displacement,
satisfies $\operatorname{Stab}(\mathcal Uu)=S(\operatorname{Stab}(u))$ by
Eq.~\eqref{eq:clifford_implementer}.  A product of
nonzero factors has
$\operatorname{Stab}(u_{\mathrm{rot}}\otimes u_{\mathrm{fin}})
=\operatorname{Stab}(u_{\mathrm{rot}})\oplus\operatorname{Stab}(u_{\mathrm{fin}})$,
because $u_1\otimes u_2\in\C(u_1'\otimes u_2')$ forces $u_1\in\C u_1'$ and
$u_2\in\C u_2'$.  A Clifford map from $\ket{\psi_d}$ to a product would
therefore split $S(D')$, contradicting
Lemma~\ref{lem:rotor-qudit-obstruction}.
\end{proof}

\begin{remark}
\label{rem:rotor-qudit-symplectic-maps}
The obstruction is arithmetic rather than a lack of mixing symplectic
maps.  Every $S\in\operatorname{Sp}(G,\rho)$ has the form
\[
 S(n,\theta\,|\,a,b)=\bigl(\epsilon n,\;
 \varphi n+\epsilon\theta-\epsilon\,\sigma_Q\bigl(v,A_{QQ}(a,b)\bigr)\;\big|\;
 nv+A_{QQ}(a,b)\bigr),
\]
with $\varphi\in\T$, $v\in\Z_d^2$, $A_{QQ}\in\mathrm{SL}(2,\Z_d)$, and a
single sign $\epsilon=\pm1$ that multiplies both $n$ and $\theta$.
Here $\sigma_Q((a,b),(a',b'))=(ab'-a'b)/d$ is the qudit part of the
exponent of $\rho$.
For instance,
$\ket n_{\mathrm{rot}}\ket j_{\mathrm{fin}}\mapsto
\ket n_{\mathrm{rot}}\ket{j+n}_{\mathrm{fin}}$ implements the case
$\epsilon=1$, $\varphi=0$, $A_{QQ}=I$, and $v=(1,0)$.
These shears are the rotor--qudit
counterparts of the oscillator--rotor shears $S_{a,c}$ of
Eq.~\eqref{eq:oscillator-rotor-shear}.
Consider
$D_{r,t}=\langle(r,0\,|\,1,0),(0,\tfrac1t\,|\,0,0)\rangle$ with $t\mid r$,
which has $K=rd/t$.
When $\gcd(r,d)=1$, the shear with $v=(-r^{-1}\bmod d,0)$ maps $D_{r,t}$ to
$\langle(r,0),(0,\tfrac1t)\rangle\oplus\{0\}$.
When $h:=\gcd(r,d)>1$, Lemma~\ref{lem:rotor-qudit-obstruction} applies with
$g_0=\tfrac dh(r,0\,|\,1,0)$.
The code of Sec.~\ref{ex:rotor-qudit} is the case $r=t=d$.
For $d=2$, the parametrization above coincides, up to conventions, with
the one-rotor, one-qubit case of the hybrid phase-space classification of
Luo et al.~\cite{luo2026hybridqubitrotor}.
Lemma~\ref{lem:rotor-qudit-block-structure} and this remark generalize
that case to qudits of any dimension $d$.
\end{remark}

\subsection{Proof of Proposition~\ref{prop:hybrid-intrinsic-nondecoupling}}
\label{appx:proof_hybrid_intrinsic}
Write points of $G$ as $z=(x,n,j;\xi,\theta,b)$, as in
Sec.~\ref{ex:torsion}, and let $f=(j,b)\in\Z_2^2$ collect the qubit
labels.
Let $S\in\operatorname{Sp}(G,\rho)$ and $D:=D_{\mathrm{int}}$.

\textit{Step 1: action on discrete labels.}
The identity component of $G$ is
$G_0=\{(x,0,0;\xi,\theta,0)\}\cong\R^2\times\T$, and
$G/G_0\cong\Z\times\Z_2^2$ through the labels $(n,f)$.
Being a continuous automorphism, $S$ preserves $G_0$ and induces an
automorphism of $\Z\times\Z_2^2$.
This automorphism preserves the torsion subgroup $\{0\}\times\Z_2^2$, and
every homomorphism from $\Z_2^2$ to $\Z$ vanishes.
Hence $S$ acts on these labels as
\begin{equation}
 (n,f)\mapsto(\epsilon n,\,Ff+nh),
 \qquad \epsilon=\pm1,\quad F\in\mathrm{GL}(2,\Z_2),\quad h\in\Z_2^2.
 \label{eq:hybrid-intrinsic-discrete-action}
\end{equation}
Since $S$ preserves commutation phases, it also maps $D^\perp$ onto
$S(D)^\perp$.
Hence $S(D)^\perp/S(D)\cong D^\perp/D\cong\Z_4^2$.

\textit{Step 2: stabilizers without rotor shift.}
Let $H:=D\cap\{n=0\}$.
An element $m_1w_1+m_2w_2+m_3w_3$ has rotor shift $2m_3$, so
$H=\langle w_1,w_2\rangle$.
By Eq.~\eqref{eq:hybrid-intrinsic-discrete-action}, $S$ preserves the
condition $n=0$, so $S(H)=S(D)\cap\{n=0\}$.
The qubit labels of $S(w_2)$ are $F(0,1)\neq0$.
Thus $S(H)$ contains an element with nonzero qubit labels.

\textit{Step 3: oscillator versus rotor and qubit.}
Suppose that $S(D)=D_1\oplus D_2$, with $D_1$ in the oscillator phase
space and $D_2$ in the rotor--qubit phase space.
Then $S(D)\cap\{n=0\}=D_1\oplus(D_2\cap\{n=0\})$.
The group $D_2\cap\{n=0\}$ is a discrete subgroup of the compact group
$\{0\}\times\T\times\Z_2^2$, hence finite.
It is trivial because $S(D)\cong D$ is torsion free by
Eq.~\eqref{eq:hybrid-intrinsic-group-structure}.
Thus $S(H)=D_1$ has no qubit labels, which contradicts Step~2.

\textit{Step 4: qubit versus oscillator and rotor.}
Suppose that $S(D)=D_1\oplus D_2$, with $D_1$ in the qubit phase space
and $D_2$ in the oscillator--rotor phase space.
Then $D_1$ is finite, hence trivial.
Every element of $S(D)$ then has zero qubit labels, which again
contradicts Step~2.

\textit{Step 5: rotor versus oscillator and qubit.}
Suppose that $S(D)=D_1\oplus D_2$, with $D_1$ in the rotor phase space
and $D_2$ in the oscillator--qubit phase space.
By Eq.~\eqref{eq:hybrid-intrinsic-discrete-action}, the rotor shifts of
$S(D)$ form $2\Z$, and they all come from $D_1$.
Since $D_1$ is torsion free, it meets $\{0\}\times\T$ trivially and is
generated by a single label $(2,\varphi)$.
The symplectic shear $(n,\theta)\mapsto(n,\theta-\varphi n/2)$ maps
$D_1$ onto $2\Z\times\{0\}$, whose logical group is $\Z_2^2$ by
Eq.~\eqref{eq:rotor-css-logical-group}.
The logical group of $S(D)$ is the product of the logical groups of
$D_1$ and $D_2$, and it has order $K^2=16$.
Hence $D_2^\perp/D_2$ has order four.
By the alternating Smith normal form described after
Lemma~\ref{lem:full_lift_pfaffian}, it is $\Z_2^2$.
The logical group of $S(D)$ is then $\Z_2^4$, which has no element of
order four.
This contradicts $S(D)^\perp/S(D)\cong\Z_4^2$ from Step~1.

A decomposition into three single-system lattices splits, in particular,
the oscillator from the rotor and the qubit, which Step~3 excludes.
This proves Proposition~\ref{prop:hybrid-intrinsic-nondecoupling}.

\subsection{Borel selections for syndrome decoders}
\label{appx:decoding_borel_section}

\begin{lemma}
\label{lem:decoding_borel_section}
Let $f:X\to Y$ be a Borel map between standard Borel spaces with at most
countable fibers.  Then $f(X)$ is Borel, and $f$ admits a Borel right
inverse on $f(X)$.
\end{lemma}
\begin{proof}
The Borel relation $\{(y,x):f(x)=y\}\subseteq Y\times X$ has countable
vertical sections.  The Lusin--Novikov uniformization theorem
\cite[Theorem~18.10]{kechris1995classical} expresses it as a countable
union of graphs of partial Borel maps with Borel domains.  Their domains
cover $f(X)$; choosing the first available value on their disjointized
domains gives the required right inverse.
\end{proof}

The elementary LCA groups considered here are second-countable, so $G$
and its quotients by closed subgroups are standard Borel spaces.  The
discrete groups $D$ and $D^\perp$ are countable.  Thus the lemma applies
to $\pi_D$, $\kappa$, and $\varphi$, as well as to their restrictions
to Borel sets.

\subsection{Hermite coefficients of Schwartz functions}
\label{appx:proof_fe_coefficients}

Put $\omega_{k,n}:=1+|k|+|n|$ for the basis~\eqref{eq:fe_basis}.

\begin{proposition}
\label{prop:fe_coefficients}
The map $\varphi\mapsto(\varphi_{k,n,w})$,
$\varphi_{k,n,w}:=\langle e_{k,n,w},\varphi\rangle$, is a topological
isomorphism of $\mathscr S(M)$ onto the space of rapidly decreasing sequences
indexed by $\mathbb Z_{\geq0}^p\times\mathbb Z^q\times\mathbb Z_{\dia c}$.  Its Fr\'echet
topology is equivalently given by the seminorms
\begin{equation}
 p_R(\varphi):=\sup_{k,n,w}\omega_{k,n}^{\,R}\,|\varphi_{k,n,w}|,
 \qquad R\in\mathbb Z_{\geq0}.
 \label{eq:fe_coefficient_seminorms}
\end{equation}
The coefficient map $u\mapsto(u_{k,n,w})$, with
$u_{k,n,w}:=\langle u,e_{k,n,w}\rangle$, identifies the continuous anti-dual
$\mathscr S'(M)$ with the sequences satisfying
\begin{equation}
 |u_{k,n,w}|\leq C\,\omega_{k,n}^{\,R_0}
 \qquad\text{for some }C>0,\ R_0\in\mathbb Z_{\geq0}.
 \label{eq:fe_distribution_growth}
\end{equation}
The corresponding functional is uniquely given by the absolutely
convergent series
$\langle u,\varphi\rangle=\sum_{k,n,w}u_{k,n,w}\overline{\varphi_{k,n,w}}$.
\end{proposition}

\begin{proof}
For a single slice $\varphi(\cdot,n,w)\in\mathscr S(\mathbb R^p)$ the
equivalence between the seminorms \eqref{eq:schwartz_seminorms} and the
weighted Hermite seminorms is \cite[Sec.~2.2, Eqs.~(12)--(13) and
Prop.~2.3]{keyl2016schwartz}; the constants there depend only on $p$ and $R$,
hence not on $(n,w)$.  Inserting the polynomial weights in $n$ and taking the
maximum over the finite set $\mathbb Z_{\dia c}$ gives \eqref{eq:fe_coefficient_seminorms}.
Conversely a rapidly decreasing sequence has a Hermite series converging in
every seminorm \eqref{eq:schwartz_seminorms}, because the $q_R$-seminorms of
$h_k$ grow at most polynomially in $|k|$.

Continuity of $u$ gives $|\langle u,\varphi\rangle|\leq Cp_{R_0}(\varphi)$ for
some $R_0$, and $\varphi=e_{k,n,w}$ yields
\eqref{eq:fe_distribution_growth}.  Conversely, if
\eqref{eq:fe_distribution_growth} holds, choose an integer $L>p+q$; then
\[
 \sum_{k,n,w}|u_{k,n,w}\varphi_{k,n,w}|
 \leq C\,p_{R_0+L}(\varphi)\sum_{k,n,w}\omega_{k,n}^{-L}<\infty,
\]
which is a continuous conjugate-linear functional.  Finite truncations of the
Hermite series converge in $\mathscr S(M)$, so the coefficients determine $u$.
\end{proof}

\subsection{Smoothing by the Gaussian regularizer}
\label{appx:proof_fe_smoothing}

\begin{proposition}
\label{prop:fe_smoothing}
For every $\Delta>0$ the operator $\mathbf E_\Delta$ is a positive injective
contraction on $L^2(M)$; it acts continuously on $\mathscr S(M)$, and its
transpose extension to $\mathscr S'(M)$, defined by
$\langle\mathbf E_\Delta u,\varphi\rangle:=\langle u,\mathbf E_\Delta\varphi\rangle$,
is injective and maps $\mathscr S'(M)$ continuously into $\mathscr S(M)$.
Moreover, \eqref{eq:fe_smoothing_convergence} holds for every
$u\in\mathscr S'(M)$.
\end{proposition}

\begin{proof}
Since $0<a_\Delta\leq1$ we get $p_R(\mathbf E_\Delta\varphi)\leq p_R(\varphi)$.
The coefficients of $\mathbf E_\Delta u$ are $a_\Delta u$, and
\eqref{eq:fe_distribution_growth} gives, for every $R$,
\[
 \sup_{k,n,w}\omega_{k,n}^{\,R}\,|a_\Delta(k,n)u_{k,n,w}|
 \leq C\sup_{k,n}\omega_{k,n}^{\,R+R_0}
   e^{-\Delta|k|-\pi\Delta|n|^2/2}<\infty,
\]
because exponential and Gaussian decay beat every polynomial.  Hence
$\mathbf E_\Delta u\in\mathscr S(M)$ by
Proposition~\ref{prop:fe_coefficients}, and injectivity holds because no
multiplier vanishes.  For continuity into $\mathscr S(M)$, note that the set
$\mathcal B_{R,\Delta}:=\{\omega_{k,n}^{\,R}a_\Delta(k,n)e_{k,n,w}\}$
is bounded in $\mathscr S(M)$, since its $p_L$-seminorms are bounded by
$\sup\omega^{R+L}a_\Delta<\infty$, and that
\[
 p_R(\mathbf E_\Delta u)
 =\sup_{\varphi\in\mathcal B_{R,\Delta}}|\langle u,\varphi\rangle|,
\]
which is a strong dual seminorm.

For the convergence, let $\mathcal B_0\subseteq\mathscr S(M)$ be bounded,
let $L>p+q$, and put
$C_0:=\sup_{\varphi\in\mathcal B_0}p_{R_0+L}(\varphi)<\infty$.  Then
\[
 \sup_{\varphi\in\mathcal B_0}
 |\langle\mathbf E_\Delta u-u,\varphi\rangle|
 \leq CC_0\sum_{k,n,w}|a_\Delta(k,n)-1|\,\omega_{k,n}^{-L}
 \longrightarrow0
\]
by dominated convergence, since $|a_\Delta-1|\leq1$ and each summand tends to
$0$.
\end{proof}

\subsection{Factors of the normalized characteristic function}
\label{appx:proof_fe_three_factors}

\begin{lemma}
\label{lem:fe_three_factors}
With the Hermite normalization of Section~\ref{subsec:fe_regularizer},
\begin{equation}
 \Xi^{\mathbb R}_\Delta(x,\xi)
 =e^{-\pi ix\cdot\xi}
  \exp\!\left[-\frac{\coth\Delta}{4}
   \bigl(|x|^2+4\pi^2|\xi|^2\bigr)\right],
 \qquad
 |\Xi^{\mathbb R}_\Delta(x,\xi)|
 \leq e^{-(|x|^2+|\xi|^2)/(4\Delta)} .
 \label{eq:fe_oscillator_factor}
\end{equation}
For the rotor factor,
\begin{equation}
 \Xi^{\mathbb Z}_\Delta(n,\theta)=0 \quad(n\neq0),
 \qquad
 \Xi^{\mathbb Z}_\Delta(0,\theta)
 =\frac{\sum_{m\in\mathbb Z^q}e^{-\pi\Delta|m|^2}e^{2\pi im\cdot\theta}}
        {\vartheta_q(\Delta)}
 =\frac{\sum_{\ell\in\mathbb Z^q}e^{-\pi|\ell-\theta|^2/\Delta}}
        {\sum_{\ell\in\mathbb Z^q}e^{-\pi|\ell|^2/\Delta}},
 \label{eq:fe_rotor_factor}
\end{equation}
and there is $C_q<\infty$ with
\begin{equation}
 |\Xi^{\mathbb Z}_\Delta(0,\theta)|
 \leq C_q\,e^{-\pi\operatorname{dist}(\theta,0)^2/(2\Delta)}
 \qquad(0<\Delta\leq1).
 \label{eq:fe_rotor_bound}
\end{equation}
Finally $\Xi^{\mathrm{fin}}(0,1)=1$ and $\Xi^{\mathrm{fin}}(w,\chi)=0$ otherwise.
\end{lemma}

\begin{proof}
Oscillator factor.  In one mode let $b=(x+\partial_x)/\sqrt2$ and
$D(\eta)=e^{\eta b^\dagger-\overline\eta b}$.  Comparing with
$\mathbf D(x,\xi)=\mathbf M_\xi\mathbf T_x$ and
$\mathbf T_x\psi(t)=\psi(t+x)$ gives
$\mathbf D(x,\xi)=e^{-\pi ix\xi}D(\eta)$ with
$\eta=(-x+2\pi i\xi)/\sqrt2$.  The Fock diagonal element is
$\langle h_j,D(\eta)h_j\rangle=e^{-|\eta|^2/2}L_j(|\eta|^2)$ with $L_j$
the Laguerre polynomial, and for $\varrho=e^{-2\Delta}$ the generating
function $\sum_j\varrho^jL_j(t)=(1-\varrho)^{-1}
\exp(-\varrho t/(1-\varrho))$ \cite[Eq.~18.12.13, $\alpha=0$]{DLMF} gives
\[
 (1-\varrho)\sum_{j\geq0}\varrho^j\langle h_j,D(\eta)h_j\rangle
 =\exp\!\left[-\frac{1+\varrho}{2(1-\varrho)}|\eta|^2\right]
 =e^{-\frac12\coth\Delta\,|\eta|^2}.
\]
Since $|\eta|^2=(x^2+4\pi^2\xi^2)/2$, taking the product over the $p$ modes
proves the identity in \eqref{eq:fe_oscillator_factor}.  The bound follows
from $\tanh\Delta\leq\Delta$, i.e.\ $\coth\Delta\geq\Delta^{-1}$, together
with $4\pi^2\geq1$.

Rotor factor.  $\mathbf T_n\delta_m=\delta_{m-n}$, while
$e^{-\pi\Delta|\mathbf Q|^2}$ and $\mathbf M_\theta$ are diagonal; so every
diagonal entry of the product vanishes unless $n=0$, and for $n=0$ the
diagonal entries give the first ratio in \eqref{eq:fe_rotor_factor}.  Poisson
summation applied to the modulated Gaussian gives the second ratio.  For the
bound, let $t\in[-1/2,1/2]^q$ represent $\theta$ and
$\delta:=|t|=\operatorname{dist}(\theta,0)$.  For every $\ell\in\mathbb Z^q$
one has $|\ell-t|^2\geq\delta^2$ and, coordinatewise from
$|\ell_j-t_j|\geq|\ell_j|/2$, also $|\ell-t|^2\geq|\ell|^2/4$; averaging the
two gives $|\ell-t|^2\geq\tfrac12\delta^2+\tfrac18|\ell|^2$.  The denominator
in \eqref{eq:fe_rotor_factor} is at least $1$, whence
\[
 |\Xi^{\mathbb Z}_\Delta(0,\theta)|
 \leq e^{-\pi\delta^2/(2\Delta)}\sum_{\ell}e^{-\pi|\ell|^2/(8\Delta)}
 \leq e^{-\pi\delta^2/(2\Delta)}
   \underbrace{\sum_{\ell}e^{-\pi|\ell|^2/8}}_{=:C_q<\infty}
\]
for $0<\Delta\leq1$.

Finite factor.  Again $\mathbf T_w\delta_g=\delta_{g-w}$ has no diagonal
entry unless $w=0$, and then the trace is $\sum_{g\in\mathbb Z_{\dia c}}\chi(g)$,
which is $C_{\dia c}$ for $\chi=1$ and $0$ otherwise, by character
orthogonality.  Divide by
$\operatorname{Tr}\mathbf 1_{\ell^2(\mathbb Z_{\dia c})}=C_{\dia c}$.
\end{proof}

\subsection{Proof of Lemma~\ref{lem:fe_lattice_sums}}
\label{appx:proof_fe_lattice_sums}

\begin{proof}
By \eqref{eq:fe_pointwise_decay} only the sets
$X:=(\Gamma\setminus\{0\})\cap H$ and $X:=(u+\Gamma)\cap H$ contribute; an
empty $X$ gives the zero sum.  On $H$ the function $\operatorname{rad}$ is
the distance to the identity for the product metric of
$\mathbb R^{2p}\times\mathbb T^q$.  Since $\Gamma$ is discrete and $u+\Gamma$
is closed and avoids $0$, in either case there is a gap $\delta_X>0$ with
$\operatorname{rad}(z)\geq\delta_X$ on $X$.

Next, $X$ is contained in a translate of the closed discrete group
$\Gamma\cap H$, so a fixed Gaussian sum over $X$ is finite.  Indeed, let
$Q:=[0,1]^{2p}\times\mathbb T^q\subseteq H$.  Any translate of $Q$ meets $X$
in at most $B:=|(\Gamma\cap H)\cap(Q-Q)|<\infty$ points, because subtracting
one of them from the others injects into this finite set; and the region
$\{\operatorname{rad}<L\}$ is covered by at most $C'(1+L)^{2p}$ translates of
$Q$.  Summing over the shells $j\leq\operatorname{rad}(z)<j+1$ gives
\begin{equation}
 S_X:=\sum_{z\in X}e^{-c\operatorname{rad}(z)^2/2}
 \leq B\,C'\sum_{j\geq0}(j+2)^{2p}e^{-cj^2/2}<\infty.
 \label{eq:fe_shell_sum}
\end{equation}
Finally, for $0<\Delta\leq1$ split the exponent,
\[
 e^{-c\operatorname{rad}(z)^2/\Delta}
 =e^{-c\operatorname{rad}(z)^2/(2\Delta)}
  e^{-c\operatorname{rad}(z)^2/(2\Delta)}
 \leq e^{-c\delta_X^2/(2\Delta)}\,
  e^{-c\operatorname{rad}(z)^2/2},
\]
so that $\sum_{z\in X}|\Xi_\Delta(z)|\leq CS_Xe^{-c\delta_X^2/(2\Delta)}$.
This proves both estimates with $a=c\delta_X^2/2$.  For an arbitrary fixed
$\Delta>0$, the identity in~\eqref{eq:fe_oscillator_factor} and the inequalities
$\coth\Delta\geq1$, $4\pi^2\geq1$, and
$|\Xi^{\mathbb Z}_\Delta|\leq1$ bound the summand on $H$ by
$e^{q/16}e^{-\operatorname{rad}(z)^2/4}$, since
$\operatorname{dist}(\theta,0)^2\leq q/4$; the same shell argument gives
absolute convergence.
\end{proof}

\subsection{Proof of Lemma~\ref{lem:fe_matrix_coefficients}}
\label{appx:proof_fe_matrix_coefficients}

\begin{proof}
The coefficient is a finite sum over $g\in\mathbb Z_{\dia c}$ and a sum over
$m\in\mathbb Z^q$
of integrals
\[
 \int_{\mathbb R^p}\overline{\varphi(t,m,g)}\,
 e^{2\pi i(t\cdot\xi+m\cdot\theta)}\chi(g)\,
 \psi(t+x,m+n,g+w)\,dt .
\]
Decay in $(x,n)$ follows from
$1+|x|+|n|\leq(1+|t|+|m|)(1+|t+x|+|m+n|)$ and the rapid decay of
$\varphi,\psi$.  Decay in $\xi$ follows by integrating by parts with
$(1-\nabla_t^2)^Je^{2\pi it\cdot\xi}=(1+4\pi^2|\xi|^2)^Je^{2\pi it\cdot\xi}$;
the product rule produces only finitely many derivatives of $\varphi$ and
$\psi$.  Choosing $J$ and the decay orders large enough leaves a summable and
integrable weight $(1+|t|+|m|)^{-p-q-1}$, which gives
\eqref{eq:fe_matrix_coefficient_decay}.  All compact and finite phases have
modulus $1$.  Finally, in
$G\cong\mathbb R^{2p}\times\mathbb Z^q\times K_{\mathrm{comp}}$, where
$K_{\mathrm{comp}}:=\mathbb T^q\times\mathbb Z_{\dia c}\times\widehat{\mathbb Z_{\dia c}}$
is compact,
the counting argument of
Lemma~\ref{lem:fe_lattice_sums} bounds the number of points of $\Gamma$ of
size at most $j+1$ by $C(j+2)^{2p+q}$; summing
\eqref{eq:fe_matrix_coefficient_decay} over unit shells with
$L>2p+q+1$ proves \eqref{eq:fe_matrix_coefficient_summability}.
\end{proof}

\subsection{Proof of Lemma~\ref{lem:fe_trace_expansion}}
\label{appx:proof_fe_trace_expansion}

\begin{proof}
Equation~\eqref{eq:fe_encoding_normalization} gives
\[
 \mathsf T_{\Delta,\tilde u}
 =\bigl(\mathbf E_\Delta\mathbf D(\tilde u)^\dagger
   \operatorname{Enc}\bigr)\operatorname{Enc}_\Delta^\dagger,
\]
whose first factor has Schwartz columns by
Proposition~\ref{prop:fe_smoothing}; so the product is bounded of rank at
most $K$ and its trace is defined.  For $\tilde u=0$ it is
$\operatorname{Enc}_\Delta\operatorname{Enc}_\Delta^\dagger$, of rank $K$.

To compute the trace, abbreviate $i=(k,n,w)$, $e_i=e_{k,n,w}$,
$a_i=a_\Delta(k,n)$, $\omega_i=\omega_{k,n}$.  By
\eqref{eq:fe_averaging_form},
\[
 \langle e_i,\mathsf T_{\Delta,\tilde u}e_i\rangle
 =\frac{a_i^2}{\sigma_D}
  \sum_{s\in D}\langle\mathbf D(\tilde u)e_i,
    \widetilde{\mathbf D}(s)e_i\rangle .
\]
Lemma~\ref{lem:fe_matrix_coefficients}, continuity of the fixed operator
$\mathbf D(\tilde u)$ on $\mathscr S(M)$, and
Proposition~\ref{prop:fe_coefficients} give constants $J,C$ depending on
$D$ and $\tilde u$ with
$\sum_{s\in D}|\langle\mathbf D(\tilde u)e_i,
 \widetilde{\mathbf D}(s)e_i\rangle|\leq C\omega_i^{\,J}$, whence
\begin{equation}
 \sum_i\sum_{s\in D}a_i^2\,
  |\langle\mathbf D(\tilde u)e_i,\widetilde{\mathbf D}(s)e_i\rangle|
 \leq C\sum_{k,n,w}
   e^{-2\Delta|k|-\pi\Delta|n|^2}\omega_{k,n}^{\,J}<\infty .
 \label{eq:fe_fubini_bound}
\end{equation}
Absolute convergence of the double series justifies summing over $s$ before
summing over $i$; absolute convergence of the scalar traces alone would not.
By \eqref{eq:ccr_semi_displacement_general} and
\eqref{eq:tilde_displacement},
$\mathbf D(\tilde u)^\dagger\widetilde{\mathbf D}(s)
 =\zeta_{\tilde u}(s)\mathbf D(s-\tilde u)$ for a unimodular
$\zeta_{\tilde u}(s)$, and since $\mathbf E_\Delta$ is Hilbert--Schmidt and
$\mathbf D(z)$ bounded, cyclicity is legitimate in
\[
 \operatorname{Tr}\bigl(\mathbf E_\Delta\mathbf D(z)\mathbf E_\Delta\bigr)
 =\operatorname{Tr}\bigl(\mathbf A_\Delta\mathbf D(z)\bigr)
 =\mathcal Z_\Delta\,\Xi_\Delta(z).
\]
Applying this term by term gives \eqref{eq:fe_trace_expansion}; absolute
convergence of that series is Lemma~\ref{lem:fe_lattice_sums}, with the
$s=\tilde u$ term separated off when $\tilde u\in D$.
\end{proof}

\subsection{Proof of Theorem~\ref{thm:fe_asymptotic_isometry}}
\label{appx:proof_fe_asymptotic_isometry}

\begin{proof}
For a normalized encoding, the logical operators of
Theorem~\ref{thm:logical_displacement_algebra} become $K\times K$ matrices:
fix representatives $\tilde u\in D^\perp$ of $u\in\mathsf L_D$, with
$\tilde 0=0$. Then
\eqref{eq:logical_displacement_definition} reads
\begin{equation}
 \widehat{\mathbf D(\tilde u)}\operatorname{Enc}
 =\operatorname{Enc}\,\mathbf L_u,
 \qquad
 \mathbf L_0=\mathbf 1,
 \qquad
 \operatorname{Tr}(\mathbf L_u^\dagger\mathbf L_v)=K\delta_{u,v}.
 \label{eq:fe_logical_intertwining}
\end{equation}
Indeed, the invariance~\eqref{eq:fe_form_dual_invariance} and the
factorization $\Pi_D=\operatorname{Enc}\operatorname{Enc}^\dagger$ imply
$\mathbf L_u\mathbf L_u^\dagger=\mathbf1$: injectivity of
$\operatorname{Enc}$ and surjectivity of $\operatorname{Enc}^\dagger$ allow
cancellation of the two encoding maps. Thus these matrices are unitary,
and the trace orthogonality of
\eqref{eq:logical_displacement_orthogonality} gives the last identity.
Hence every $B_0\in M_K(\mathbb C)$ expands as
\begin{equation}
 B_0=\frac1K\sum_{u\in\mathsf L_D}
   \operatorname{Tr}(\mathbf L_u^\dagger B_0)\,\mathbf L_u .
 \label{eq:fe_logical_expansion}
\end{equation}

\emph{Step 1: the logical Weyl coefficients.}  Using only bounded finite-rank
operators on $L^2(M)$ and $K\times K$ matrices,
\begin{align}
 t_u(\Delta):=\operatorname{Tr}
   \bigl(\mathbf L_u^\dagger\mathsf G_\Delta\bigr)
 &=\operatorname{Tr}\bigl(\operatorname{Enc}_\Delta
    \mathbf L_u^\dagger\operatorname{Enc}_\Delta^\dagger\bigr)
 =\operatorname{Tr}\bigl(\mathbf E_\Delta
    \mathbf D(\tilde u)^\dagger\Pi_D\mathbf E_\Delta\bigr) \notag\\
 &=\frac{\mathcal Z_\Delta}{\sigma_D}
   \sum_{s\in D}\zeta_{\tilde u}(s)\,\Xi_\Delta(s-\tilde u),
 \label{eq:fe_logical_coefficient}
\end{align}
where the first equality is
$\operatorname{Tr}_{\mathbb C^K}(BA)=\operatorname{Tr}_{L^2(M)}(AB)$
for finite-rank compositions, with both traces unnormalized; the second uses
\eqref{eq:fe_logical_intertwining} and
\eqref{eq:fe_encoding_normalization}, and the third is
Lemma~\ref{lem:fe_trace_expansion}.  Note that
$\operatorname{Tr}(\mathbf D(\tilde u)^\dagger\Pi_D\mathbf A_\Delta)$ is
never treated as an ordinary trace on $L^2(M)$ without smoothing on both
sides.

\emph{Step 2: separating the identity coefficient.}  For $u=0$ take
$\tilde u=0$.  The $s=0$ term of \eqref{eq:fe_logical_coefficient} equals
$1$, and \eqref{eq:fe_lattice_sum} bounds the rest, so
\[
 t_0(\Delta)=\operatorname{Tr}\mathsf G_\Delta
 =\frac{\mathcal Z_\Delta}{\sigma_D}
   \bigl(1+O(e^{-a_0/\Delta})\bigr).
\]
For $u\neq0$ the representative $\tilde u$ lies in
$D^\perp\setminus D$, so the coset $D-\tilde u$ avoids $0$ and
\eqref{eq:fe_coset_sum} gives
$|t_u(\Delta)|\leq\mathcal Z_\Delta\sigma_D^{-1}C_ue^{-a_u/\Delta}$.

\emph{Step 3: reassembling the matrix.}  There are only $K^2$ cosets; let
$a>0$ be no larger than any of the finitely many rates $a_u$, including
$a_0$, and write
$t_u(\Delta)=\mathcal Z_\Delta\sigma_D^{-1}(\delta_{u,0}+r_u(\Delta))$ with
$|r_u(\Delta)|\leq C_ue^{-a/\Delta}$.  Substituting into
\eqref{eq:fe_logical_expansion},
\[
 \mathsf G_\Delta
 =\frac{\mathcal Z_\Delta}{K\sigma_D}
  \Bigl(\mathbf 1+\sum_{u\in\mathsf L_D}r_u(\Delta)\mathbf L_u\Bigr),
 \qquad
 \Bigl\|\sum_{u}r_u(\Delta)\mathbf L_u\Bigr\|_{\mathrm{op}}
 \leq\sum_{u}C_u\,e^{-a/\Delta},
\]
because every $\mathbf L_u$ is unitary.  This is
\eqref{eq:fe_gram_estimate}; taking matrix elements gives
\eqref{eq:fe_isometry_estimate}, and
\eqref{eq:fe_partition_function} gives the explicit form of $\nu_\Delta$.
\end{proof}

\section{Graph lattices and LCA codes}
\label{app:graph_lattices}

This appendix recalls the graph-lattice construction of the LCA codes of
Ref.~\cite{yr64-ypxt} in the notation of Sec.~\ref{sec:GSC}.
A graph lattice projects isomorphically onto a full-rank lattice in the
oscillator phase space.
Graph lattices therefore require at least one oscillator.
Each graph lattice is isomorphic to $\Z^{2p}$ and hence torsion-free.
Its syndrome space $\widehat D\cong\T^{2p}$ is therefore a torus.
Theorem~\ref{thm:main_dual} constructs such a lattice together with its
symplectic dual.
Equation~\eqref{eq:logical-dimension-pfaffian} then expresses the code
dimension as the Pfaffian of the $2p\times2p$ matrix $\varTheta-Z$.
This formula needs $q=0$ and a torsion-free lattice, whereas
Lemma~\ref{lem:full_lift_pfaffian} applies to every lattice.
Here $T$ and $T_{\mathrm{dual}}$ denote the embedding maps of $D$ and
$D^\perp$ defined below, each with $2p$ columns.
Adjoining the $2k$ qudit periods $c_je_{a_j}$ and $c_je_{b_j}$ to the
columns of $T$ gives a square stabilizer generator matrix in the sense of
Lemma~\ref{lem:full_lift_pfaffian}.
Its Gram matrix has the graph-lattice Gram matrix $\varTheta$ as its
upper-left block, as for the single-mode code of Sec.~\ref{ex:lca}.

\subsection{Graph lattices and the LCA construction}\label{app:lca-intro}

For the graph-lattice construction below we specialize to $q=0$, so that
$M=\mathbb R^p\times\mathbb Z_{\dia c}$, $n=2p\geq2$, and $J_{\cv}=J_{\mathrm{osc}}$. Let $\mathscr{T}_n$ be the space of $n\times n$ real antisymmetric matrices, which we call \emph{symplectic Gram matrices}.

Following~\cite[Def.~2.1]{li2004strong}, \cite[Def.~4.1]{rieffel} and the related discussion, let $L=\R^n$ and regard $\Z^n$ as the standard lattice in $L^*$. An element $\varTheta\in\mathscr T_n$ determines an element of $\bigwedge^2L$. An \emph{embedding map} for $\varTheta$ is a linear map $T:L^*\to H^*$, with $H^*$ as in Eq.~\eqref{eq:h} of Sec.~\ref{sec:GSC} for $q=0$, satisfying the following conditions:

\begin{enumerate}
   \item  $T\left(\mathbb{Z}^n\right) \subseteq \mathbb{R}^p \times \mathbb{R}^{* p}  \times \mathbb{Z}^k \times \mathbb{Z}^k$. Then we can think of $T\left(\mathbb{Z}^n\right)$ as in $G$ via composing $\left.T\right|_{\mathbb{Z}^n}$ with the natural covering map $\pi\colon\mathbb{R}^p \times \mathbb{R}^{* p}  \times \mathbb{Z}^k \times \mathbb{Z}^k \rightarrow G$ of Eq.~\eqref{eq:full_lift_covering} for $q=0$.
   \item $T\left(\mathbb{Z}^n\right)$ is a lattice in $G$.
   \item[(2$'$)] Equivalently, the map
   \[
      \widetilde{T}:=\gamma\circ T:L^*\longrightarrow
      \mathbb{R}^p\times\mathbb{R}^{*p}
   \]
   is invertible, where $\gamma$ is the projection of $H^*$ onto
   $\mathbb{R}^p\times\mathbb{R}^{*p}$.  In the present setting $q=0$,
   so there is no $\mathbb{R}^q$ factor.
   \item The form $J$ on $H^*$ is pulled back by $T$ to the form $\varTheta$ on $L^*$, i.e. 
   \begin{equation}
   \label{eq:embedding-condition}
       T^{\intercal} J T= \varTheta~.
   \end{equation}
\end{enumerate}

Set $D=T(\mathbb Z^n)$. We call $D$ the \emph{graph lattice} associated with the embedding map $T$. By definition of $T$, it is a lattice in
$$G=\mathbb{R}^p \times \widehat{\mathbb{R}^{p}}  \times\left(\mathbb{Z}_{c_1} \times \cdots \times \mathbb{Z}_{c_k}\right) \times\left(\mathbb{Z}_{c_1} \times \cdots \times \mathbb{Z}_{c_k}\right).
$$ 

The covolume of $D$ can be read off directly from the real
projection $\widetilde T$.  Indeed, set
\[
 G_{\mathrm{osc}}:=\mathbb R^p\times\mathbb R^{*p},\qquad
 G_{\mathrm{fin}}:=\Z_{\dia c}\times\widehat{\Z_{\dia c}},\qquad
 \Lambda_{\mathrm{osc}}:=\widetilde T(\mathbb Z^{2p})\subseteq G_{\mathrm{osc}}.
\]
Since $\widetilde T$ is invertible, the projection
$D\to\Lambda_{\mathrm{osc}}$ is an isomorphism. Thus $D$ is the graph
of a homomorphism $\Lambda_{\mathrm{osc}}\to G_{\mathrm{fin}}$; that is, there is a unique homomorphism
$\varphi:\Lambda_{\mathrm{osc}}\to G_{\mathrm{fin}}$ such that
$D=\{(\lambda,\varphi(\lambda)):\lambda\in\Lambda_{\mathrm{osc}}\}$. Moreover,
\[
 \widetilde T\bigl([0,1)^{2p}\bigr)\times G_{\mathrm{fin}}
\]
is a measurable fundamental domain for $D$ in $G$. Thus, by the
definition of covolume in~\eqref{eq:covolume_definition}, using the
self-Plancherel Haar measure $\mu$ on $G$,
\begin{equation}
\label{eq:embedding-lattice-covolume}
 \begin{aligned}
 |G/D|_\mu
 &=\mu\!\left(\widetilde T\bigl([0,1)^{2p}\bigr)\times G_{\mathrm{fin}}\right)\\
 &=\lvert\Z_{\dia c}\rvert\,\lvert\det\widetilde T\rvert
 =\left(\prod_{j=1}^k c_j\right)\lvert\det\widetilde T\rvert.
 \end{aligned}
\end{equation}

When $\varTheta$ is integral, Eq.~\eqref{eq:embedding-condition} shows that $D$ is isotropic. The resulting GLCA code is an \emph{LCA code}. These codes were introduced in~\cite{yr64-ypxt}. When the finite factor is absent, so that $M=\mathbb R^p$, these reduce to GKP codes~\cite{10.1103/physreva.64.012310}.

\subsection{Dual graph lattices and the $\mathrm{SO}(n, n \mid \mathbb{Z})$ action}
\label{app:general-lca}

Let $\mathrm{O}(n, n \mid \mathbb{R})$ be the group of linear transformations of the space $\mathbb{R}^{2 n}$ preserving the quadratic form $x_1 x_{n+1}+x_2 x_{n+2}+\cdots+x_n x_{2 n}$, and $\mathrm{SO}(n, n \mid \mathbb{Z})$ be the subgroup of $\mathrm{O}(n, n \mid \mathbb{R})$ consisting of matrices with integer entries and determinant 1. We can write an element $g$ of $\mathrm{O}(n, n \mid \mathbb{R})$ as
\begin{equation}
g=\left(\begin{array}{ll}
\mathsf{A} & \mathsf{B} \\
\mathsf{C} & \mathsf{D}
\end{array}\right),
\end{equation}
where $\mathsf{A}, \mathsf{B}, \mathsf{C}, \mathsf{D}$ are $n \times n$ matrices satisfying
\begin{equation}
\mathsf{A}^{\intercal} \mathsf{C}+\mathsf{C}^{\intercal} \mathsf{A}=0=\mathsf{B}^{\intercal} \mathsf{D}+\mathsf{D}^{\intercal} \mathsf{B}, \quad \mathsf{A}^{\intercal} \mathsf{D}+\mathsf{C}^{\intercal} \mathsf{B}=\dia{1}_n .
\end{equation}

We now define a (partial) action of $\mathrm{SO}(n, n \mid \mathbb{Z})$ on the space of symplectic Gram matrices $\mathscr{T}_n$ as
\begin{equation}
g \cdot \varTheta:=(\mathsf{A} \varTheta+\mathsf{B})(\mathsf{C} \varTheta+\mathsf{D})^{-1}
\end{equation}
whenever $\mathsf{C} \varTheta+\mathsf{D}$ is invertible. It is known that for each $g \in \mathrm{SO}(n, n \mid \mathbb{Z}),$ this action is defined on a dense open subset of $\mathscr{T}_n$~\cite{rieffel1999morita,li2004strong}.

For the embeddings constructed by the following theorem, a suitable
$g$ relates the symplectic Gram matrices $\varTheta$ and
$-g\cdot\varTheta$ of a graph lattice and its symplectic dual.
When $\varTheta$ is integral, the lattice defined by $T$ is a stabilizer
lattice, and the columns of the dual embedding matrix $T_{\mathrm{dual}}$ specify
generators of the logical displacements modulo stabilizers and phases.
The following theorem restates Theorem~2 of Ref.~\cite{yr64-ypxt}.

\begin{theorem} \label{thm:main_dual}
Let $Z\in\mathscr T_n$ have rational entries. Then there exist $g\in\mathrm{SO}(n,n\mid\mathbb Z)$ and a finite abelian group $\Z_{\dia c}$ such that, for every $\varTheta\in\mathscr T_n$ with $\varTheta-Z$ invertible, there are embedding maps $T,T_{\mathrm{dual}}:L^*\to H^*$ for the symplectic Gram matrices $\varTheta$ and $-g\cdot\varTheta$, respectively. Here $M=\R^p\times\Z_{\dia c}$, $H^*=\R^p\times\R^{*p}\times\R^k\times\R^{*k}$, and
\[
T_{\mathrm{dual}}(\Z^n)=T(\Z^n)^\perp.
\]
Moreover, writing $\widetilde T=\gamma\circ T$ for the continuous projection of $T$, the construction determines $\widetilde T$, up to left multiplication by a $J_{\cv}$-symplectic matrix, through
\begin{equation}
\widetilde T^{\intercal}J_{\cv}\widetilde T=\varTheta-Z.
\label{eq:projected-embedding-identity}
\end{equation}
\end{theorem}

The theorem therefore supplies both the stabilizer lattice and its symplectic dual for a broad class of LCA codes.

The same construction also determines the logical dimension. Let $\varTheta$ be an integral matrix as in Theorem~\ref{thm:main_dual}, write $D=T(\Z^n)$, and let $K$ denote the dimension of the logical Hilbert space. Combining the rank formula~\eqref{eq:rank_covolume_formula} with the graph-lattice covolume formula~\eqref{eq:embedding-lattice-covolume} gives
\begin{equation}
K
=\left(\prod_{j=1}^k c_j\right)|\det\widetilde T|.
\label{eq:logical-dimension-from-graph-covolume}
\end{equation}
On the other hand, taking determinants in~\eqref{eq:projected-embedding-identity} and using $\det J_{\cv}=1$ yields
\[
\det(\varTheta-Z)=(\det\widetilde T)^2.
\]
Since $\det A=\operatorname{Pf}(A)^2$ for every even-dimensional antisymmetric matrix $A$, we obtain
\begin{equation}
K=\left(\prod_{j=1}^k c_j\right)
\bigl|\operatorname{Pf}(\varTheta-Z)\bigr|.
\label{eq:logical-dimension-pfaffian}
\end{equation}
With the compatible orientation convention for which $\operatorname{Pf}(\varTheta-Z)>0$, the absolute-value signs may be omitted. Thus~\eqref{eq:logical-dimension-pfaffian} agrees with the formula in~\cite[Appendix~E]{yr64-ypxt}.
Here $\mathrm{Pf}(A)$ denotes the Pfaffian of the antisymmetric matrix $A$.
Recall that the Pfaffian of a $2p \times 2p$ antisymmetric matrix $A = (a_{ij})$ is given by
\begin{equation}
\operatorname{Pf}(A)
=
\sum_{\pi \in S_{2p}} \operatorname{sgn}(\pi)
\prod_{s=1}^{p} a_{\pi(2s-1)\, \pi(2s)},
\end{equation}
where the sum runs over all permutations $\pi \in S_{2p}$ satisfying $\pi(2s-1) < \pi(2s)$ for all $1 \le s \le p$ and $\pi(1) < \pi(3) < \cdots < \pi(2p - 1)$. Here $\operatorname{sgn}(\pi)\in\{+1,-1\}$ is the signature of $\pi$.

For example, the usual GKP codes arise by taking $M=\R^p$, so that $n=2p$, the finite factor is absent, and $Z=0$; see~\cite{10.1103/physreva.64.012310,Conrad2022gottesmankitaev,10.1103/prxquantum.3.010335}. Let $\varTheta\in\mathscr T_n$ be integral and nondegenerate. Theorem~\ref{thm:main_dual} then applies with
\[
g= \begin{pmatrix}
0 & \dia{1}_{2p} \\
\dia{1}_{2p} & 0
\end{pmatrix} \in \mathrm{SO}(n, n \mid \mathbb{Z}).
\]
Thus $D=T(\Z^n)$ is a GKP lattice in $\R^p\times\widehat{\R^p}$, its logical dimension is $|\operatorname{Pf}(\varTheta)|$, and $D^\perp=T_{\mathrm{dual}}(\Z^n)$ supplies the logical displacement operators.

The single-mode code of Sec.~\ref{ex:lca} has $n=2$,
$\varTheta=\theta J_{\mathrm{osc}}$, and $Z=-(d/c)J_{\mathrm{osc}}$.
Equation~\eqref{eq:logical-dimension-pfaffian} then gives
$K=c\,|\theta+d/c|=|c\theta+d|$, in agreement with the full-lift
computation there.

The integer lifts of Sec.~\ref{ex:mixed-qudit} have a finite
counterpart of this split-orthogonal duality.  Under the fractional-linear
$\mathrm{SO}(2k,2k\mid\Z)$ action, the block-swap element
$g_0=\left(\begin{smallmatrix}0&I_{2k}\\I_{2k}&0\end{smallmatrix}\right)$
sends $\varTheta$ to $\varTheta^{-1}$, so the dual Gram matrix
$-\varTheta^{-1}$ of the integer lift in
Eq.~\eqref{eq:mixed-qudit-dual-basis} is
$-g_0\cdot\varTheta$.  The minus sign is part of our symplectic-duality
convention.

\section{$\mathrm C^*$-algebras, Gelfand--Naimark theorem and direct integrals}\label{appx:c_star}

\subsection*{Group $\mathrm C^*$-algebras}

A (concrete) $\mathrm C^*$-algebra is a complex subalgebra
$\mathfrak{A}\subseteq\mathcal B(\mathcal H)$ of the bounded operators on a Hilbert
space $\mathcal H$, closed in the operator norm and under the adjoint
$a\mapsto a^\dagger$. It is \emph{unital} if it has a multiplicative
identity and \emph{commutative} if $ab=ba$ for all
$a,b\in\mathfrak{A}$.

Let $H$ be a discrete abelian group. A normalized $2$-cocycle is a map
$\omega:H\times H\to\mathbb T$ satisfying
\[
 \omega(h,k)\omega(h+k,l)=\omega(h,k+l)\omega(k,l),
 \qquad \omega(0,h)=\omega(h,0)=1.
\]
The twisted left regular unitaries on $\ell^2(H)$ are
\[
 (L_\omega(h)f)(k)=\omega(h,k-h)f(k-h),
 \qquad L_\omega(h)L_\omega(k)=\omega(h,k)L_\omega(h+k).
\]
The \emph{twisted group $\mathrm C^*$-algebra} $\mathrm C^*(H,\omega)$ is the
operator-norm closure of the linear span of these unitaries in
$\mathcal B(\ell^2(H))$. Since $H$ is abelian, this also agrees with the
universal $\mathrm C^*$-algebra generated by unitaries satisfying the displayed
relations. The choice $\omega=1$ gives the ordinary group $\mathrm C^*$-algebra
$\mathrm C^*(H)$. In the main text, the choices $(H,\omega)=(D,\beta)$ and
$(D^\perp,\overline\beta)$ give $\mathfrak{A}_D$ and $\mathfrak{B}_D$, respectively.

For each $h\in H$, the \emph{canonical unitary} $u_h\in\mathrm C^*(H,\omega)$
is the distinguished generator associated with the group element $h$;
in the regular realization above, it is $L_\omega(h)$, or equivalently
the image of the point mass $\delta_h$ in the twisted group algebra.
These generators satisfy
\[
 u_0=1,\qquad u_hu_k=\omega(h,k)u_{h+k},\qquad
 u_h^\dagger u_h=u_hu_h^\dagger=1,
\]
and their finite linear combinations are norm dense in
$\mathrm C^*(H,\omega)$. Thus ``canonical'' refers to the specified
generators indexed by group elements. A unital representation of the algebra
sends them to concrete unitaries obeying the same relations.
In the untwisted case $\omega=1$, their multiplication reduces to
$u_hu_k=u_{h+k}$.

The commutative Gelfand--Naimark theorem states that every unital
commutative $\mathrm C^*$-algebra $\mathfrak{A}$ is isometrically
$*$-isomorphic to $C(X)$,
where $X$ is its compact Hausdorff space of characters and the isomorphism
is $a\mapsto(x\mapsto x(a))$~\cite[Section~2.3.5]{BratteliRobinson1987}.
For a discrete abelian group $H$, this gives
\[
 \mathrm C^*(H)\cong C(\widehat H).
\]
With the Fourier convention of~\eqref{eq:fourier_transform_convention},
the identification is
\[
 a\longmapsto\widehat a,\qquad
 \widehat a(\chi)=\sum_{h\in H}a(h)\overline{\chi(h)},
 \qquad \widehat u_h(\chi)=\overline{\chi(h)}.
\]
The sum is initially defined for $a\in\ell^1(H)$ and extends to the
$\mathrm C^*$-completion. In particular, for an isotropic lattice $D$,
Lemma~\ref{lem:phase_triv} removes the multiplier and identifies
$\mathfrak{A}_D\cong \mathrm C^*(D)\cong C(\widehat D)$. Thus the stabilizer algebra acts
by continuous functions on the syndrome space $\widehat D$.

\subsection*{Direct integrals of syndrome fibers}\label{appx:direct_integral}

Only the direct integral associated with a finite-rank Hermitian vector
bundle, as in Appendix~\ref{sec:app_prelim}, is needed here. Let $\mathscr E\to X$ be such a bundle over a
compact metrizable space, and let $\mu$ be a Borel probability measure.
 The direct integral
\[
 L^2(X,\mathscr E;\mu)
 :=\int_X^\oplus E_x\,d\mu(x)
\]
consists of measurable sections $s$ with
$\int_X\|s(x)\|_x^2\,d\mu(x)<\infty$, identified when they agree
$\mu$-almost everywhere. Its scalar inner product is
\[
 \langle s,t\rangle
 =\int_X\langle s(x),t(x)\rangle_x\,d\mu(x),
\]
conjugate-linear in the first argument and linear in the second.
Continuous sections are dense in this Hilbert space, and
$\|s\|_{L^2}\leq\sup_{x\in X}\|s(x)\|_x$ for continuous $s$.

For $f\in L^\infty(X,\mu)$, multiplication
$(M_fs)(x)=f(x)s(x)$ is a bounded diagonal operator.
In Section~\ref{a_sec:glca_ncg}, $X=\widehat D$, $E_x=E_{D,\chi}$,
and $\mu$ is normalized Haar measure. The Bloch--Floquet transform
identifies $L^2(M)$ with this direct integral and sends the stabilizer
$\widetilde{\mathbf D}(w)$ to multiplication by $\overline{\chi(w)}$.
Pointwise evaluation of an arbitrary $L^2$ section is not defined at a
prescribed character when that singleton has measure zero. The canonical
continuous sections supplied by $\Omega_D\cong\Gamma(\mathscr E_D)$
instead define every fiber, including the codespace at $\chi=1$.

\section{Projective modules, vector bundles, and trace pairing}
\label{sec:app_prelim}

A module is finitely generated projective if it is a direct summand of
a finite free module. For a compact Hausdorff space $X$, the Serre--Swan
theorem identifies these modules over $C(X)$ with modules
$\Gamma(\mathscr E)$ of continuous sections of finite-rank complex vector
bundles~\cite[Theorem~1.6.3]{Rosenberg1994}. Here a vector bundle is a
locally trivial family of finite-dimensional vector spaces, and the
module action is $(f\cdot s)(x)=f(x)s(x)$. A Hermitian metric equips
its fibers with scalar inner products. Concretely, for some projection
$P=P^2=P^\dagger\in M_n(C(X))$,
\[
 \Gamma(\mathscr E)\cong P\,C(X)^n,
 \qquad E_x\cong\operatorname{Ran}P(x).
\]
The rank is locally constant, but can differ between connected components.
For the syndrome bundle, Corollary~\ref{cor:vector_bundle} establishes
constant rank on all of $\widehat D$ by displacement covariance.

Let $\mathfrak{A}$ be a unital $\mathrm C^*$-algebra with a tracial state $\tau$.
Extend it using the ordinary, unnormalized matrix trace:
\[
 \tau_n((a_{ij})):=\sum_{i=1}^n\tau(a_{ii}).
\]
A finitely generated projective right module has the form
$P\mathfrak{A}^n$; for a left module the corresponding model is
$\mathfrak{A}^nP$, with row vectors.
Its \emph{trace pairing}, or $\tau$-dimension, is
\begin{equation}\label{eq:appendix_module_trace_pairing}
 \tau([\mathcal E]):=\tau_n(P).
\end{equation}
This depends only on the module isomorphism class: equivalent projection
representatives satisfy $P=V^\dagger V$ and $Q=VV^\dagger$ for a
rectangular matrix $V$, so traciality gives $\tau_n(P)=\tau_m(Q)$.
The brackets $[\mathcal E]$ denote that isomorphism class; the notation
$\tau_*([P])$ in Section~\ref{sec:PSC_basic} denotes the same pairing.
For commutative $\mathfrak{A}$, left and right scalar actions agree, so the column
model $P C(X)^n$ also describes the left section modules used in the paper.

If $\tau(f)=\int_X f\,d\mu$ for a probability measure $\mu$, then
\[
 \tau([\Gamma(\mathscr E)])
 =\int_X\operatorname{Tr}P(x)\,d\mu(x)
 =\int_X\dim E_x\,d\mu(x).
\]
Thus the pairing equals the fiber dimension whenever the rank is constant,
without a connectedness assumption on $X$. For $\mathfrak{A}_D\cong C(\widehat D)$,
character orthogonality identifies the canonical trace
$\mathrm{Tr}_D(a)=a(0)$ with
$\tau_D(\widehat a)=\int_{\widehat D}\widehat a(\chi)\,d\chi$.
Consequently,
\[
 \mathrm{Tr}_D([\Omega_D])
 =\tau_D([\Omega_D])
 =\operatorname{rank}\mathscr E_D
 =\dim\mathcal C,
\]
which is the trace pairing used in Theorem~\ref{thm:code_dimension}.

\section{Morita equivalence and inner-product conventions}\label{appx:morita}

We use scalar Hilbert-space inner products linear in the second argument,
as in $\langle\xi,\eta\rangle_{L^2(M)}
=\int_M\overline{\xi(t)}\eta(t)\,d\mu_M(t)$.
A left Hilbert $\mathfrak{A}$-module has an $\mathfrak{A}$-valued
sesquilinear inner product
${}_{\mathfrak{A}}\langle\xi,\eta\rangle$ linear in the first argument;
a right Hilbert $\mathfrak{B}$-module has a $\mathfrak{B}$-valued inner
product $\langle\xi,\eta\rangle_{\mathfrak{B}}$ linear in the second.
Their module rules are
\begin{align*}
 {}_{\mathfrak{A}}\langle a\cdot\xi,\eta\rangle
 &=a\,{}_{\mathfrak{A}}\langle\xi,\eta\rangle,&
 {}_{\mathfrak{A}}\langle\xi,a\cdot\eta\rangle
 &={}_{\mathfrak{A}}\langle\xi,\eta\rangle a^\dagger,\\
 \langle\xi,\eta\cdot b\rangle_{\mathfrak{B}}
 &=\langle\xi,\eta\rangle_{\mathfrak{B}} b,&
 \langle\xi\cdot b,\eta\rangle_{\mathfrak{B}}
 &=b^\dagger\langle\xi,\eta\rangle_{\mathfrak{B}}.
\end{align*}
Both inner products are conjugate symmetric and positive definite, and
the module is complete in the norm given by the square root of the norm
of its inner product with itself. It is \emph{full} when the closed
linear span of its inner-product values is the whole coefficient algebra.

An $\mathfrak{A}$--$\mathfrak{B}$ equivalence bimodule is a full left
Hilbert $\mathfrak{A}$-module and a full right Hilbert
$\mathfrak{B}$-module with commuting actions, satisfying
\begin{align*}
 \langle a\cdot\xi,\eta\rangle_{\mathfrak{B}}
 &=\langle\xi,a^\dagger\cdot\eta\rangle_{\mathfrak{B}},\\
 {}_{\mathfrak{A}}\langle\xi\cdot b,\eta\rangle
 &={}_{\mathfrak{A}}\langle\xi,\eta\cdot b^\dagger\rangle,\\
 {}_{\mathfrak{A}}\langle\xi,\eta\rangle\cdot\zeta
 &=\xi\cdot\langle\eta,\zeta\rangle_{\mathfrak{B}}.
\end{align*}
The two module norms agree. The algebras are \emph{strongly Morita
equivalent} when such a bimodule exists. If both algebras are unital,
the equivalence bimodule is finitely generated projective on both
sides~\cite[Section~3]{rieffel}.

In Section~\ref{a_sec:glca_ncg}, the equivalence bimodule is Rieffel's module $\Omega_D$,
with $\mathfrak{A}=\mathfrak{A}_D=\mathrm C^*(D,\beta)$ and
$\mathfrak{B}=\mathfrak{B}_D=\mathrm C^*(D^\perp,\overline\beta)$. The actions
$u_w\cdot\xi=\mathbf D(w)\xi$ and
$\xi\cdot u_z=\mathbf D(z)^\dagger\xi$ have the corresponding coefficient
inner products
\[
 {}_{\mathfrak{A}_D}\langle\xi,\eta\rangle(w)
 =\langle\mathbf D(w)\eta,\xi\rangle_{L^2(M)},
 \qquad
 \langle\xi,\eta\rangle_{\mathfrak{B}_D}(z)
 =\langle\xi,\mathbf D(z)\eta\rangle_{L^2(M)},
\]
with the compatible Haar measures of Rieffel's construction
\cite[Theorem~2.15]{rieffel}. These are linear in $\xi$ on the left and
in $\eta$ on the right, exactly as required. After phase trivialization
and Fourier transformation, the scalar fiber inner product is therefore
\[
 \langle[\xi]_\chi,[\eta]_\chi\rangle_\chi
 ={}_{\mathrm R}\langle\eta,\xi\rangle(\chi),
\]
with the arguments reversed as in
\eqref{eq:rieffel_fiber_scalar_inner_product}. This makes the fiber
inner product linear in its second argument and agrees with the
direct-integral and distributional conventions of the main text.

\end{document}